\documentclass[12pt]{amsart}
\usepackage{amssymb,amscd}
\usepackage{mathrsfs}
\usepackage[all,cmtip]{xy}

\usepackage[colorlinks=true,allcolors = blue]{hyperref} % my new

\let\mc\mathcal
\usepackage{eucal}

\let\mathcal\mc

\let\frak\mathfrak

\def\>{\relax\ifmmode\mskip.666667\thinmuskip\relax\else\kern.111111em\fi}
\def\:{\relax\ifmmode\mskip.333333\thinmuskip\relax\else\kern.0555556em\fi}
\def\?{\relax\ifmmode\mskip-.666667\thinmuskip\relax\else\kern-.111111em\fi}
\def\<{\relax\ifmmode\mskip-.333333\thinmuskip\relax\else\kern-.0555556em\fi}
\def\vsk#1>{\vskip#1\baselineskip}
\def\vv#1>{\vadjust{\vsk#1>}\ignorespaces}
\def\vvn#1>{\vadjust{\nobreak\vsk#1>\nobreak}\ignorespaces}
\def\vvgood{\vadjust{\penalty-500}} \let\alb\allowbreak
\def\fratop{\genfrac{}{}{0pt}1}
\def\satop#1#2{\fratop{\scriptstyle#1}{\scriptstyle#2}}
\let\dsize\displaystyle  \let\ssize\scriptstyle
\let\sssize\scriptscriptstyle 
\let\phan\phantom \let\vp\vphantom \let\hp\hphantom
\def\stackrel#1#2{\mathrel{\mathop{\kern 0pt#2}\limits^{#1}}}

\newbox\tsbox
\def\itemflat#1{\par\hangindent2\parindent\indent\kern\parindent
\llap{#1\enspace}\ignorespaces}

\font\cyr=wncyr10 scaled 1200
\font\scyr=wncyr7 scaled 1200
\font\sscyr=wncyr5 scaled 1200

\def\Bcyr{{\mathchoice{\hbox{\cyr B}}{\hbox{\cyr B}}{\hbox{\scyr B}}
{\hbox{\sscyr B}}}}

\let\Medskip\medskip
\def\medskip{\par\Medskip}
\let\Bigskip\bigskip
\def\bigskip{\par\Bigskip}

\let\Maketitle\maketitle
\def\maketitle{\Maketitle\thispagestyle{empty}\let\maketitle\empty}

\let\mc\mathcal
\let\nc\newcommand

\newtheorem{thm}{Theorem}[section]
\newtheorem{cor}[thm]{Corollary}
\newtheorem{lem}[thm]{Lemma}
\newtheorem{prop}[thm]{Proposition}
\newtheorem{conj}[thm]{Conjecture}
\newtheorem{defn}[thm]{Definition}

\numberwithin{equation}{section}

\theoremstyle{definition}
\newtheorem{rem}[thm]{Remark}
\newtheorem*{example}{Example}

\def\beq{\begin{equation}}
\def\eeq{\end{equation}}
\def\be{\begin{equation*}}
\def\ee{\end{equation*}}

\nc{\bea}{\begin{eqnarray*}}
\nc{\eea}{\end{eqnarray*}}
\nc{\bean}{\begin{eqnarray}}
\nc{\eean}{\end{eqnarray}}
\let\al\alpha
\let\bt\beta
\let\gm\gamma
\let\Gm\Gamma
\let\dl\delta
\let\Dl\Delta

\let\ka\kappa
\let\la\lambda

\let\pho\phi
\let\phi\varphi
\let\si\sigma
\let\Si\Sigma

\let\Ups\Upsilon

\let\Om\Omega

\let\der\partial
\let\Hat\widehat

\let\ox\otimes
\let\Tilde\widetilde
\let\bra\langle
\let\ket\rangle

\let\ge\geqslant

\let\le\leqslant

\let\on\operatorname
\let\bi\bibitem
\let\bs\boldsymbol

\def\C{{\mathbb C}}
\def\Z{{\mathbb Z}}
\def\R{{\mathbb R}}

\def\MM{{\mathbb M}}
\def\PP{{\mathbb P}}
\def\CP{\C\:\PP}

\def\Bc{{\mc B}}

\def\Ec{{\mc E}}
\def\F{{\mc F}}
\def\Hc{{\mc H}}
\def\Ic{{\mc I}}
\def\Jc{{\mc J}}
\def\Kc{{\mc K}}
\def\Mc{{\mc M}}

\def\Oc{{\mc O}}

\def\Zc{{\mc Z}}

\def\Sg{\mathfrak S}

\def\+#1{^{\{#1\}}}
\def\lsym#1{#1\alb\dots\relax#1\alb}
\def\lc{\lsym,}
\def\lox{\lsym\ox}

\def\End{\on{End}}

\def\Res{\on{Res}}

\def\tr{\on{tr}}

\def\cirs{{\raise.2ex\hbox{$\sssize\circ$}}}
\def\buls{{\raise.2ex\hbox{$\sssize\bullet$}}}
\def\zeno{{\raise.1ex\hbox{$\sssize\varnothing$}}}
\def\zenoo{{\raise.1ex\hbox{$\sssize\varnothing\<,0$}}}
\def\Hss{{\raise.2ex\hbox{$\sssize H$}}}
\def\pls{{\raise.2ex\hbox{$\sssize+$}}}
\def\dvs{{\sssize\div}}
\def\lzs{{\sssize\lozenge}}
\def\dmd{{\raise.4ex\hbox{$\sssize\diamond$}}}
\def\trg{{\raise.3ex\hbox{$\sssize\vartriangle$}}}
\def\rst{{\mathchoice{\star}{\star}{\raise.2ex\hbox{$\ssize\star$}}
{\raise.2ex\hbox{$\sssize\star$}}}}
\def\Kon{\Kc^{\raise.4ex\hbox{$\sssize\:\zz=\bs1$}}}
\def\pnd{{\raise.2ex\hbox{$\sssize\#\<$}}}
\def\0{{\mathchoice{\raise.16ex\rlap{$\>\wr$}0}{\raise.16ex\rlap{$\>\wr$}0}
{\raise.1282ex\rlap{$\ssize\>\wr$}0}{\raise.09ex\rlap{$\sssize\>\wr$}0}}}
\def\1{\text{\slshape\bfseries1}}

\def\Srs{\mathscr S}
\def\Srsl{\Srs_\bla}
\def\Srskl{\Srs^{\sssize K}_{\?H_\bla}}
\def\Srsol{\Srs^{\:\cirs}_{\?\bla}}
\def\Hrs{\mathscr H}
\def\Hrsl{\Hrs^{}_\bla}
\def\Hco{\Hc^{\raise.2ex\hbox{$\sssize\oplus$}}}
\def\mmk{{\raise.2ex\hbox{$\sssize\MM_\ka$}}}

\def\Hcs{{\raise.2ex\hbox{$\sssize\Hc$}}}
\def\Kcs{{\raise.2ex\hbox{$\sssize\Kc$}}}
\def\Ocs{{\raise.2ex\hbox{$\sssize\Oc$}}}
\def\shat#1{{\lower.#1ex\hbox{$\sssize\Hat{\hp v}$}}}
\def\cirsh{\sssize\Hat\circ}
\def\Srhl{\Srs_{\!\<\bla}^{\:\smash{\shat2}\<}}
\def\Srhol{\Srs_{\!\<\bla}^{\>\smash{\cirsh}\?}}
\def\Srol{\Srs_{\!\<\bla}^{\:\Ocs}}
\def\Srolh{\Srs_{\!H_\bla}^{\:\Ocs}}
\def\Sroll{\Srs_{\!\<\bla}^{\:\Ocs_{\!\<L}}}

\def\aa{a,\:a}
\def\ab{a,\:b}
\def\ba{b,\:a}

\def\ii{i,\:i}
\def\ij{i,\:j}
\def\ik{i,\:k}

\def\ji{j,\:i}
\def\jj{j,\:j}

\def\jl{j,\>l}
\def\ki{k,\>i}

\def\kl{k,\>l}

\def\IJ{I\<,\:J}

\def\JK{J,\:K}

\let\ga\gamma
\let\Ga\Gamma

\def\Cxs{\C^\times}

\def\bla{{\bs\la}}
\def\Il{{\Ic_{\bla}}}
\def\Fla{\F_\bla}
\def\tfl{{T^*\?\Fla}}

\def\Upsl{\Ups_{\!\bla}}

\def\et{\tilde e}

\def\acuv#1{\acute{#1}\vp{#1}}

\def\gmd{\,\acute{\!\gm}\vp\gm}
\def\GGd{\acute\GG\vp\GG}
\def\fdd{\acuv f}
\def\hdd{\acuv h}
\def\Pdd{\acuv P}

\def\xdd{\acuv x}

\def\zdd{\acuv z}

\def\zzd{\>\acute{\?\zz}\vp\zz}

\def\mkh{\mu^{\sssize K}_{H_\bla}}
\def\mko{\mu^{\sssize K}_{\sssize\!\Hrs^{}_{\<\bla}}}
\def\mho{\mu^{\sssize H}_{\sssize\!\Hrs^{}_{\<\bla}}}
\def\muk{\mu^{\sssize\:\Kc}_\bla}
\def\muko{\mu^{\sssize\:\Kc^\circ}_\bla}
\def\muht{\mu^{\:\smash{\shat5}}_\bla}
\def\muho{\mu^{\:\smash{\cirsh}}_\bla}
\def\muo{\mu^{\sssize\Oc}_\bla}
\def\muoh{\mu^{\sssize\Oc}_{H_\bla}}
\def\muol{\mu^{\sssize\Oc_{\?L}}_\bla}
\def\pii{\pi\sqrt{\<-1}}
\def\piit{\pi\:\sqrt{\<-1}}

\def\xxx{x_1\lc x_n}
\def\yyy{y_1\lc y_n}
\def\zzz{z_1\lc z_n}

\def\Imi{I^{\:\min}}

\def\Imil{\Imi_\bla}

\def\Sla{S_{\la_1}\!\<\lsym\times S_{\la_N}}

\def\Bt{\Tilde B}
\def\Ct{\Tilde C}
\def\Gt{\Tilde G}

\def\Kt{\Tilde K}

\def\Rb{\bar R}
\def\Qb{\bar Q}
\def\Ao{A^{\?\cirs}}
\def\Aob{\hbox{$\Ao$}}

\def\Co{C^{\:\cirs}}
\def\Cto{\Ct^{\:\cirs}}
\def\CH{C^{\sssize H}}

\def\FH{F^{\sssize H}}
\def\Go{G^{\:\cirs}}
\def\Gp{G^{\:\pls}}
\def\Gtp{\Gt^{\:\pls}}
\def\GH{G^{\sssize H}}

\def\Jco{\Jc^{\:\cirs}}

\def\Jch{\Jc^{\<\sssize H}}
\def\Jchh{\hat\Jc{}^{\<\sssize H}}
\def\Jcht{\tilde{\Jc\;\,}\!\!\?{}^{\<\sssize H}}
\def\Ko{K^\cirs}
\def\Kco{\Kc^{\:\cirs}}
\def\KH{K^{\sssize H}}
\def\Lo{L^{\?\cirs}}
\def\Mo{M^{\:\cirs}}
\def\PHd{\Pdd^{\sssize H}}
\def\Ro{R^{\:\cirs}}
\def\Uo{U^\cirs}
\def\Wo{W^\cirs}
\def\WWo{\WW^\cirs}
\def\Xo{X^\cirs}

\def\Xt{\Tilde X}

\def\sh{\hat s}

\def\Mt{{\,\Tilde{\!M}}}

\def\Wto{\:\Tilde{\<W\<}\:^\cirs}
\def\Xt{\Tilde X}

\def\Hd{\dsize H}

\def\Wh{{\:\Hat{\<W\<}\:}}
\def\Who{\Wh{\vp W}^\cirs}

\def\Psh{\Hat\Psi}
\def\Psho{\Hat\Pso}
\def\Pshoi{\Psho{}^{\raise.5ex\hbox{$\sssize\<-1\?$}}}

\def\Pht{\Tilde\Phi}

\def\Omo{\Om^\cirs}

\def\Pho{\Phi^\cirs}
\def\Pst{\Tilde\Psi}
\def\Pso{\Psi^\cirs}

\def\Psho{\Psh^\cirs}

\def\Psb{\bar\Psi}
\def\Psbo{\Psb^\cirs}

\def\Psp{\Psi^\pnd}
\def\Pspt{\Tilde\Psi^\pnd}
\def\Pspo{\Psi^\buls}
\def\Psd{\Psi^\dmd}
\def\Psdo{\Psi^\trg}
\def\Psz{\Psi^{\:\0}}
\def\Pszo{\Psi^{\zeno\:}}
\def\Pszoo{\Psi^{\zenoo\:}}

\def\Wb{\,\,\smash{\overline{\!\!W\!\?}}\;\vp W}

\def\nobla{\nabla^{\:\cirs}}

\def\Fppm_#1{F_{\!{\sssize\#},#1}}
\def\Bcyo{\Bcyr^{\<\cirs}}
\def\Bcyh{\Bcyr^{\<\sssize H}}
\def\Bcyho{\Bcyr^{\<\sssize H\!,0}}
\def\cv/{\raise.6ex\hbox{\tiny CV}}

\def\bul{\mathbin{\raise.2ex\hbox{$\sssize\bullet$}}}
\def\intt{\mathchoice
{\mathop{\raise.2ex\rlap{$\,\,\ssize\backslash$}{\intop}}\nolimits}
{\mathop{\raise.3ex\rlap{$\,\sssize\backslash$}{\intop}}\nolimits}
{\mathop{\raise.1ex\rlap{$\sssize\>\backslash$}{\intop}}\nolimits}
{\mathop{\rlap{$\sssize\:\backslash$}{\intop}}\nolimits}}

\def\GZ/{Gelfand-Zetlin}
\def\KZ/{{\slshape KZ\/}}
\def\qKZ/{{\slshape qKZ\/}}
\def\qKZB/{{\slshape qKZB\/}}
\def\XXX/{{\slshape XXX\/}}
\def\XXZ/{{\slshape XXZ\/}}

\def\EE{{\bs E}}

\def\FF{{\bs F}}
\def\kk{{\bs k}}
\def\lb{{\bs l}}
\def\mb{{\bs m}}
\def\pp{{\bs p}}
\def\qq{{\bs q}}
\def\rr{{\bs r}}
\def\ss{{\bs s}}
\def\TT{{\bs t}}

\def\xx{{\bs x}}
\def\yy{{\bs y}}
\def\zz{{\bs z}} 
\def\GG{{\bs\Ga}}

\def\glN{{\frak{gl}_N}}
\def\Sym{\on{Sym}}

\def\St{\mathrm{Stab}}
\def\Stop{\St^{\sssize\<\mathrm{o\<p}}}

\def\Dlv#1{\Dl^{\!\raise.14ex\hbox{$\ssize#1$}}}
\def\Dlx{\Dlv{\xx}}
\def\Dly{\Dlv{\yy}}
\def\Dlz{\Dlv{\zz}}

\def\Cnn{(\C^N)^{\ox\:n}}
\def\Cnnl{\Cnn_{\:\bla}}
\def\Ctw{\C^{\:2}}
\def\Ctn{(\Ctw)^{\ox\:n}}
\def\Ctnl{\Ctn_{\:\bla}}

\def\WW{{\check W}}

\begin{document}

\hrule width0pt
\vsk->

\title[Landau--\:Ginzburg mirror for a partial flag variety]
%[\qKZ/ difference equations for a partial flag variety]
{Landau--\:Ginzburg mirror, quantum differential
\\[2pt]
equations and \qKZ/ difference equations
\\[2pt]
for a partial flag variety}

\author[Vitaly Tarasov and Alexander Varchenko]
{Vitaly Tarasov$\>^\circ$ and Alexander Varchenko$\>^\star$}

\maketitle

\begin{center}
{\it $^{\star}\<$Department of Mathematics, University
of North Carolina at Chapel Hill\\ Chapel Hill, NC 27599-3250, USA\/}

%\vsk.5>
%{\it $^{\star}\<$Faculty of Mathematics and Mechanics, Lomonosov Moscow State
%University\\ Leninskiye Gory 1, 119991 Moscow GSP-1, Russia\/}

%\vsk.5>
%{\it $^{\star}\<$Moscow Center of Fundamental and Applied Mathematics\\
%Leninskiye Gory 1, 119991 Moscow GSP-1, Russia\/}

\vsk.5>
{\it $\kern-.4em^\circ\<$Department of Mathematical Sciences,
Indiana University\,--\>Purdue University Indianapolis\kern-.4em\\
402 North Blackford St, Indianapolis, IN 46202-3216, USA\/}

%\vsk.5>
%{\it $^\circ\<$St.\,Petersburg Branch of Steklov Mathematical Institute\\
%Fontanka 27, St.\,Petersburg, 191023, Russia\/}
\end{center}

{\let\thefootnote\relax
\footnotetext{\vsk-.8>\noindent
$^\circ\<${\sl E\>-mail}:\enspace vtarasov@iupui.edu\>, % vt@pdmi.ras.ru\>,
supported in part by Simons Foundation grants \rlap{430235, 852996}
\\
$^\star\<${\sl E\>-mail}:\enspace anv@email.unc.edu\>,
supported in part by NSF grants DMS-1665239, DMS-1954266}}

\vsk>
{\leftskip3pc \rightskip\leftskip \parindent0pt \Small
{\it Key words\/}:
Dynamical differential equations, \qKZ/ difference equations,
quantum differential equations,
$q$-
hyper\-geometric solutions
\vsk.6>
{\it 2010 Mathematics Subject Classification\/}: 14N35, 53D45, 14D05, 33C70
\par}

\begin{abstract}
We consider the system of quantum differential equations for a partial flag
variety and construct a basis of solutions in the form of multidimensional
hypergeometric functions, that is, we construct a Landau\:--\:Ginzburg mirror
for that partial flag variety. In our construction, the solutions are labeled
by elements of the $\>K\?$-theory algebra of the partial flag variety.

To establish these facts we consider the equivariant quantum differential
equations for a partial flag variety and introduce a compatible system of
difference equations, which we call the \qKZ/ equations. We construct a basis of
solutions of the joint system of the equivariant quantum differential equations
and \qKZ/ difference equations in the form of multidimensional hypergeometric
functions. Then the facts about the non-equivariant quantum differential
equations are obtained from the facts about the equivariant quantum differential
equations by a suitable limit.

Analyzing these constructions we obtain a formula for the fundamental Levelt
solution of the quantum differential equations for a partial flag variety.
\end{abstract}

{\small\tableofcontents\par}

\setcounter{footnote}{0}
\renewcommand{\thefootnote}{\arabic{footnote}}

\section{Introduction}
\label{intro}

The genus zero Gromov\:--\<Witten invariants of a partial flag variety $\,V\<$
answer enumerative questions about rational curves in $\,V\<$ and are put
together in various ways to define rich mathematical structures such as the
quantum cohomology algebra and the system of quantum differential equations.
Those structures are part of the so\:-called \>`$\!A$-model' of $V$.
Landau\:--\:Ginzburg mirror symmetry seeks to describe these structures in terms
of a mirror dual \>`$\<B$-model' associated with $\>V$. The way originated by
A.\,Givental in \cite{Gi1,Gi2} seeks to encode the data from the $\<A$-model by
oscillating integrals of a superpotential function $\>W$, which is defined on
a `mirror dual' affine variety $\>V^o$. In particular, these oscillating
integrals are supposed to give all solutions of the system of quantum
differential equations.

\vsk.2>
This Landau\:--\:Ginzburg mirror symmetry has been established by A.\,Givental
for full flag varieties and projective spaces in \cite{Gi1, Gi2, GKi}\:.
The construction for Grassmannians has been performed by R.\>J.\,Marsh and
K.\,Rietsch in \cite{MR}\:. A standard problem, related to generalizations of
Givental's approach, is to check that the number of critical points of
the superpotential equals the dimension of the cohomology algebra of $V$, and
then check that the oscillating integrals, which are intrinsically labeled
by the critical points, indeed generate a basis of solutions of the system of
quantum differential equations. While the first part is of algebraic nature,
the second part is of global topological nature (to determine integration cycles
exiting critical points)\:. For example, in \cite{MR} a superpotential for
a Grassmannian was constructed, and it was shown that the number of its critical
points is the correct one, but the part that the oscillating integrals indeed
give a basis of solutions was not addressed.

\vsk.2>
In this paper, we suggest a new `hypergeometric Landau\:--\:Ginzburg mirror
symmetry model' for a partial flag variety $\>V\<$ and construct the full set of
solutions of the system of quantum differential equation of $\>V\<$ in the form
of multidimensional hypergeometric functions. In these functions the role of
Givental's exponential of the superpotential is played by what we call the
master function, which is the product of Gamma functions multiplied by the
exponential of a linear function, and the role of the mirror dual $\>V^o\<$ is
played by the complement in an affine space to the set of all poles of the
product of Gamma functions composing the master function. In particular, our
$\>V^o\<$ {\it is not\/} an algebraic variety, but a complex analytic variety.

\vsk.2>
It is interesting to note that in our construction, the solutions of the system
of quantum differential equations of $\>V\<$ are naturally labeled by elements
of the $\>K\?$-theory algebra of $\>V$. In particular, that observation
suggests that the monodromy and Stokes phenomenon of the system of quantum
differential equations of $\>V\<$ may be described in terms of
the $\>K\?$-theory algebra of $\>V$, in accordance with the philosophy
of B.\,Dubrovin, see \cite{D1, D2, CDG, TV6, TV7, TV2, CV}.

\vsk.2>
Consider the example of the $\,m\:$-\:dimensional complex projective space
$\>V\?=\CP^{\>m}$. In the body of the paper, this example corresponds
\vvn.1>
to the case $\,n=m+1\,$, $\>N=2\,$, $\,\bla=(1,m)\,$. The cohomology algebra
is $\,\Hd^*(V;\C)=\C[x]\big/\bra\:x^{m+1}\:\ket\,$, where $\,x\,$ is the first
Chern class of the tautological line bundle. The quantum multiplication $\,*_p$
on $\,\Hd^*(V;\C)\>$ depends on a parameter $\,p\,$,
\vvn-.1>
\be
x^i\?*_px^j=\:x^{i+j}\quad \text{for}\;\;i+j\le m\,,\qquad
x^i\?*_px^j=\:p\>x^{i+j-m-1}\quad \text{for}\;\;i+j>m\,,\kern-2em
\ee
$i,j=0\lc m\,$. The parameter $\,p\,$ corresponds to $\,p_2/\<p_1$ in the body
\vvgood
of the paper. The quantum differential equation is
\vvn-.4>
\beq
\label{DE1}
-\:\ka\:p\>\frac{\der I}{\der p}\,=\,x*_p\<I\,,
\vv.2>
\eeq
where $\,\ka\,$ is the parameter of the differential equation and $\,I\,$ is
\vvn.1>
the unknown $\,\Hd^*(V;\C)\:$-valued function. The $\>K\?$-theory algebra is
$\,K(V;\C)=\C[X,X^{-1}]\big/\<\bigl\bra(X\?-1)^{m+1}\:\bigr\ket\,$, where
$\>X$ is the class of the tautological line bundle.
% with basis $\,1,X\lc X^{\:m}$.

\vsk.3>
The $\,\Hd^*(V;\C)\:$-valued weight function is
$\;W(t,x)=\sum_{\:i=0}^{\:m}\:t^{m-i}\:x^i\:$. The master function is
\vvn.3>
\be
%\Phi(t,p,\ka)\,=\,\bigl(\<(-\:\ka)^{-m-1}\:p\:\bigr)^{t/\?\ka}\>
%\bigl(\:\Gm(-\:t/\<\ka)\bigr)^{m+1}\:.\kern-1em
\Phi(t,p,\ka)\,=\,(\ka^{m+1}\!/p\:)^{t/\?\ka}\>
\bigl(\:\Gm(\:t/\<\ka)\bigr)^{m+1}\:.\kern-1em
\vv.3>
\ee
For a univariate Laurent polynomial $\,P(X)\,$, define
\vvn.2>
\be
\Psi_P(p,\ka)\,=\,\ka^{-1}
\sum_{r=0}^\infty\,\Res_{\>t\>=\:-\:r\ka\>}\bigl( % r\ka
\Phi(t,p,\ka)\>P(e^{\:2\:\pii\;t/\?\ka})\>W(t,x)\bigr)\,.\kern-1.4em
\vvgood
\ee
The function $\,\Psi_P(p,\ka)\,$ depends only on the class of $\,P\:$ in
$\,K(V;\C)\>$. Alternatively, $\,\Psi_P(p,\ka)\,$ can be written as
an integral over an appropriate contour $\,C\>$ encircling the poles of
$\,\Phi(t,p,\ka)\,$ in the positive direction,
\vvn.1>
\be
\Psi_P(p,\ka)\,=\,\frac1{2\:\piit}\,\int_{\:C}
\Phi(t,p,\ka)\>P(e^{\:2\:\pii\;t/\?\ka})\>W(t,x)\,d\:t\,.\kern-1em
\vv.3>
\ee
For instance, one can take the parabola
$\,C=\,\{\,\ka\>\bigl(1-s^2\?+s\>\sqrt{\<-1}\,\bigr)\ \,|\ \,s\in\R\,\}\,$.
% s^2\?-1-s\>

\begin{thm}
\label{th1}
For any Laurent polynomial $\,P$, the $\:\Hd^*(V)$-\:valued function
$\,\Psi_P(p,\ka)\:$ is an entire function of $\,\:\log\:p\>$ that solves
\vv.06>
the quantum differential equation \eqref{DE1}\:. If the classes of Laurent
polynomials $\,P_0\lc P_m$ give a basis of $\;K(V;\C)\>$, then the functions
$\,\Psi_{P_0}(p,\ka)\lc\Psi_{P_m}(p,\ka)\>$ give a basis of solutions of
the quantum differential equation.
\end{thm}

Theorem \ref{th1} follows from Proposition \ref{mkoiso},
see also Proposition \ref{PsiPsolo}. The fact that solutions
$\,\Psi_{P_0}(p,\ka)\lc\Psi_{P_m}(p,\ka)\>$ of Theorem \ref{th1}
form a basis follows from the following determinant formula.

\begin{thm}
\label{th2}
We have
\vvn.3>
\be
\det\:\biggl(\>\sum_{r=0}^\infty\,\Res_{\>t\:=\:-\:r\ka}\bigl(t^{\:i}\:
e^{-\:
2\:\pii\,\:j\:t/\?\ka}\,\Phi(t,p,\ka)\bigr)\!\biggr)_{\!\ij=\:0}^{\!m}\!=\,
\bigl(2\:\piit\,\bigr)^{m\:(m+1)\</2}\>\ka^{\:(m+1)(m+2)\</2}\:. % -\:\ka
\kern-.6em
\vv.2>
\ee
\end{thm}
Theorem \ref{th2} follows from Theorem \ref{detYo},
see also formula \eqref{detPso2}\;.

\vsk.3>
Notice that in \cite{Gu}\:, D.\,Guzzetti considers a scalar linear differential
equation of order $m+1$ equivalent to the quantum differential equation
\eqref{DE1} and constructs a basis of solutions in the form of similar
integrals.

\vsk.2>
Our construction of the basis of solutions of the system of quantum differential
equations for a partial flag variety is done in several steps and is based on
constructions from representation theory.
\vsk.2>

We consider the joint compatible system of rational \qKZ/ difference equations
and dynamical differential equations for sections of the trivial bundle
$\pi\,:\,\Cnn\times\C^n\times \C^N \,\to\, \C^n\times \C^N$. This system is
defined in terms of the Yangian $Y_h(\frak{gl}_N)$ action on $\Cnn$ and depends
on the Yangian deformation parameter $h$. In this paper, we describe the
$h\to\infty$ limit of this system of difference and differential equations.
This is our first main result, see Section \ref{sec h to inf}. We call
the obtained equations the limiting \qKZ/ difference equations and limiting
dynamical differential equations.

\vsk.2>
In \cite{TV1,TV6} we constructed solutions of the joint system of initial \qKZ/
difference equations and dynamical differential equations in the form of
multidimensional hypergeometric functions, see Theorem \ref{thm cy}.
In this paper, we calculate the $h\to\infty$ limit of these solutions and obtain
solutions of the joint system of limiting equations, see Theorem \ref{thmcyo}.
This is our second main result.

\vsk.2>
According to the general theory in \cite{MO} the initial system of \qKZ/
difference equations and dynamical differential equations is identified with
the system of \qKZ/ difference equations and equivariant quantum differential
equations in the equivariant cohomology of cotangent bundles of partial flag
varieties. The precise formulae for that identification can be found in
\cite{GRTV, RTV1}. It was expected that a suitable $\,{h\to\infty}\,$ limit
of these equations gives a system of difference and differential equations
related to the equivariant quantum cohomology of partial flag varieties
themselves instead of their cotangent bundles. In particular, it was shown
in \cite{BMO} that the $\,{h\to\infty}\,$ limit of the equivariant quantum
differential equations for the cotangent bundle of the full flag variety gives
the equivariant quantum differential equations of the corresponding full
flag variety. Also, the appropriate $\,{h\to\infty}\,$ limit of the Yangian
$R\:$-matrix associated with the quantum cohomology of the cotangent bundle
of a Grassmannian, was used in \cite{GK, GKS} to calculate the quantum
multiplication in the cohomology of the Grassmannian itself.

\vsk.2>
In this paper, we identify our limiting dynamical differential equations for
sections of the bundle $\,\pi\,$ with the equivariant quantum differential
equations for partial flag varieties. Under this identification, our limiting
\qKZ/ difference equations become a system of difference equations in
the equivariant cohomology of partial flag varieties, hence, a new system
of difference equations compatible with the equivariant quantum differential
equations. At the same time, the multidimensional hypergeometric solutions of
the limiting equations for sections of the bundle $\,\pi\,$ become solutions of
the equivariant quantum differential equation and \qKZ/ difference equations
for partial flag varieties. This is our third main result,
see Theorem \ref{thm main}.

\vsk.2>
The particular case of a projective space was considered in \cite{TV7,CV},
where the joint system of the equivariant quantum differential equation
and compatible \qKZ/ difference equations, together with
hypergeometric solutions, was used to analyze the Stokes phenomenon for
the equivariant quantum differential equation of a projective space.
We expect similar applications of our compatible \qKZ/ equations for partial
flag varieties.

\vsk.2>
The paper is organized as follows. In Section \ref{DQKZ} we define the initial
joint system of \qKZ/ and dynamical equations. In Section \ref{sec h to inf}
we describe the $\,{h\to\infty}\,$ limit of these equations. In Section
\ref{sec solns} we describe multidimensional
%$q\:$-
hypergeometric solutions of the initial \qKZ/ and dynamical equations.
In Section \ref{sec lim solns} we obtain solutions of the limiting equations by
taking the $\,{h\to\infty}\,$ limit of the solutions of the initial equations.
In Section \ref{sQde} we identify the limiting equations with the equation in
equivariant cohomology of partial flag varieties.
Appendices \ref{AppA}\>--\>\ref{AppC} contain technical information.

\vsk.2>
In Appendix \ref{AppE} we show that the monodromy of the system of dynamical
differential equations \eqref{DEQ} is abelian under a certain resonance
condition for the equivariant parameters. This is analogous to the corresponding
property of the quantum differential equation of the Hilbert scheme of points
in the plane, see \cite{OP}\:.

\vsk.2>
Notice that Proposition \ref{SPsiPpro} provides an equivariant Gamma theorem
for a partial flag variety $\,\Fla\,$, cf.~Theorem~B.2 and formula (11.19)
in \cite{TV6}\:. Notice also that in Theorems \ref{thmdet}, \ref{detY},
\ref{thmdeto}, \ref{detYo} we prove different determinant formulas which
imply that the functions entering the determinants form bases in the spaces
of solutions of the corresponding differential and difference equations.
The Levelt fundamental solutions are discussed in Sections \ref{secLev} and
\ref{secLevo}.

\vsk.2>
The authors thank G.\,Cotti, A.\,Givental, V.\,Gorbounov, L.\,Mihalcea,
K.\,Rietsch, R.\,Rim\'anyi for useful discussions. The second author thanks
Max Planck Institute for Mathematics in Bonn for hospitality in May--June of
2022.

\section{Dynamical and \qKZ/ equations}
\label{DQKZ}

\subsection{Notations}
\label{Nts}
Fix $\,N, n\in\Z_{>0}\,$.
Let $\,\bla\in\Z^N_{\ge 0}\,$, $\,|\:\bla\:|=\la_1+\ldots+\la_N=n\,$.
Let $\,I=(I_1\lc I_N)\,$ be a partition of $\,\{1\lc n\}\,$ into disjoint
subsets $\,I_1\lc I_N\,$. Denote by $\,\Il\>$ the set of all partitions
$\,I\>$ with $\,|I_j|=\la_j\,$, $\,\:j=1\lc N\>$.

\vsk.2>
Consider the space $\,\C^N\:$ with the standard basis
$\,v_i=(0\lc0,1_i,0\lc0)\,$, $\;i=1\lc N$, and
the tensor product $\,\Cnn\>$ with the basis $\,(v_I)_I\>$,
\vvn.2>
\be
v_I\>=\,v_{i_1}\lox v_{i_n}\,,\kern-2em
\ee
where $\,i_j=m\,$ if $\,j\in I_m\,$.

\vsk.2>
The space $\,\Cnn$ is a module over the Lie algebra $\glN$ with basis
\vv.2>
$\,e_{\ij}\,$, $i,j=1\lc N$. The $\glN$-module $\Cnn$ has weight decomposition
$\,\Cnn=\sum_{|\bla|=n}\Cnnl$, where $\,\Cnnl$ is the subspace with basis
$\,(v_I)_{I\in\Il}\:$.

\subsection{Difference \qKZ/ equations}
\label{sec qkz}

Define the $\,R$-\:matrix acting on $\:(\C^N)^{\ox 2}$,
\vvn.2>
\be
R(u;h)\,=\,\frac{u-h\:P}{u-h}\,,
\vv.3>
\ee
where $P$ is the permutation of factors of $(\C^N)^{\ox 2}$ and
$u,h\in\C\>$.
The $R\:$-matrix satisfies the Yang-Baxter and unitarity equations,
\vvn.4>
\begin{gather}
\label{YB}
R^{(12)}(u-v;h)R^{(13)}(u;h)R^{(23)}(v;h)\,=\,
R^{(23)}(v;h)R^{(13)}(u;h)R^{(12)}(u-v;h)\,,
\\[2pt]
\label{unit}
R^{(12)}(u;h)R^{(21)}(-u;h)\,=\,1\,.
\\[-14pt]
\notag
\end{gather}
The first equation is an equation in $\End((\C^N)^{\ox 3})$.
The superscript indicates the factors of $\:(\C^N)^{\ox 3}$
on which the corresponding operators act.

\vsk.2>
Let $\,\zz=(z_1\lc z_n)\in\C^n$, $\,\qq=(q_1\lc q_N)\in\C^N$, and $\ka\in\C\>$.
Define the \qKZ/ operators $\,K_1\lc K_n\>$ acting on $\Cnn$:
\vvn.3>
\begin{align}
\label{K}
\kern1.2em
K_a(\zz\:;h;\qq\:;\ka)\>&{}=\,
R^{(\aa-1)}(z_a\<-z_{a-1}+\ka;h)\,\dots\,R^{(a,1)}(z_a\<-z_1+\ka;h)\,\times{}
\\[3pt]
\notag
& {}\>\times\<\;q_1^{e_{1,1}^{(a)}}\!\dots\,q_N^{e_{N,N}^{(a)}}\,
R^{(a,n)}(z_a\<-z_n\:;h)\,\dots\,R^{(\aa+1)}(z_a\<-z_{a+1}\:;h)\,.
\kern-1.1em
\\[-14pt]
\notag
\end{align}
The \qKZ/ operators preserve the weight decomposition of $\,\Cnn$ and
%form a discrete flat connection with parameter $\ka$,
have the property
\vvn.3>
\begin{align}
\label{flatK}
K_a(z_1\lc z_b+\ka\lc z_n;h;\qq\:;\ka)\,K_b(\zz\:;h;\qq\:;\ka) & {}\,={}
\\[3pt]
\notag
{}=\,K_b(z_1\lc z_a+\ka\lc z_n;h;\qq\:;\ka) & \,K_a(\zz\:;h;\qq\:;\ka)
\kern-1.4em
\\[-15pt]
\notag
\end{align}
for all $a,b=1\lc n$, see \cite{FR}. We say that the collection of operators
$\,K_1\lc K_n\>$ with this property define a {\it discrete flat connection\/},
see \cite{TV2}\:.

\vsk.2>
The system of difference equations with step $\ka$,
\vvn.2>
\beq
\label{Ki}
f(z_1\lc z_a+\ka\lc z_n;h;\qq\:;\ka)\,=\,
K_a(\zz\:;h;\qq\:;\ka)\,f(\zz\:;h;\qq\:;\ka)\,,\qquad a=1\lc n\>,
\vv.3>
\eeq
for a \,$ (\C^N)^{\ox n}\<$-valued
function $f(\zz,h,\qq\:;\ka)$ is called the \qKZ/ {\it equations\/}.

\begin{rem}
\end{rem}

\subsection{Differential dynamical equations}
\label{Dde}
Define the linear operators $\,X_1\lc X_N$ \,acting on \,$\Cnn$,
called the dynamical Hamiltonians:
\begin{align*}
X_i(\zz\:;h;\qq)\,=\,\sum_{a=1}^n\,z_a\:e^{(a)}_{\ii}-\>
h\>\Bigl(\>\frac{\et_{i,i}(1-\et_{\ii})}2\,&{}+\!
\sum_{1\le a<b\le n}\,\sum_{k=1}^N\,e_{\ik}^{(a)}\>e_{\ki}^{(b)}\>+{}
\\[3pt]
&{}+\>\sum_{\satop{j=1}{j\ne i}}^N\,\frac{q_j}{q_i\<-q_j}\,
(\et_{\ij}\>\et_{\ji}\<-\et_{\ii})\<\Bigr)\>,\kern-1.4em
\notag
\\[-28pt]
\notag
\end{align*}
where \,$\et_{s,\:t} = \sum_{a=1}^n e_{s,\:t}^{(a)}$.
The differential operators
\beq
\label{nady}
\nabla_{\qq,\ka,i}\,=\,\ka\>q_i\:\frac\der{\der\:q_i} - X_i(\zz\:;h;\qq)\,,
\qquad i=1\lc N\>,
\eeq
preserve the weight decomposition of $\Cnn$ and pairwise commute,
see \cite{TV3}, also \cite[Section 3.4]{GRTV}, \cite[Section 7.1]{RTV1},
\cite{MTV}. The operators $\nabla_{\qq,\ka,i}$ define the $\Cnn$-valued
{\it dynamical connection\/}.

\vsk.2>
The system of differential equations
\vvn.3>
\beq
\label{DEQ}
\ka\>q_i\:\frac{\der f}{\der\:q_i}\,=\,X_i(\zz\:;h;\qq)\>f\,,\qquad i=1\lc N\>,
\kern-1em
\vv.1>
\eeq
for a $\,\Cnn$-valued function $\,f(\zz\:;h;\qq\:;\ka)\,$ is called
the {\it dynamical equations\/}.

\begin{thm} [{\cite{TV3}}]
\label{thm qkz}
The joint system of dynamical and \qKZ/ equations with the same parameter $\ka$
is compatible.
\end{thm}

\begin{lem}
\label{lem new X}
The dynamical Hamiltonians $\,\:X_1\lc X_N\:$ acting on $\:\,\Cnn\<$
can be written in the form
\vvn-.4>
\beq
\label{Xi}
X_i(\zz\:;h;\qq)\,=\,\sum_{a=1}^n\,z_a\:e^{(a)}_{\ii}-\>
h\!\sum_{1\le a<b\le n}\,\sum_{\satop{j=1}{j\ne i}}^N\,
\Bigl(\,\frac{e_{\ij}^{(a)}\:e_{\ji}^{(b)}}{1-q_j/q_i}\,-\,
\frac{e_{\ji}^{(a)}\:e_{\ij}^{(b)}}{1-q_i/q_j}\,\Bigr)\>.
\kern-2.4em
\eeq
\end{lem}
\begin{proof}
The proof is by direct verification.
\end{proof}

\section{Limit $\,{h\to\infty}\,$}
\label{sec h to inf}

\subsection{Limit of \qKZ/ operators}
Introduce the $R\:$-matrix $\Ro\<(u)$ acting on $(\C^N)^{\ox 2}$:
\vvn.3>
\be
\Ro\<(u)\,=\,P+\>u\!\sum_{1\le i<j\le N}\!e_{\ii}\ox e_{\jj}\,.
\vv-.2>
\ee
For example, for $\,N=2\,\:$ the matrix is
\vvn-.1>
\be
\Ro\<(u) =
\left(\,
\begin{matrix}
1 & 0 & 0 &0
\\
0 & u & 1 &0
\\
0 & 1 & 0 & 0
\\
0 & 0 & 0 & 1
\end{matrix}
\,\right).
\ee

\begin{lem}
\label{lem 3.1}
We have
\beq
\label{RR0}
\Ro\<(u)\,=\,\lim_{h\to\infty}\bigl(
(\<-\:h)^{\,\sum_{i<j}e^{(1)}_{i,i}e^{(2)}_{j,j}}\,R(u;h)\;
(\<-\:h)^{\:-\!\sum_{i<j}e^{(1)}_{j,j}e^{(2)}_{i,i}}\,\bigr)\,.
\vv.2>
\eeq
Moreover, $\Ro\<(u)$ satisfies the Yang-Baxter and unitarity equations,
\vvn.2>
\begin{gather}
\label{YB0}
\bigl(\Ro\<(u-v)\bigr)^{(12)}\bigl(\Ro\<(u)\bigr)^{(13)}
\bigl(\Ro\<(v)\bigr)^{(23)}=\,
\bigl(\Ro\<(v)\bigr)^{(23)}\bigl(\Ro\<(u)\bigr)^{(13)}
\bigl(\Ro\<(u-v)\bigr)^{(12)},
\\[2pt]
\label{unit0}
\bigl(\Ro\<(u)\bigr)^{(12)}\bigl(\Ro\<(-u)\bigr)^{(21)}=\,1\,.
\end{gather}
\end{lem}
\begin{proof}
The proof of formula \eqref{RR0} is straightforward. Equations \eqref{YB0},
\eqref{unit0} follow from \eqref{RR0} and \eqref{YB}, \eqref{unit}, respectively.
\end{proof}

Let $\,\pp=(p_1\lc p_N)\in\C^N$.
Define the \qKZ/ operators $\,\Ko_1\lc\Ko_n\:$ acting on $\,\Cnn\,$:
\vvn.4>
\begin{align}
\label{Ko}
\kern2em
\Ko_a(\zz\:;\pp\:;\ka)\>&{}=\,\bigl(\Ro\<(z_a\<-z_{a-1}+\ka)\bigr)^{(\aa-1)}
\<\dots\,\bigl(\Ro\<(z_a\<-z_1+\ka)\bigr)^{(a,1)}\times{}
\\
\notag
& {}\>\times\,p_1^{e_{1,1}^{(a)}}\!\dots\,p_N^{e_{N,N}^{(a)}}\,
\bigl(\Ro\<(z_i\<-z_n)\bigr)^{(a,n)}\<\dots\,
\bigl(\Ro\<(z_a\<-z_{a+1})\bigr)^{(\aa+1)}\,.\kern-2em
\end{align}

\begin{lem}
\label{lem K-lim}
For $a=1\lc n$, we have
\vvn.1>
\be
\Ko_a(\zz\:;\pp\:;\ka)\,=\lim_{h\to\infty}
\bigl(\<(\<-\:h)^{\,\sum_{b<c,\,i<j}e^{(b)}_{j,j}e^{(c)}_{i,i}}\,
\Kt_i(\zz\:;h;\pp\:;\ka)\,
(\<-\:h)^{\:-\sum_{b<c,\,i<j}e^{(b)}_{j,j}e^{(c)}_{i,i}}\,\bigr)
\,,\kern-2em
\vv.1>
\ee
where the operators $\,\Kt_1\lc\Kt_n$ are obtained from the \qKZ/ operators
$\,K_1\lc K_n$, see \eqref{K}, by the substitution
\vvn-.3>
\beq
\label{q->tq}
q_i\>\mapsto\,p_i\,(\<-\:h)^{\,\sum_{b=1}^n
\left(\sum_{j>i}e^{(b)}_{j,j}\,-\,\sum_{j<i}e^{(b)}_{j,j}\>\right)}\,,\qquad
i=1\lc N\,.\kern-2em
\eeq
\end{lem}
\begin{proof}
The lemma is proved by direct verification using Lemma \ref{lem 3.1}.
\end{proof}

\begin{cor}
\label{cor qKZ}
The \qKZ/ operators $\,\Ko_1\lc\Ko_n\>$ preserve the weight decomposition of
$\,\Cnn$ and define a discrete flat connection,
\vvn.1>
\be
\Ko_a(z_1\lc z_b+\ka\lc z_n;\pp\:;\ka)\,\Ko_b(\zz\:;\pp\:;\ka)\,=\,
\Ko_b(z_1\lc z_a+\ka\lc z_n;\pp\:;\ka)\,\Ko_a(\zz\:;\pp\:;\ka)
\vv.1>
\ee
for all $\,a,b=1\lc n\,$, cf.~\eqref{flatK}\:.
\end{cor}

\begin{rem}
The $R\:$-matrix $\,\Ro\<(u)\,$ for $\,N=2\,$ is analogous to the $R\:$-matrix
of the five-vertex model and the $q\:$-boson model, and to the $R\:$-matrices
considered in \cite{GK,GKS} in relation to
%generalized
the quantum cohomology theory for Grassmannians.
\end{rem}

\begin{rem}
Since $\,\:\det\:\Ro\<(u)=1\,$, the \qKZ/ operators $\,\Ko_1\lc\Ko_n\:$ are
\vvn.07>
invertible for all $\,\zz,\pp,\ka\,$ provided $\,p_i\<\ne 0\,$ for all
\vvn.06>
$\,i=1\lc N\>$. That is, the discrete flat connection on $\,\Cnn\:$ defined
by $\,\Ko_1\lc\Ko_n\>$ is regular for all $\,\zz\,$, unlike the discrete flat
\vvn.05>
connection defined by the \qKZ/ operators $\,K_1\lc K_n\>$, see \eqref{K}\:.
\end{rem}

\subsection{Dynamical operators $\,\Xo_1\lc\Xo_{\?N}\,$}
\label{secXo}

Recall \,$\bla=(\la_1\lc\la_N)\,$. Let
\,$\la^{(i)}\<=\sum_{j=1}^i\,\la_i\,$, $\;i=1\lc N\,$, so that
\vvn.06>
$\,\la^{(N)}\<=n\,$. For $\,1\le i<j\le n\,$, set
$\;\la_{\bra\ij\ket}=\:
\sum_{k=i}^{j-1}\,(\la_k\<+\la_{k+1})\:=\:
\la_i\<+\la_j\<+2\sum_{k=i+1}^{j-1}\la_k\,$.

\vsk.2>
Recall $\;I=(I_1\lc I_N)\,$ and the vectors $v_I\<\in\<\Cnn$.
For $\,\si\in S_n\>$ and $\,I=(I_1\lc I_N)\,$, set
$\,\si(I)=\bigl(\si(I_1)\lc\si(I_N)\bigr)\,$. For \,$a,b=1\lc n\,$,
let $\,s_{\ab}\,$ be the transposition of $\,a\:,b\,$.

\vsk.2>
Given $\,I\<\in\Il\,$, let $\,\:a\in I_i\,$ and $\,\:b\in I_j\,$.
A pair $\,(a,b)\,$ is called {\it $\,I\?$-disordered\/}
if either $\,a<b\,$, $\,i>j\,$, or $\,a>b\,$, $\,i<j\,$.
A pair $\,(a,b)\,$ is called {\it $\,I\?$-ordered\/}
\vvn.06>
if either $\,a<b\,$, $\,i<j\,$, or $\,a>b\,$, $\,i>j\,$.
A pair $\,(a,b)\,$ is called {\it $\,I\?$-flat\/} if $\,i=j\,$.

\vsk.2>
Denote $\,M=\{\:\min\:(a,b)+1\lc\max\:(a,b)-1\:\}\,$,
\vvn.1>
$\,k=\min\:(i,j)\,$, $\,l=\max\:(i,j)\,$. If $\,i\ne j\,$, so that $\,k<l\,$,
set \,$m_{\ab,\:I}=\la_{\bra\kl\:\ket}\,$ and
\vvn-.6>
\be
r_{\ab,\:I}\>=\,|\:M\cap I_i\>|+|\:M\cap I_j\>|+
2\>\Bigl|\>M\cap\!\bigcup_{r=k+1}^{l-1}\! I_r\,\Bigr|\,.
\vv-.2>
\ee
Clearly, $\,m_{\ab.\:I}\ge 2\;$ and
$\,\:0\le r_{\ab,\:I}\le m_{\ab,\:I}-2\;$ if $\,\:i\ne j\,$.
Set $\,I_{\:[\ab\:]}=\>\bigcup_{\:r=k}^{\,l} I_r\,$.

\vsk.3>
Set
\vvn-.3>
\beq
\label{Imil}
\Imil\>=\,\bigl(\<(\la_1\lc\la^{(1)})\,,(\la^{(1)}\<+1\lc\la^{(2)})\,,
\,\dots\,(\la^{(N-1)}\<+1\lc n)\bigr)\,.
\vv.3>
\eeq
For $\,I\<\in\Il\,$, let $\,\si_I\in S_n\,$ be the element of minimal length
such that $\,\si_I(\Imil)=I\,$. Notice that
\vvn-.2>
\beq
\label{siI}
|\:\si_I\:|\,=\,
\bigl|\:\{\:(a,b)\ |\ a\in I_i\,,\;b\in I_j\,,\;a<b\,,\;i>j\>\}\bigr|\,.
\vv.2>
\eeq

\begin{lem}
\label{lemapp}
We have
\vsk.2>
\itemflat{\rm a)}
If the pair $\,(a,b)\:$ is $\,I\?$-disordered, then
$\,|\:\si_{s_{a\<,b}(I)}\:|=|\:\si_I\:|-r_{\ab,\:I}-1\,$.
\vsk.2>
\itemflat{\rm b)}
If the pair $\,(a,b)\:$ is $\,I\?$-ordered, then
$\,|\:\si_{s_{a\<,b}(I)}\:|=|\:\si_I\:|+r_{\ab,\:I}+1\,$.
\vsk.2>
\itemflat{\rm c)}
If the pair $\,(a,b)\:$ is $\,I\?$-flat, then $\,s_{\ab}(I)=I\,$.
\end{lem}

\begin{cor}
\label{corapp}
If $\,|\:\si_{s_{a\<,b}(I)}\:|<|\:\si_I\:|\,$, then
the pair $\,(a,b)\:$ is $\,I\?$-disordered.
\vvn.1>
If $\,|\:\si_{s_{a\<,b}(I)}\:|>|\:\si_I\:|\,$, then
the pair $\,(a,b)$ is $\,I\?$-ordered.
\end{cor}

A pair $\,(a,b)\,$ is called {\it $\,I\?$-admissible of the first kind\/}
\vvn.1>
if ${\,|\:\si_{s_{a\<,b}(I)}\:|=|\:\si_I\:|-1\,}$,
and {\it $\,I\?$-admis\-sible of the second kind\/} if
$\;|\:\si_{s_{a\<,b}(I)}\:|=|\:\si_I\:|+m_{\ab,\:I}-1\,$.

\begin{lem}
\label{ladm1}
A pair $\,(a,b)$ is $\,I\?$-admissible of the first kind if and only if
\vvn.1>
it is $\,I\?$-disordered and the intersection
$\;\{\:\min\:(a,b)+1\lc\max\:(a,b)-1\:\}\cap I_{\:[\ab\>]}$ \,is empty.
\end{lem}

\begin{lem}
\label{ladm2}
A pair $\,(a,b)$ is $\,I\?$-admissible of the second kind if and only if
\vvn.1>
it is $\,I\?$-ordered and the intersection
$\;\{\:1\lc\min\:(a,b)-1,\,\max\:(a,b)+1\lc n\:\}\cap I_{\:[\ab\>]}$
\,is empty.
\end{lem}

\begin{example}
Let $\,N=4\,$, $\,n=5\,$, $\,\bla=(2,1,1,1)\,$,
$\,I=\bigl(\{\:1,3\:\},\{4\},\{2\},\{5\}\bigr)\,$.
Then the $\,I\<$-ad\-mis\-sible pairs of the first kind are $\,(3,2)\,$,
$(2,3)\,$, $(4,2)\,$, $(2,4)\,$, and the $\,I\<$-admissible pairs of the second
kind are $\,(1,4)\,$, $(4,1)\,$, $(1,5)\,$, $(5,1)\,$, $(2,5)\,$, $(5,2)\,$.
\end{example}

For all $\,i,j,a,b\,$, define linear operators $\,Q_{\ij}^{\>\ab}\>$
acting on $\,\Cnn$ by the rule
\begin{alignat*}2
Q_{\ij}^{\>\ab}\>v_I\> &{}=\,v_{s_{a\<,b}(I)}\,,\quad &&
\text{if $\,a\in I_i\,$, $\,b\in I_j\,$, and
the pair $\,(a,b)\,$ is $\,I\<$-admissible,}
\\[1pt]
Q_{\ij}^{\>\ab}\>v_I\> &{}=\,0\,, &&\text{otherwise.}
\end{alignat*}
Recall the linear operators $\,e_{\ii}^{(a)}\>$ acting on $\,\Cnn\>$.

\vsk.2>
The dynamical operators $\,\Xo_1\lc\Xo_{\?N}\,$, acting on $\,\Cnn$ are
given by the formula
\vvn.3>
\begin{align}
\label{Xo}
\Xo_i(\zz\:;\pp)\,=\sum_{a=1}^n z_a\,e_{\ii}^{(a)}\>
&{}+\!\sum_{1\le b<a\le n}\<\biggl(\>\sum_{j=i+1}^N\?Q_{\ij}^{\>\ab}\:-
\sum_{j=1}^{i-1}\>\frac{p_i}{p_j}\,Q_{\ij}^{\>\ab}\>\biggr)+{}
\\[3pt]
\notag
&{}+\!\sum_{1\le a<b\le n}\<\biggl(\>
\sum_{j=i+1}^N\>\frac{p_j}{p_i}\>Q_{\ij}^{\>\ab}\:-
\sum_{j=1}^{i-1}\>Q_{\ij}^{\>\ab}\>\biggr)\>.
\end{align}
\vsk.2>\noindent
Notice that the operators $\,\Xo_1\lc\Xo_{\?N}\>$ preserve the weight decomposition
of $\Cnn$.

\vsk.3>
\begin{example}
Let $\,N=2\,$, $\,n=3\,$. The operators $\,\Xo_1\>,\:\Xo_2\>$
preserve the subspace spanned by the vectors
\vvn.1>
\be
v_{(\{1\},\{2,3\})}\:=\>v_1\<\ox v_2\<\ox v_2\,,\qquad
v_{(\{2\},\{1,3\})}\:=\>v_2\<\ox v_1\<\ox v_2\,,\qquad
v_{(\{3\},\{1,2\})}\:=\>v_2\<\ox v_2\<\ox v_1\,.
\vv.4>
\ee
of weight $\,\bla=(1,2)\,$.
The matrices of $\,\Xo_1\>,\:\Xo_2\>$ in this basis are
\vvn.3>
\be
\quad
\Xo_1=\left(\,
\begin{matrix}
z_1 & 1 & 0
\\
0 & z_2 & 1
\\
p_2/p_1 & 0 & z_3
\end{matrix}
\,\right), \qquad
\Xo_2=\left(\,
\begin{matrix}
z_2\:+z_3 & -\:1 & 0
\\
0 & z_1\:+z_3 & -\:1
\\
-\>p_2/p_1 & 0 & z_1\:+z_2
\end{matrix}
\,\right).
\vv.5>
\ee
\end{example}

\begin{example} Let $\,N=5\,$, $\,n=6\,$,
$\,I=\bigl(\{\:2,5\},\{6\},\{3\},\{1\},\{4\}\bigr)\,$.
%$\,v_I\<=v_4\<\ox v_1\<\ox v_3\<\ox v_5\<\ox v_1\<\ox v_2\,$.
Then
\vvn.4>
\begin{align*}
\Xo_1(\zz\:;\pp)\,v_I\,=\,{}& (z_2+z_5)\,v_I\:+\:
v_{(\{\:1,5\},\{6\},\{3\},\{2\},\{4\})}+\:
v_{(\{\:2,4\},\{6\},\{3\},\{1\},\{5\})}+{}
\\[6pt]
&\hp{(z_2+z_5)\,v_I}\:+\,v_{(\{\:2,3\},\{6\},\{5\},\{1\},\{4\})}+\>
\frac{p_2}{p_1}\,v_{(\{\:5,6\},\{2\},\{3\},\{1\},\{4\})}\,,
\\[4pt]
\Xo_2(\zz\:;\pp)\,v_I\,=\,{}& z_5\,v_I\:+\:
v_{(\{\:2,5\},\{4\},\{3\},\{1\},\{6\})}+\:
v_{(\{\:2,5\},\{3\},\{6\},\{1\},\{4\})}-
\,\frac{p_2}{p_1}\,v_{(\{\:5,6\},\{2\},\{3\},\{1\},\{4\})}\,.
\end{align*}
\end{example}

\subsection{Limit of dynamical Hamiltonians}
Let the operators $\,\Xt_1\lc\Xt_N$ be obtained from the dynamical Hamiltonians
$X_1\lc X_N$, see \eqref{Xi}, by substitution \eqref{q->tq}\,. In more detail,
the operators $\,\Xt_1\lc\Xt_N$ preserve the weight decomposition of $\Cnn$
and the action of $\,\Xt_i(\zz\:;h;\pp)$ on $\,\Cnnl$ coincides with the action
of $\,X_i(\zz\:;h;\qq)\,$ for $\,\qq=(q_1\lc q_N)$,
\beq
\label{q->tqt}
q_j\>=\,p_j\,(\<-\:h)^{\,\sum_{k>j}\la_k\,-\,\sum_{k<j}\la_k}\,,\qquad j=1\lc N\,.
\kern-2em
\eeq

\begin{lem}
\label{lem X-lim}
For $i=1\lc N$, we have
\beq
\Xo_i(\zz\:;\pp)\,=\lim_{h\to\infty}
\bigl(\<(\<-\:h)^{\,\sum_{b<c,\,j<k}e^{(b)}_{k,k}e^{(c)}_{j,j}}\,\Xt_i(\zz\:;h;\pp)\,
(\<-\:h)^{\,-\sum_{b<c,\,j<k}e^{(b)}_{k,k}e^{(c)}_{j,j}}\,\bigr)
\,.\kern-2em
\eeq
\end{lem}

\begin{proof}
The lemma is proved by direct verification using Lemma \ref{lem new X},
the equality
\vvn.3>
\be
\sum_{b<c,\,j<k}\!e^{(b)}_{k,k}\,e^{(c)}_{j,j}\,v_I\,=\,
|\:\si_I\:|\,v_I\,,
\ee
following from \eqref{siI}, and Lemma \ref{lemapp}.
\end{proof}

\begin{cor} The differential operators
\beq
\label{nadyn}
\nobla_{\<\pp,\ka,i}\,=\,
\ka\>p_i\:\frac\der{\der\:p_i} - \Xo_i(\zz\:;\pp)\,,
\qquad i=1\lc N\>,
\eeq
preserve the weight decomposition of $\,\Cnn\!$ and pairwise commute.
\end{cor}

\subsection{Limiting equations}
The system of difference equations with step $\ka$,
\vvn.2>
\beq
\label{Kio}
f(z_1\lc z_a+\ka\lc z_n;\pp)\,=\,\Ko_a(\zz\:;\pp\:;\ka)\,f(\zz\:;\pp)\,,
\qquad a=1\lc n\>,\kern-2em
\vv.2>
\eeq
for a $\:(\C^N)^{\ox n}\<$-valued function $f(\zz,\pp)$ will be called
the {\it limiting \qKZ/ equations\/}.

\vsk.2>
The system of differential equations with parameter $\ka$,
\vvn.2>
\beq
\label{DEQo}
\ka\>p_i\:\frac{\der f}{\der\:p_i}\,=\,\Xo_i(\zz\:;\pp)\>f\,,\qquad i=1\lc N\>,
\eeq
for a \,$\Cnn$-valued function $f(\zz\:;\pp)$ will be called the
{\it limiting dynamical equations\/}.

\begin{thm}
\label{thm qkz n}
The joint systems of limiting dynamical and limiting \qKZ/ equations
with the same parameter $\ka$ is compatible.
\end{thm}
\begin{proof}
The theorem follows from Theorem \ref{thm qkz} and
Lemmas \ref{lem K-lim}, \ref{lem X-lim}.
\end{proof}

\section{%$q\:$-
Hypergeometric solutions of the dynamical and \qKZ/ equations}
\label{sec solns}

\subsection{Weight functions $W_I(\TT\:;\zz\:;h)$}
\label{dwf}
For $I\in \Il$, we define the weight functions
\vvn.1>
$W_I(\TT\:;\zz\:;h)$, cf.~\cite{TV1,TV4,RTV1,TV6}.
The functions \,$W_I(\TT\:;\zz\:;h)$ \,here coincide with those defined \,in
\cite[Section~3.4]{TV6}.

\vsk.2>
Recall \,$\bla=(\la_1\lc\la_N)\,$ and $\,I=(I_1\lc I_N)$. Set
\;$\bigcup_{\>k=1}^{\,j}I_k=\>\{\>i^{(j)}_1\!\lsym<i^{(j)}_{\la^{(j)}}\}$\>.
Consider the variables \,$t^{(j)}_a$, \,$j=1\lc N-1$, \,$a=1\lc\la^{(j)}$.
\,Set \,$t^{(N)}_a\!=z_a$, \,$a=1\lc n$\>.
Denote $\,\TT^{(j)}=(t^{(j)}_1\lc t^{(j)}_{\la^{(j)}})\,$ and
$\,\TT=(\TT^{(1)}\lc\TT^{(N-1)})\,$.

\vsk.2>
For $\,I\<\in\Il\,$, define
\beq
\label{SiI}
\Si_I\>=\,(z_{i^{(1)}_1}\lc z_{i^{(1)}_{\la^{\vp1\smash{(1)}}}}\?,
z_{i^{(2)}_1}\lc z_{i^{(2)}_{\la^{\vp1\smash{(2)}}}}\?,\,\ldots\,,\,
z_{i^{(N-1)}_1}\lc z_{i^{(N-1)}_{\la^{\vp1\smash{(N-1)}}}})\,,
\eeq
so that, $\,\TT=\Si_I\,$ reads in detail as
\be
t^{(k)}_a\<=\>z_{i^{(k)}_a}, \qquad k=1\lc N-1\,,\quad i=a\lc\la^{(k)}\>.
\ee

\vsk.2>
For a function $\,f(x_1\lc x_k)$ of some variables, denote
\vvn.2>
\be
\Sym_{\:x_1\lc x_k}\<f(x_1\lc x_k)\,=\,
\sum_{\si\in S_k}f(x_{\si(1)}\lc x_{\si(k)})\,.
\ee

\vsk-.4>
The weight functions are
\vvn.3>
\beq
\label{hWI}
W_I(\TT\:;\zz\:;h)\,=\,
\Sym_{\>t^{(1)}_1\!\lc\,t^{(1)}_{\la^{(1)}}}\,\ldots\;
\Sym_{\>t^{(N-1)}_1\!\lc\,t^{(N-1)}_{\la^{(N-1)}}}
U_I(\TT\:;\zz\:;h)\,,
\vv-.4>
\eeq
\begin{align*}
& U_I(\TT\:;\zz\:;h)\,={}
\\[4pt]
\notag
&{}\!=\,\prod_{j=1}^{N-1}\,\prod_{a=1}^{\la^{(j)}}\,\biggl(
\prod_{\satop{c=1}{i^{(j+1)}_c\?<\>i^{(j)}_a}}^{\la^{(j+1)}}
\!\!(t^{(j)}_a\?-t^{(j+1)}_c)
\prod_{\satop{d=1}{i^{(j+1)}_d>\>i^{(j)}_a}}^{\la^{(j+1)}}
\!\!(t^{(j)}_a\?-t^{(j+1)}_d-h )\,\prod_{b=a+1}^{\la^{(j)}}
\frac{t^{(j)}_b\?-t^{(j)}_a\?-h}{t^{(j)}_b\?-t^{(j)}_a}\,\biggr)\,.
\kern-1em
\end{align*}

\begin{example}
Let $N=2$, $n=2$, $\bla=(1,1)$, $I=(\{1\},\{2\})$, $J=(\{2\},\{1\})$. Then
\vvn.1>
\be
W_I(\TT\:;\zz\:;h)\,=\,t^{(1)}_1\?-z_2-h\,, \qquad
W_J(\TT\:;\zz\:;h)\,=\,t^{(1)}_1\?-z_1\,.
\vv.2>
\ee
\end{example}

\begin{lem}
[\cite{RTV1,TV6}]
\label{c3t}
For any $\,I\in \Il\,$, $\,i=1\lc N-1\,$, and $\,a=1\lc n-1\,$, we have
\begin{align}
\label{m3t}
& W_{s_{a\<,a+1}(I)}(\TT\:;\zz\:;h)\,={}
\\[4pt]
\notag
\!&{}=\,\frac{z_a\<-z_{a+1}\<-h}{z_a\<-z_{a+1}}\,
W_I(\TT\:;z_1\lc z_{a+1},z_a\lc z_n;h)\>+\>\frac{h}{z_a\<-z_{a+1}}\,
W_I(\TT\:;\zz\:;h)\,.
\end{align}
\end{lem}

\vsk0>
\begin{rem}
Define the operators \,$\sh_1\lc\sh_{n-1}$ acting on functions of $\,\zzz\,$:
\vvn.2>
\beq
\label{Sa}
\sh_a\:f(\zz)\,=\,\frac{z_a\<-z_{a+1}\<-h}{z_a\<-z_{a+1}}\,
f(z_1\lc z_{a+1},z_a\lc z_n)\>+\>\frac h{z_a\<-z_{a+1}}\,f(\zz)\,.
\kern-1.6em
\vv.3>
\eeq
The assignment \,$s_{\aa+1}\mapsto\sh_a\,,\;\;a=1\lc n-1\,$, yields
a representation of the symmetric group $\,S_n\,$.
\end{rem}

Set
\vvn->
\beq
\label{cla}
c_\bla(\TT\:;h)\,=\>\prod_{i=1}^{N-1}\,\prod_{a=1}^{\la^{(i)}}\>
\prod_{\satop{b=1}{b\ne a}}^{\la^{(i)}}\,(t^{(i)}_a\<-t^{(i)}_b\<-h)\,.
\vv.2>
\eeq
Recall the permutations $\,\si_I\,$, $\,I\<\in\Il\,$, defined in Section
\ref{secXo}.

\begin{lem}[\cite{RTV1}] % Lemma 3.1
\label{WIz}
We have $\;W_J(\Si_I\:;\zz\:;h)=0\>$ unless $\,I=J$ or
$\,|\:\si_I\:|>|\:\si_J\:|\,$, and
\vvn.4>
\beq
\label{WII}
W_I(\Si_I\:;\zz\:;h)\,=\,c_\bla(\Si_I\:;h)\;
\prod_{j=1}^{N-1}\>\prod_{k=j+1}^N\,\prod_{a\in I_j}\>\biggl(\,
\prod_{\satop{b\in I_k}{b<a}}\:(z_a\<-z_b)\,
\prod_{\satop{b\in I_k}{b>a}}\:(z_a\<-z_b\<-h)\<\biggr)\>.\kern-1.4em
\vv.3>
\eeq
\end{lem}

\begin{lem}
\label{WIJ}
For any $\,J\<\in\Il\>$, the polynomial $\;W_J(\Si_I\:;\zz\:;h)\>$
is divisible by
\vvn.41>
\be
c_\bla(\Si_I\:;h)\;\prod_{j=1}^{N-1}\>\prod_{k=j+1}^N\,
\prod_{a\in I_j}\>\prod_{\satop{b\in I_k}{b>a}}\,(z_a\<-z_b\<-h)\,.\kern-1.4em
\vv.2>
\ee
\end{lem}
\noindent
Lemma \ref{WIJ} is a version of \cite[\:Lemma 3.1, item (I)\:]{RTV2}\:.
A stronger more technical analogue of Lemma \ref{WIJ} is given by
Lemma \ref{lemW} in Section \ref{details}.

\vsk.4>
Let $\,\si_0\,$ be the longest permutation,
$\,\si_0(i)=n+1-i\,$, $\;i-1\lc n\,$.
\vvn.4>
For $\,I\<\in\Il\,$, define
\beq
\label{WcI}
\WW_I(\TT\:;\zz\:;h)\,=\,W_{\si_0(I)}(\TT\:;z_n\lc z_1\:;h)\,.
\eeq
Set
\vvn-.3>
\beq
\label{RQ}
R_\bla(\zz)\,=\:\prod_{i=1}^{N-1}\,\prod_{a=1}^{\la^{(i)}}\,
\prod_{b=\la^{(i)}+1}^{\la^{(i+1)}}\!\?(z_a\<-z_b)\,,\kern1.9em
Q_\bla(\zz\:;h)\,=\:\prod_{i=1}^{N-1}\,\prod_{a=1}^{\la^{(i)}}\,
\prod_{b=\la^{(i)}+1}^{\la^{(i+1)}}\!\?(z_a\<-z_b\<-h)\,,\kern-1.2em
\vv.6>
\eeq
For $\,\si\<\in\<S_n\,$, denote $\,\zz_\si\<=(z_{\si(1)}\lc z_{\si(n)})\,$.

\begin{prop}[{\cite[Lemma~3.4\:]{RTV1}}]
\label{lemorth}
The functions $\,W_I(\TT\:;\zz\:;h)\>$ and $\,\WW_J(\TT\:;\zz\:;h)\>$
are biorthogonal\:,
\vvn-.4>
\beq
\label{orth}
\sum_{I\in\Il}\,\frac{W_J(\Si_I\:;\zz\:;h)\,\WW_K(\Si_I\:;\zz\:;h)}
{c_\bla^{\>2}(\Si_I\:;h)\>R_\bla(\zz_{\si_I}\<)\>Q_\bla(\zz_{\si_I};h)}
\,=\,\dl_{\JK}\,.
\kern-2em
\vvn.4>
\eeq
\end{prop}

\begin{cor}
\label{cororth}
We have
\vvn.4>
\beq
\label{corth}
\sum_{I\in\Il}\,W_I(\Si_J\:;\zz\:;h)\,\WW_I(\Si_K;\zz\:;h)\,=\,\dl_{\JK}\,
c_\bla^{\>2}(\Si_J\:;h)\>R_\bla(\zz_{\si_J}\<)\>Q_\bla(\zz_{\si_J};h)\,.
\kern-2em
\vvn.2>
\eeq
\end{cor}

\subsection{Master function}
\label{sMF}
Let $\,\pho(x,h,\ka)=\Ga(x/\ka)\,\Ga\bigl(\<(h-x)/\ka\bigr)\,$.
\vvn.3>
Define the {\it master function\/}:
\begin{align}
\label{Phi}
\Phi_\bla(\TT\:;\zz\:;h;\qq\:;\ka)\,=\,
(e^{\:\pii\,(n\:-\<\la_N)}q_N)^{\>\sum_{a=1}^nz_a\</\<\ka}\,
\prod_{i=1}^{N-1}\>\Bigl(\>e^{\:\pii\,(\la_{i+1}-\la_i)}\>
\frac{q_i}{q_{i+1}}\>\Bigr)^{\sum_{a=1}^{\la^{(i)}}\:t^{(i)}_a\!\</\ka}
\times{}&
\kern-1em
\\
\notag
{}\times\<{}\prod_{i=1}^{N-1}\,\prod_{a=1}^{\la^{(i)}}\,\biggl(\,
\prod_{\satop{b=1}{b\ne a}}^{\la^{(i)}}\,
\frac1{(t_a^{(i)}\?-t_b^{(i)}\?-h)\,\pho(t_a^{(i)}\?-t_b^{(i)}\<,h,\ka)}\,
\prod_{c=1}^{\la^{(i+1)}}\<\pho(t_a^{(i)}\?-t_c^{(i+1)}\<,h,\ka)\<\biggr)\,&,
\kern-1em
\\[-20pt]
\notag
\end{align}
where $\,\la^{(N)}\?=n\,$ and $\,t^{(N)}_a\?=z_a\,$, $\;a=1\lc n\,$.
\vvn.1>
The function $\,\:\Phi_\bla(\TT\:;\zz\:;h;\qq\:;\ka)\,$ is symmetric in
the variables $\;t^{(i)}_1\?\lc t^{(i)}_{\la^{(i)}}\:$ for each \,$i=1\lc N-1$.

\subsection{Jackson integrals}
\label{secJ}
Let $\,\la\+1\<=\sum_{i=1}^{N-1}\la^{(i)}\>$. Consider the space
$\,\C^{\la\+1}\!\?\times\C^n\!\times\alb\C\times\alb\C^N$
with coordinates \,$\TT,\zz,h,\qq$\,. The lattice \,$\ka\>\Z^{\la\+1}\!$
naturally acts on this space by shifting the \,$\TT\:$-coordinates.

\vsk.2>
For a function of $\,f(\TT)$ \,and a point \,$\ss\<\in\C^{\la\+1}\!$,
define $\;\Res_{\>\TT=\ss}f(\TT)\,$ to be the iterated
\vvn.4>
residue,
\be
\Res_{\>\TT=\ss}f(\TT)\>=\,\Res_{t^{(1)}_1\<=\>s^{(1)}_1}\ldots\>
\Res_{t^{(1)}_{\la^{(1)}}\<=\>s^{(1)}_{\la^{(1)}}}\ldots\;
\Res_{t^{(N-1)}_1\<=\>s^{(N-1)}_1}\ldots\>
\Res_{t^{(N-1)}_{\la^{(N-1)}}\<=\>s^{(N-1)}_{\la^{(N-1)}}}f(\TT)\,.\kern-.5em
\ee

\vsk.4>
Let $\,L'\:$ be the complement in $\,\C^n\!\times\C\,$ of the union of
the hyperplanes
\vvn.3>
\beq
\label{zzh}
h\>=\>m\:\ka\,,\qquad z_a\<-z_b\>=\>m\:\ka\,,\qquad
z_a\<-z_b\<+h\>=\>m\:\ka\,,
\vv.3>
\eeq
for all \,$a,b=1\lc n$\,, \,$a\ne b$\,, and all $\,m\in\Z_{\le 0}\,$.
Let \,$L''\!\subset\C^N$ be the domain
\vvn.3>
\beq
\label{q/q}
|\:q_{i+1}/q_i|<1\,,\qquad i=1\lc N-1\,,
\vv.3>
\eeq
with additional cuts fixing a branch of $\,\:\log\:q_i\:$ for all \,$i=1\lc N$.
Set \,$L=L'\!\times L''\!\subset\<\C^n\!\times\C\times\C^N$.

\vsk.2>
Let $\,\lb=\bigl(\>l^{\:(1)}_{\>1}\!\lc l^{\:(1)}_{\:\la^{(1)}},\,\ldots\,,
l^{\:(N-1)}_{\>1}\!\lc l^{\:(N-1)}_{\:\la^{(N-1)}}\bigr)\in\Z^{\la\+1}\?$.
By convention, set $\,l^{\:(i)}_a\!=0\,$ for $\,a>\<\la^{(i)}$ or $\,i=N\>$.

\vsk.2>
Let $\,f(\:\TT\:;\zz\:;h;\qq)$ be a polynomial in \,$\TT$ and a holomorphic
\vv.1>
function of \,$\zz,h,\qq\,$ in \,$L$\,. For \,$(\zz\:;h;\qq)\in L$\,, define
\vvn.2>
\beq
\label{MSi}
\Mc_J(\Phi_\bla\:f)(\zz\:;h;\qq\:;\ka)\,=\!
\sum_{\lb\in\Z^{\:\la\+1}\!\!\!\!}\Res_{\>\TT\>=\>\Si_J-\:\lb\ka\>}\bigl(
\Phi_\bla(\TT\:;\zz\:;h;\qq\:;\ka)\>f(\TT\:;\zz\:;h;\qq)\bigr)\,.
\eeq
This sum is called the {\it Jackson integral over the discrete cycle\/}
$\;\Hat\Si_J\<\subset
\C^{\la\+1}\!\?\times\C^n\!\times\alb\C\times\alb\C^N\?$,
\vvn.4>
\be
\Hat\Si_J\>=\,\{(\:\Si_J-\lb\ka\:;\zz\:;h\:;\qq\:)\;\;|\,\,\,\,
\lb\<\in\Z^{\la\+1}\!,\;\;(\zz\:;h\:;\qq\:)\in L\>\}\,.
\vv.4>
\ee
Notice that the master function $\,\Phi_\bla(\TT\:;\zz\:;h;\qq\:;\ka)\,$
has only simple poles, and
% in \eqref{MSi},
\vvn.2>
\begin{align*}
& \Res_{\>\TT\>=\>\Si_J-\:\lb\ka\>}\bigl(
\Phi_\bla(\TT\:;\zz\:;h;\qq\:;\ka)\>f(\TT\:;\zz\:;h;\qq)\bigr)\,={}
\\[5pt]
&\;{}=\,f(\Si_J\<-\lb\ka\:;\zz\:;h;\qq)\,
\Res_{\>\TT\>=\>\Si_J-\:\lb\ka\>}\Phi_\bla(\TT\:;\zz\:;h;\qq\:;\ka)\,.
\\[-17pt]
\end{align*}
A closed expression for the residue
$\,\Res_{\>\TT\>=\>\Si_J-\:\lb\ka\>}\Phi_\bla(\TT\:;\zz\:;h;\qq\:;\ka)\,$
%of the master function
is given by Lemma \ref{ResPh} below.

\vsk.3>
Recall $\,c_\bla(\TT\:;h)\,$, see \eqref{cla}\:. Set
\vvn-.5>
\beq
\label{Mla}
M_\bla(\zz\:;\ka)\,=\,\prod_{i=1}^{N-1}\>\prod_{a=1}^{\la^{(i)}}\,
\prod_{b=\la^{(i)}+1}^{\la^{(i+1)}}\!
\frac{\sin\:\bigl(\pi\>(z_a\?-z_b)/\ka\bigr)}
{\pi\>e^{\:\pii\>(z_a+\:z_b)/\ka}}\;,\kern-2em
\vv-.6>
\eeq
and
\vvn-.3>
\beq
\label{Ath}
A_\bla(\TT\:;\zz\:;h;\ka)\,=\,\prod_{i=1}^{N-1}\,\prod_{a=1}^{\la^{(i)}}\;
\biggl(\,\prod_{\satop{b=1}{b\ne a}}^{\la^{(i)}}\>
\frac{\Gm\bigl(1+(t^{(i)}_b\!-t^{(i)}_a)/\ka\bigr)}
{\Gm\bigl(\<(t^{(i)}_b\!-t^{(i)}_a\!+h)/\ka\bigr)}\;
\prod_{c=1}^{\la^{(i+1)}}\:
\frac{\Gm\bigl(\<(t^{(i+1)}_c\!-t^{(i)}_a\!+h)/\ka\bigr)}
{\Gm\bigl(1+(t^{(i+1)}_c\!-t^{(i)}_a)/\ka\bigr)}\,\biggr)\>.\kern-.7em
\eeq

\begin{lem}
\label{ResPh}
If $\;\lb\not\in\<\Z_{\ge0}^{\la\+1}\?$, then
$\;\Res_{\>\TT\>=\>\Si_J-\:\lb\ka\>}\Phi_\bla(\TT\:;\zz\:;h;\qq\:;\ka)=0\,$.
For $\;\lb\in\<\Z_{\ge0}^{\la\+1}\?$,
\vvn.2>
\begin{align}
\label{ResPhlb}
& \Res_{\>\TT\>=\>\Si_J-\:\lb\ka}\Phi_\bla(\TT\:;\zz\:;h;\qq\:;\ka)\,={}
\\[7pt]
\notag
&{}=\,\frac{\ka^{\:\la\+1}\?A_\bla(\Si_J\<-\lb\ka\:;\zz\:;h;\ka)}
{M_\bla(\zz_{\si_J}\:;\ka)\>c_\bla(\Si_J\<-\lb\ka;h)}\;
\prod_{i=1}^N\>q_i^{\>\sum_{a\in J_i}\?z_a\</\<\ka}\,
\prod_{i=1}^{N-1}\,(q_{i+1}/q_i)^{\>\sum_{a=1}^{\la^{(i)}}l_a^{(i)}}\!.
\\[-20pt]
\notag
\end{align}
In particular,
\vvn-.6>
\begin{align}
\label{ResPhi}
\Res_{\>\TT\>=\>\Si_J}\?\Phi_\bla(\TT\:;\zz\:;h;\qq\:;\ka)\, &{}=\,
\frac{\bigl(\ka\>\Gm(h/\<\ka)\bigr)^{\<\la\+1}}{c_\bla(\Si_J\:;h)}\;
\prod_{i=1}^N\>\bigl(\>e^{\:\pii\,(n-\la_i)}\>q_i\:\bigr)
^{\>\sum_{a\in J_i}\?z_a\</\<\ka}\:\times{}
\\[2pt]
\notag
&\>{}\times\,\,
\prod_{i=1}^{N-1}\prod_{j=i+1}^N\,\:\prod_{a\in J_i}\,
\prod_{b\in J_j}\,\Ga\bigl(\<(z_a\<-z_b)/\ka\bigr)\,
\Ga\bigl(\<(z_b\<-z_a\<+h)/\ka\bigr)\,.\kern-1em
\\[-22pt]
\notag
\end{align}
\end{lem}

By Lemma \ref{ResPh}, the actual summation in formula \eqref{MSi} is only
over the positive cone of the lattice,
\vvn-.2>
\beq
\label{MSi<}
\Mc_J(\Phi_\bla\:f)(\zz\:;h;\qq\:;\ka)\,=\>
\sum_{\lb\in\Z_{\ge0}^{\:\la\+1}\!}f(\Si_J\<-\lb\ka\:;\zz\:;h;\qq)\,
\Res_{\>\TT\>=\>\Si_J-\:\lb\ka}\Phi_\bla(\TT\:;\zz\:;h;\qq\:;\ka)\,.
\kern-2em
\eeq

\subsection{Solutions of the dynamical and \qKZ/ equations}
For \,$J\<\in\Il$\,, define
\vvn.4>
\beq
\label{mcF}
\Psi_{\?J}(\zz\:;h;\qq\:;\ka)\,=\,
\bigl(\ka\>\Gm(h/\<\ka)\bigr)^{\<-\:\la\+1}\>\Om_\bla(h;\qq\:;\ka)\,
\sum_{I\in\Il}\,\Mc_J(\Phi_\bla W_I)(\zz\:;h;\qq\:;\ka)\,v_I\,,
\vv-.7>
\eeq
where
\vvn-.4>
\beq
\label{Oml}
\Om_\bla(h;\qq\:;\ka)\,=\,
\prod_{i=1}^{N-1}\prod_{j=i+1}^N (1-q_j/q_i)^{h\:\la_i/\<\ka}\>.\kern-1em
\vv.2>
\eeq

\begin{defn}
\label{hdef}
Say that a function $\,f(\qq)\,$ is {\it holomorphic in the unit polydisk
\vvn.06>
around\/} $\,\qq=\0\,$ if $\,f(\qq)=g(q_2/q_1\lc q_N/q_{N-1})\,$ for
\vvn.1>
a function $\,g(s_1\lc s_{N-1})\,$ holomorphic in $\,s_1\lc s_{N-1}\,$,
provided $\,|\:s_i\:|<1\,$ for all $\,i=1\lc N-1\,$.
Denote $\,f(\0)=g(0\lc 0)\,$.
\end{defn}

\noindent
In the given definition, we described homogeneous functions of $\,\qq\,$
that have behave regularly as the ratios of the subsequent coordinates
approaches zero. The symbol $\,\0\:$ used in the definition is formal
and does not represent any point of $\,\C^n$.

\begin{thm}
\label{thm cy}
The $\,\Cnnl$-valued function $\,\Psi_{\?J}(\zz\:;h;\qq\:;\ka)\:$ \,is
a solution of the joint system of dynamical differential equations \eqref{DEQ}
\vvn.3>
and \qKZ/ difference equations \eqref{Ki}. It has the form
\begin{align}
\label{PsiPsd}
\Psi_{\?J}(\zz\:;h;\qq\:;\ka)\, &{}=\,\Psd_{\?J}(\zz\:;h;\qq\:;\ka)\;
\prod_{i=1}^N\,\bigl(\>e^{\:\pii\,(n-\la_i)}\>q_i\:\bigr)
^{\>\sum_{a\in J_i}\?z_a\</\<\ka}\:\times{}
\\[3pt]
\notag
&\>{}\times\,\,\prod_{i=1}^{N-1}\prod_{j=i+1}^N\,\:\prod_{a\in J_i}\>
\biggl(\,\prod_{b\in J_j}\>\frac{\Ga\bigl(1+(z_b\<-z_a\<+h)/\ka\bigr)}
{\sin\bigl(\pi\>(z_a\<-z_b)/\ka\bigr)}\;
\prod_{\satop{c\in J_j}{c<a}}\>\frac1{z_a\<-z_c\<-h}\:\biggr)\>,\kern-.8em
\end{align}
where the function $\,\Psd_{\?J}(\zz\:;h;\qq\:;\ka)\>$ is entire in $\,\zz,h\>$
and holomorphic in $\,\qq\>$ in the unit polydisk around $\,\qq=\0\,$.
In more detail,
\vvn.2>
\begin{align}
\label{Psd}
\Psd_{\?J}(\zz\:;h;\qq\>;\ka)\,&{}=\,
\prod_{i=1}^{N-1}\prod_{j=i+1}^N\,\:\prod_{a\in J_i}\>\prod_{b\in J_j}\>
\frac{-\>\pi\>\ka}{\Gm\bigl(1+(z_b\<-z_a)/\ka\bigr)}\,\times{}
\\[3pt]
\notag
&\>{}\times\,\biggl(\Psd_{\<J,\:0}\:(\zz\:;h)\>+\!\!
\sum_{\satop{\,\mb\in\Z_{\ge0}^{N-1}\!}{\mb\ne0}}\!\?
\Psd_{\<J,\:\mb}(\zz\:;h;\ka)\,
\prod_{i=1}^{N\<-1}\>(q_{i+1}/q_i)^{m_i}\<\biggr)\:,\kern-2em
\\[-22pt]
\notag
\end{align}
where
\vvn-.4>
\begin{align}
\label{Psd0}
\Psd_{\<J,\:0}\:(\zz\:;h)\,=\,\frac1{c_\bla(\Si_J\:;h)}\;
\prod_{i=1}^{N-1}\prod_{j=i+1}^N\,\:\prod_{a\in J_i}\>
\prod_{\satop{b\in J_j}{b>a}}\>\frac1{z_a\<-z_b\<-h}\,\times{}\!\kern-2em &
\\[5pt]
\notag
{}\times\Bigl(\:W_J(\Si_J\:;\zz\:;h)\,v_J\,+\!
\sum_{\satop{I\<\in\Il}{|\si_I\<|<|\si_J\<|}}\!\!W_I(\Si_J\:;\zz\:;h)\,v_I\>\Bigr)
\kern-2em &
\end{align}
is a polynomial in $\>\zz,h\,$, and
$\,\Psd_{\<J,\:\mb}(\zz\:;h;\ka)$ for $\,\mb\ne0\>$ are rational functions
of $\,\:\zz,h,\ka\>$ with at most simple poles
\vvn.1>
on the hyperplanes $\,z_a\<-z_b\<\in\<\ka\>\Z_{>0}\>$ for $\,a\in\<J_i\,$,
$\,b\in\<J_j\,$, $\,1\le i<j\le N\>$. Furthermore, for any transposition
\,$s_{\ab}\<\in\<S_n\,$,
\vvn.3>
\beq
\label{Psdab}
\Psd_{\?J}(\zz\:;h;\qq\:;\ka)\big|_{\:z_a=z_b}=\,
\Psd_{\<s_{a\<,b}(J)}(\zz\:;h;\qq\:;\ka)\big|_{\:z_a=z_b}\:.
\vv.2>
\eeq
\end{thm}
\begin{proof}
The fact that $\,\Psi_{\?J}(\zz\:;h;\qq\:;\ka)\,$ solves the dynamical equations
\eqref{DEQ} is proved in \cite[Theorem 8.4]{TV6}\:, and the fact that
$\,\Psi_{\?J}(\zz\:;h;\qq\:;\ka)$ solves the \qKZ/ equations \eqref{Ki} is proved
in \cite[Theorem 1.5.2]{TV1}\:, cf.~\cite{TV4}\:.

\vsk.2>
Analytic properties of $\,\Psd_{\?J}(\zz\:;h;\qq\:;\ka)\,$ are proved
in Section \ref{details}. The fact that the right\:-hand side of
formula \eqref{Psd0} is a polynomial in $\,\zz,h\,$ follows from
Lemmas \ref{WIz}, \ref{WIJ}.
\end{proof}

The functions $\,\Psi_{\?J}(\zz\:;h;\qq\:;\ka)$ \,are called
the {\it multidimensional
%\,$q\:$-
hypergeometric solutions\/} of the dynamical equations. In \cite{TV5},
we constructed another type of solutions of the dynamical equations
using
%called the {\it
multidimensional hypergeometric integrals.
%solutions\/}.

\vsk.2>
The next theorem computes the determinant of coordinates of
solutions $\,\Psi_{\?J}(\zz\:;h;\qq\:;\ka)\,$ and is analogous
to \cite[Theorem~11.3]{TV6}\:.

\begin{thm}
\label{thmdet}
Let \,$n\ge 2$\,. Then
\vvn.7>
\begin{align}
\label{detPsi}
\det\:\bigl( &\>\Om_\bla(\qq\:,\ka)\,
\Mc_J(\Phi_\bla W_I)(\zz\:;h;\qq\:;\ka)\bigr)_{\IJ\in\>\Il}\<={}
\\[9pt]
\notag
&\!\<{}=\,\bigl(\ka\>\Gm(h/\<\ka)\bigr)^{\<\la\+1d_\bla}\>
\Bigl(\:e^{\:2\pii\,(n-1)\>d^{(2)}_\bla}\>\prod_{i=1}^N\,
q_i^{\:\smash{d^{(1)}_{\bla,i}}\vp|}\>\Bigr)^
{\sum_{a=1}^n z_a\</\?\ka}\times{}
\\[4pt]
\notag
&\?\times\,\prod_{a=1}^{n-1}\,\prod_{b=a+1}^n
\biggl(\:\frac{\pi\:\ka^2\>
\Gm\bigl(\<(z_a\<-z_b\<+h)/\ka\bigr)\>\Gm\bigl(1+(z_b\<-z_a\<+h)/\ka\bigr)}
{\sin\bigl(\pi\>(z_a\<-z_b)/\ka\bigr)}\>\biggr)^{\!d^{(2)}_\bla}\!\<,
\\[-12pt]
\notag
\end{align}
where $\;\la\+1\<=\sum_{i=1}^{N-1}\:(N\?-i)\>\la_i\>$,
\beq
\label{dla12}
d_\bla\>=\,\frac{n\:!}{\la_1\:!\ldots\la_N\:!}\;,\qquad
d^{(1)}_{\bla\:,\:i}\>=\,
\frac{\la_i\>(n-1)\:!}{\la_1\:!\ldots\la_N\:!}\;,\qquad
d^{(2)}_\bla\>=\,\frac{(n-2)\:!}{\la_1\:!\ldots\la_N\:!}\, %% factor 2 ??
\sum_{i=1}^{N-1}\sum_{j=i+1}^N\la_i\>\la_j\,.\kern-2em
\vv.1>
\eeq
\end{thm}
\begin{proof}
Denote by $\,F(\zz\:;\qq)\,$ the determinant in the left\:-hand side of
formula \eqref{detPsi}\:. By Theorem \ref{thm cy}, it solves the differential
equations
\vvn.5>
\be
\Bigl(\:\ka\>q_i\:\frac{\der}{\der q_i}\>-\>
\tr\:X_i(\zz\:;h;\qq)|_{\:\Cnnl}\Bigr)\>F(\zz\:;\qq)\,=\,0\,, \qquad
i=1\lc N\>,\kern-2em
\vv.4>
\ee
where $\,X_i(\zz\:;h;\qq)|_{\:\Cnnl}\,$ are the restrictions of dynamical
Hamiltonians \eqref{Xi} to the invariant subspace $\,\Cnnl$. Since
\vv.06>
$\;\tr\:X_i(\zz\:;h;\qq)|_{\:\Cnnl}=d^{(1)}_{\bla\:,\:i}\>\sum_{a=1}^n z_a\,$,
the function $\,F(\zz\:;\qq)\,$ equals the product of powers of $\,q_1\lc q_n\:$
in the right\:-hand side of formula \eqref{detPsi} multiplied by a factor that
does not depend on $\,\qq\,$. This factor can be found by taking the limit
\,$q_{i+1}/q_i\to 0\,$ for all $\,i=1\lc N-1\,$, using Theorem \ref{thm cy}.
\end{proof}

\begin{rem}
By Theorem \ref{thm cy}, the determinant $\,F(\zz\:;\qq)\,$ in
Theorem \ref{thmdet} solves the difference equations
\vvn.2>
\beq
\label{FdetK}
F(z_1\lc z_a+\ka\lc z_n;\qq)\,=\,
\det K_a(\zz\:;h;\qq\:;\ka)|_{\:\Cnnl}\,F(\zz\:;\qq)\,,\qquad a=1\lc n\>,
\vv.2>
\eeq
where $\,K_a(\zz\:;h;\qq\:;\ka)|_{\:\Cnnl}$ are the restrictions of the \qKZ/
operators \eqref{K} to the invariant subspace $\,\Cnnl$. Equations
\vvn.1>
\eqref{FdetK} determine the product of Gamma functions in the right\:-hand side
of formula \eqref{detPsi} up to a $\:\ka\:$-periodic function of $\,\zzz\,$.
\end{rem}

\subsection{Proof of Theorem \ref{thm cy}}
\label{details}
Recall $\,\la\+1\<=\sum_{i=1}^{N-1}\la^{(i)}$.
\vvn.1>
Set $\,\la\+2=\,\sum_{i=1}^{N-1}\,\bigl(\la^{(i)}\bigr)^2$ and
$\,\la_{\{2\}}=\sum_{1\le i<j\le N}\>\la_i\>\la_j\,$. Notice the homogeneity
properties
\vvn.2>
\begin{gather}
\label{homogen}
\Phi_\bla(\TT\:;\zz\:;h;\qq\:;\ka)\,=\,\ka^{\:\la\+1\<-\:\la\+2}
\Phi_\bla(\TT\:/\ka;\zz/\ka;h/\ka;\qq\:;\:1)\,,\kern-1em
\\[2pt]
\notag
\Psi_{\?J}(\zz\:;h;\qq\:;\ka)\,=\,\ka^{\:\la\+1\<+\:\la_{\{2\}}}
\Psi_{\?J}(\zz/\ka;h/\ka;\qq\:;\:1)\,.\kern-2em
\\[-15pt]
\notag
\end{gather}
To simplify writing, we assume in this section that $\,\ka=1\,$ and omit
the corresponding argument in all functions. The general case can be recovered
by the homogeneity.

\vsk.3>
Recall \,$\Imil,\,A_\bla(\TT\:;\zz\:;h)$, see \eqref{Imil}\:, \eqref{Ath}\:.
For $\,\lb\<\in\<\Z_{\ge0}^{\la\+1}\?$, define
\vvn.4>
\beq
\label{BAc}
B_{\:\lb}(\zz\:;h)\,=\,
\frac{A_\bla(\Si_{\Imil}\<-\lb\:;\zz\:;h)}{c_\bla(\Si_{\Imil}\<-\lb;h)}\;.
\vv-.5>
\eeq
Set
\vvn-.1>
\beq
\label{Zc}
\Zc_\bla\,=\,\{\>\lb\<\in\Z^{\la\+1}_{\ge 0}\,\:|\ \,
l^{\:(i)}_{\:a}\?\ge l^{\:(i+1)}_{\:a}\<,\ \;i=1\lc N-1\,,
\ \;a=1\lc\la^{(i)}\:\}\,.\kern-2em
\vv.6>
\eeq

\begin{lem}
\label{Blem}
If $\;\lb\<\not\in\Zc_\bla\,$, then $\,B_{\:\lb}(\zz\:;h)=0\,$.
For $\;\lb\<\in\Zc_\bla\,$,
\vvn.4>
\begin{align}
\label{Blz}
B_{\:\lb}(\zz\:;h)\,=\,\prod_{i=1}^{N-1}\,
\prod_{a=1}^{\la^{(i)}}\;\biggl(\,\prod_{\satop{b=1}{b\ne a}}^{\la^{(i)}}\,
\frac{\Gm(z_b-z_a\?+l^{\:(i)}_{\:a}\!\<-l^{\:(i)}_{\:b}\?+1)\,
\Gm(z_b-z_a\?+l^{\:(i)}_{\:a}\!\<-l^{\:(i+1)}_{\:b}\?+h)\,}
{\Gm(z_b-z_a\?+l^{\:(i)}_{\:a}\!\<-l^{\:(i)}_{\:b}\?+h+1)\,
\Gm(z_b-z_a\?+l^{\:(i)}_{\:a}\!\<-l^{\:(i+1)}_{\:b}\?+1)} &{}\,\times{}
\kern-.8em
\\[-2pt]
\notag
{}\times\,\frac{\Gm(\:l^{\:(i)}_{\:a}\!\<-l^{\:(i+1)}_{\:a}\?+h)}
{(\:l^{\:(i)}_{\:a}\!\<-l^{\:(i+1)}_{\:a})\:!}\;
\prod_{c=\la^{(i)}+1}^{\la^{(i+1)}}
\frac{\Gm(z_c\<-z_a\?+l^{\:(i)}_{\:a}\!-l^{\:(i+1)}_{\:c}\?+h)}
{\Gm(z_c\<-z_a\?+l^{\:(i)}_{\:a}\!-l^{\:(i+1)}_{\:c}\?+1)} &\,\biggr)\>.
\kern-.8em
\end{align}
\end{lem}

\vsk.4>
Recall $\,\zz_\si\<=(z_{\si(1)}\lc z_{\si(n)})\,$ and the permutations
\vv.1>
$\,\si_I\,$, $\,I\<\in\Il\,$, defined in Section \ref{secXo}.
Let $\,\rr=(r_1\lc r_{N-1})\,$. For $\,I,\<J\<\in\Il\,$, define
\vvn-.2>
\beq
\label{Bcf}
\Bc_{\IJ}(\zz\:;h;\rr)\,=\>\sum_{\lb\in\Zc_\bla\!}\,
B_{\:\lb}(\zz\:;h)\>W_I(\Si_{\Imil}\<-\lb\:;\zz_{\si_J^{-1}};h)\,
\prod_{i=1}^{N-1}\:r_i^{\>\sum_{a=1}^{\la^{(i)}}l_a^{(i)}}\!.\kern-1em
\vv-.4>
\eeq
Set
\vvn-.4>
\beq
\label{Mt}
\Mt_\bla(\zz\:;h)\,=\,\prod_{i=1}^{N-1}\,\prod_{a=1}^{\la^{(i)}}\,
\biggl(\sin\:(\pi\:h)\prod_{b=a+1}^{\la^{(i)}}\:
\frac{\sin\:\bigl(\pi\>(z_a\<-z_b)\bigr)}{z_a\<-z_b}
\prod_{c=\la^{(i)}+1}^{\la^{(i+1)}}\:
\sin\:\bigl(\pi\>(z_a\<-z_c\<-h)\bigr)\!\biggr)\>,\kern-.2em
\vv.3>
\eeq

\begin{prop}
\label{Bfhol}
For any $\,I,\<J\<\in\Il\,$, the function
\vvn.13>
$\,\Mt_\bla(\zz\:;h)\>\Bc_{\IJ}(\zz\:;h;\rr)\>$ is entire in $\,\zz,h\,$
and holomorphic in $\,\rr\:$ provided $\,|\:r_i\:|<1\>$ for all
$\,i=1\lc N-1\>$.
\end{prop}
\begin{proof}
By Stirling's formula, ratios of Gamma functions appearing in
formula \eqref{Blz} have the following asymptotics as $\,k\to+\:\infty\,$
over integers,
\vvn.5>
\beq
\label{Stirk}
\frac{\Gm(\al+k)}{\Gm(\bt+k)}\,=\,k^{\>\al-\bt}\>\bigl(1+o(1)\bigr)\,,\qquad
\frac{\Gm(\al-k)}{\Gm(\bt-k)}\,=\,k^{\>\al-\bt}\,
\frac{\sin\:(\pi\:\al)}{\sin\:(\pi\:\bt)}\>\bigl(1+o(1)\bigr)\,,\kern-2em
\vv.4>
\eeq
provided $\,\al\,$ and $\,\bt\,$ are not integers in the second case, and these
asymptotics can be differentiated with respect to $\,\al\,$ and $\,\bt\,$.

\vsk.2>
The right\:-hand side of formula \eqref{Bcf} allows one to present
$\,\Mt_\bla(\zz\:;h)\>\Bc_{\IJ}(\zz\:;h;\rr)\,$ as a power series in $\,\rr\,$:
\vvn-.9>
\be
\Mt_\bla(\zz\:;h)\>\Bc_I(\zz_{\si_J}\:;\zz\:;h;\rr)\,=\?
\sum_{\kk\in\Z^{N-1}_{\ge 0}}\!
\Bt_{\IJ,\:\kk}(\zz\:;h)\,\prod_{i=1}^{N-1}\,r_i^{\>k_i}\:.
\vv-.1>
\ee
Formulae \eqref{Blz}\:, \eqref{Stirk} show that for given $\,\zz,h\,$,
this series and its formal derivatives with respect to $\,\zz,h\,$ converge
provided $\,|\:r_i\:|<1\>$ for all $\,i=1\lc N-1\,$. Therefore, the function
$\,\Mt_\bla(\zz\:;h)\>\Bc_{\IJ}(\zz\:;h;\rr)\,$ is holomorphic
in $\,\rr\>$ provided $\,|\:r_i\:|<1\>$ for all $\,i=1\lc N-1\,$,
and is holomorphic in $\,\zz,h\,$ whenever all the coefficients
$\,\Bt_{\IJ,\:\kk}(\zz\:;h)\,$ are holomorphic in $\,\zz,h\,$.

\vsk.2>
It remains to show that all the coefficients $\,\Bt_{\IJ,\:\kk}(\zz\:;h)\,$
\vv.06>
are entire in $\,\zz,h\,$. To this end, it suffices to show that
for any $\,\lb\<\in\<\Zc_\bla\,$, the product $\,\Mt_\bla(\zz\:;h)\>
B_{\:\lb}(\zz\:;h)\>W_I(\Si_{\Imil}\<-\lb\:;\zz_{\si_J^{-1}};h)\,$
is an entire function of $\,\zz,h\,$.

\vsk.2>
Take $\,\lb\<\in\<\Zc_\bla\,$ and $\,a\ne b\,$.
Since $\,l^{\:(i)}_{\:a}\?\ge l^{\:(i+1)}_{\:a}$ and
$\,l^{\:(i)}_{\:b}\?\ge l^{\:(i+1)}_{\:b}$, the ratio
\vvn.2>
\be
\frac{\sin\:\bigl(\pi\>(z_a\?-z_b)\bigr)\,
\Gm(z_a\<-z_b+l^{\:(i)}_{\:b}\!\<-l^{\:(i)}_{\:a}\?+1)\>
\Gm(z_b-z_a\<+l^{\:(i)}_{\:a}\!\<-l^{\:(i)}_{\:b}\?+1)}
{(z_a\<-z_b)\,\Gm(z_a-z_b\?+l^{\:(i)}_{\:b}\!\<-l^{\:(i+1)}_{\:a}\?+1)\>
\Gm(z_b-z_a\<+l^{\:(i)}_{\:a}\!\<-l^{\:(i+1)}_{\:b}\?+1)}\kern-2em
\vv.3>
\ee
is an entire function of $\,z_a\:,z_b\,$,
and if $\,\:l^{\:(i)}_{\:b}\?>l^{\:(i+1)}_{\:b}$, then the ratio
\be
\frac{\Gm(z_b-z_a\?+l^{\:(i)}_{\:a}\!\<-l^{\:(i+1)}_{\:b}\?+h)}
{\Gm(z_b-z_a\?+l^{\:(i)}_{\:a}\!\<-l^{\:(i)}_{\:b}\?+h+1)}\kern-2em
\vv.2>
\ee
is a polynomial in $\,z_a\:,z_b\:,h\,$. Set
\vvn-.5>
\beq
\label{Fzh}
F_\lb(\zz\:;h)\,=\,\prod_{i=1}^{N-1}\,
\prod_{a=1}^{\la^{(i)}}\,\:\prod_{\satop{b=1}
{b\ne a\rlap{$\ssize{},\,l^{\:(i)}_{\:b}\?=\>l^{\:(i+1)}_{\:b}$}}}^{\la^{(i)}}
\,(z_b-z_a\<+l^{\:(i)}_{\:a}\!\<-l^{\:(i)}_{\:b}\?+h)\,.\kern-2em
\vv-.2>
\eeq
Then by formulae \eqref{Blz}\:, \eqref{Mt}\:, the product
\vv.1>
$\,\Mt_\bla(\zz\:;h)\alb\>B_{\:\lb}(\zz\:;h)\>F_\lb(\zz\:;h)\,$ is
an entire function of $\,\zz,h\,$. Since by Lemma \ref{lemW} below, the ratio
$\,W_I({\Si_{\Imil}\<-\lb\:;{}}\alb\zz_{\si_J^{-1}};h)/F_\lb(\zz\:;h)\,$
is a polynomial in $\,\zz,h\,$, the product $\,\Mt_\bla(\zz\:;h)\>
B_{\:\lb}(\zz\:;h)\>W_I(\Si_{\Imil}\<-\lb\:;\zz_{\si_J^{-1}};h)\,$
is an entire function of $\,\zz,h\,$ too. Proposition \ref{Bfhol} is proved.
\end{proof}

\begin{lem}
\label{lemW}
Let $\,F_\lb(\zz\:;h)\>$ be given by \eqref{Fzh}\:. Then the ratio
$\,W_I({\Si_{\Imil}\<-\lb\:;{}}\alb\zz_{\si_J^{-1}};h)/F_\lb(\zz\:;h)\,$
is a polynomial in $\,\zz,h\,$.
\end{lem}
\begin{proof}
The proof is by inspection, similarly to the proof of Lemma \ref{WIJ}.
\end{proof}

\begin{proof}[Proof of Theorem \ref{thm cy}]
Recall that we assume $\,\ka=1\,$ and skip the corresponding argument
in all functions. By formulae \eqref{MSi<}\:, \eqref{BAc}\,--\,\eqref{Bcf}\:,
and Lemmas \ref{ResPh}, \ref{Blem},
\vvn.5>
\begin{align}
\label{MJfB}
\Mc_J(\Phi_\bla\:W_I){}&{}(\zz\:;h;\qq)\,=\,
\Bc_{\IJ}(\zz_{\si_J};h;q_2/q_1\lc q_N/q_{N-1})\>\times{}
\\[5pt]
\notag
& {}\,\times\,\prod_{i=1}^N\>\bigl(\>e^{\:\pii\,(n-\la_i)}\>q_i\:\bigr)
^{\>\sum_{a\in J_i}\?z_a}\:\prod_{i=1}^{N-1}\prod_{j=i+1}^N\,\prod_{a\in J_i}
\,\prod_{b\in J_j}\>\frac{\pi}{\sin\bigl(\pi\>(z_a\<-z_b)\bigr)}\;.
\\[-13pt]
\notag
\end{align}
Then by formulae \eqref{mcF}\:, \eqref{PsiPsd}\:, \eqref{MJfB}\:, we have
\vvn.6>
\begin{align}
\label{yaPsd}
\Psd_{\?J}(\zz\:;h;\qq)\,&{}=\>\sum_{I\in\Il}\:
\Bc_{\IJ}(\zz_{\si_J};h;q_2/q_1\lc q_N/q_{N-1})\,v_I\,\times{}
\\
\notag
&{}\>\times\,\frac{\Om_\bla(\qq)}{\bigl(\Gm(h)\bigr)^{\la\+1}}\,
\prod_{i=1}^{N-1}\prod_{j=i+1}^N\,\:\prod_{a\in J_i}\>
\biggl(\,\prod_{b\in J_j}\>\frac\pi{\Ga(z_b\<-z_a\<+h+1)}\;
\prod_{\satop{c\in J_j}{c<a}}(z_a\<-z_c\<-h)\?\biggr).\kern-.6em
\\[-14pt]
\notag
\end{align}
By Proposition \ref{Bfhol}, the function in the right-hand side of formula
\eqref{yaPsd} is holomorphic in $\,\qq\,$ in the unit polydisk around $\,\0\,$,
and may have poles in $\,\zz,h\,$ at most at the hyperplanes
$\,z_a\<-z_b\<\in\Z_{\:\ne\:0}\,$ for $\,a\ne b\,$,
$\,z_a\<-z_b\<-h\<\in\Z_{\ge 0}\,$ for $\,a\in J_i\,$, $\,b\in J_j\,$,
$\,i<j\,$, and $\,h\in\Z_{>0}\,$. To complete the proof of
Theorem \ref{thm cy}, it remains to show that those poles
do not actually occur.

\vsk.2>
For the hyperplanes $\,z_a\<-z_b\<\in\Z_{\:\ne\:0}\,$, $\,a\ne b\,$,
\vvn.1>
we will show the regularity of the function
$\,M_\bla(\zz_{\si_J}\<)\>\Psi_{\?J}(\zz\:;h;\qq)\,$.
\vv.1>
Observe that $\,M_\bla(\zz_{\si_J}\<)\>\Psi_{\?J}(\zz\:;h;\qq)\,$ is regular
at the hyperplanes $\,z_a\<=z_b\,$ for all $\,a\ne b\,$, and solves \qKZ/
difference equations \eqref{Ki}\:. Since all the \qKZ/ operators
\vv.1>
$\,K_1(\zz\:;h;\qq)\lc K_n(\zz\:;h;\qq)\,$ and their inverses are regular
at the hyperplanes $\,z_a\<-z_b\<\in\Z\,$ for $\,a\ne b\,$, the function
$\,M_\bla(\zz_{\si_J}\<)\>\Psi_{\?J}(\zz\:;h;\qq)\,$ is regular at all hyperplanes
$\,z_a\<-z_b\<\in\Z\,$ for $\,a\ne b\,$.

\vsk.2>
To deal with the hyperplanes $\,z_a\<-z_b\<-h\<\in\Z_{\ge 0}\,$ and
$\,h\in\Z_{>0}\,$, we will show that for given $\,\zz\,$, the function
$\,\Psd_{\?J}(\zz\:;h;\qq)\,$ is entire in $\,h\,$, and it suffices to do it
for generic $\,\zz\,$.

\vsk.2>
To simplify writing, we will omit the argument $\,\zz\,$
in all functions for a while. Let $\,\rr=(r_1\lc r_{N-1})$ and
$\,\rr_{\!*}=\:(1,r_1\:,\:r_1r_2\>\lc\:r_1\ldots r_{N-1})\,$.
\vvn.1>
Denote $\,F(h;\rr)=\Psd_{\?J}(h;\rr_{\!*})\,$, so that
$\,\Psd_{\?J}(h;\qq)=F(h;q_2/q_1\lc q_N/q_{N-1})\,$.

\vsk.2>
Recall the dynamical operators $\,X_1\lc X_N\,$, see \eqref{Xi}\:. Set
\vvn-.1>
\be
X_{(i)}(h;\rr)\,=\,\sum_{j=i+1}^N\Bigl(X_j(h;\rr_{\!*})\:-
\sum_{a\in J_j\!}\,z_a\:\Bigr)\Big|_{\:\Cnnl}\,,\qquad i=1\lc N-1\,.\kern-2em
\vv.1>
\ee
The dynamical differential equations \eqref{DEQ} for $\,\Psi_{\?J}(h;\qq)\,$
are equivalent to the following equations for $F(h;\rr)$,
\vvn-.3>
\beq
\label{Feq}
r_i\:\frac\der{\der\:r_i}\>F(h;\rr)\,=\,X_{(i)}(h;\rr)\>F(h;\rr)\,,\qquad
i=1\lc N-1\,.\kern-2em
\eeq

\vsk.4>
The operators $\,X_{(1)}(h;\rr)\lc X_{(N-1)}(h;\rr)\,$ are linear functions in
\vvn.1>
$\,h\,$ and rational functions in $\,\rr\,$, regular provided $|\:r_i\:|<1\,$
for all $\,i=1\lc N-1\,$.
\vv.08>
The eigenvalues of the operator $\,X_{(i)}(h;0\:)\,$ for a given $\,i\,$ are
\vv.08>
$\,\sum_{j=i+1}^N\>\bigl(\>\sum_{a\in I_j\!}z_a-
\sum_{a\in J_j\!}z_a\:\bigr)\,$, $\;I\<\in\Il\,$. Hence, one of the eigenvalues
of $\,X_{(i)}(h;0\:)\,$ equals zero and all other eigenvalues are not integers
\vv.05>
for generic $\,\zz\,$. Therefore, a solution $\,F(h;\rr)\,$ of equations
\eqref{Feq} holomorphic $|\:r_i\:|<1\,$ for all $\,i\,$, is uniquely determined
by the value $\,F(h;0)\,$, and $\,F(h,\rr)\,$ is holomorphic in $\,h\,$
whenever $\,F(h;0)\,$ is. The value $\,F(h;0)=\Psd(h;\0)\,$ can be found from
formulae \eqref{mcF}\:, \eqref{ResPhi}\:, \eqref{Blz}\:, \eqref{Bcf}\:,
\eqref{MJfB}\:, \eqref{yaPsd}\:. It is a polynomial in $\,h\,$
by Lemmas \ref{WIz}, \ref{WIJ}. Therefore, $\,F(h,\rr)\,$ is
an entire function of $\,h\,$, and so is $\,\Psd(h;\qq)\,$.

\vsk.2>
Formulae \eqref{Psd}\:, \eqref{Psd0}, \eqref{Psdab}\:, and the rationality
of $\,\Psd_{\<J,\:\mb}(\zz\:;h)\,$ follow from formulae \eqref{mcF}\:,
\eqref{BAc}\,--\,\eqref{Bcf}\:, \eqref{MJfB}\:, \eqref{yaPsd}\:,
and Lemma \ref{WIz}. The properties of poles of
$\,\Psd_{\<J,\:\mb}(\zz\:;h)\,$ are determined by the analytic properties of
\,$\Psd_{\?J}(\zz\:;h;\qq)\,$. The fact that $\,\Psd_{\<J,\:0}\:(\zz\:;h)\,$
is a polynomial in $\,\zz,h\,$ follows from Lemma \ref{WIJ}.
\end{proof}

\subsection{Solutions of the dynamical and \qKZ/ equations parametrized
by Laurent polynomials}
\label{seclaur}
In the sequel, we will use the following notation. For any variable $\,x\,$,
there is the companion $\,\xdd\,$, and for any function $\,f(\xdd)\,$,
we set $\,\fdd(x;\ka)=f(e^{\:2\:\pii\,x/\<\ka})\,$.
The convention for collections of variables is similar. For instance,
\vv.1>
$\,\zzd=(\zdd_1\lc\zdd_n)\,$ and $\,\fdd(\zz\:;\ka)=
f(e^{\:2\:\pii\,z_1\</\<\ka}\lc e^{\:2\:\pii\,z_n\</\<\ka})\,$.

\vsk.2>
%Let $\,\yy=(\yyy)\,$.
Introduce the variables $\gm_{\ij}\,$, $\,i=1\lc N$, $\;j=1\lc\la_i\,$. Denote
\vvn.2>
\beq
\label{GG}
\GG\>=\,(\gm_{1,1}\lc\gm_{1,\>\la_1}\:\lc\gm_{N,1}\lc\gm_{N\?,\>\la_N})\,,\qquad
\GGd^{\pm1}=\,(\gmd_{1,1}^{\pm1}\lc\gmd_{N\?,\>\la_N}^{\pm1})\,.\kern-2em
\vv.2>
\eeq
Recall $\,\zz_\si\<=(z_{\si(1)}\lc z_{\si(n)})\,$, and the permutations
$\,\si_I\,$, $\,I\<\in\Il\,$, defined in Section \ref{secXo}.
For a Laurent polynomial $\,P(\:\GGd\:;\zzd;\hdd)\,$, set
\vvn.4>
\beq
\label{PPI}
\Psi_P(\zz\:;h;\qq\:;\ka)\,=\,
\sum_{J\in\Il}\,\Pdd(\zz_{\si_J};\zz\:;h;\ka)\,\Psi_{\?J}(\zz\:;h;\qq\:;\ka)\,,
\kern-2em
\vv.1>
\eeq
where the functions $\,\Psi_{\?J}(\zz\:;h;\qq\:;\ka)\,$,
are given by \eqref{mcF}.
\begin{prop}
\label{PsiPsol1}
The function $\,\Psi_P(\zz\:;h;\qq\:;\ka)\>$ is a solution of the joint system
of dynamical differential equations \eqref{DEQ} and \qKZ/ difference equations
\eqref{Ki}\:.
\end{prop}
\begin{proof}
By Theorem~\ref{thm cy}, the function $\,\Psi_P(\zz\:;h;\qq\:;\ka)\,$
solves the system of equations \eqref{DEQ} and \eqref{Ki} since
$\,\Pdd(\zz_{\si_J};\zz\:;h;\ka)\,$ does not depend on $\,\qq\,$ and is
$\,\ka\:$-periodic function in each of the variables $\,\zzz\,$.
\end{proof}

Let $\,\C[\:\GGd^{\pm1}]^{\:S_\bla}$ be the space of Laurent polynomials in
$\,\GGd$ symmetric in $\,\gmd_{i,1}\lc\gmd_{i,\>\la_i}$ for each $\,i=1\lc N\>$.

\begin{prop}
\label{PsiPsol2}
For any
$\>P\?\in\C[\:\GGd^{\pm1}]^{\:S_\bla}\?\ox\C[\:\zzd^{\pm1}\?,\hdd^{\pm1}]\,$,
the function $\,\Psi_P(\zz\:;h;\qq\:;\ka)\>$ is holomorphic in $\,\zz,h,\qq\,$
provided $\,z_a\<-z_b\<+h\not\in\<\ka\>\Z_{\le0}\,$ for all $\,a,b=1\lc n\,$,
$\,a\ne b\,$, and $\,\:|\:q_{i+1}/q_i|<1\,$ for all $\,i=1\lc N-1\,$,
\vvn.1>
with a branch of $\,\:\log\:q_i\:$ fixed for each \,$i=1\lc N$.
The singularities of $\,\Psi_P(\zz\:;h;\qq\:;\ka)\>$ at the hyperplanes
$\;z_a\<-z_b\<+h\in\<\ka\>\Z_{\le0}\,$ are simple poles.
\end{prop}
\begin{proof}
By Theorem~\ref{thm cy}, we need only to show that
$\,\Psi_P(\zz\:;h;\qq\:;\ka)\,$ is regular at the hyperplanes
$\,z_a\<-z_b\<\in\<\ka\>\Z\,$, $\,a\ne b\,$, where it might have simple poles.

\vsk.2>
The regularity of $\,\Psi_P(\zz\:;h;\qq\:;\ka)\,$ at the hyperplanes
$\,z_a\<=z_b\,$ for all $a\ne b$ follows from formula \eqref{Psdab}\:.
Since $\,\Psi_P(\zz\:;h;\qq\:;\ka)\,$ solves \qKZ/ difference
equations \eqref{Ki} and all the \qKZ/ operators
$\,K_1(\zz\:;h;\qq\:;\ka)\lc K_n(\zz\:;h;\qq\:;\ka)\,$ and their inverses
are regular at the hyperplanes $\,z_a\<-z_b\<\in\<\ka\>\Z\,$, the function
$\,\Psi_P(\zz\:;h;\qq\:;\ka)\,$ is regular at all hyperplanes
$\,z_a\<-z_b\<\in\<\ka\>\Z\,$, $\,a\ne b\,$.
\end{proof}

Denote by $\>\Srsl\:$ the space of solutions of the system of dynamical
\vv.1>
differential equations \eqref{DEQ} and \qKZ/ difference equations \eqref{Ki}
\vv.06>
spanned over $\,\C\,$ by the functions $\,\Psi_P(\zz\:;h;\qq\:;\ka)\,$,
$\>P\?\in\C[\:\GGd^{\pm1}]^{\:S_\bla}\?\ox\C[\:\zzd^{\pm1}\?,\hdd^{\pm1}]\,$.
The space $\>\Srsl\:$ is a $\,\C[\:\zzd^{\pm1}\?,\hdd^{\pm1}]\:$-\:module
with $f(\zzd\:;\hdd)\,$ acting as multiplication by $\fdd(\zz\:;h;\ka)\,$.

\vsk.3>
Define the algebra
\vvn-.4>
\beq
\label{Krel}
\Kc_\bla\:=\,\C[\:\GGd^{\pm1}]^{\:S_\bla}\?\ox
\C[\:\zzd^{\pm1}\?,\hdd^{\pm1}]\>\Big/\Bigl\bra
\,\prod_{i=1}^N\prod_{j=1}^{\la_i}\,(u-\gmd_{\ij})\,=\,
\prod_{a=1}^n\,(u-\zdd_a)\Bigr\ket\,,\kern-1.6em
\vv.1>
\eeq
where $\,u\,$ is a formal variable. By \eqref{PPI}\:,
\vvn.2>
the assignment $\,P\mapsto\Psi_P\>$ defines a homomorphism
\beq
\label{muk}
\muk:\:\Kc_\bla\to\:\Srsl\,,\qquad Y\<\mapsto\Psi_Y\>,
\eeq
of $\,\C[\:\zzd^{\pm1}\?,\hdd^{\pm1}]\:$-\:modules.

\vsk.2>
By Propositions \ref{PA1}, \ref{PA2}, the algebra $\,\Kc_\bla\,$ is
a free $\,\C[\:\zzd^{\pm1}\?,\hdd^{\pm1}]\:$-\:module generated
by the classes
\vvn-.4>
\beq
\label{YI}
Y_I(\:\GG)\,=\>V_I(\gmd_{1,1}^{-1}\lc\gmd_{1,\>\la_1}^{-1}\:\lc
\gmd_{N\?,1}^{-1}\lc\gmd_{N\?,\>\la_N}^{-1})\,,
\qquad I\<\in\Il\,,\kern-2em
\vv.2>
\eeq
where the polynomials $\,V_{\<I}\>$ are defined by formula \eqref{VIx}\:.
Introduce the coordinates of solutions $\,\Psi_{Y_I}\>$:
\beq
\label{PsiYI}
\Psi_{Y_I}\<(\zz\:;h;\qq\:;\ka)\,=\>
\sum_{J\in\Il}\,\Psb_{\IJ}(\zz\:;h;\qq\:;\ka)\,v_J\,.
\vv-.2>
\eeq
\begin{thm}
\label{detY}
Let \,$n\ge 2$\,. Then
\begin{align}
\label{detPsiY}
\,\det\:\bigl(\>\Psb_{\IJ}(\zz\:;h;\qq\:;\ka)\bigr)_{\IJ\in\>\Il}\<=
\,\Bigl(\:e^{\:\pii\,(n-1)\>d^{(2)}_\bla}\>\prod_{i=1}^N\,
q_i^{\:\smash{d^{(1)}_{\bla,i}}\vp|}\>\Bigr)^{\sum_{a=1}^n z_a\</\<\ka}\,\,
\prod_{j=2}^{n-1}\,j^{\:(n-j)\>d^{(2)}_\bla}\times{}\!\? &
\\
\notag
{}\times\,\prod_{a=1}^{n-1}\,\prod_{b=a+1}^n\bigl(\:2\:\piit\;\ka^2\,
\Gm\bigl(\<(z_a\<-z_b\<+h)/\ka\bigr)\>\Gm\bigl(1+(z_b\<-z_a\<+h)/\ka\bigr)
\bigr)^{\>d^{(2)}_\bla}\!\<, &
\\[-15pt]
\notag
\end{align}
where $\;d^{(1)}_{\bla\:,\:i}\,,\,d^{(2)}_\bla$ are given by
formulae \,\eqref{dla12}\:.
\end{thm}
\begin{proof}
The statement follows from Theorem \ref{thmdet} and formula \eqref{detVI}\:.
\end{proof}

\begin{cor}
\label{muk=}
The map $\,\muk:\:\Kc_\bla\<\to\Srsl\>$ is an isomorphism of
$\;\C[\:\zzd^{\pm1}\?,\hdd^{\pm1}]\:$-\:modules.
\end{cor}

\begin{rem}
The algebra $\,\Kc_\bla\:$ is the equivariant $\>K\?$-theory algebra
$\,K_{T\<\times\Cxs}\<(\tfl\>;\C)\>$ of the cotangent bundle of the partial
flag variety $\,\Fla\>$, see the notation in Section \ref{sQde}.
\end{rem}

\subsection{Levelt fundamental solution}
\label{secLev}

Recall Definition \ref{hdef} of a function $f(\qq)\,$ holomorphic in the unit
polydisk around $\,\0\,$.
%if $\,f(\qq)=g(q_2/q_1\lc q_N/q_{N-1})\,$ for a
%function $\,g(s_1\lc s_{N-1})\,$ holomorphic in the polydisk
%$\{\:|\:s_1\:|<1\:\lc|\:s_{N-1}\:|<1\:\}$, and $\,f(\0)=g(0\lc0)\,$.
%\vsk.3>
The dynamical Hamiltonians $\,X_1(\zz\:;h;\qq)\lc X_n(\zz\:;h;\qq)\,$ given by
\eqref{Xi} are holomorphic in $\,\qq\,$ in the unit polydisk around $\,\0\,$ and
\vvn.2>
\beq
\label{Xi0}
X_i(\zz\:;h;\0)\,=\,
\sum_{a=1}^n\,z_a\:e^{(a)}_{\ii}-\>h\?\sum_{1\le a<b\le n}\?\biggl(\,
\sum_{j=1}^{i-1}\,e_{\ji}^{(a)}\:e_{\ij}^{(b)}\,-\?
\sum_{j=i+1}^N\?e_{\ij}^{(a)}\:e_{\ji}^{(b)}\>\biggr)\:.
\kern-.2em
\eeq
Notice that for $\,I\<\in\Il\,$,
\vvn.1>
\be
X_i(\zz\:;h;\0)\,v_I\>=\>\sum_{\:a\in I_i}z_a\:v_I\:+\!\<
\sum_{\satop{J\in\:\Il}{|\si_J\<|<|\si_I\<|}}\!\xi_{\:i,\:\IJ}\,v_J\,,
\ee
where the coefficients $\,\xi_{\:i,\:\IJ}\>$ take values $\;0\:,\pm\>h\,$.
\vv.1>
Therefore, the eigenvalues of the restriction of the operator
$\,X_i(\zz\:;h;\0)\,$ on $\,\Cnnl\,$ are $\,\sum_{\:a\in I_i}\<z_a\,$,
$\,I\<\in\Il\,$. A more detailed statement is given by Proposition \ref{egvX0}
below.

\vsk.2>
Recall the function $\,\Psd_{\<I\<,\:0}\:(\zz\:;h)\,$, $\,I\<\in\Il\,$,
given by \eqref{Psd0}\:.

\begin{prop}
\label{egvX0}
Given $\,I\<\in\Il\,$, we have
\vvn.16>
$\;X_i(\zz\:;h;\0)\,\Psd_{\<I\<,\:0}\:(\zz\:;h)\:=\:
\sum_{\:a\in I_i}\?z_a\,\Psd_{\<I\<,\:0}\:(\zz\:;h)\,$, and
$\,\Psd_{\<I\<,\:0}\:(\zz\:;h)\ne 0\>$ provided $\,z_a\<\ne z_b\,$ for all pairs
$\,a,b\>$ such that $\,a<b\>$ and $\,\si_I^{-1}(a)>\si_I^{-1}(b)\,$.
\end{prop}
\begin{proof}
The first part of the statement follows from Theorem \ref{thm cy}.
The nonvanishing of $\,\Psd_{\<I\<,\:0}\:(\zz\:;h)\,$ is implied by
formula \eqref{Psd0} and Lemmas \ref{WIz}, \ref{WIJ}.
\end{proof}

For $\,I\<\in\Il\,$, set
$\,\EE_I(\zz)=\bigl(E^{(1)}_I(\zz)\lc E^{(N\<-1)}_I(\zz)\bigr)\,$,
\vv.06>
where $\,E^{(i)}_I\<(\zz)=\sum_{j=1}^i\:\sum_{\:a\in I_j}\<z_a\,$
is the eigenvalue of
%the operators
$\,X_1(\zz\:;h;\0)\lsym+X_i(\zz\:;h;\0)\,$ on $\,\:\Psd_{\<I\<,\:0}\:(\zz\:;h)\,$.
For $\,I\<,J\<\in\Il\,$, denote by \,$D_{\IJ}$ the set of points $\,\zz\,$
such that $\EE_I(\zz)-\EE_J(\zz)\in\Z^{N\<-1}_{\ge 0}$ and
\vvn.07>
$\,\EE_I(\zz)\ne\EE_J(\zz)\,$.
Set $\,D_\bla\<=\:\bigcup_{\:\IJ\in\:\Il}D_{\IJ}\,$.

\begin{thm}
\label{Bthm}
{\rm(\:i\:)}\enspace
For any $\,\zz,h\>$ such that $\,z_a\<-z_b\not\in\<\ka\>\Z_{\:\ne\:0}\>$
\vvn.1>
for all $\,\:1\le a<b\le n\,$, there exists
an $\,\:\End\:\bigl(\Cnnl\bigr)$-\:valued function $\,\:\Psp(\zz\:;h;\qq)\,$,
\vvn.1>
holomorphic in $\,\qq\:$ in the unit polydisk around $\,\0\,$, such that
$\,\:\Psp(\zz\:;h;\0)\>$ is the identity operator and the function
\vvn.2>
\beq
\label{Psipnd}
\Psh(\zz\:;h;\qq\:;\ka)\,=\,\Psp(\zz\</\<\ka\:;h/\ka\:;\qq)
\,\:\prod_{i=1}^N\,q_i^{\>\smash{X_i(\zz\</\<\ka\:;\:h\</\<\ka\:;\>\0)}}\>,
\kern-1.2em
\vv.1>
\eeq
solves dynamical differential equations \eqref{DEQ}.
\vv.1>
For given $\,\zz,h\,$, the function $\,\:\Psp(\zz\:;h;\qq)\:$ with the specified
properties is unique if and only if $\,\zz\not\in\<\ka\:D_\bla\,$.
Furthermore, $\,\,\det\:\Psp(\zz\:;h;\qq)=1\,$ and
\vvn-.3>
\beq
\label{detPsh}
\det\:\Psh(\zz\:;h;\qq\:;\ka)\,=\,\prod_{i=1}^N\,
q_i^{\:\smash{d^{(1)}_{\bla,i}}\:\sum_{a=1}^n z_a\</\<\ka}\,,\kern-1.4em
\eeq
where $\;d^{(1)}_{\bla\:,\:1}\:\lc d^{(1)}_{\bla\:,\:N}$ are given
by formula \,\eqref{dla12}\:.
\vsk.3>
\noindent
{\rm(\:ii\:)}\enspace
Define the function $\,\Psp(\zz\:;h;\qq)$ for generic $\,\zz,h\>$ as in item
\vvn.1>
\>{\rm(\:i\:)}\:. Then $\,\Psp(\zz\:;h;\qq)$ is holomorphic in $\,\zz\>$
if $\,{z_a\<-z_b\not\in\Z_{\:\ne\:0}}\>$ for all $\,\:{1\le a<b\le n}\,$,
\vvn.1>
and is entire in $\,h\,$. The singularities of $\;\Psp(\zz\:;h;\qq)$
at the hyperplanes $\,z_a\<-z_b\in\Z_{\:\ne\:0}\>$ are simple poles.
\end{thm}

The theorem is proved in Section \ref{pfBthm}. An explicit expression for
the function $\,\:\Psp(\zz\:;h;\qq)\,$ is given by formula \eqref{PspIJ}\:.

\vsk.2>
Following \cite[Chapter 2]{AB}\:, we will call
\vvn.1>
$\,\Psh(\zz\:;h;\qq\:;\ka)\,$ the {\it Levelt fundamental solution\/} of
dynamical differential equations~\eqref{DEQ} on $\:\Cnnl\:$, see also
\cite[Section~6.2]{CV}\:.

\vsk.3>
For a $\,\Cnnl$-valued solution $\,\Psi(\zz\:;h;\qq\:;\ka)\,$ of dynamical
differential equations \eqref{DEQ}\:, define its {\it principal term\/}
\vvn-.1>
\beq
\label{princ}
\Psz(\zz\:;h;\ka)\,=\,\Psh(\zz\:;h;\qq\:;\ka)^{-1}\,\Psi(\zz\:;h;\qq\:;\ka)\,.
\kern-1em
\vv.5>
\eeq
By Theorem \ref{Bthm}, the principal term does not depend on $\,\qq\,$.

\vsk.3>
Set
\vvn-.7>
\beq
\label{CGG}
C_\bla(\zz\:;\ka)\,=\,\prod_{i=1}^N\>e^{\:\pii\,(n-\la_i)
\sum_{a=\smash{\la^{\?(i-1)}}\<+1}^{\la^{(i)}}z_a\</\<\ka}\kern-2em
\vv-.5>
\eeq
and
\vvn-.2>
\beq
\label{GGG}
G_\bla(\zz\:;h;\ka)\,=\,\prod_{i=1}^{N-1}\>\prod_{a=1}^{\la^{(i)}}\,
\prod_{b=\la^{(i)}+1}^{\la^{(i+1)}}\!
\Gm\bigl(\<(z_a\<-z_b)/\ka\bigr)\,\Gm\bigl(\<(z_b\<-z_a\<+h)/\ka\bigr)\,.\kern-.5em
\vv.3>
\eeq

\begin{prop}
\label{PsiPpr}
For a Laurent polynomial $\,P(\:\GGd\:;\zzd;\hdd)\,$, the principal term of
the solution $\,\Psi_P(\zz\:;h;\qq\:;\ka)\,$, given by \eqref{PPI}\:, equals
\vvn.2>
\beq
\label{PsiP0}
\Psz_{\?P}(\zz\:;h;\ka)\,=\?\sum_{\IJ\in\:\Il}\!\Pdd(\zz_{\si_J};\zz\:;h;\ka)
\,C_\bla(\zz_{\si_J};\ka)\,G_\bla(\zz_{\si_J};h;\ka)\,
\frac{W_I(\Si_J\:;\zz\:;h)}{c_\bla(\Si_J\:;h)}\;v_I\,.
\kern-1em
\vv.3>
\eeq
Here $\,c_\bla(\Si_J\:;h)\,$ is given by formula \eqref{cla}\:.
\end{prop}

The proposition is proved in Section \ref{pfBthm}.

\subsection{Proofs of Theorem \ref{Bthm} and Proposition \ref{PsiPpr}}
\label{pfBthm}
To simplify writing, we assume in this section that $\,\ka=1\,$, similarly
to Section \ref{details}, and omit the corresponding argument in all functions.
The general case can be recovered by the homogeneity properties
\eqref{homogen}\:.

\begin{lem}
\label{lemA}
Given $\,A\in\End\:\bigl(\Cnnl\bigr)\,$, assume that the function
\be
F_A(\zz\:;h;\qq)\,=\,\prod_{i=1}^N\,q_i^{\>\smash{X_i\<(\zz;h;\:\0)}}\,\:A\;
\prod_{i=1}^N\,q_i^{\smash{-\<X_i\<(\zz;h;\:\0)}}
\ee
is holomorphic in $\,\qq\,$ in the unit polydisk around $\,\0\,$,
and $\,F_A(\zz\:;h;\0)=1\,$, the identity operator. Then $\,A=1\>$ and
$\,F_A(\zz\:;h;\qq)=1\,$, provided $\,\zz\not\in D_\bla\,$. On the other hand,
if $\,\zz\in D_\bla\,$, then there exists $\,A\ne 1\>$ such that the function
$\,F_A(\zz\:;h;\qq)$ has the requested properties and is not a constant function
of $\,\qq\,$.
\end{lem}
\begin{proof}
The statement follows from a basic linear algebra reasoning.
\end{proof}

\begin{proof}[Proof of Theorem \ref{Bthm}]
For the uniqueness statement, assume that the function $\,\Psp(\zz\:;h;\qq)\,$
is as required in Theorem \ref{Bthm}. Then the function
$\,\Psp_1(\zz\:;h;\qq)\,$ has the same requested properties
if and only if the product
\vvn.1>
\be
E(\zz\:;h)\,=\,\prod_{i=1}^N\,q_i^{\smash{-\<X_i\<(\zz;h;\:\0)}}\,
\Psp(\zz\:;h;\qq)^{-1}\>\Psp_1(\zz\:;h;\qq)\,\:
\prod_{i=1}^N\,q_i^{\>\smash{X_i\<(\zz;h;\:\0)}}
\vv.1>
\ee
does not depend on $\,\qq\,$, whilst the function
\vvn.2>
\be
F(\zz\:;h;\qq)\,=\,\prod_{i=1}^N\,q_i^{\>\smash{X_i\<(\zz;h;\:\0)}}\>
E(\zz\:;h)\,\prod_{i=1}^N\,q_i^{\smash{-\<X_i\<(\zz;h;\:\0)}}\>=\,
\Psp(\zz\:;h;\qq)^{-1}\>\Psp_1(\zz\:;h;\qq)\,.\kern-.6em
\vv.2>
\ee
is holomorphic in $\,\qq\,$ in the unit polydisk around $\,\0\,$,
and $\,F(\zz\:;h;\0)\,$ is the identity operator.
Thus the uniqueness statement follows from Lemma \ref{lemA}.

\vsk.2>
Recall the functions $\,\Psd_{\!J}(\zz\:;h;\qq)\,$ and
$\,\Psd_{\<J,\:0}\:(\zz\:;h)\,$, see \eqref{Psd}\:, \eqref{Psd0}\:.
By Lemma \ref{WIz},
\vvn.2>
\beq
\label{Psd0=}
\Psd_{\<J,\:0}\:(\zz\:;h)\,=\,\prod_{i=1}^{N-1}\prod_{j=i+1}^N\,\:
\prod_{a\in J_i}\>\prod_{\satop{b\in J_j}{b>a}}\:\frac1{z_a\<-z_b\<-h}\;
\sum_{I\<\in\Il}\>\frac{W_I(\Si_J\:;\zz\:;h)}{c_\bla(\Si_J\:;h)}\;v_I\,.
\kern-2em
\vv.1>
\eeq
Let $\,\Psp(\zz\:;h;\qq)\,$ to be the operator such that for any
$\,I\<\in\Il\,$,
\vvn.2>
\begin{align}
\label{PspI}
\Psp(\zz\:;h;\qq)\::\:v_I\,\mapsto\>\sum_{J\in\Il}\,&\Psd_{\!J}(\zz\:;h;\qq)
\,\prod_{i=1}^{N-1}\prod_{j=i+1}^N\,\:\prod_{a\in J_i}\>\prod_{b\in J_j}\>
\frac{\Gm(z_b\<-z_a)}\pi\,\times{}
\\[4pt]
\notag
{}\times\,{}&\frac{\WW_I(\Si_J\:;\zz\:;h)}{c_\bla(\Si_J\:;h)}\;
\prod_{i=1}^{N-1}\prod_{j=i+1}^N\,\:\prod_{a\in J_i}\>
\prod_{\satop{b\in J_j}{b<a}}\>\frac1{z_a\<-z_b\<-h}\kern-1.5em
\\[-16pt]
\notag
\end{align}
We will verify below that the function $\,\Psp(\zz\:;h;\qq)\,$ is as required
in Theorem \ref{Bthm}.

\vsk.3>
By Theorem \ref{thm cy}, the functions $\,\Psd_{\!J}(\zz\:;h;\qq)\,$ are
\vv.1> holomorphic in $\,\qq\,$ in the unit polydisk around $\,\0\,$, hence
\vv.04>
so does $\,\Psp(\zz\:;h;\qq)\,$. Then $\,\Psp(\zz\:;h;\0)\,$ is the identity
operator by formulae \eqref{Psd}\:, \eqref{Psd0=}\:, and orthogonality relation
\eqref{orth}\:.

\vsk.2>
By formulae \eqref{orth}\:, \eqref{Psd0=}\:, and Proposition \ref{egvX0},
\vvn,1>
\be
\prod_{i=1}^N\:q_i^{\>\smash{X_i(\zz\:;h;\:\0)}}\:v_I\>=\?
\sum_{\JK\<\in\:\Il}\<\frac{\WW_I(\Si_J\:;\zz\:;h)\,W_K(\Si_J\:;\zz\:;h)}
{c_\bla^{\>2}(\Si_J\:;h)\>R_\bla(\zz_{\si_J}\<)\>Q_\bla(\zz_{\si_J};h)}\,\,
\prod_{i=1}^N\,q_i^{\>\sum_{\smash{a\in J_i}}\?z_a}\>v_K\,.\kern-1em
\vv.4>
\ee
Then by formulae \eqref{PspI}\:, \eqref{corth}\:, \eqref{PsiPsd}\:,
the functions
\begin{align*}
\Psp(\zz\:;h;\qq)\>\prod_{i=1}^N\:q_i^{\>\smash{X_i(\zz\:;h;\:\0)}}\:v_I
\>=\>\sum_{J\in\Il}\Psi_{\?J}(\zz\:;h;\qq)\;
\frac{\WW_I(\Si_J\:;\zz\:;h)}{c_\bla(\Si_J\:;h)}\;
\prod_{i=1}^N\,e^{\:\pii\,(\la_i-n)\>\sum_{a\in J_i}\?z_a}\:\times{}
\kern-1.2em&
\\[3pt]
\notag
{}\times\>\prod_{i=1}^{N-1}\prod_{j=i+1}^N\,\:\prod_{a\in J_i}\>
\prod_{b\in J_j}\>\frac{-\>1}{\Ga(1+z_a\<-z_b)\>\Ga(1-z_a\<+z_b\<+h)}
\;\kern-1.2em&\kern1.2em\!\!\kern-1.2em
\end{align*}
are linear combinations of the solutions $\,\Psi_{\?J}(\zz\:;h;\qq)\,$ of
\vv.1>
dynamical differential equations \eqref{DEQ} with coefficients independent
\vv.1>
of $\,\qq\,$. Hence, $\,\Psh(\zz\:;h;\qq)=\Psp(\zz\:;h;\qq)\>
\prod_{i=1}^N\,q_i^{\>\smash{X_i(\zz\:;h;\:\0)}}\>$ solves dynamical
differential equations \eqref{DEQ}\:.

\vsk.2>
The determinant $\;\det\:\Psh(\zz\:;h;\qq)\,$ satisfies the equations
\vvn.2>
\be
\Bigl(\:q_i\:\frac{\der}{\der q_i}\>-\>
\tr\:X_i(\zz\:;h;\qq)|_{\:\Cnnl}\Bigr)\>
\det\:\Psh(\zz\:;h;\qq)\,=\,0\,, \qquad i=1\lc N\>,\kern-2em
\vv.1>
\ee
Since $\,\tr\:X_i(\zz\:;h;\qq)|_{\>\Cnnl}\<=\:\tr\:X_i(\zz\:;h;\0)|_{\>\Cnnl}
\<=\:d^{(1)}_{\bla,i}\,(z_1\<\lsym+ z_n)\,$,
\vv.1>
where $\,d^{(1)}_{\bla\:,\:i}\>$ are given by \eqref{dla12}\:,
the determinant $\;\det\:\Psp(\zz\:;h;\qq)\,$ does not depend on \,$\qq\,$.
Therefore,
\vvn.1>
\be
\det\:\Psp(\zz\:;h;\qq)\>=\>\det\:\Psp(\zz\:;h;\0)\>=\>1
\kern1.6em\text{and}\kern1.6em
\det\:\Psh(\zz\:;h;\qq)\,=\,\prod_{i=1}^N\,
q_i^{\:\smash{d^{(1)}_{\bla,i}}\:\sum_{a=1}^n z_a}.\kern-.4em
\ee

\vsk.1>
The functions $\,\Psd_{\!J}(\zz\:;h;\qq)\,$ in formula \eqref{PspI} are entire
in $\,\zz,h\,$ by Theorem \ref{thm cy}, and the expressions
\vvn-.1>
\be
\frac{\WW_I(\Si_J\:;\zz\:;h)}{c_\bla(\Si_J\:;h)}\;
\prod_{i=1}^{N-1}\prod_{j=i+1}^N\,\:\prod_{a\in J_i}\>
\prod_{\satop{b\in J_j}{b<a}}\>\frac1{z_a\<-z_b\<-h}\kern-1.5em
\ee
are polynomials by Lemma \ref{WIJ} applied to the functions
\vv.1>
$\,\WW_I(\TT\:;\zz\:;h)\,$. Hence, $\,\Psp(\zz\:;h;\qq)\,$ is holomorphic
in $\,\zz,h\,$ provided $\,z_a\<-z_b\<\not\in\Z\,$ for all $\,a\ne b\,$.
Moreover, $\,\Psp(\zz\:;h;\qq)\,$ is regular for $\,z_a\<=z_b\,$
by formula \eqref{Psdab}\:, and the singularities at the hyperplanes
$\;z_a\<-z_b\<\in\Z_{\:\ne\:0}\,$ are simple poles.
Theorem \ref{Bthm} is proved.
\end{proof}

\begin{proof}[Proof of Proposition \ref{PsiPpr}]
Denote by $\,\Psb_P(\zz\:;h)\,$ the right-hand side of
\vvn.1>
formula \eqref{PsiP0}\:. Then formula \eqref{PPI}\:,
Proposition \ref{egvX0}, and Theorem \ref{thm cy} yield
$\,\Psh(\zz\:;h;\qq)\,\Psb_P(\zz\:;h)=\Psi_P(\zz\:;h;\qq)\,$.
\vvn.1>
Hence by definition \eqref{princ} of the principal term,
$\,\Psz_{\?P}(\zz\:;h)=\,\Psb_P(\zz\:;h)\,$.
%Proposition \ref{PsiPpr} is proved.
\end{proof}

Consider the entries of $\,\Psp(\zz\:;h;\qq)\,$
\vvn.3>
in the standard basis $\,\{\:v_I\,,\alb\,I\<\in\Il\:\}\,$ of $\,\Cnnl$,
\beq
\label{entrij}
\Psp(\zz\:;h;\qq)\::\:v_J\,\mapsto\>
\sum_{I\in\Il}\,\Psp_{\<\IJ}(\zz\:;h;\qq)\,v_I\,.\kern-2em
\eeq
Recall the function $\,A(\TT\:;\zz\:;h;\ka)\,$ at $\,\ka=1\,$,
see \eqref{Ath}\:,
\vvn.3>
\beq
\label{A}
A(\TT\:;\zz\:;h)\,=\,\prod_{i=1}^{N-1}\,\prod_{a=1}^{\la^{(i)}}\;\biggl(\,
\prod_{\satop{b=1}{b\ne a}}^{\la^{(i)}}\>
\frac{\Gm(t^{(i)}_b\!-t^{(i)}_a\!+1)}{\Gm(t^{(i)}_b\!-t^{(i)}_a\!+h)}\;
\prod_{c=1}^{\la^{(i+1)}}\:\frac{\Gm(t^{(i+1)}_c\!-t^{(i)}_a\!+h)}
{\Gm(t^{(i+1)}_c\!-t^{(i)}_a\!+1)}\,\biggr)\>,\kern-1.4em
\vv-.2>
\eeq
where $\,\la^{(N)}\?=n\,$ and $\,t^{(N)}_a\?=z_a\,$,
$\;a=1\lc n\,$. Notice that
\vvn.2>
\be
A(\Si_I\:;\zz\:;h)\,=\,\bigl(\Gm(h)\bigr)^{\la\+1}\>
\prod_{i=1}^{N-1}\prod_{j=i+1}^N\,\:\prod_{a\in I_i}\>\prod_{b\in I_j}
\,\frac{\Gm(z_b\<-z_a\<+h)}{\Gm(z_b\<-z_a\<+1)}\;,\kern-1.2em
\ee
where $\,\la\+1\<=\sum_{i=1}^{N-1}\la^{(i)}\>$. Recall the function
$\,\Om_\bla(h;\qq\:;\ka)\,$ at $\,\ka=1\,$, see \eqref{Oml}\:,
\vvn.2>
\beq
\label{Omh}
\Om_\bla(h;\qq)\,=\,\prod_{i=1}^{N-1}\prod_{j=i+1}^N (1-q_j/q_i)^{h\:\la_i}\>.
\vv-.3>
\eeq
For $\,\lb\<\in\Z^{\:\la\+1}\!$, set
\beq
\label{Jc}
\Jc_{\IJ,\,\lb}(\zz\:;h)\,=\>\sum_{K\in\Il}\>
\frac{A(\Si_K\?-\lb\:;\zz\:;h)\,W_I(\Si_K\?-\lb;\zz\:;h)\,\WW_J(\Si_K;\zz\:;h)}
{R_\bla(\zz_{\si_K}\<)\,Q_\bla(\zz_{\si_K};h)\,A(\Si_K;\zz\:;h)\,
c_\bla(\Si_K\?-\lb\:;h)\,c_\bla(\Si_K;h)}\;.\kern-1em
\vv.4>
\eeq
\begin{prop}
\label{lemPsp}
We have
\vvn-.3>
\beq
\label{PspIJ}
\Psp_{\<\IJ}(\zz\:;h;\qq)\,=\,\Om_\bla(h;\qq)\!
\sum_{\lb\in\Z_{\ge0}^{\:\la\+1}\!\!}\Jc_{\IJ,\,\lb}(\zz\:;h)\,
\prod_{i=1}^{N-1}\,(q_{i+1}/q_i)^{\>\sum_{a=1}^{\la^{(i)}}l_a^{(i)}}.\kern-1em
\eeq
\end{prop}
\begin{proof}
The statement follows from formula \eqref{PspI} and Theorem \ref{thm cy}.
\end{proof}

\subsection{The map $\,\Bcyr_\bla\,$}
\label{secB}

In Section \ref{seclaur}, we introduced the space $\>\Srsl$ of
$\,\Cnnl\<$-valued solutions of the joint system of dynamical differential
equations \eqref{DEQ} and \qKZ/ difference equations \eqref{Ki} spanned over
$\,\C\,$ by the functions $\,\Psi_P(\zz\:;h;\qq\:;\ka)\,$ labeled by Laurent
polynomials in $\,\:\GGd\<,\zzd,\hdd\,$; we also defined the map
\vvn.1>
\be
\muk:\:\Kc_\bla\to\:\Srsl\,,\qquad Y\<\mapsto\Psi_Y\>,
\vv.1>
\ee
see \eqref{muk}\:. In Section \ref{secLev} we introduced the Levelt fundamental
solution $\,\:\Psh(\zz\:;h;\qq\:;\ka)\,$ of dynamical differential equations
\eqref{DEQ}\:, see \eqref{Psipnd}\:, \eqref{PspI}\:. Denote by $\>\Srhl\>$
the space of $\,\Cnnl\<$-valued solutions of dynamical differential equations
\eqref{DEQ} spanned over $\,\C\,$ by the functions
\vvn.06>
$\,\:\Psh(\zz\:;h;\qq\:;\ka)\>v\,$, $\,v\in\Cnnl$.
Since $\,\:\det\:\Psh(\zz\:;h;\qq\:;\ka)\ne 0\,$, see \eqref{detPsh}\:,
there is an isomorphism
\vvn-.3>
\be
\muht:\Cnnl\to\>\Srhl\>,\qquad v\:\mapsto\Psh(\zz\:;h;\qq\:;\ka)\>v\,.
\kern-1.6em
\ee

\vsk.5>
Let $\,L\<\subset\<\C^n\!\times\C\,$ be the complement of the union of
the hyperplanes
\vvn.4>
\beq
\label{zzhZ}
z_a\<-z_b\<+h\in\<\ka\>\Z_{\le0}\,,\qquad z_a\<-z_b\<\in\<\ka\>\Z_{\ne0}\,,
\qquad a,b=1\lc n\,,\quad a\ne b\,.\kern-2em
\eeq

\vsk.4>
Denote by $\,\Oc_L\:$ the ring of functions of $\,\zz,h\,$ holomorphic
\vvn.07>
in $\,L\,$. Let $\>\Sroll$ be the space of $\,\Cnnl\<$-valued solutions
\vv.06>
of dynamical differential equations \eqref{DEQ}
holomorphic in $\,\qq\,$ provided $\,\:|\:q_{i+1}/q_i|<1\,$ for all
\vv.06>
$\,i=1\lc N-1\,$, with a branch of $\,\:\log\>q_i\>$ fixed for each
\,$i=1\lc N$, and holomorphic in $\,\zz,h\,$ in $\,L\,$.
Both spaces $\>\Srsl\>$ and $\>\Srhl\>$ are subspaces of $\>\Sroll$,
see Proposition \ref{PsiPsol2} and Theorem \ref{Bthm}.
Let
\vvn.4>
\beq
\label{muol}
\muol\<:\:\Cnnl\?\ox_{\:\C}\Oc_L\to\>\Sroll,\qquad
v\:\mapsto\Psh(\zz\:;h;\qq\:;\ka)\>v\,,
\vv.3>
\eeq
be the $\>\Oc_L\:$-\:linear extension of the map $\,\muht\,$.

\vsk.3>
Recall the functions
\vvn-.2>
\be
%\label{CGG2}
C_\bla(\zz\:;\ka)\,=\,\prod_{i=1}^N\>e^{\:\pii\,(n-\la_i)
\sum_{a=\smash{\la^{\?(i-1)}}\<+1}^{\la_i}z_a\</\<\ka}\:\kern-2em
\vv-.3>
\ee
and
\vvn.1>
\be
%\label{GGG2}
G_\bla(\zz\:;h;\ka)\,=\,\prod_{i=1}^{N-1}\>\prod_{a=1}^{\la^{(i)}}\,
\prod_{b=\la^{(i)}+1}^{\la^{(i+1)}}\!
\Gm\bigl(\<(z_a\<-z_b)/\ka\bigr)\,\Gm\bigl(\<(z_b\<-z_a\<+h)/\ka\bigr)\,,\kern-.5em
\vv.5>
\ee
see \eqref{CGG}\:, \eqref{GGG}\:. Define a map
\vvn.2>
\begin{gather}
\notag
\Bcyr_\bla\::\:\Kc_\bla\to\:\Cnnl\?\ox_{\:\C}\Oc_L\,,\kern-.6em
\\[4pt]
\label{BcyrP}
[P\:]\,\mapsto\?
\sum_{\IJ\in\:\Il}\?\Pdd(\zz_{\si_J};\zz\:;h;\ka)\,
C_\bla(\zz_{\si_J};\ka)\,G_\bla(\zz_{\si_J};h;\ka)\,
\frac{W_I(\Si_J\:;\zz\:;h)}{c_\bla(\Si_J\:;h)}\;v_I\,,\kern-1em
\\[-14pt]
\notag
\end{gather}
where $\,[P\:]\in\Kc_\bla\,$ stands for the class of the Laurent polynomial
\vv.06>
$\,P(\:\GGd\:;\zzd\:;\hdd)\,$. By Proposition \ref{PsiPpr}, the map
$\,\:\Bcyr_\bla\>$ sends the class $\,Y\!\in\Kc_\bla\,$ to the principal term
of the solution $\,\:\Psi_Y\:$ of the joint system of dynamical differential
equations \eqref{DEQ} and \qKZ/ difference equations \eqref{Ki}.

\begin{prop}
\label{triangle}
The map $\;\Bcyr_\bla\::\:\Kc_\bla\to\:\Cnnl\?\ox_{\:\C}\Oc_L\:$
is well-defined and the following diagram is commutative,
\vvn-.4>
\beq
\label{cd}
\xymatrix{\Kc_\bla\ar^-{\;\Bcyr_\bla}[rr]
\ar_{\smash{\lower.8ex\llap{$\ssize\muk\;\;\,$}}}[dr]&&
\rlap{$\Cnnl\?\ox_{\:\C}\Oc_L$}\phan{\Kc_\bla}
\ar^{\smash{\lower.8ex\rlap{$\;\ssize\muol$}}}[dl]\\
&\!\Sroll\:&}
\eeq
\end{prop}
\begin{proof}
By Lemma \ref{WIJ}, poles of the sum in the right-hand side of \eqref{BcyrP}
are at most those of the function $\,G_\bla(\zz\:;h;\ka)\,$ and, therefore,
in addition to hyperplanes \eqref{zzhZ} can occur only at the hyperplanes
$\,z_a\<=z_b\>$, $\,a\ne b$. However, the sums
\vvn.1>
\be
\sum_{J\in\Il}\Pdd(\zz_{\si_J};\zz\:;h;\ka)\,
C_\bla(\zz_{\si_J};\ka)\,G_\bla(\zz_{\si_J};h;\ka)\,
\frac{W_I(\Si_J\:;\zz\:;h)}{c_\bla(\Si_J\:;h)}
\vv.2>
\ee
are regular at the hyperplanes $\,z_a\<=z_b\>$ for all $\,a,b\,$,
by the standard reasoning. Hence, the map $\,\:\Bcyr_\bla\>$ is well-defined.

\vsk.2>
The commutativity of diagram \eqref{cd} follows from
Proposition \ref{PsiPpr}.
\end{proof}

\begin{rem}
Since $\,\:\Psh(\zz\:;h;\qq\:;\ka)\,$ is holomorphic in $\,\zz,h\,$ in $\,L\,$,
and $\,\bigl(\det\:\Psh(\zz\:;h;\qq\:;\ka)\bigr)^{\?-1}\:,$
is entire in $\,\zz,h\,$, see \eqref{detPsh}\:, the inverse matrix
\vv.03>
$\,\:\Psh(\zz\:;h;\qq\:;\ka)^{-1}\:$ is holomorphic in $\,\zz,h\,$ in $\,L\,$.
Therefore, for every $\,\:\Psi\in\Sroll$, its principal term
$\,\:\Psz\<=\Psh^{-1}\>\Psi\,$, defined by \eqref{princ}\:,
\vv.06>
belongs to $\,\Cnnl\?\ox_{\:\C}\Oc_L\,$, and there is an isomorphism
\vv.06>
$\,\Sroll\:\!\to\Cnnl\?\ox_{\:\C}\Oc_L\,$, $\;\Psi\mapsto\Psz\:$.
The inverse map equals $\,\muol\,$, see \eqref{muol}\:, so that,
$\,\Sroll\?=\Srhl\?\ox_{\:\C}\Oc_L\,$.
\end{rem}

\subsection{Example $\,\bla=(1,n-1)\,$}
\label{example}
Throughout this section, let $\,N=2\,$ and $\,\bla=(1,n-1)\,$.
Denote by $\,[\:a\:]$ the element
$\bigl(\{a\},\{\:1\lc a-1,a+1\lc n\:\}\bigr)\<\in\Il$\,.
\vv-.04>
The space $\,\Ctnl\:$ has a basis $\,v_{[\:1\:]}\lc v_{[\:n\:]}\:$,
where $\,v_{[a]}\<=v_2^{\ox(a-1)}\<\ox v_1\<\ox v_2^{\ox(n-a)}\>$.
Clearly $\,e_{1,1}^{(a)}v_{[\:b\:]}=\:\dl_{\ab}\,v_{[\:b\:]}\,$ and
$\,e_{2,2}^{(a)}v_{[\:b\:]}=(1-\dl_{\ab})\,v_{[\:b\:]}\,$.

\vsk.3>
The \qKZ/ operators $\,K_1\lc K_n\>$, see \eqref{K}\:, are
\vvn.4>
\begin{align}
\label{K1}
\kern1.2em
K_a(\zz\:;h;\qq\:;\ka)\,={}&\,
R^{(\aa-1)}(z_a\<-z_{a-1}+\ka;h)\,\dots\,R^{(a,1)}(z_a\<-z_1+\ka;h)\>\times{}
\\[3pt]
\notag
&\?{}\times\<\;q_1^{e_{1,1}^{(a)}}\:q_2^{e_{2,2}^{(a)}}\,
R^{(a,n)}(z_a\<-z_n\:;h)\,\dots\,R^{(\aa+1)}(z_i\<-z_{i+1}\:;h)\,.
\\[-12pt]
\notag
\end{align}
The $\:R\:$-matrices in the right-hand side % of \eqref{K1}
preserve the subspace $\Ctnl\?\subset\<\Ctn$, acting there as follows,
\vvn.1>
\begin{gather*}
R^{(\ab)}(z\:;h)\,v_{[a]}\,=\,
\frac z{z-h}\,v_{[a]}\:-\:\frac h{z-h}\,v_{[\:b\:]}\,,\kern1.6em
R^{(\ab)}(z\:;h)\,v_{[\:b\:]}\,=\,
\frac z{z-h}\,v_{[\:b\:]}\:-\:\frac h{z-h}\,v_{[a]}\,,\kern-.2em
\\[9pt]
R^{(\ab)}(z\:;h)\,v_{[c]}\>=\>v_{[c]}\,,\qquad c\ne a,b\,.\kern-1em
\\[-12pt]
\end{gather*}
The \qKZ/ difference equations \eqref{Ki} are
\vvn.4>
\beq
\label{Ki1}
f(z_1\lc z_a+\ka\lc z_n;h;\qq\:;\ka)\,=\,
K_a(\zz\:;h;\qq\:;\ka)\,f(\zz\:;h;\qq\:;\ka)\,,\qquad a=1\lc n\>.\kern-1em
\eeq

\vsk.6>
The dynamical Hamiltonians $\,X_1\:,\:X_2\>$, see \eqref{Xi}\:,
act on $\,\Ctnl\:$ as follows
\vvn.2>
\begin{align}
\label{Xi1}
&X_1(\zz\:;h;\qq)\>v_{[\:a\:]}\>=\,z_a\>v_{[\:a\:]}\:-\>
h\>\sum_{b=1}^{a-1}\:v_{[\:b\:]}\:-\>
\frac{h\:q_2}{q_1\<-q_2}\,\sum_{\satop{b=1}{b\ne a}}^n\>v_{[\:b\:]}\>,
\\[-8pt]
\notag
&X_2(\zz\:;h;\qq)\>v_{[\:a\:]}\>=\>
\Bigl(\<-\>X_1(\zz\:;h;\qq)\>+\sum_{b=1}^n\,z_b\Bigr)\,v_{[\:a\:]}\,,
\\[-20pt]
\notag
\end{align}
and the dynamical differential equations \eqref{DEQ} are
\vvn.5>
\begin{align}
\label{DEQ1}
\ka\>q_1\:\frac\der{\der\:q_1}\>\Psi(\zz\:;h;\qq\:;\ka)\,&{}=\>
X_1(\zz\:;h;\qq)\>\Psi(\zz\:;h;\qq\:;\ka)\,,\kern-1em
\\[7pt]
\notag
\ka\>q_2\:\frac\der{\der\:q_2}\>\Psi(\zz\:;h;\qq\:;\ka)\,&{}=\>
X_2(\zz\:;h;\qq)\>\Psi(\zz\:;h;\qq\:;\ka)\,.\kern-1em
\end{align}

\vsk.2>
In this section, we use the variable $\,t=t^{(1)}_1\:$.
The substitution $\,\:\TT=\Si_{\:[a]}\>$ reads as $\,t\:=z_a\,$.
The weight functions are
\vvn-.1>
\beq
\label{Wa}
W_{[a]}(t\:;\zz\:;h)\,=\,
\prod_{b=1}^{a-1}\,(t-z_b)\prod_{b=a+1}^n\?(t-z_b\<-h)\,,
\qquad a=1\lc n\,.\kern-2em
\vv.2>
\eeq
We have $\,c_\bla(t\:;h)=1\,$, see \eqref{cla}\:.
The permutations $\,\:\si_{[1]}\lc\si_{[n]}\,$ are
\vvn.3>
\be
\si_{[a]}(1)=a\,,\qquad \si_{[a]}(b)=b-1\,,\quad b=1\lc a-1\,,
\qquad\si_{[a]}(b)=b\,,\quad b=a+1\lc n\,,
\vv.3>
\ee
and $\,\:|\:\si_{[a]}\:|=a-1\,$. We have
\vvn-.4>
\begin{gather}
\label{Wab}
W_{[a]}(z_a;\zz\:;h)\,=\,
\prod_{b=1}^{a-1}\,(z_a\<-z_b)\prod_{b=a+1}^n\?(z_a\<-z_b\<-h)\,,
\\[6pt]
\notag
W_{[a]}(z_b;\zz\:;h)\>=\>0\,,\qquad b=1\lc a-1\,,
\\[-14pt]
\notag
\end{gather}
and $\;W_{[a]}(z_b;\zz\:;h)\,$ is divisible by the product
\vvn.07>
$\;\prod_{b=a+1}^n(z_a\<-z_b\<-h)\,$ for any $\,b=1\lc n\,$,
cf.~Lemmas \ref{WIz}, \ref{WIJ}

\vsk.2>
The functions $\,\WW_{[a]}(t\:;\zz\:;h)\,$, see \eqref{WcI}\:, are
\vvn.2>
\beq
\label{Wca}
\WW_{[a]}(t\:;\zz\:;h)\,=\,
\prod_{b=1}^{a-1}\,(t-z_b\<-h)\prod_{c=a+1}^n\?(t-z_b)\,,
\qquad a=1\lc n\,.\kern-2em
\vv-.3>
\eeq
Set
\vvn-.5>
\be
\Rb_a(\zz)\,=\,\prod_{\satop{b=1}{b\ne a}}^n\,(z_a\<-z_b)\,,\qquad
\Qb_a(\zz\:;h)\,=\,\prod_{\satop{b=1}{b\ne a}}^n\,(z_a\<-z_b\<-h)\,,
\qquad a=1\lc n\,.\kern-1em
\ee
Then $\,\Rb_a(\zz)=R_\bla(\zz_{\si_{[a]}})\,,$ and
$\,\Qb_a(\zz\:;h)=Q_\bla(\zz_{\si_{[a]}};h)\,$, where the functions
\vvn.1>
$\,R_\bla(\zz)\,$, $\,Q_\bla(\zz\:;h)\,$, are given by \eqref{RQ}\:.
Biorthogonality relations \eqref{orth}\,, \eqref{corth} become
\vvn.3>
\begin{gather*}
\sum_{c=1}^n\,\frac{W_{[a]}(z_c;\zz\:;h)\,\WW_{[b]}(z_c;\zz\:;h)}
{\Rb_c(\zz)\>\Qb_c(\zz\:;h)}\,=\,\dl_{\ab}\,,\kern-1.8em
\\
\sum_{c=1}^n\,W_{[c]}(z_a;\zz\:;h)\,\WW_{[c]}(z_b;\zz\:;h)
\,=\,\dl_{\ab}\,\Rb_a(\zz)\>\Qb_a(\zz\:;h)\,.\kern-1em
\end{gather*}

\vsk.1>
The master function, see \eqref{Phi}\:, is
\vvn.3>
\begin{align*}
& \Phi_\bla(t\:;\zz\:;h;\qq\:;\ka)\,={}
\\[4pt]
& \!\<{}=\,(e^{\:\pii}q_2)^{\>\sum_{a=1}^n\<z_a\</\?\ka}\,
(e^{\:\pii\,(n-2)}q_1/q_2)^{\:t/\?\ka}\,\prod_{a=1}^n\,
\Gm\bigl(\<(t-z_a)/\ka\bigr)\>\Gm\bigl(\<(z_a\<-t+h)/\ka\bigr)\,.\kern-.9em
\\[-20pt]
\end{align*}
%\begin{align*}
%\Phi_\bla(t\:;\zz\:;h;\qq\:;\ka)\,{}=\,
%(e^{\:\pii}q_2)^{\>\sum_{a=1}^n\<z_a\</\?\ka}\,
%(e^{\:\pii\,(n-2)}q_1/q_2)^{\:t/\?\ka}\:\times{}\kern-1.7em\,\:&
%\\[3pt]
%{}\times\,\prod_{a=1}^n\,
%\Gm\bigl(\<(t-z_a)/\ka\bigr)\>\Gm\bigl(\<(z_a\<-t+h)/\ka\bigr)\,.\kern-1.7em&
%\\[-13pt]
%\end{align*}
The
%$q\:$-
hypergeometric solutions \eqref{mcF} of the joint system of
% dynamical
differential equations \eqref{DEQ1} and
% \qKZ/
difference equations \eqref{Ki1} have the form
\vvn.1>
\be
\Psi_{[a]}(\zz\:;h;\qq\:;\ka)\,=\,
\frac{(1-q_2/q_1)^{h\</\?\ka}}{\ka\>\Gm(h/\<\ka)}\,
\sum_{b=1}^n\,\sum_{l=0}^\infty\,
\Res_{\>t\:=\;z_a-\:l\ka}\Phi_\bla(t\:;\zz\:;h;\qq\:;\ka)\;
W_{[\:b\:]}(z_a\<-l\ka\:;h)\>v_{[\:b\:]}\,,
\vv.2>
\ee
where the residues of the master function are
\vvn.1>
\begin{align*}
\Res_{\>t\:=\;z_a-\:l\ka}\Phi_\bla(t\:;\zz\:;h;\qq\:;\ka)\,&{}=\,
q_1^{\:z_a\</\<\ka}\:q_2^{\>\sum_{c=1,\>c\ne a}^n\<z_c\</\?\ka}\,
\prod_{\satop{c=1}{c\ne a}}^n\,\frac{\pi\>e^{\:\pii\>(z_a+\:z_c)/\?\ka}}
{\sin\:\bigl(\pi\>(z_a\?-z_c)/\<\ka\bigr)}\,\times{}
\\[3pt]
&{}\>\times\,\frac{\ka\>(q_2/q_1)^{\:l}\,\Gm(l+h/\<\ka)}{l\:!}\,
\prod_{\satop{c=1}{c\ne a}}^n\,\frac{\Gm\bigl(l+(z_c\<-z_a\<+h)/\ka\bigr)}
{\Gm\bigl(l+1+(z_c\<-z_a)/\ka\bigr)}\;.
\\[-16pt]
\end{align*}
Determinant formula for coordinates of the
%$\>q\:$-
hypergeometric solutions, see \eqref{detPsi}\:, is
\vvn.5>
\begin{align}
\label{detPsi1}
\det\:\biggl(\>\sum_{l=0}^\infty\,
\Res_{\>t\:=\;z_a-\:l\ka}\Phi_\bla &{}(t\:;\zz\:;h;\qq\:;\ka)\,
W_{[\:b\:]}(z_a\<-l\ka\:;h)\?\biggr)_{\!\?\ab\:=1}^{\!\<n}\<={}
\\[10pt]
\notag
{}=\,\:\frac{\<\bigl(\:q_1\>q_2^{\:n-1}\bigr)^{\sum_{a=1}^n z_a\</\?\ka}\>
\bigl(\ka\>\Gm(h/\<\ka)\bigr)^n\?}{(1-q_2/q_1)^{n\:h\</\?\ka}}\;
& \prod_{a=1}^{n-1}\,\prod_{b=a+1}^n\!\frac{\pi\:\ka^2\>
\Gm\bigl(\<(z_a\<-z_b\<+h)/\ka\bigr)\>\Gm\bigl(1+(z_b\<-z_a\<+h)/\ka\bigr)}
{e^{-\pii\>(z_a+\:z_b)/\?\ka}\:\sin\bigl(\pi\>(z_a\<-z_b)/\ka\bigr)}\;.
\kern-.4em
\\[-11pt]
\notag
\end{align}
By formulae \eqref{Wa}, \eqref{Wab}, and the Vandermonde determinant
formula, equality \eqref{detPsi1} transforms to
\begin{align}
\label{detPsi2}
\det\:\biggl(\>\sum_{l=0}^\infty\,\sum_{c=1}^n\,
\Res_{\>t\:=\;z_c-\:l\ka}\bigl(\:
t^{\:a-1}\:e^{-\:2\:\pii\,\:(b-1)\>t/\?\ka}\,
\Phi_\bla(t\:;\zz\:;h;\qq\:;\ka)\bigr)\?\biggr)_{\!\?\ab\:=1}^{\!\<n}\<={}
\hp{=\sum\:} &
\\[8pt]
\notag
{}=\,\:\frac{\<\bigl(\:q_1\>q_2^{\:n-1}\bigr)^{\sum_{a=1}^n z_a\</\?\ka}\>
\bigl(\<-\:2\:\piit\,\:\bigr)^{\<n\:(n-1)/2}\:\ka^{\:n\:(n+1)/2}\>
\bigl(\:\Gm(h/\<\ka)\bigr)^n\?}{(1-q_2/q_1)^{n\:h\</\?\ka}}\,\times{} &
\\[1pt]
\notag
{}\times\>\prod_{a=1}^{n-1}\,\prod_{b=a+1}^n\!
\Gm\bigl(\<(z_a\<-z_b\<+h)/\ka\bigr)\>\Gm\bigl(\<(z_b\<-z_a\<+h)/\ka\bigr)\>&\:.
\end{align}

\vsk.3>
Let $\,\gmd=\gmd_{1,1}\,$.
% The variables $\,\:\gmd_{\:2,1}\lc\gmd_{\:2,\:n-1}\>$ will not be used
% in this section.
For a Laurent polynomial $\,P(\gmd\:;\zzd\:;\hdd)\,$, we have
\vvn.4>
\be
\Pdd(\gm\:;\zz\:;h;\ka)\,=\,P(e^{\:2\:\pii\,\gm\</\<\ka}\:;\:
e^{\:2\:\pii\,z_1\</\<\ka}\lc e^{\:2\:\pii\,z_n\</\<\ka}\:;\:
e^{\:2\:\pii\,h\</\<\ka})\,.\kern-1em
\vv.3>
\ee
The solution $\,\Psi_P\>$ of differential equations \eqref{DEQ1} and difference
equations \eqref{Ki1} corresponding to $\,P(\gmd\:;\zzd\:;\hdd)\,$ is
\vvn-.5>
\beq
\label{PsiP1}
\Psi_P(\zz\:;h;\qq\:;\ka)\,=\,\sum_{a=1}^n\,
\Pdd(z_a;\zz\:;h;\ka)\,\Psi_{[a]}(\zz\:;h;\qq\:;\ka)\,.\kern-1.2em
\vv.1>
\eeq
By Proposition \ref{PsiPsol2}, % the function
$\,\Psi_P(\zz\:;h;\qq\:;\ka)\,$ is holomorphic in $\,\zz,h,\qq\,$
\vv.07>
provided $\,z_a\<-z_b\<+h\not\in\<\ka\>\Z_{\le0}\,$ for all $\,a,b=1\lc n\,$,
$\,a\ne b\,$, and $\,\:|\:q_2/q_1|<1\,$ with branches of $\,\:\log\:q_1\:$ and
\vv.06>
$\,\:\log\:q_2\:$ fixed. The singularities of $\,\Psi_P(\zz\:;h;\qq\:;\ka)\>$
at the hyperplanes $\;z_a\<-z_b\<+h\in\<\ka\>\Z_{\le0}\,$ are simple poles.

\vsk.3>
The solution $\,\Psi_P(\zz\:;h;\qq\:;\ka)\,$ can be written as an integral
\vv.16>
over a suitable contour $\,C\>$ encircling the poles of the product
$\,\prod_{a=1}^n\<\Gm\bigl({(t-z_a)/\<\ka}\bigr)\,$ counterclockwise
\vv.14>
and separating them from the poles of the product
$\,\prod_{a=1}^n\<\Gm\bigl(\<(z_a\<-t+h)/\<\ka\bigr)\,$,
\vvn.5>
\begin{align}
\label{PsiPint}
& \Psi_P(\zz\:;h;\qq\:;\ka)\,={}
\\[4pt]
\notag
&\!{}=\,\frac{(1-q_2/q_1)^{h\</\?\ka}}{2\:\piit\,\ka\>\Gm(h/\<\ka)}\;
\int\limits_{\!\!C\,}\<\Pdd(t\:;\zz\:;h;\ka)\,\Phi_\bla(t\:;\zz\:;h;\qq\:;\ka)\,
\sum_{a=1}^n\>W_{[a]}(t\:;\zz\:;h)\,v_{[a]}\,d\:t\,.
\\[-14pt]
\notag
\end{align}
For instance, if $\,h/\ka\,$ is sufficiently large positive real,
the integral can be taken over the parabola
\vvn-.1>
\be
C\,=\,\{\,h/2-\ka\>\bigl(\:s^2\?-s\>\sqrt{\<-1}\,\bigr)\ \,|\ \,s\in\R\,\}\,.
\vv.5>
\ee
Formula \eqref{PsiPint} can be used to give an alternative proof of
analytic properties of the function $\,\:\Psi_P(\zz\:;h;\qq\:;\ka)\,$.

\vsk.3>
Let $\>\Srsl\:$ be the space of solutions of the joint system of dynamical
\vv.1>
differential equations \eqref{DEQ1} and \qKZ/ difference equations \eqref{Ki1}
\vv.06>
spanned over $\,\C\,$ by the functions $\,\Psi_P(\zz\:;h;\qq\:;\ka)\,$
corresponding to Laurent polynomials $\,P(\gmd\:;\zzd\:;\hdd)\,$.
\vvn.1>
The space $\>\Srsl\:$ is a $\,\C[\:\zzd^{\pm1}\?,\hdd^{\pm1}]\:$-\:module
%a function
with $\,f(\zzd\:;\hdd)\,$ acting as multiplication by
$\,f(e^{\:2\:\pii\,z_1\</\<\ka}\lc e^{\:2\:\pii\,z_n\</\<\ka}\:;\:
e^{\:2\:\pii\,h\</\<\ka})\,$.
%$\,\fdd(\zz\:;h)\,$.

\vsk.2>
For $\,\bla=(1,n-1)\,$, the algebra $\,\Kc_\bla\,$, see \eqref{Krel},
can be presented as follows
\vvn.3>
\beq
\label{Krel11}
\Kc_\bla\:=\,\C[\:\gmd^{\pm1}\?,\zzd^{\pm1}\?,\hdd^{\pm1}]\>\>\Big/\Bigl\bra
\,\prod_{a=1}^n\,(\gmd-\zdd_a)\>=\>0\,\Bigr\ket\,.\kern-1.6em
\vv.3>
\eeq
The function $\,\Psi_P(\zz\:;h;\qq\:;\ka)\,$ depends only on the class
\vv.1>
of the Laurent polynomial $\,P\:$ in $\>\Kc_\bla\,$.
The assignment $\,P\mapsto\Psi_P\>$ defines the homomorphism
\vvn.3>
\be
\muk\<:\Kc_\bla\to\:\Srsl\,,\qquad Y\<\mapsto\Psi_Y\>,\kern-1em
\vv.3>
\ee
of $\,\C[\:\zzd^{\pm1}\?,\hdd^{\pm1}]\:$-\:modules. Formula \eqref{detPsi2}
implies that the homomorphism $\,\muk\,$ is an isomorphism.

\vsk.3>
The algebra $\,\Kc_\bla\:$ is the equivariant $\>K\?$-theory algebra
\vvn.1>
$\,K_{T\<\times\Cxs}\<(T^*\CP^{\>n-1}\<;\C)\>$ of the cotangent bundle of
the projective space $\,\CP^{\>n-1}$, see the notation in Section \ref{sQde}.

\vsk.3>
Recall Definition \ref{hdef} of a function $\,f(\qq)\,$ holomorphic in the unit
\vv.07>
polydisk around $\,\qq=\0\,$.
%if $\,f(\qq)=g(q_2/q_1)\,$ for a function $\,g(s)\,$
%holomorphic provided $\,|\:s\:|<1\,$, and $\,f(\0)=g(0)\,$.
For example,
\vv.07>
the dynamical Hamiltonians $\,X_1\:,\:X_2\>$, see \eqref{Xi1}\:, are holomorphic
in the unit polydisk around $\,\qq=\0\,$, and $\,X_1(\zz\:;h;\0)\,$,
$\,X_2(\zz\:;h;\0)\,$ act on $\,\Ctnl\,$ as follows
\vvn.16>
\be
X_1(\zz\:;h;\0)\>v_{[\:a\:]}\>=\,z_a\>v_{[\:a\:]}\:-\>
h\>\sum_{b=1}^{a-1}\:v_{[\:b\:]}\,,\kern1.4em
X_2(\zz\:;h;\0)\>v_{[\:a\:]}\>=\>
\Bigl(\<-\>X_1(\zz\:;h;\0)\>+\sum_{b=1}^n\,z_b\Bigr)\,v_{[\:a\:]}\,,
\ee

\vsk.16>
The Levelt fundamental solution $\,\:\Psh(\zz\:;h;\qq\:;\ka)\,$,
see \eqref{Psipnd}\:, of differential equations \eqref{DEQ1} is
\vvn.3>
\be
\Psh(\zz\:;h;\qq\:;\ka)\,=\,\Psp(\zz/\ka;h/\ka;\qq)\;
q_1^{\>\smash{X_1(\zz\:;h;\:\0)}/\ka}\:q_2^{\>\smash{X_2(\zz\:;h;\:\0)}/\ka}\>,
\kern-1.2em
\vv.9>
\ee
where the $\,\:\End\:\bigl(\Cnnl\bigr)$-\:valued function
$\,\Psp(\zz\:;h;\qq)\,$ is as follows,
\vvn.3>
\begin{gather*}
\Psp(\zz\:;h;\qq)\::\:v_{[\:b\:]}\,\mapsto\>
\sum_{a=1}^n\,\Psp_{\ab}(\zz\:;h;\qq)\,v_{[a]}\,,\kern-1em
\\[-2pt]
\begin{aligned}
& \Psp_{\ab}(\zz\:;h;\qq)\,={}
\\
&{}\!\?=\,(1-q_2/q_1)^h\,\sum_{l=0}^\infty\,(q_2/q_1)^l\,
\sum_{c=1}^n\>\frac{W_{[a]}(z_c\<-l\:;\zz\:;h)\,\WW_{[b]}(z_c;\zz\:;h)}
{\Rb_c(\zz)\>\Qb_c(\zz\:;h)}\,\prod_{d=1}^n\,\prod_{m=0}^{l-1}\,
\frac{z_d\<-z_c\<+h+m}{z_d\<-z_c\<+1+m}\;.\kern-1.4em
\end{aligned}
\\[-19pt]
\end{gather*}
Furthermore, one has $\;\det\:\Psp(\zz\:;h;\qq)=1\,\:$ and
$\;\det\:\Psh(\zz\:;h;\qq\:;\ka)\,=\,
\bigl(\:q_1\>q_2^{\:n-1}\bigr)^{\sum_{a=1}^n z_a\</\?\ka}\:$.

\vsk.3>
For a solution $\,\Psi(\zz\:;h;\qq\:;\ka)\,$ of dynamical differential
equations \eqref{DEQ1}\:, its principal term, see \eqref{princ}\:, is
\be
\Psz(\zz\:;h;\ka)\,=\,\Psh(\zz\:;h;\qq\:;\ka)^{-1}\,\Psi(\zz\:;h;\qq\:;\ka)\,.
\vv.6>
\ee
The principal term of the solution $\,\Psi_P(\zz\:;h;\qq\:;\ka)\,$,
corresponding to a Laurent polynomial $\,P(\gmd\:;\zzd\:;\hdd)\,$,
see \eqref{PsiP0}\:, equals
\vvn.3>
\begin{align}
\label{PsiP10}
\Psz_{\?P}(\zz\:;h;\ka)\,=\,\sum_{a=1}^n\>v_{[a]}\,
\sum_{b=1}^n\>W_{[a]}(z_b;\zz\:;h)\,\Pdd(z_b;\zz\:;h;\ka)
\;e^{\:\pii\>\left(\<(n-2)\:z_b\:+\sum_{c=1}^nz_c\<\right)}\:\times{}&
\\[-3pt]
\notag
{}\times\>\prod_{\satop{c=1}{c\ne b}}^n\,
\Gm\bigl(\<(z_b\<-z_c)/\ka\bigr)\,\Gm\bigl(\<(z_c\<-z_b\<+h)/\ka\bigr)\>&\:.
\end{align}

\vsk.3>
Let $\,L\,$ be the complement of the union of the hyperplanes
\vvn.5>
\be
z_a\<-z_b\<\in\<\ka\>\Z_{\ne0}\,,\qquad z_a\<-z_b\<+h\in\<\ka\>\Z_{\le0}\,,\qquad
a,b=1\lc n\,,\quad a\ne b\,,\kern-.6em
\vv.5>
\ee
cf.~\eqref{zzhZ}\:. Denote by $\,\Oc_L\:$ the ring of
functions of $\,\zz,h\,$ holomorphic in $\,L\,$. The map
\vvn.6>
\begin{gather}
\label{Bcyr1}
\Bcyr_\bla\::\:\Kc_\bla\to\:\Cnnl\?\ox_{\:\C}\Oc_L\,,\kern-1em
\\[4pt]
\notag
\begin{aligned}
{[P\:]}\,\mapsto\:\sum_{a=1}^n\>v_{[a]}\,
\sum_{b=1}^n\>W_{[a]}(z_b;\zz\:;h)\,\Pdd(z_b;\zz\:;h;\ka)
\;e^{\:\pii\>\left.\left(\<(n-2)\:z_b\:+\sum_{c=1}^nz_c\<\right)
\<\right/\<\ka}\:\times{}&
\kern-1em
\\[-6pt]
\notag
{}\times\>\prod_{\satop{c=1}{c\ne b}}^n\,
\Gm\bigl(\<(z_b\<-z_c)/\ka\bigr)\,\Gm\bigl(\<(z_c\<-z_b\<+h)/\ka\bigr)\>&\:,
\kern-1em
\end{aligned}
\end{gather}
sends the class $\,[P\:]\in\Kc_\bla\,$ of the Laurent polynomial
$\,P(\gmd\:;\zzd\:;\hdd)\,$ to the principal term of the solution
$\,\:\Psi_P\:$ of the joint system of differential equations \eqref{DEQ1}
and difference equations \eqref{Ki1}.

\vsk.2>
Let $\>\Sroll$ be the space
of $\,\Cnnl\<$-valued solutions of dynamical differential equations \eqref{DEQ1}
\vv.06>
holomorphic in $\,\qq\,$ provided $\,\:|\:q_2/q_1|<1\,$ with branches of
\vvn.04>
$\,\:\log\:q_1\:$ and $\,\:\log\:q_2\:$ fixed, and holomorphic in $\,\zz,h\,$
in $\,L\,$. The space $\>\Srsl\>$ is a subspace of $\>\Sroll$.

\vsk.3>
Since the matrix $\,\:\Psh(\zz\:;h;\qq\:;\ka)\,$ is holomorphic in $\,\zz,h\,$
in $\,L\,$, and $\,\bigl(\det\:\Psh(\zz\:;h;\qq\:;\ka)\bigr)^{\?-1}\:$ is entire
in $\,\zz,h\,$, the inverse matrix $\,\:\Psh(\zz\:;h;\qq\:;\ka)^{-1}\:$ is also
holomorphic in $\,\zz,h\,$ in $\,L\,$. Thus the map
\vvn-.1>
\be
\muol\<:\:\Cnnl\?\ox_{\:\C}\Oc_L\to\>\Sroll,\qquad
v\:\mapsto\Psh(\zz\:;h;\qq\:;\ka)\>v\,,
\vv.7>
\ee
gives an isomorphism of $\,\Oc_L\:$-modules.
Furthermore, the following diagram is commutative,
\be
\xymatrix{\Kc_\bla\ar^-{\;\Bcyr_\bla}[rr]
\ar_{\smash{\lower.8ex\llap{$\ssize\muk\;\;\,$}}}[dr]&&
\rlap{$\Cnnl\?\ox_{\:\C}\Oc_L$}\phan{\Kc_\bla}
\ar^{\smash{\lower.8ex\rlap{$\;\ssize\muol$}}}[dl]\\
&\!\Sroll\:&}
\vv-.7>
\ee
see Proposition \ref{triangle}.

\section{Limit $\,{h\to\infty}\,$ for solutions of the dynamical and \qKZ/
equations}
\label{sec lim solns}
\subsection{Limiting weight functions $\Wo_{\?I}(\TT\:;\zz)$}
\label{sec WC}
For $\,I\<\in\Il\,$, define the {\it limiting weight functions\/}
\begin{gather}
\notag
\\[-12pt]
\label{WoI}
\Wo_{\?I}(\TT\:;\zz)\,=\,
\Sym_{\>t^{(1)}_1\!\lc\,t^{(1)}_{\la^{(1)}}}\,\ldots\;
\Sym_{\>t^{(N-1)}_1\!\lc\,t^{(N-1)}_{\la^{(N-1)}}}\Uo_I(\TT\:;\zz)\,,
\\[4pt]
\notag
\Uo_I(\TT\:;\zz)\,=\,\prod_{j=1}^{N-1}\,\prod_{a=1}^{\la^{(j)}}\,\biggl(
\prod_{\satop{c=1}{i^{(j+1)}_c\?<\>i^{(j)}_a}}^{\la^{(j+1)}}
\!\!(t^{(j)}_a\?-t^{(j+1)}_c)\,\prod_{b=a+1}^{\la^{(j)}}
\frac{1}{t^{(j)}_b\?-t^{(j)}_a}\,\biggr)\,.
\end{gather}

\vsk.2>
Recall the element \,$\Imil\>$ and the permutation $\,\si_I\,$. Set
\vvn.2>
\be
d_I\>=\:\sum_{j=1}^{N-1}\la^{(j)}(\la^{(j+1)}-1) - |\:\si_I\:|\,.
\ee

\begin{lem}
\label{lemWWo}
For $\,I\<\in\Il\,$, we have
\vvn.2>
\be
\Wo_{\?I}(\TT\:;\zz)\,=\,\lim_{h\to\infty}(-h)^{-d_I}\,W_I(\TT\:;\zz\:;h)\,.
\ee
\end{lem}
\begin{proof}
The statement follows from formulae \eqref{hWI}, \eqref{WoI} by induction
on the length of \,$\si_I\,$.
\end{proof}

\begin{example}
Let $N=2$, $n=2$, $\bla=(1,1)$, $I=(\{1\},\{2\})$, $J=(\{2\},\{1\})$. Then
\vvn.2>
\be
\Wo_{\?I}(\TT\:;\zz)\,=\,1\,,\qquad \Wo_{\!\<J}(\TT\:;\zz)\,=\,t^{(1)}_1\?-z_1\,.
\vv.4>
\ee
\end{example}

\begin{lem}
\label{Womax}
We have $\;\Wo_{\Imil}(\TT\:;\zz)=1\,$ and
\vvn-.3>
\beq
\label{WoIma}
\Wo_{\si_0(\Imil)}(\TT\:;\zz)\,=\,\prod_{j=1}^{N-1}\,\prod_{a=1}^{\la^{(j)}}\,
\prod_{b=\smash{\la^{(j)}}+1}^{\la^{(j+1)}}\!(t_a^{(j)}\?-z_{n-b+1})\,,
\vv.2>
\eeq
where $\,\si_0\<\in S_n\>$ is the longest permutation,
$\,\si_0(a)=n+1-a\,$, $\;a=1\lc n\,$.
\end{lem}

\vsk.1>
Let $\;\Upsl=\{\,i\ |\ \,\la_{i+1}\<\ne 0\,\}\subset\{\:1\lc N-1\:\}\,$.

\begin{lem}
\label{WImis}
For the transpositions $\;s_{\la^{(i)}\?,\,\la^{(i)}\<+1}\,,\,\;i\in\Upsl$,
we have
\beq
\label{Wos}
\Wo_{\?s_{\la^{(i)}\!,\>\la^{(i)}\<+1}(\Imil)}(\TT\:;\zz)\,=\,
\sum_{j=1}^{\la^{(i)}}\;(\:t^{(i)}_j\?-z_j)\,.
\eeq
\end{lem}

\noindent
Lemmas \ref{Womax} and \ref{WImis} are proved in Appendix \ref{AppB}

\vsk.4>
The following lemma describes the $\,h\to\infty\,$ limit of the three-term
relations of Lemma \ref{c3t} for functions $\,\Wo_{\?I}(\TT\:;\zz\:;h)\,$.

\begin{lem}
\label{DlW}
For any $\,I\in \Il\,$, $\,i=1\lc N-1\,$, and $\,a=1\lc n-1\,$, we have
\vvn.3>
\be
\Wo_{s_{a\<,a+1}(I)}(\TT\:;\zz)\,=\,
\frac{\Wo_{\?I}(\TT\:;\zz)-\Wo_{\?I}(\TT\:;z_1\lc z_{a+1},z_a\lc z_n)}{z_{a+1}\<-z_a}\;,
\vv.3>
\ee
if $\,|\:s_{\aa+1}\:\si_I\:|<|\si_I|\,$, and
$\;\Wo_{\?I}(\TT\:;z_1\lc z_{a+1},z_a\lc z_n)\,=\,\Wo_{\?I}(\TT\:;\zz)\,$, otherwise.
\end{lem}
\begin{proof}
Notice that $\,\si_{s_{a\<,a+1}(I)}=\:s_{\aa+1}\:\si_I\,$ unless
\vvn.1>
$\,s_{\aa+1}(I)=I\,$. Now the statement follows from Lemmas \ref{lemWWo}
and \ref{c3t}.
\end{proof}

\begin{rem}
Define the operators \,$\Dl_1\lc\Dl_{n-1}$ acting on functions of $\,\zzz\,$:
\vvn.3>
\be
%\label{Dla}
\Dl_{\:a}\:f(\zz)\,=\,\frac{f(\zz)-f(z_1\lc z_{a+1},z_a\lc z_n)}{z_a\<-z_{a+1}}\;.
\kern-1.6em
\ee
They satisfy the nil\:-\:Coxeter relations,
\vvn.3>
\be
%\label{nilCx}
\Dl_{\:a}^2=\:0\,,\kern1.8em
\Dl_{\:a}\>\Dl_{\:b}=\:\Dl_{\:b}\>\Dl_{\:a}\,,\quad |\:a-b\:|>1\,,
\kern 1.8em\Dl_{\:a}\>\Dl_{a+1}\>\Dl_{\:a}=\:\Dl_{a+1}\>\Dl_{\:a}\>\Dl_{a+1}\,,
\kern-4em
\vv.2>
\ee
for any $\,a,b\,$.

\vsk.2>
Denote $\,\yy=(\yyy)\,$. Let the functions $\,\Wto_{\?\!I}(\yy;\zz)\,$
be obtained from $\,\Wo_{\?I}(\TT\:;\zz)\,$ by the substitution
$\,t^{(i)}_j\!=y_j\,$, $\;i=1\lc N-1\,$, $\,j=1\lc\la^{(i)}\>$.
By Lemmas \ref{Womax}, \ref{DlW}, \ref{sibla}, and Proposition \ref{Sgxyx},
\vv.06>
the functions $\,\Wto_{\?\!I}(\yy;\zz)\,$ coincide with the $A\:$-type double
Schubert polynomials $\,\:\Sg_\si(\yy\:;\zz)\,$,
\beq
\label{WdS}
\Wto_{\?\!I}(\yy;\zz)\,=\,\Sg_{\si_I}\<(\yy\:;\zz)\,.\kern-1em
%(-1)^{|\si_I|}\Sg_{\si_I^{-1}}(\zz\:;\yy)\,.
\eeq
\end{rem}

For $\,I\<\in\Il\,$, define
\vvn-.1>
\beq
\label{WocI}
\WWo_{\?I}(\TT\:;\zz)\,=\,\Wo_{\?\si_0(I)}(\TT\:;z_n\lc z_1)\,,
\vv.4>
\eeq
where $\,\si_0\,$ is the longest permutation. Recall $\,R_\bla(\zz)\,$,
see \eqref{RQ}\:, and $\,\zz_\si\<=(z_{\si(1)}\lc z_{\si(n)})\,$.

\vsk.3>
Lemmas \ref{WIzo} and \ref{lemorto} below follow by Lemma \ref{lemWWo} from
Lemma \ref{WIz}, Proposition \ref{lemorth}, and Corollary \ref{cororth},
respectively. Lemmas \ref{WIzo} and \ref{lemorto} are equivalent to the
vanishing and orthogonality properties of the double Schubert polynomials.

\begin{lem}
\label{WIzo}
For $\,I,J\<\in\Il\,$, we have $\;\Wo_{\!\<J}(\Si_I\:;\zz\:;h)=0\>$ unless $\,I=J$ or
$\,|\:\si_I\:|>|\:\si_J\:|\,$, and
\vvn.2>
\beq
\label{WIIo}
\Wo_{\?I}(\Si_I\:;\zz)\,=\,\prod_{j=1}^{N-1}\>\prod_{k=j+1}^N\,
\prod_{a\in I_j}\>\prod_{\satop{b\in I_k}{b<a}}\,(z_a\<-z_b)\,.\kern-1.4em
\eeq
\end{lem}

\begin{lem}
\label{lemorto}
The functions $\,\Wo_{\?I}(\TT\:;\zz)\:$ and $\;\WWo_{\!\<J}(\TT\:;\zz)\:$
are biorthogonal\:,
\vvn.5>
\beq
\label{ortho}
\sum_{I\in\Il}\,\frac{\Wo_{\!\<J}(\Si_I\:;\zz)\,\WWo_{\!\?K}(\Si_I\:;\zz)}
{R_\bla(\zz_{\si_I}\<)}\,=\,\dl_{\JK}\,,\kern1.8em
\sum_{I\in\Il}\,\Wo_{\?I}(\Si_J\:;\zz)\,\WWo_{\!I}(\Si_K;\zz)\,=\,
\dl_{\JK}\,R_\bla(\zz_{\si_J}\<)\,.\kern-.4em
\vv.6>
\eeq
\end{lem}

\subsection{Limiting master function}
\label{limmaster}
Define the {\it limiting master function\/}:
\vvn.4>
\begin{align}
\label{Phio}
\kern.4em\Pho_\bla(\TT\:;\zz\:;\pp\:;\ka)\,&{}=\,
(\:\ka^{\>\la_N\<-\:n}\:p_N)^{\>\sum_{a=1}^nz_a\</\<\ka}\,\>
\prod_{i=1}^{N-1}\,\bigl(\:\ka^{\>\la_i+\la_{i+1}}\>p_i/p_{i+1}\:\bigr)
^{\sum_{a=1}^{\la^{(i)}}\:t^{(i)}_a\!\</\ka}\times{}
\\[2pt]
\notag
&{}\times\>\prod_{i=1}^{N-1}\,\prod_{a=1}^{\la^{(i)}}\,\biggl(\,
\prod_{\satop{b=1}{b\ne a}}^{\la^{(i)}}\,
\frac1{\Gm\bigl(\<(t_a^{(i)}\?-t_b^{(i)})/\ka\bigr)}\,\prod_{c=1}^{\la^{(i+1)}}
\<\Gm\bigl(\<(t_a^{(i)}-t_c^{(i+1)})/\ka\bigr)\<\biggr)\,,
\\[-19pt]
\notag
\end{align}
where $\,\la^{(N)}\?=n\,$ and $\,t^{(N)}_a\?=z_a\,$, $\;a=1\lc n\,$.
The last formula can be also written as
\vvn.2>
\begin{align*}
\Pho_\bla(\TT\:;\zz\:;\pp\:;\ka)\,&{}=\,
(\:\ka^{\>\la_N\<-\:n}\:p_N)^{\>\sum_{a=1}^nz_a\</\<\ka}\,\>
\prod_{i=1}^{N-1}\,\bigl(\:\ka^{\>\la_i+\la_{i+1}}\>p_i/p_{i+1}\:\bigr)
^{\sum_{a=1}^{\la^{(i)}}\:t^{(i)}_a\!\</\ka}\times{}
\\[4pt]
&{}\times\>\prod_{i=1}^{N-1}\,\prod_{a=1}^{\la^{(i)}}\,\biggl(\,
\prod_{b=1}^{a-1}\,\frac{(t_a^{(i)}\?-t_b^{(i)})\:
\sin\bigl(\pi\>(t_b^{(i)}\?-t_a^{(i)})/\ka\bigr)}{\pi\ka}\,
\prod_{c=1}^{\la^{(i+1)}}\<\Gm\bigl(\<(t_a^{(i)}-t_c^{(i+1)})/\ka\bigr)
\<\biggr)\,,
\\[-14pt]
\end{align*}
Denote by $\,\Pht_\bla(\TT\:;\zz\:;h;\pp\:;\ka)\,$ the function obtained from
\vvn.1>
the master function $\,\Phi_\bla(\TT\:;\zz\:;h;\qq\:;\ka)\,$ by the substitution
\vvn-.1>
\beq
\label{qp}
q_i\>=\,p_i\,h^{\>\sum_{j=i+1}^N\la_j\,-\,\sum_{j=1}^{i-1}\:\la_j}\,
e^{\:\pii\,(\:\la_i\<-\:n)}\>,\qquad i=1\lc N\>,\kern-2em
\vv.2>
\eeq
cf.~\eqref{q->tqt}\:.

\vsk.2>
Recall $\,\:\la\+1\<=\sum_{i=1}^{N-1}\la^{(i)}$,
$\,\la\+2=\,\sum_{i=1}^{N-1}\,\bigl(\la^{(i)}\bigr)^2$,
$\,\la_{\{2\}}=\sum_{1\le i<j\le N}\>\la_i\>\la_j\,$.

\begin{lem}
%[\cite{TV6}]
\label{lemPho}
For $\,|\<\arg\:(h/\ka)\:|<\pi$, \,we have
\be
\Pho_\bla(\TT\:;\zz\:;\pp\:;\ka)\,=\,
\lim_{h\to\infty}\,(\<-\:h)^{\:\la\+2\<-\:\la\+1}
\bigl(\:\Gm(h/\ka)\bigr)^{\?-\:\la\+1\<-\:\la_{\{2\}}}\,\:
\Pht_\bla(\TT\:;\zz\:;h;\pp\:;\ka)\,.\kern-1.6em
\vv.2>
\ee
\end{lem}
\begin{proof}
The statement follows from formulae \eqref{Phi}\:, \eqref{Phio}
by Stirling's formula
\vvn.1>
\beq
\label{Stir}
\frac{\Gm(\al+h/\ka)}{\Gm(\bt+h/\ka)}\;\sim\,(h/\ka)^{\>\al-\bt}\>,
\qquad h\to\infty\,,\quad |\<\arg\:(h/\ka)\:|<\pi\,.\kern-2em
\vv-1.5>
\eeq
\vv.9>
\end{proof}

\subsection{Solutions of the limiting dynamical and \qKZ/ equations}
\label{sec 5.4}
For a polynomial $\,f(\TT\:;\zz)\,$, define the Jackson integral
\vvn.2>
\beq
\label{MSio}
\Mc_J(\Pho_\bla f)(\zz\:;\pp\:;\ka)\,=\!
\sum_{\lb\in\Z^{\la\+1}\!\!\!\!}\Res_{\>\TT\>=\>\Si_J-\:\lb\ka\>}
\bigl(\Pho_\bla(\TT\:;\zz\:;\pp\:;\ka)\>f(\TT\:;\zz)\bigr)\,,\qquad J\<\in\Il\,.
\kern-1.3em
\vv.1>
\eeq
Notice that the master function $\,\Pho_\bla(\TT\:;\zz\:;\pp\:;\ka)\,$
has only simple poles, and
%\eqref{MSio}
\vvn.4>
\be
\Res_{\>\TT\>=\>\Si_J-\:\lb\ka\>}
\bigl(\Pho_\bla(\TT\:;\zz\:;\pp\:;\ka)\>f(\TT\:;\zz)\bigr)\,=\,
f(\Si_J\<-\lb\ka\:;\zz)\,
\Res_{\>\TT\>=\>\Si_J-\:\lb\ka\>}\Pho_\bla(\TT\:;\zz\:;\pp\:;\ka)\,.
\kern-1em
\vv.4>
\ee
A closed expression for the residue
$\,\Res_{\>\TT\>=\>\Si_J-\:\lb\ka\>}\Pho_\bla(\TT\:;\zz\:;\pp\:;\ka)\,$
%of the master function
is given by Lemma \ref{ResPho} below.

\vsk.3>
Set
\vvn-.7>
\beq
\label{Mlao}
\Mo_\bla(\zz\:;\ka)\,=\,\pi^{-\la_{\{2\}}}\>
\prod_{i=1}^{N-1}\>\prod_{a=1}^{\la^{(i)}}\,
\prod_{b=\la^{(i)}+1}^{\la^{(i+1)}}\!
\sin\:\bigl(\pi\>(z_a\?-z_b)/\ka\bigr)\,,\kern-2em
\vv-.5>
\eeq
cf.~\eqref{Mla}\:, and
\vvn-.2>
\beq
\label{Ato}
\Ao\<(\TT\:;\zz\:;\ka)\,=\,\prod_{i=1}^{N-1}\,
\prod_{a=1}^{\la^{(i)}}\;\biggl(\,\prod_{\satop{b=1}{b\ne a}}^{\la^{(i)}}
\>\Gm\bigl(1+(t^{(i)}_b\!-t^{(i)}_a)/\ka\bigr)
\prod_{c=1}^{\la^{(i+1)}}\:
\frac1{\Gm\bigl(1+(t^{(i+1)}_c\!-t^{(i)}_a)/\ka\bigr)}\,\biggr)\>.
\eeq

\begin{lem}
\label{ResPho}
If $\;\lb\not\in\<\Z_{\ge0}^{\la\+1}\?$, then
$\;\Res_{\>\TT\>=\>\Si_J-\:\lb\ka\>}\Pho_\bla(\TT\:;\zz\:;\pp\:;\ka)=0\,$.
For $\;\lb\in\<\Z_{\ge0}^{\la\+1}\?$,
\vvn.3>
\begin{align}
\label{ResPholb}
& \Res_{\>\TT\>=\>\Si_J-\:\lb\ka}\Pho_\bla(\TT\:;\zz\:;\pp\:;\ka)\,=\,
\frac{\Ao_\bla(\Si_J\<-\lb\ka\:;\zz\:;\ka)}{\Mo_\bla(\zz_{\si_J}\:;\ka)}\,
\times{}
\\[4pt]
\notag
&{}\,\times\,\ka^{\:\la\+1}\>
\prod_{i=1}^N\,\bigl(\:\ka^{\>\sum_{j=i+1}^N\la_j\,-\,\sum_{j=1}^{i-1}\:\la_j}
\>p_i\:\bigr)^{\>\sum_{a\in J_i}\?z_a\</\<\ka}\,
\prod_{i=1}^{N-1}\,\bigl(\<(-\:\ka)^{-\la_i-\la_{i+1}}\>
p_{i+1}/p_i\:\bigr)^{\>\sum_{a=1}^{\la^{(i)}}l_a^{(i)}}\!.
\\[-19pt]
\notag
\end{align}
In particular,
\vvn.3>
\begin{align}
\label{ResPhoi}
& \Res_{\>\TT\>=\>\Si_J}\?\Pho_\bla(\TT\:;\zz\:;\pp\:;\ka)\,={}
\\[4pt]
\notag
&{}=\,\ka^{\:\la\+1}\>\prod_{i=1}^N\,\bigl(\:
\ka^{\>\sum_{j=i+1}^N\la_j\,-\,\sum_{j=1}^{i-1}\:\la_j}\>
p_i\:\bigr)^{\>\sum_{a\in J_i}\?z_a/\ka}\;
\prod_{i=1}^{N-1}\prod_{j=i+1}^N\,\:\prod_{a\in J_i}\,
\prod_{b\in J_j}\,\Ga\bigl(\<(z_a\<-z_b)/\ka\bigr)\,.\kern-.06em
\end{align}
\end{lem}

\vsk.1>
By Lemma \ref{ResPho}, the actual summation in formula \eqref{MSio}
is only over the positive cone of the lattice,
\beq
\label{MSio<}
\Mc_J(\Phi_\bla\:f)(\zz\:;\pp\:;\ka)\,=\>
\sum_{\lb\in\Z_{\ge0}^{\:\la\+1}\!}f(\Si_J\<-\lb\ka\:;\zz\:;\pp)\,
\Res_{\>\TT\>=\>\Si_J-\:\lb\ka}\Pho_\bla(\TT\:;\zz\:;\pp\:;\ka)\,.
\kern-2em
\eeq

\begin{lem}
\label{lemresh}
For $\,|\<\arg\:(h/\ka)\:|<\pi$, \,we have
\vvn.3>
\begin{align*}
& \Res_{\>\TT\>=\>\Si_J-\:\lb\ka}\Pho_\bla(\TT\:;\zz\:;\pp\:;\ka)\,={}
\\[3pt]
&\!{}=\>\lim_{h\to\infty}\>
(\<-\:h)^{\:\la\+2\<-\:\la\+1}
\bigl(\:\Gm(h/\ka)\bigr)^{\?-\:\la\+1\<-\:\la_{\{2\}}}\>
\Res_{\>\TT\>=\>\Si_J-\:\lb\ka}\Pht_\bla(\TT\:;\zz\:;h;\pp\:;\ka)\,.
\\[-13pt]
\end{align*}
The convergence is uniform in $\,\lb$ and locally uniform in $\zz,\pp\,$.
\end{lem}
\begin{proof}
The claim follows from formulae \eqref{ResPhlb}\:, \eqref{qp}\:,
\eqref{ResPholb} \:by Lemmas \ref{Stira}, \ref{Stirb}
that are detalization of Stirling's formula \eqref{Stir}\:.
\end{proof}

Define
\vvn-.1>
\beq
\label{Pso}
\Pso_{\!J}(\zz\:;\pp\:;\ka)\,=\,\ka^{-\la\+1}\>\Omo_\bla(\pp\:;\ka)\,
\sum_{I\in\Il} \Mc_J(\Pho_\bla \Wo_{\?I})(\zz\:;\pp\:;\ka)\,v_I\,,
\vv-.3>
\eeq
where
\beq
\label{Omlo}
\Omo_\bla(\pp\:;\ka)\,=\,e^{\:\sum_{\:i<j}^{\cirs}p_j\</(\ka\:p_i)}
\vv.4>
\eeq
the sum taken over all pairs $\,1\le i<j\le N\>$ such that
\vv.1>
$\,\la_i\<=1\,$ and $\,\la_s\<=0\,$ for all $\,s=i+1\lc j\,$.

\vsk.2>
Recall the function $\,\Psi_{\?J}(\zz\:;h;\qq\:;\ka)\,$, see \eqref{mcF}\:.
Let $\,\Pst_{\?J}(\zz\:;h;\pp\:;\ka)\,$ be the function obtained from
$\,\Psi_{\?J}(\zz\:;h;\qq\:;\ka)\,$ by substitution \eqref{qp}\:.

\begin{prop}
\label{lem5.5}
For $\,|\<\arg\:(h/\ka)\:|<\pi$, \,we have
\vvn.2>
\beq
\label{limPs}
\Pso_{\!J}(\zz\:;\pp\:;\ka)\,=\>\lim_{h\to\infty}\>
\bigl(-\:h\:\Gm(h/\ka)\bigr)^{\?-\la_{\{2\}}}\>
(\<-\:h)^{\:\sum_{b<c,\,j<k}e^{(b)}_{k,k}e^{(c)}_{j,j}}\;
\Pst_{\?J}(\zz\:;h;\pp\:;\ka)\,.\kern-2em
\vv.24>
\eeq
The convergence is locally uniform in $\zz,\pp\,$.
\end{prop}
\begin{proof}
The statement follows from formulae \eqref{MSi<}\:, \eqref{mcF}\:,
\eqref{Oml}\:, \eqref{MSio<}\:, \eqref{Pso}\:, \eqref{Omlo}\:,
and Lemmas \ref{lemWWo}, \ref{lemresh}.
\vv-.5>
\end{proof}

\begin{defn}
\label{entdef}
Say that a function $\,f(\pp)\,$ is {\it entire in} $\,\pp_{\!\dvs}$
\vvn.06>
if $\,f(\pp)=g(p_2/p_1\lc p_N/p_{N-1})\,$ for an entire function
$\,g(s_1\lc s_{N-1})\,$. Denote $\,f(\0)=g(0\lc 0)\,$.
\end{defn}

\begin{thm}
\label{thmcyo}
The $\,\Cnnl$-valued function $\,\Pso_{\?J}(\zz\:;\pp\:;\ka)\:$ \,is a solution
of the joint system of dynamical differential equations \eqref{DEQo} and \qKZ/
difference equations \eqref{Kio}. It has the form
\vvn-.5>
\begin{align}
\label{PsiPsdo}
\Pso_{\?J}(\zz\:;\pp\:;\ka)\,=\,\Psdo_{\?J}(\zz\:;\pp\:;\ka) & \;
\prod_{i=1}^N\,\bigl(\:\ka^{\>\sum_{j=i+1}^N\la_j\,-\,\sum_{j=1}^{i-1}\:\la_j}
\>p_i\:\bigr)^{\>\sum_{a\in J_i}\?z_a\</\<\ka}\:\times{}
\\[1pt]
\notag
& {}\,\;\times\,\prod_{i=1}^{N-1}\prod_{j=i+1}^N\,\:\prod_{a\in J_i}\>
\prod_{b\in J_j}\>\frac1{\sin\bigl(\pi\>(z_a\<-z_b)/\ka\bigr)}\;,
\\[-14pt]
\notag
\end{align}
where the function $\,\Psdo_{\?J}(\zz\:;\pp\:;\ka)\>$ is entire in $\,\zz$
and is entire in $\,\pp_{\!\dvs}\:$. In more detail,
\vvn.2>
\begin{align}
\label{Psdo}
\Psdo_{\?J}(\zz\:;\pp\:;\ka)\,&{}=\,
\prod_{i=1}^{N-1}\prod_{j=i+1}^N\,\:\prod_{a\in J_i}\>\prod_{b\in J_j}\>
\frac\pi{\Gm\bigl(1+(z_b\<-z_a)/\ka\bigr)}\,\times{}
\\[3pt]
\notag
&\>{}\times\,\biggl(\Psdo_{\<J,\:0}\:(\zz)\>+\!\!
\sum_{\satop{\,\mb\in\Z_{\ge0}^{N-1}\!}{\mb\ne0}}\!\?
\Psdo_{\<J,\:\mb}(\zz\:;\ka)\,
\prod_{i=1}^{N\<-1}\>(\:p_{i+1}/p_i)^{m_i}\<\biggr)\:,\kern-2em
\\[-26pt]
\notag
\end{align}
where
\beq
\label{Psd0o}
\Psdo_{\<J,\:0}\:(\zz)\,=\,\Wo_{\!\<J}(\Si_J\:;\zz)\,v_J\,+\!
\sum_{\satop{I\<\in\Il}{|\si_I\<|<|\si_J\<|}}\!\!\Wo_{\?I}(\Si_J\:;\zz)\,v_I
\kern-2em
\eeq
is a polynomial in $\>\zz\,$, and
$\,\Psdo_{\<J,\:\mb}(\zz\:;\ka)$ for $\,\mb\ne0\>$ are rational functions
of $\,\:\zz,\ka\>$ with at most simple poles
\vvn.1>
on the hyperplanes $\,z_a\<-z_b\<\in\<\ka\>\Z_{>0}\>$ for $\,a\in\<J_i\,$,
$\,b\in\<J_j\,$, $\,1\le i<j\le N\>$. Furthermore, for any transposition
\,$s_{\ab}\<\in\<S_n\,$,
\vvn.3>
\beq
\label{Psdoab}
\Psdo_{\?J}(\zz\:;\pp\:;\ka)\big|_{\:z_a=z_b}=\,
\Psdo_{\<s_{a\<,b}(J)}(\zz\:;\pp\:;\ka)\big|_{\:z_a=z_b}\:.
\vv.3>
\eeq
\end{thm}
\begin{proof}
The statement follows from Theorem \ref{thm cy}, Lemmas \ref{lem K-lim},
\ref{lem X-lim}, \ref{lemWWo}, \ref{ResPho}, \ref{lemresh}, and
Proposition \ref{lem5.5}. See also \cite[Section 11]{TV6}\:.
\end{proof}

The functions $\,\Pso_{\?J}(\zz\:;\pp\:;\ka)$ \,are called
the {\it multidimensional
%\,$q\:$-
hypergeometric solutions\/} of the dynamical equations \ref{DEQo}
and \qKZ/ difference equations \eqref{Kio}.

\vsk.3>
The next theorem computes the determinant of coordinates of solutions
$\,\Pso_{\?J}(\zz\:;\pp\:;\ka)\,$ and is analogous to
\cite[formula~(11.23)]{TV6}\:.

\begin{thm}
\label{thmdeto}
Let \,$n\ge 2$\,. Then
\vvn.3>
\begin{align}
\label{detPso}
\det\: & \bigl(\>\Omo_\bla(\pp\:,\ka)\,
\:\Mc_J(\Pho_\bla\Wo_{\?I})(\zz\:;\pp\:;\ka)\bigr)_{\IJ\in\>\Il}\<={}
\\[6pt]
\notag
&{}\?=\,\ka^{\:\la\+1d_\bla}\,
\prod_{i=1}^N\,p_i^{\,\smash{d^{(1)}_{\bla,i}\>\sum_{a=1}^n z_a\</\<\ka}}\;
\prod_{a=1}^{n-1}\,\prod_{b=a+1}^n\<\biggl(\:
\frac{\pi\:\ka}{\sin\bigl(\pi\>(z_a\<-z_b)/\ka\bigr)}\:\biggr)
^{\!d^{(2)}_\bla},\kern-2.6em
\end{align}
where $\;\la\+1\<=\sum_{i=1}^{N-1}\:(N\?-i)\>\la_i\>$ and
$\;d_\bla\,,\,d^{(1)}_{\bla\:,\:i}\,,\,d^{(2)}_\bla$ are given
by formulae \,\eqref{dla12}\:.
\end{thm}
\begin{proof}
The statement follows from Lemma \ref{lemWWo}, Proposition \ref{lem5.5}, and
formula \eqref{detPsi}\:.
\vsk.2>
Alternatively, denote by $\,F(\zz\:;\pp)\,$ the determinant in the left\:-hand
side of formula \eqref{detPso}\:. By Theorem \ref{thmcyo}, it solves
the differential equations
\vvn.3>
\be
\Bigl(\:\ka\>p_i\:\frac{\der}{\der p_i}\>-\>
\tr\:\Xo_i(\zz\:;\pp)|_{\:\Cnnl}\Bigr)\>F(\zz\:;\pp)\,=\,0\,, \qquad
i=1\lc N\>,\kern-2em
\vv.4>
\ee
where $\,\Xo_i(\zz\:;\pp)|_{\:\Cnnl}\:$ are the restrictions of dynamical
Hamiltonians \eqref{Xo} to the invariant subspace $\,\Cnnl$. Since
\vv.06>
$\;\tr\:\Xo_i(\zz\:;h;\qq)|_{\>\Cnnl}=d^{(1)}_{\bla\:,\:i}\>\sum_{a=1}^n z_a\,$,
\vv.1>
the function $\,F(\zz\:;\pp)\,$ equals the product of powers of $\,p_1\lc p_n\:$
\vv.1>
in the right\:-hand side of formula \eqref{detPsi} multiplied by a factor that
does not depend on $\,\pp\,$. This factor can be found by taking the limit
\,$p_{i+1}/p_i\to 0\,$ for all $\,i=1\lc N-1\,$, using Theorem \ref{thmcyo}.
\end{proof}

\begin{rem}
By Theorem \ref{thmcyo}, the determinant $\,F(\zz\:;\pp)\,$ in
Theorem \ref{thmdeto} solves the difference equations
\vvn.1>
\beq
\label{FdetKo}
F(z_1\lc z_a+\ka\lc z_n;\pp)\,=\,
\det\:\Ko_a(\zz\:;\pp\:;\ka)|_{\:\Cnnl}\,F(\zz\:;\pp)\,,\qquad a=1\lc n\>,
\vv.2>
\eeq
where $\,\Ko_a(\zz\:;h;\qq\:;\ka)|_{\:\Cnnl}$ are the restrictions of \qKZ/
\vv-.13>
operators \eqref{Ko} to the invariant subspace $\,\Cnnl$. Since
\vv.1>
$\;\det\:\Ko_a(\zz\:;\pp\:;\ka)=(-1)^{\:d^{(2)}_\bla}$, equations
\eqref{FdetKo} determine the product of sines in the right\:-hand side of
formula \eqref{detPso} up to a $\:\ka\:$-periodic function of $\,\zzz\,$.
\end{rem}

\subsection{Solutions parametrized by Laurent polynomials}
\label{seclauro}
Recall the notation from Section \ref{seclaur}.
For a Laurent polynomial $\,P(\:\GGd\:;\zzd)\,$, set
\vvn.3>
\beq
\label{PPIo}
\Pso_{\?P}(\zz\:;\pp\:;\ka)\,=\,
\sum_{J\in\Il}\,\Pdd(\zz_{\si_J};\zz\:;\ka)\,\Pso_{\?J}(\zz\:;\pp\:;\ka)\,.
\kern-2em
\eeq

Let $\,\Pst_{\?P}(\zz\:;h;\pp\:;\ka)\,$ be the function obtained
from $\,\Psi_{\?P}(\zz\:;h;\qq\:;\ka)\,$, see \eqref{PPI}\:,
by substitution \eqref{qp}\:.

\begin{lem}
\label{lem5.13}
For $\,|\<\arg\:(h/\ka)\:|<\pi$, \,we have
\vvn.2>
\beq
\label{limPsP}
\Pso_{\!P}(\zz\:;\pp\:;\ka)\,=\>\lim_{h\to\infty}\>
\bigl(-\:h\:\Gm(h/\ka)\bigr)^{\?-\la_{\{2\}}}\>
(\<-\:h)^{\:\sum_{b<c,\,j<k}e^{(b)}_{k,k}e^{(c)}_{j,j}}\;
\Pst_{\?P}(\zz\:;h;\pp\:;\ka)\,.\kern-2em
\vv.24>
\eeq
The convergence is locally uniform in $\zz,\pp\,$.
\end{lem}
\begin{proof}
The statement follows from formulae \eqref{PPI}\:, \eqref{PPIo}\:,
and Proposition \ref{lem5.5}.
\end{proof}

\begin{prop}
\label{PsiPsolo}
The function $\,\Pso_{\?P}(\zz\:;\pp\:;\ka)\>$ is a solution of the joint system
\vvn.06>
of limiting dynamical differential equations \eqref{DEQo} and \qKZ/ difference
equations \eqref{Kio}\:. Furthermore, for
$\>P\?\in\C[\:\GGd^{\pm1}]^{\:S_\bla}\?\ox\C[\:\zzd^{\pm1}]\,$,
the function $\,\Pso_{\?P}(\zz\:;\pp\:;\ka)\>$ is entire in $\,\zz\,$ and
is holomorphic in $\,\pp\,$ provided a branch of $\,\:\log\:p_i\:$ is fixed
for each \,$i=1\lc N$.
\end{prop}
\begin{proof}
The statement follows from Propositions \ref{PsiPsol1}, \ref{PsiPsol2},
and Lemma \ref{lem5.13}.
\end{proof}

Denote by $\>\Srsol\:$ the space of solutions of the system of dynamical
\vv.1>
differential equations \eqref{DEQo} and \qKZ/ difference equations \eqref{Kio}
\vv.06>
spanned over $\,\C\,$ by the functions $\,\Pso_{\?P}(\zz\:;\pp\:;\ka)\,$,
$\>P\?\in\C[\:\GGd^{\pm1}]^{\:S_\bla}\?\ox\C[\:\zzd^{\pm1}]\,$.
The space $\>\Srsol\:$ is a $\,\C[\:\zzd^{\pm1}]\:$-\:module with
$f(\zzd)\,$ acting as multi\-pli\-ca\-tion by $\fdd(\zz)\,$.

\vsk.3>
Define the algebra
\vvn-.4>
\beq
\label{Krelo}
\Kco_\bla\:=\,\C[\:\GGd^{\pm1}]^{\:S_\bla}\?\ox\C[\:\zzd^{\pm1}]\>\Big/
\Bigl\bra\,\prod_{i=1}^N\prod_{j=1}^{\la_i}\,(u-\gmd_{\ij})\,=\,
\prod_{a=1}^n\,(u-\zdd_a)\Bigr\ket\,,\kern-1.6em
\vv.1>
\eeq
where $\,u\,$ is a formal variable. By \eqref{PPIo}\:,
\vvn.2>
the assignment $\,P\mapsto\Pso_{\?P}\>$ defines a homomorphism
\beq
\label{muko}
\muko:\:\Kco_\bla\to\:\Srsol\,,\qquad Y\<\mapsto\Pso_{\<Y}\>,
\eeq
of $\,\C[\:\zzd^{\pm1}]\:$-\:modules.

\vsk.2>
By Proposition \ref{PA1}, \ref{PA2}, the algebra $\,\Kco_\bla\,$ is
\vvn.3>
a free $\,\C[\:\zzd^{\pm1}]\:$-\:module generated by the classes
\beq
\label{YIo}
Y_I(\:\GG)\,=\>V_I(\gmd_{1,1}^{-1}\lc\gmd_{1,\>\la_1}^{-1}\:\lc
\gmd_{N\?,1}^{-1}\lc\gmd_{N\?,\>\la_N}^{-1})\,,
\qquad I\<\in\Il\,,\kern-2em
\vv.2>
\eeq
where the polynomials $\,V_{\<I}\>$ are defined by formula \eqref{VIx}\:.
Introduce the coordinates of solutions $\,\Pso_{\<Y_I}\>$:
\vvn-.2>
\be
%\label{PsiYIo}
\Pso_{\<Y_I}\<(\zz\:;\pp\:;\ka)\,=\>
\sum_{J\in\Il}\,\Psbo_{\?\IJ}(\zz\:;\pp\:;\ka)\,v_J\,.
\vv-.2>
\ee
\begin{thm}
\label{detYo}
Let \,$n\ge 2$\,. Then
\vvn.3>
\begin{align}
\label{detPsoY}
& \det\:\bigl(\Psbo_{\IJ}(\zz\:;\pp\:;\ka)\bigr)_{\IJ\in\>\Il}\<={}
\\[1pt]
\notag
&\,\:{}=\,\bigl(2\:\piit\;\ka\:\bigr)^{n\:(n-1)\>d^{(2)}_\bla\!\</\:2}\:
\ka^{\:\la\+1d_\bla}\>
\Bigl(\:e^{-\pii\,(n-1)\>d^{(2)}_\bla}\>\prod_{i=1}^N\,
p_i^{\:\smash{d^{(1)}_{\bla,i}}\vp|}\>\Bigr)^{\sum_{a=1}^n z_a\</\<\ka}\;
\prod_{j=2}^{n-1}\,j^{\:(n-j)\>d^{(2)}_\bla}\:,
\\[-25pt]
\notag
\end{align}
where $\;d_\bla\,,\,d^{(1)}_{\bla\:,\:i}\,,\,d^{(2)}_\bla$ are given by
formulae \,\eqref{dla12}\:.
\end{thm}
\begin{proof}
The statement follows from Theorem \ref{thmdeto} and formula \eqref{detVI}\:.
Alternatively, the statement follows from Theorem \ref{detY} and
Lemma \ref{lem5.13}\:.
\end{proof}

\begin{cor}
\label{muko=}
The map $\,\muko:\:\Kco_\bla\<\to\Srsol\>$ is an isomorphism of
$\;\C[\:\zzd^{\pm1}]\:$-\:modules.
\end{cor}

\begin{rem}
The algebra $\,\Kco_\bla\:$ is the equivariant $\>K\?$-theory algebra
$\,K_T(\Fla\>;\C)\>$ of the partial flag variety $\,\Fla\>$,
see Section \ref{s6.6}.
\end{rem}

\subsection{Levelt fundamental solution}
\label{secLevo}

Recall Definition \ref{entdef} of a function $\,f(\pp)\,$ entire
\vvn.07>
in $\,\pp_{\!\dvs}$.
%if $\,f(\pp)=g(p_2/p_1\lc p_N/p_{N-1})\,$ for an entire function
%$\,g(s_1\lc s_{N-1})\,$, and $\,f(\0)=g(0\lc 0)\,$.
%\vsk.3>
The dynamical Hamiltonians \,$\Xo_1(\zz\:;\pp)\lc\Xo_n(\zz\:;\pp)\,$ given by
\eqref{Xo} are entire in $\,\pp_{\!\dvs}$ and
\vvn.2>
\beq
\label{Xo0}
\Xo_i(\zz\:;\0)\,=\,
\sum_{a=1}^n\,z_a\,e^{(a)}_{\ii}-\sum_{1\le a<b\le n}\?\biggl(\,
\sum_{j=1}^{i-1}\>Q_{\ij}^{\>\ab}\:-\?
\sum_{j=i+1}^N\?Q_{\ij}^{\>\ba}\>\biggr)\,.
\kern-.2em
\eeq
Notice that for $\,I\<\in\Il\,$,
\vvn.1>
\be
\Xo_i(\zz\:;\0)\,v_I\>=\>\sum_{\:a\in I_i}z_a\:v_I\:+\!\<
\sum_{\satop{J\in\:\Il}{|\si_J\<|<|\si_I\<|}}\!\xi_{\:i,\:\IJ}\,v_J\,,
\ee
where the coefficients $\,\xi_{\:i,\:\IJ}\>$ take values $\;0\:,\pm\>1\,$.
\vv.1>
Therefore, the eigenvalues of the restriction of the operator
$\,\Xo_i(\zz\:;\0)\,$ on $\,\Cnnl\,$ are $\,\sum_{\:a\in I_i}\<z_a\,$,
$\,I\<\in\Il\,$. A more detailed statement is given
by Proposition \ref{egvXo0} below.

\vsk.2>
Recall the function $\,\Psdo_{\<I\<,\:0}\:(\zz)\,$, $\,I\<\in\Il\,$,
given by \eqref{Psd0o}\:.

\begin{prop}
\label{egvXo0}
Given $\,I\<\in\Il\,$, we have
\vvn.16>
$\;\Xo_i(\zz\:;\0)\,\Psdo_{\<I\<,\:0}\:(\zz)\:=\:
\sum_{\:a\in I_i}\?z_a\,\Psdo_{\<I\<,\:0}\:(\zz)\,$, and
$\,\Psdo_{\<I\<,\:0}\:(\zz)\ne 0\>$ provided $\,z_a\<\ne z_b\,$ for all pairs
$\,a,b\>$ such that $\,a<b\>$ and $\,\si_I^{-1}(a)>\si_I^{-1}(b)\,$.
\end{prop}
\begin{proof}
The first part of the statement follows from Theorem \ref{thmcyo}.
The nonvanishing of $\,\Psdo_{\<I\<,\:0}\:(\zz)\,$ is implied by
formula \eqref{Psd0o} and Lemma \ref{WIzo}.
\end{proof}

For $\,I\<\in\Il\,$, set
$\,\EE_I(\zz)=\bigl(E^{(1)}_I(\zz)\lc E^{(N\<-1)}_I(\zz)\bigr)\,$,
\vv.05>
where $\,E^{(i)}_I\<(\zz)=\sum_{j=1}^i\:\sum_{\:a\in I_j}\<z_a\,$
is the eigenvalue of
%the operators
$\,\Xo_1(\zz\:;\0)\lsym+\Xo_i(\zz\:;\0)\,$ on $\,\:\Psdo_{\<I\<,\:0}\:(\zz)\,$.
\vv.07>
For $\,I\<,J\<\in\Il\,$, denote by \,$D_{\IJ}$ the set of points $\,\zz\,$
such that $\EE_I(\zz)-\EE_J(\zz)\in\Z^{N\<-1}_{\ge 0}$ and
$\,\EE_I(\zz)\ne\EE_J(\zz)\,$.
Set $\,D_\bla\<=\:\bigcup_{\:\IJ\in\:\Il}D_{\IJ}\,$.

\begin{thm}
\label{Bthmo}
{\rm(\:i\:)}\enspace
For any $\,\zz\>$ such that $\,z_a\<-z_b\not\in\<\ka\>\Z_{\:\ne\:0}\>$
\vvn.1>
for all $\,\:1\le a<b\le n\,$, there exists an
$\,\:\End\:\bigl(\Cnnl\bigr)$-\:valued function $\,\Pspo\<(\zz\:;\pp\:;\ka)\,$,
\vvn.1>
entire in $\,\pp_{\!\dvs}\:$ such that $\,\Pspo\<(\zz\:;\0\:;\ka)\>$ is
the identity operator and the function
\vvn-.4>
\beq
\label{Psopnd}
\Psho\<(\zz\:;\pp\:;\ka)\,=\,\Pspo\<(\zz\:;\pp\:;\ka)\,\:
\prod_{i=1}^N\,p_i^{\>\smash{\Xo_i(\zz\:;\0)}/\ka}\>,\kern-1.2em
\vv.1>
\eeq
solves dynamical differential equations \eqref{DEQo}.
\vv.1>
For given $\,\zz\,$, the function $\,\Pspo\<(\zz\:;\pp\:;\ka)\>$ with the
specified properties is unique if and only if $\,\zz\not\in\<\ka\:D_\bla\,$.
Furthermore, $\,\,\det\:\Psp(\zz\:;\pp\:;\ka)=1\,$ and
\vvn-.4>
\beq
\label{detPsho}
\det\:\Psho\<(\zz\:;\pp\:;\ka)\,=\,\prod_{i=1}^N\,
p_i^{\:\smash{d^{(1)}_{\bla,i}}\:\sum_{a=1}^n z_a\</\<\ka}\,,\kern-1.4em
\eeq
where $\;d^{(1)}_{\bla\:,\:1}\:\lc d^{(1)}_{\bla\:,\:N}$ are given
by formula \,\eqref{dla12}\:.
\vsk.3>
\noindent
{\rm(\:ii\:)}\enspace
Define the function $\>\Pspo\<(\zz\:;\pp\:;\ka)$ for generic $\,\zz\>$ as in
\vvn.1>
item \>{\rm(\:i\:)}\:. Then $\>\Pspo\<(\zz\:;\pp\:;\ka)$ is holomorphic
in $\,\zz\>$ if $\,{z_a\<-z_b\not\in\<\ka\>\Z_{\:\ne\:0}}\>$ for all
$\,\:{1\le a<b\le n}\,$. The singularities of $\,\Pspo\<(\zz\:;\pp\:;\ka)$
at the hyperplanes $\,z_a\<-z_b\in\<\ka\>\Z_{\:\ne\:0}\>$ are simple poles.
\end{thm}
\begin{proof}
The proof of the uniqueness statement is similar to that in Theorem \ref{Bthm}.
To prove the existence part, we give an explicit expression for the function
$\,\:\Pspo\<(\zz\:;\pp\:;\ka)\,$, see formulae \eqref{PspIJo}\:,
\eqref{Pspo}\:.

\vsk.2>
Recall the function $\,\Ao\<(\TT\:;\zz\:;\ka)\,$ at $\,\ka=1\,$,
see \eqref{Ato}\:,
\vvn.4>
\beq
\label{Ato1}
\Ao\<(\TT\:;\zz)\,=\,\prod_{i=1}^{N-1}\,
\prod_{a=1}^{\la^{(i)}}\;\biggl(\,\prod_{\satop{b=1}{b\ne a}}^{\la^{(i)}}
\>\Gm(1+t^{(i)}_b\!-t^{(i)}_a)\prod_{c=1}^{\la^{(i+1)}}\:
\frac1{\Gm(1+t^{(i+1)}_c\!-t^{(i)}_a)}\,\biggr)\>,\kern-2em
\eeq
where $\,\la^{(N)}\?=n\,$ and $\,t^{(N)}_a\?=z_a\,$, $\;a=1\lc n\,$.
Notice that
\vvn.3>
\be
\Ao\<(\Si_I\:;\zz)\,=\>
\prod_{i=1}^{N-1}\prod_{j=i+1}^N\,\:\prod_{a\in I_i}\>\prod_{b\in I_j}
\,\frac1{\Gm(z_b\<-z_a\<+1)}\;.\kern-1.2em
\vv-.2>
\ee
For $\,\lb\<\in\Z^{\:\la\+1}\!$, set
\vvn-.1>
\beq
\label{Jco}
\Jco_{\IJ,\,\lb}(\zz)\,=\>\sum_{K\in\Il}\>\frac{\Ao\<(\Si_K\?-\lb\:;\zz)
\,\Wo_{\?I}(\Si_K\?-\lb;\zz)\,\WWo_{\!J}(\Si_K;\zz)}
{\Aob\<(\Si_K;\zz)\,R_\bla(\zz_{\si_K}\<)}\;.\kern-1em
\eeq
\vsk.3>
Recall the function $\,\Omo_\bla(\pp\:;\ka)=
e^{\:\sum_{\:i<j}^{\cirs}p_j\</(\ka\:p_i)}$,
\vv.1>
where the sum is taken over all pairs $\,1\le i<j\le N\>$ such that
\vv.16>
$\,\la_i\<=1\,$ and $\,\la_s\<=0\,$ for all $\,s=i+1\lc j\,$,
see \eqref{Omlo}\:. Set
\beq
\label{PspIJo}
\Pspo_{\?\IJ}(\zz\:;\pp\:;\ka)\>=\,\Omo_\bla(\pp\:;\ka)\!\<
\sum_{\lb\in\Z_{\ge0}^{\:\la\+1}\!\!}\!\<
\ka^{\>|\si_I|\:-\:|\si_J|}\,\Jco_{\IJ,\,\lb}(\zz/\ka)\,
\prod_{i=1}^{N-1}\>\bigl(\<(-\:\ka)^{-\la_i-\la_{i+1}}\>
p_{i+1}/p_i\:\bigr)^{\>\sum_{a=1}^{\la^{(i)}}l_a^{(i)}}\!.\kern-.8em
\vv.1>
\eeq
Let $\,\Pspo\<(\zz\:;\pp\:;\ka)\,$ be the linear operator with the entries
\vvn.1>
$\,\:\Pspo_{\?\IJ}(\zz\:;\pp\:;\ka)\,$ in the standard basis
$\,\{\:v_I\,,\alb\,I\<\in\Il\:\}\,$ of $\,\Cnnl$,
\vvn-.2>
\beq
\label{Pspo}
\Pspo\<(\zz\:;\pp\:;\ka)\::\:v_J\,\mapsto\>
\sum_{I\in\Il}\,\Pspo_{\?\IJ}(\zz\:;\pp\:;\ka)\,v_I\,.\kern-1em
\eeq

To verify that the function $\,\Pspo\<(\zz\:;\pp\:;\ka)\,$ is as required in
Theorem \ref{Bthmo}, recall the function $\,\:\Psp(\zz\:;h;\qq)\,$, introduced
in Theorem \ref{Bthm} and its entries $\,\:\Psp_{\<\IJ}(\zz\:;h;\qq)\,$.
Let the functions $\,\:\Pspt_{\<\IJ}(\zz\:;h;\pp\:;\ka)\,$ be obtained from
$\,\:\Psp_{\<\IJ}(\zz/\ka\:;h/\ka;\qq)\,$ by substitution \eqref{qp}\:.
\vv.1>
Formulae \eqref{Jc}\:, \eqref{PspIJ}\:, \eqref{Jco}\:, \eqref{PspIJo}\:,
and Lemma \ref{Stirb} imply that
\vvn.3>
\beq
\label{limPsp}
\Pspo_{\?\IJ}(\zz\:;\pp\:;\ka)\,=\>
\lim_{\,h\to\infty}\:\Pspt_{\<\IJ}(\zz\:;h;\pp\:;\ka)\kern-1em
\vv.1>
\eeq
locally uniformly in $\,\zz,\pp\,$. Therefore, the properties of
$\,\:\Psp(\zz\:;h;\qq)\,$ established in Theorem \ref{Bthm} \:yield
the properties of $\,\:\Pspo\<(\zz\:;\pp\:;\ka)\,$ required
in Theorem \ref{Bthmo}.
\end{proof}

\vsk.2>
Following \cite[Chapter 2]{AB}\:, we will call
\vvn.1>
$\,\Psho\<(\zz\:;\pp\:;\ka)\,$ the {\it Levelt fundamental solution\/} of
dynamical differential equations~\eqref{DEQo} on $\:\Cnnl\:$, see also
\cite[Section~6.2]{CV}\:.

\vsk.3>
For a solution $\,\Pso(\zz\:;\pp\:;\ka)\,$ of dynamical differential
equations \eqref{DEQo}\:, define its {\it principal term\/}
\vvn-.1>
\beq
\label{princo}
\Pszo(\zz\:;\ka)\,=\,\Psho\<(\zz\:;\pp\:;\ka)^{-1}\,\Pso(\zz\:;\pp\:;\ka)\,.
\kern-1em
\vv.5>
\eeq
By Theorem \ref{Bthmo}, the principal term does not depend on $\,\pp\,$.

\vsk.3>
Set
\vvn-.7>
\beq
\label{CGGo}
\Co_\bla(\zz\:;\ka)\,=\,\prod_{i=1}^N\>
\ka^{\:\left(\:\sum_{j=i+1}^N\la_j\,-\,\sum_{j=1}^{i-1}\:\la_j\<\right)
\sum_{a=\smash{\la^{\?(i-1)}}\<+1}^{\la^{(i)}}z_a\</\<\ka}\kern-2em
\vv-.5>
\eeq
and
\vvn-.2>
\beq
\label{GGGo}
\Go_\bla(\zz\:;\ka)\,=\,\prod_{i=1}^{N-1}\>\prod_{a=1}^{\la^{(i)}}\,
\prod_{b=\la^{(i)}+1}^{\la^{(i+1)}}\!\Gm\bigl(\<(z_a\<-z_b)/\ka\bigr)\,.
\kern-.5em
\vv.3>
\eeq

\begin{prop}
\label{PsiPpro}
For a Laurent polynomial $\,P(\:\GGd\:;\zzd)\,$, the principal term of
the solution $\,\Pso_{\?P}(\zz\:;\pp\:;\ka)\,$, given by \eqref{PPIo}\:, equals
\vvn.5>
\beq
\label{PsoP0}
\Pszo_{\?P}(\zz\:;\ka)\,=\?\sum_{\IJ\in\:\Il}\?\Pdd(\zz_{\si_J};\zz\:;\ka)\,
\Co_\bla(\zz_{\si_J};\ka)\,\Go_\bla(\zz_{\si_J};\ka)\,\Wo_I(\Si_J\:;\zz)
\,v_I\,.
\kern-2em
\vv.2>
\eeq
\end{prop}
\begin{proof}
Denote by $\,\Psbo_{\?P}(\zz\:;\ka)\,$ the right-hand side of
formula \eqref{PsoP0}\:. Then formula \eqref{PPIo}\:,
Proposition \ref{egvXo0}, and Theorem \ref{thmcyo} yield
$\,\Psho\<(\zz\:;\pp\:;\ka)\,\Psbo_{\?P}(\zz\:;\ka)=
\Pso_{\?P}(\zz\:;\pp\:;\ka)\,$.
\vvn.1>
Hence by definition \eqref{princo} of the principal term,
$\,\Pszo_{\?P}(\zz\:;\ka)=\,\Psbo_{\?P}(\zz\:;\ka)\,$.
\end{proof}

\subsection{The map $\,\Bcyo_\bla\,$}
\label{secBo}

In Section \ref{seclauro}, we introduced the space $\>\Srsol$ of
$\,\Cnnl\<$-valued solutions of the joint system of dynamical differential
equations \eqref{DEQo} and \qKZ/ difference equations \eqref{Kio} spanned over
$\,\C\,$ by the functions $\,\Pso_{\?P}(\zz\:;\pp\:;\ka)\,$ labeled by Laurent
polynomials in $\,\:\GGd\<,\zzd\,$; we also defined a map
\vvn-.1>
\beq
\label{muko2}
\muko\!:\:\Kco_\bla\to\:\Srsol\,,\qquad Y\<\mapsto\Pso_Y\>,
\vv.1>
\eeq
see \eqref{muko}\:. In Section \ref{secLevo} we introduced the Levelt
fundamental solution $\,\:\Psho\<(\zz\:;\pp\:;\ka)\,$ of dynamical differential
equations \eqref{DEQo}\:, see \eqref{Psopnd}\:, \eqref{Pspo}\:.
Denote by $\>\Srhol\>$ the space of $\,\Cnnl\<$-valued solutions of dynamical
differential equations \eqref{DEQo} spanned over $\,\C\,$ by the functions
$\,\:\Psho\<(\zz\:;\pp\:;\ka)\>v\,$, $\,v\in\Cnnl$.
Since $\,\:\det\:\Psho\<(\zz\:;\pp\:;\ka)\ne 0\,$, see \eqref{detPsho}\:,
\vvn.3>
there is an isomorphism
\beq
\label{muho}
\muho:\Cnnl\to\>\Srhol\>,\qquad v\:\mapsto\Psho\<(\zz\:;\pp\:;\ka)\>v\,.
\kern-1.6em
\eeq

\vsk.4>
Let $\,\Lo\!\<\subset\<\C^n\:$ be the complement of the union of
the hyperplanes
\vvn.2>
\beq
\label{zzZo}
z_a\<-z_b\<\in\<\ka\>\Z_{\ne0}\,,\qquad a,b=1\lc n\,,\quad a\ne b\,.\kern-.6em
\vv.1>
\eeq
Denote by $\,\Oc\>$ the ring of functions of $\,\zz\,$ holomorphic in
$\,\Lo\,$. Let $\>\Srol\>$ be the space of $\,\Cnnl\<$-valued solutions
of dynamical differential equations \eqref{DEQo} holomorphic in $\,\zz\,$
in $\,\Lo\>$. Both spaces $\>\Srsol\>$ and $\>\Srhol\>$ are subspaces of
$\>\Srol\,$,
\vvn.4>
see Proposition \ref{PsiPsolo} and Theorem \ref{Bthmo}. Let
\beq
\label{muo}
\muo\::\:\Cnnl\?\ox_{\:\C}\Oc\,\to\>\Srol\:,\qquad
v\:\mapsto\Psho\<(\zz\:;\pp\:;\ka)\>v\,,
\vv.4>
\eeq
be the $\>\Oc\:$-\:linear extension of the map $\,\muho\,$.

\vsk.2>
Define a map
\vvn-.2>
\begin{gather}
\label{Bcyo}
\Bcyo_\bla\::\:\Kco_\bla\to\:\Cnnl\?\ox_{\:\C}\Oc\,,\kern-.6em
\\[6pt]
\label{BcyoP}
[P\:]\,\mapsto\?\sum_{\IJ\in\:\Il}\?\Pdd(\zz_{\si_J};\zz\:;\ka)\,
\Co_\bla(\zz_{\si_J};\ka)\,\Go_\bla(\zz_{\si_J};\ka)\,
W_I(\Si_J\:;\zz)\,v_I\,,
\kern-1em
\\[-16pt]
\notag
\end{gather}
where $\,[P\:]\in\Kco_\bla\,$ stands for the class of the Laurent polynomial
\vv.07>
$\,P(\:\GGd\:;\zzd)\,$ and the functions
$\,\Co_\bla(\zz_{\si_J};\ka)\,$, $\,\Go_\bla(\zz_{\si_J};\ka)\,$ are
\vv.1>
given by \eqref{CGGo}\:, \eqref{GGGo}\:, respectively.
By Proposition \ref{PsiPpro}, the map $\,\:\Bcyo_\bla\>$ sends the class
$\,Y\!\in\Kco_\bla\,$ to the principal term of the solution $\,\:\Pso_Y\:$
\vv.06>
of the joint system of dynamical differential equations \eqref{DEQo} and
\qKZ/ difference equations \eqref{Kio}.

\begin{prop}
\label{trianglo}
The map $\;\Bcyo_\bla\::\:\Kco_\bla\to\:\Cnnl\?\ox_{\:\C}\Oc\,$
is well-defined and the following diagram is commutative,
\vvn-.5>
\beq
\label{cdo}
\xymatrix{\Kco_\bla\ar^-{\;\Bcyo_\bla}[rr]
\ar_{\smash{\lower.8ex\llap{$\ssize\mkh\;\;\,$}}}[dr]&&
\rlap{$\Cnnl\?\ox_{\:\C}\Oc$}\phan{\Kc_\bla}
\ar^{\smash{\lower.8ex\rlap{$\;\ssize\muoh$}}}[dl]\\
&\!\Srolh\,&}
\eeq
\end{prop}
\begin{proof}
Poles of the sum in the right-hand side of \eqref{BcyoP} are at most those
of the function $\,\Go_\bla(\zz\:;h;\ka)\,$ and, therefore, in addition to
hyperplanes \eqref{zzZo} can occur only at the hyperplanes $\,z_a\<=z_b\>$,
$\,a\ne b$. However, the sums
\vvn.5>
\be
\sum_{\IJ\in\:\Il}\?\Pdd(\zz_{\si_J};\zz\:;\ka)\,\Co_\bla(\zz_{\si_J};\ka)\,
\Go_\bla(\zz_{\si_J};\ka)\,W_I(\Si_J\:;\zz)\kern-.8em
\vv.1>
\ee
are regular at the hyperplanes $\,z_a\<=z_b\>$ for all $\,a\ne b\,$
\vvn.07>
by the standard reasoning. Hence, the map $\,\:\Bcyo_\bla\>$ is well-defined.

\vsk.2>
The commutativity of diagram \eqref{cdo} follows from
Proposition \ref{PsiPpro}.
\end{proof}

\begin{rem}
Since $\,\:\Psho\<(\zz\:;\pp\:;\ka)\,$ is holomorphic in $\,\zz\,$ in $\,\Lo\,$,
and $\,\bigl(\det\:\Psho\<(\zz\:;\pp\:;\ka)\bigr)^{\?-1}\:$ is entire
in $\,\zz\,$, see \eqref{detPsho}\:, the inverse matrix
\vv.06>
$\,\:\Psho(\zz\:;\pp\:;\ka)^{-1}\:$ is holomorphic in $\,\zz\,$ in $\,\Lo\,$.
Therefore, for every $\,\:\Psi\in\Srol\>$, its principal term
$\,\:\Pszo\!=\Pshoi\>\Psi\,$,
\vv.06>
defined by \eqref{princo}\:, belongs to $\,\Cnnl\?\ox_{\:\C}\Oc\,$,
and we have an isomorphism $\,\Srol\!\to\Cnnl\?\ox_{\:\C}\Oc\,$,
\vv.06>
$\;\Psi\mapsto\Pszo$. The inverse map, equals $\,\muo\,$, see \eqref{muo}\:,
so that $\,\Srol\?=\Srhol\<\ox_{\:\C}\Oc\,$.
\end{rem}

\subsection{Example $\,\bla=(1,n-1)\,$}
\label{exampleo}
Throughout this section, let $\,N=2\,$, $\,n\ge 2\,$, $\,\bla=(1,n-1)\,$.
Like in Section \ref{example}, denote by $\,[\:a\:]$ the element
$\bigl(\{a\},\{\:1\lc a-1,a+1\lc n\:\}\bigr)\<\in\Il$\,,
\vv-.04>
The space $\,\Ctnl\:$ has a basis $\,v_{[\:1\:]}\lc v_{[\:n\:]}\:$,
where $\,v_{[a]}\<=v_2^{\ox(a-1)}\<\ox v_1\<\ox v_2^{\ox(n-a)}\>$.
Clearly $\,e_{1,1}^{(a)}v_{[\:b\:]}=\:\dl_{\ab}\,v_{[\:b\:]}\,$ and
$\,e_{2,2}^{(a)}v_{[\:b\:]}=(1-\dl_{\ab})\,v_{[\:b\:]}\,$.

\vsk.3>
The \qKZ/ operators $\,\Ko_1\lc\Ko_n\>$, see \eqref{Ko}\:, are
\vvn.2>
\begin{align}
\label{Ko1}
\kern2em
\Ko_a(\zz\:;\pp\:;\ka)\>={}&\,\bigl(\Ro\<(z_a\<-z_{a-1}+\ka)\bigr)^{(\aa-1)}\dots\,
\bigl(\Ro\<(z_a\<-z_1+\ka)\bigr)^{(a,1)}\times{}
\\[1pt]
\notag
& {}\?\times\,p_1^{e_{1,1}^{(a)}}\:p_2^{e_{2,2}^{(a)}}\,
\bigl(\Ro\<(z_i\<-z_n)\bigr)^{(a,n)}\dots\,
\bigl(\Ro\<(z_a\<-z_{a+1})\bigr)^{(\aa+1)}\,.\kern-2em
\end{align}
The $\:R\:$-matrices in the right-hand side
preserve the subspace $\Ctnl\?\subset\<\Ctn$, acting there as follows,
\vvn.1>
\begin{gather*}
\bigl(\Ro\<(z)\bigr)^{(\ab)}v_{[a]}\,=\,v_{[\:b\:]}-z\:v_{[\:a\:]}\,,\kern1.6em
\bigl(\Ro\<(z)\bigr)^{(\ab)}v_{[\:b\:]}\,=\,v_{[a]}\,,\kern-.2em
\\[9pt]
\bigl(\Ro\<(z)\bigr)^{(\ab)}v_{[c]}\>=\>v_{[c]}\,,\qquad c\ne a,b\,.\kern-1em
\\[-12pt]
\end{gather*}
The \qKZ/ difference equations \eqref{Kio} are
\vvn.4>
\beq
\label{Kio1}
f(z_1\lc z_a+\ka\lc z_n;\pp\:;\ka)\,=\,
\Ko_a(\zz\:;\pp\:;\ka)\,f(\zz\:;\pp\:;\ka)\,,\qquad a=1\lc n\>.\kern-1em
\eeq

\vsk.6>
The dynamical Hamiltonians $\,\Xo_1\:,\:\Xo_2\>$, see \eqref{Xo}\:,
act on $\,\Ctnl\:$ as follows
\vvn.2>
\begin{alignat}2
\label{Xo1}
&\Xo_1(\zz\:;\pp)\>v_{[\:1\:]}\>=\,z_1\>v_{[\:1\:]}\:+\>
\frac{p_2}{p_1}\,v_{[\:n\:]}\,,
\\[4pt]
\notag
&\Xo_1(\zz\:;\pp)\>v_{[\:a\:]}\>=\,z_a\>v_{[\:a\:]}\:+\>v_{[\:a-1\:]}\,,
&& a=2\lc n\,,\kern-2em
\\
\notag
&\Xo_2(\zz\:;\pp)\>v_{[\:b\:]}\>=\>
\Bigl(\<-\>\Xo_1(\zz\:;\pp)\>+\sum_{c=1}^n\,z_c\Bigr)\,v_{[\:b\:]}\,,
\qquad && b=1\lc n\,,\kern-2em
\\[-20pt]
\notag
\end{alignat}
and the dynamical differential equations \eqref{DEQ} are
\vvn.5>
\begin{align}
\label{DEQo1}
\ka\>p_1\:\frac\der{\der\:p_1}\>\Psi(\zz\:;\pp\:;\ka)\,&{}=\>
\Xo_1(\zz\:;\pp)\>\Psi(\zz\:;\pp\:;\ka)\,,\kern-1em
\\[7pt]
\notag
\ka\>p_2\:\frac\der{\der\:p_2}\>\Psi(\zz\:;\pp\:;\ka)\,&{}=\>
\Xo_2(\zz\:;\pp)\>\Psi(\zz\:;\pp\:;\ka)\,.\kern-1em
\end{align}

\vsk.2>
In this section, we use the variable $\,t=t^{(1)}_1\:$.
The substitution $\,\:\TT=\Si_{\:[a]}\>$ reads as $\,t\:=z_a\,$.
The weight functions are
\vvn-.3>
\beq
\label{Wao}
\Wo_{[a]}(t\:;\zz)\,=\,\prod_{b=1}^{a-1}\,(t-z_b)\,,
\qquad a=1\lc n\,.\kern-2em
\eeq
The permutations $\,\:\si_{[1]}\lc\si_{[n]}\,$ are
\vvn.4>
\be
\si_{[a]}(1)=a\,,\qquad \si_{[a]}(b)=b-1\,,\quad b=1\lc a-1\,,
\qquad\si_{[a]}(b)=b\,,\quad b=a+1\lc n\,,
\vv.4>
\ee
and $\,\:|\:\si_{[a]}\:|=a-1\,$. We have
\vvn-.5>
\begin{gather}
\label{Wabo}
\Wo_{[a]}(z_a;\zz)\,=\,\prod_{b=1}^{a-1}\,(z_a\<-z_b)\,,
\\[6pt]
\notag
\Wo_{[a]}(z_b;\zz)\>=\>0\,,\qquad b=1\lc a-1\,,
\\[-20pt]
\notag
\end{gather}
cf.~Lemma \ref{WIzo}.

\vsk.2>
The functions $\,\WWo_{[a]}(t\:;\zz)\,$, see \eqref{WocI}\:, are
\vvn.2>
\beq
\label{Wcao}
\WWo_{[a]}(t\:;\zz)\,=\,\prod_{c=a+1}^n\?(t-z_b)\,,\qquad a=1\lc n\,.\kern-2em
\vv-.3>
\eeq
Set
\vvn-.7>
\be
\Rb_a(\zz)\,=\,\prod_{\satop{b=1}{b\ne a}}^n\,(z_a\<-z_b)\,,
\qquad a=1\lc n\,.\kern-2em
\ee
Then $\,\Rb_a(\zz)=R_\bla(\zz_{\si_{[a]}})\,,$, where the function
\vvn.1>
$\,R_\bla(\zz)\,$ is given by \eqref{RQ}\:.
Biorthogonality relations \eqref{ortho} become
\be
\sum_{c=1}^n\,\frac{\Wo_{[a]}(z_c;\zz)\,\WWo_{[b]}(z_c;\zz)}{\Rb_c(\zz)}\,=\,
\dl_{\ab}\,,\qquad
\sum_{c=1}^n\,\Wo_{[c]}(z_a;\zz)\,\WWo_{[c]}(z_b;\zz)\,=\,
\dl_{\ab}\,\Rb_a(\zz)\,.\kern-1.6em
\ee

The master function, see \eqref{Phio}\:, is
\vvn-.3>
\beq
\label{Pho}
\Pho_\bla(t\:;\zz\:;\pp\:;\ka)\,=\,(p_2/\ka)^{\>\sum_{a=1}^n\<z_a\</\?\ka}\,
(\ka^{\:n}p_1/p_2)^{\:t/\?\ka}\,\prod_{a=1}^n\,\Gm\bigl(\<(t-z_a)/\ka\bigr)\,.
\kern-1.4em
\eeq
The
%$q\:$-
hypergeometric solutions \eqref{mcF} of the joint system of
% dynamical
differential equations \eqref{DEQ1} and
% \qKZ/
difference equations \eqref{Ki1} have the form
\vvn.1>
\be
\Pso_{[a]}(\zz\:;\pp\:;\ka)\,=\,\frac1\ka\,\sum_{b=1}^n\,\sum_{l=0}^\infty\,
\Res_{\>t\:=\;z_a-\:l\ka}\Pho_\bla(t\:;\zz\:;\pp\:;\ka)\;
\Wo_{[\:b\:]}(z_a\<-l\ka)\>v_{[\:b\:]}\,,\kern-1.6em
\vv.2>
\ee
where the residues of the master function are
\vvn.3>
\begin{align*}
& \Res_{\>t\:=\;z_a-\:l\ka}\Pho_\bla(t\:;\zz\:;\pp\:;\ka)\,={}
\\[3pt]
&{}\,=\,\ka\>(\ka^{\:n-1}p_1)^{\:z_a\</\<\ka}\:
(p_2/\ka)^{\>\sum_{c=1,\>c\ne a}^n\<z_c\</\?\ka}\,
\frac{(-\:\ka^{-n}p_2/p_1)^{\:l}}{l\:!}\,
\prod_{\satop{c=1}{c\ne a}}^n\,\Gm\bigl(-\:l+(z_a\<-z_c)/\ka\bigr)\,.
\\[-16pt]
\end{align*}
Determinant formula for coordinates of the
%$\>q\:$-
hypergeometric solutions, see \eqref{detPso}\:, is
\vvn.5>
\begin{align}
\label{detPso1}
& \det\:\biggl(\>\sum_{l=0}^\infty\,
\Res_{\>t\:=\;z_a-\:l\ka}\Pho_\bla(t\:;\zz\:;\pp\:;\ka)\,
\Wo_{[\:b\:]}(z_a\<-l\ka)\?\biggr)_{\!\?\ab\:=1}^{\!\<n}\<={}
\\[5pt]
\notag
& \:{}=\,\pi^{\:n\:(n-1)/2}\:\ka^{\:n\:(n+1)/2}
\bigl(\:p_1\>p_2^{\:n-1}\bigr)^{\sum_{a=1}^n z_a\</\?\ka}\,
\prod_{a=1}^{n-1}\,\prod_{b=a+1}^n\,
\frac1{\sin\bigl(\pi\>(z_a\<-z_b)/\ka\bigr)}\;.\kern-1em
\\[-11pt]
\notag
\end{align}
By formula \eqref{Wao} and the Vandermonde determinant formula,
\vvn.5>
equality \eqref{detPsi1} transforms to
\begin{align}
\label{detPso2}
\det\:\biggl(\>\sum_{l=0}^\infty\,\sum_{c=1}^n\,
& \Res_{\>t\:=\;z_c-\:l\ka}\bigl(\:
t^{\:a-1}\:e^{-\:2\:\pii\,\:(b-1)\>t/\?\ka}\,
\Pho_\bla(t\:;\zz\:;\pp\:;\ka)\bigr)\?\biggr)_{\!\?\ab\:=1}^{\!\<n}\<={}
\kern-2em
\\[8pt]
\notag
{}=\,{}& \bigl(2\:\piit\,\:\bigr)^{\<n\:(n-1)/2}\:\ka^{\:n\:(n+1)/2}
\bigl(\:e^{-\:\pii\,\:(n-1)}p_1\>p_2^{\:n-1}\bigr)^{\sum_{a=1}^n z_a\</\?\ka}
\:.\kern-2em
\end{align}

\vsk.4>
Let $\,\gmd=\gmd_{1,1}\,$. For a Laurent polynomial $\,P(\gmd\:;\zzd)\,$,
we have
\vvn.4>
\be
\Pdd(\gm\:;\zz\:;\ka)\,=\,P(e^{\:2\:\pii\,\gm\</\<\ka}\:;\:
e^{\:2\:\pii\,z_1\</\<\ka}\lc e^{\:2\:\pii\,z_n\</\<\ka}\:)\,.\kern-1em
\vv.3>
\ee
The solution $\,\Pso_{\?P}\>$ of differential equations \eqref{DEQo1} and
difference equations \eqref{Kio1} corresponding to $\,P(\gmd\:;\zzd)\,$ is
\vvn-.5>
\beq
\label{PsoP1}
\Pso_{\?P}(\zz\:;\pp\:;\ka)\,=\,\sum_{a=1}^n\,
\Pdd(z_a;\zz\:;\ka)\,\Pso_{[a]}(\zz\:;\pp\:;\ka)\,.\kern-1.2em
\vv.1>
\eeq
By Proposition \ref{PsiPsolo}, % the function $\,\Psi_P(\zz\:;\pp\:;\ka)\,$
it is entire in $\,\zz\,$ and is holomorphic in $\,p_1\:,\>p_2\:$ provided
branches of $\,\:\log\:p_1\:$ and $\,\:\log\:p_2\:$ are fixed.

\vsk.3>
The solution $\,\Pso_{\?P}(\zz\:;\pp\:;\ka)\,$ can be written as an integral
\vv.16>
over a suitable contour $\,C\>$ encircling the poles of the product
$\,\prod_{a=1}^n\<\Gm\bigl({(t-z_a)/\<\ka}\bigr)\,$ counterclockwise,
\vvn.2>
\beq
\label{PsoPint}
\Pso_{\?P}(\zz\:;\pp\:;\ka)\,=\,\frac1{2\:\piit\,\ka}\;
\int\limits_{\!\!C\,}\<\Pdd(t\:;\zz\:;\ka)\,\Pho_\bla(t\:;\zz\:;\pp\:;\ka)\,
\sum_{a=1}^n\>\Wo_{[a]}(t\:;\zz)\,v_{[a]}\,d\:t\,.\kern-1em
\vv.2>
\eeq
For instance, the integral can be taken over the parabola
\vvn.4>
\be
C\,=\,\{\,\ka\>\bigl(\:A-s^2\?+s\>\sqrt{\<-1}\,\bigr)\ \,|\ \,s\in\R\,\}\,,
\vv.3>
\ee
where $\,A\,$ is a sufficiently large positive real number.
Formula \eqref{PsiPint} can be used to give an alternative proof of
analytic properties of the function $\,\:\Pso_{\?P}(\zz\:;\pp\:;\ka)\,$.

\vsk.3>
Let $\>\Srsol\:$ be the space of solutions of the joint system of dynamical
\vv.1>
differential equations \eqref{DEQo1} and \qKZ/ difference
equations \eqref{Kio1}
\vv.06>
spanned over $\,\C\,$ by the functions $\,\Pso_{\?P}(\zz\:;\pp\:;\ka)\,$
corresponding to Laurent polynomials $\,P(\gmd\:;\zzd)\,$.
\vvn.1>
The space $\>\Srsol\:$ is a $\,\C[\:\zzd^{\pm1}]\:$-\:module
with $\,f(\zzd)\,$ acting as multiplication by
$\,f(e^{\:2\:\pii\,z_1\</\<\ka}\lc e^{\:2\:\pii\,z_n\</\<\ka}\:)\,$.

\vsk.2>
For $\,\bla=(1,n-1)\,$, the algebra $\,\Kco_\bla\,$, see \eqref{Krelo},
can be presented as follows
\vvn.3>
\beq
\label{Krelo11}
\Kco_\bla\:=\,\C[\:\gmd^{\pm1}\?,\zzd^{\pm1}]\>\>\Big/\Bigl\bra
\,\prod_{a=1}^n\,(\gmd-\zdd_a)\>=\>0\,\Bigr\ket\,.\kern-1.6em
\vv.3>
\eeq
The function $\,\Pso_{\?P}(\zz\:;\pp\:;\ka)\,$ depends only on the class
\vv.1>
of the Laurent polynomial $\,P\:$ in $\>\Kco_\bla\,$.
The assignment $\,P\mapsto\Pso_P\>$ defines the homomorphism
\vvn.3>
\be
\muko\<:\Kco_\bla\to\:\Srsol\,,\qquad Y\<\mapsto\Pso_Y\>,\kern-1em
\vv.3>
\ee
of $\,\C[\:\zzd^{\pm1}]\:$-\:modules. Formula \eqref{detPso2}
implies that the homomorphism $\,\muko\,$ is an isomorphism.

\vsk.3>
The algebra $\,\Kco_\bla\:$ is the equivariant $\>K\?$-theory algebra
\vvn.1>
$\,K_T(\CP^{\>n-1}\<;\C)\>$ of the projective space $\,\CP^{\>n-1}$,
see the notation in Section \ref{sQde}.

\vsk.3>
Recall Definition \ref{entdef} of a function $\,f(\pp)\,$ entire
\vvn.07>
in $\,\pp_{\!\dvs}$.
%if $\,f(\pp)=g(p_2/p_1)\,$ for an entire function $\,g(s)\,$,
%and $\,f(\0)=g(0)\,$.
For example, the dynamical Hamiltonians $\,\Xo_1\:,\:\Xo_2\>$,
see \eqref{Xo1}\:, are entire in $\,\pp_{\!\dvs}$ and $\,\Xo_1(\zz\:;\0)\,$,
$\,\Xo_2(\zz\:;\0)\,$ act on $\,\Ctnl\,$ as follows
\be
\Xo_1(\zz\:;\0)\>v_{[\:a\:]}\>=\,z_a\>v_{[\:a\:]}\:+\>v_{[\:a-1\:]}\,,
\qquad
\Xo_2(\zz\:;\0)\>v_{[\:a\:]}\>=\>
\Bigl(\<-\>\Xo_1(\zz\:;\0)\>+\sum_{b=1}^n\,z_b\Bigr)\,v_{[\:a\:]}\,,
\kern-1em
\vv-.4>
\ee
where $\,v_{[\:0\:]}\<=0\,$.

\vsk.2>
The Levelt fundamental solution $\,\:\Psho\<(\zz\:;\pp\:;\ka)\,$,
\vvn.3>
see \eqref{Psopnd}\:, of differential equations \eqref{DEQo1} is
\be
\Psho\<(\zz\:;\pp\:;\ka)\,=\,
\Pspo\<(\zz;\pp\:;\ka)\;p_1^{\>\smash{\Xo_1(\zz\:;\:\0)}/\ka}\:
p_2^{\>\smash{\Xo_2(\zz\:;\:\0)}/\ka}\>,\kern-1.4em
\vv.3>
\ee
where the $\,\:\End\:\bigl(\Cnnl\bigr)$-\:valued function
$\,\Pspo\<(\zz\:;\pp\:;\ka)\,$ is as follows,
\vvn.3>
\begin{gather}
\label{Pspo1}
\Pspo\<(\zz\:;\pp\:;\ka)\::\:v_{[\:b\:]}\,\mapsto\>
\sum_{a=1}^n\,\Pspo_{\ab}(\zz\:;\pp\:;\ka)\,v_{[a]}\,,\kern-1em
\\[3pt]
\notag
\Pspo_{\ab}(\zz\:;\pp\:;\ka)\,=\,\sum_{l=0}^\infty\,\ka^{\:a-b}\:
\Jco_{\ab,\>l}(\zz/\ka)\>\bigl(\<(\<-\ka)^{-\:n\:}p_2/p_1\<\bigr)^l\>,
\kern-1em
\\[6pt]
\notag
\Jco_{\ab,\>l}(\zz)\,=\,\sum_{c=1}^n\,
\Wo_{[a]}(z_c\<-l\:;\zz)\,\WWo_{[b]}(z_c;\zz)\;\frac{(-1)^{n-1}}{l\:!}\,
\prod_{\satop{d=1}{d\ne c}}^n\,\prod_{m=0}^l\,\frac1{z_d\<-z_c\<+m}\;.
\kern-1em
%\Jco_{\ab,\>l}(\zz)\,=\,\sum_{c=1}^n\,
%\frac{\Wo_{[a]}(z_c\<-l\:;\zz)\,\WWo_{[b]}(z_c;\zz)}{\Rb_c(\zz)}\;
%\prod_{d=1}^n\,\frac{\Gm(z_d\<-z_c\<+1)}{\Gm(z_d\<-z_c\<+1+l)}\;.\kern-1em
%\prod_{m=0}^{l-1}\,\frac1{z_d\<-z_c\<+1+m}\;.
\\[-20pt]
\notag
\end{gather}
Furthermore, one has $\;\det\:\Pspo\<(\zz\:;\pp\:;\ka)=1\,\:$ and
$\;\det\:\Psho\<(\zz\:;\pp\:;\ka)\,=\,
\bigl(\:p_1\>p_2^{\:n-1}\bigr)^{\sum_{a=1}^n z_a\</\?\ka}\:$.

\vsk.4>
For a solution $\,\Pso(\zz\:;\pp\:;\ka)\,$ of dynamical differential equations
\vvn.4>
\eqref{DEQo1}\:, its principal term is
\be
\Pszo(\zz\:;\ka)\,=\,\Psho\<(\zz\:;\pp\:;\ka)^{-1}\,\Pso(\zz\:;\pp\:;\ka)\,,
\vv.3>
\ee
see \eqref{princo}\:. The principal term of the solution
\vv.07>
$\,\Pso_{\?P}(\zz\:;\pp\:;\ka)\,$, corresponding to a Laurent polynomial
$\,P(\gmd\:;\zzd)\,$, see \eqref{PsoP0}\:, equals
\vvn.2>
\beq
\label{PsoP10}
\Pszo_{\?P}(\zz\:;\ka)\,=\,
\sum_{a=1}^n\>v_{[a]}\,
\sum_{b=1}^n\>\Wo_{[a]}(z_b;\zz)\,\Pdd(z_b;\zz\:;\ka)
\;\ka^{\left.\left(\<n\:z_b\:-\sum_{c=1}^nz_c\<\right)\<\right/\<\ka}
\,\prod_{\satop{c=1}{c\ne b}}^n\,\Gm\bigl(\<(z_b\<-z_c)/\ka\bigr)\,.
\kern-.9em
\eeq

Let $\,\Lo\,$ be the complement of the union of the hyperplanes
$\,z_a\<-z_b\<\in\<\ka\>\Z_{\ne0}\,$, $\,a,b=1\lc n\,$, $\,a\ne b\,$,
cf.~\eqref{zzZo}\:. Denote by $\,\Oc\,$ the ring of functions of $\,\zz\,$
\vvn.4>
holomorphic in $\,\Lo\,$. The map
\begin{gather}
\label{Bcyo1}
\Bcyo_\bla\::\:\Kco_\bla\to\:\Cnnl\?\ox_{\:\C}\Oc\,,\kern-1em
\\[3pt]
\notag
{[P\:]}\,\mapsto\:\sum_{a=1}^n\>v_{[a]}\,
\sum_{b=1}^n\>\Wo_{[a]}(z_b;\zz)\,\Pdd(z_b;\zz\:;\ka)
\;\ka^{\left.\left(\<n\:z_b\:-\sum_{c=1}^nz_c\<\right)\<\right/\<\ka}
\,\prod_{\satop{c=1}{c\ne b}}^n\,\Gm\bigl(\<(z_b\<-z_c)/\ka\bigr)
\kern-1em
\end{gather}
sends the class $\,[P\:]\in\Kco_\bla\,$ of the Laurent polynomial
$\,P(\gmd\:;\zzd)\,$ to the principal term of the solution $\,\:\Pso_{\?P}\:$
of the joint system of differential equations \eqref{DEQo1} and difference
equations \eqref{Kio1}.

\vsk.2>
Let $\>\Srol$ be the space of $\,\Cnnl\<$-valued solutions of dynamical
\vv.04>
differential equations \eqref{DEQo1} holomorphic in $\,\zz\,$ in $\,\Lo\,$.
The space $\>\Srsol\>$ is a subspace of $\>\Srol$.

\vsk.3>
Since the matrix $\,\:\Psho\<(\zz\:;\pp\:;\ka)\,$ is holomorphic in
\vvn.1>
$\,\zz\,$ in $\,\Lo\,$, and
$\,\bigl(\det\:\Psho\<(\zz\:;\pp\:;\ka)\bigr)^{\?-1}\:$ is entire in $\,\zz\,$,
the inverse matrix $\,\:\Psho\<(\zz\:;\pp\:;\ka)^{-1}\,$ is also holomorphic
in $\,\zz\,$ in $\,\Lo\,$. Thus the map
\vvn.3>
\be
\muo:\:\Cnnl\?\ox_{\:\C}\Oc\>\to\>\Srol,\qquad
v\:\mapsto\Psho\<(\zz\:;\pp\:;\ka)\>v\,,
\vv.4>
\ee
gives an isomorphism of $\,\Oc\,$-modules.
Furthermore, the following diagram is commutative,
\be
\xymatrix{\Kco_\bla\ar^-{\;\Bcyo_\bla}[rr]
\ar_{\smash{\lower.8ex\llap{$\ssize\muko\;\;\,$}}}[dr]&&
\rlap{$\Cnnl\?\ox_{\:\C}\Oc$}\phan{\Kco_\bla}
\ar^{\smash{\lower.8ex\rlap{$\;\ssize\muo$}}}[dl]\\
&\!\Srol\:&}
\vv-.7>
\ee
see Proposition \ref{trianglo}.

\section{Equations for partial flag varieties}
%% Equivariant quantum differential and difference
\label{sQde}

\subsection{Equivariant cohomology of partial flag varieties}
\label{sec:equiv}

Consider the partial flag variety $\,\Fla\:$ parametrizing chains of subspaces
\be
0\,=\,F_0\subset F_1\lsym\subset F_N\>=\,\C^n\kern-1.4em
\ee
with $\,\:\dim F_i\:/F_{i-1}=\la_i\,$, $\,i=1\lc N\>$.

\vsk.2>
Given a basis of $\,\C^n$, the group $\,GL_n(\C)\,$ acts on $\,\C^n$.
Let \,$T\subset GL_n(\C)$ \,be the torus of diagonal matrices.
Let $\,\zzz\,$ the Chern roots corresponding to the factors of $\,T\,$,
and for $\,i=1\lc N\>$, let $\,\gm_{i,1}\lc\gm_{i,\>\la_i}\,$ \,be the Chern
roots of the bundle over \,$\Fla$ \,with fiber $\,F_i\:/F_{i-1}\,$. Denote
$\,\:\GG\<=(\gm_{1,1}\lc\gm_{1,\>\la_1}\:\lc\gm_{N,1}\lc\gm_{N\?,\>\la_N})\,$
\vv.05>
and $\,\zz\,=\,(\zzz)$, cf.~\eqref{GG}\:. Let $\,\C[\:\GG\:]^{\:S_\bla}$ be the
space of polynomials in $\,\:\GG\>$ symmetric in $\,\gm_{i,1}\lc\gm_{i,\>\la_i}$
for each $\,i=1\lc N\>$.

\vsk.2>
Consider the equivariant cohomology algebra $\,\Hd^*_T(\Fla\>;\C)\,$. Then
\vvn.2>
\beq
\label{Hrel}
H^*_T(\Fla\>;\C)\,=\,\C[\:\GG\:]^{\:S_\bla}\?\ox\C[\zz]\>\Big/\Bigl\bra
\,\prod_{i=1}^N\,\prod_{j=1}^{\la_i}\,(u-\ga_{\ij})\,=\,\prod_{a=1}^n\,(u-z_a)
\Bigr\ket\,,\kern-2em
\vv.2>
\eeq
where $\,u\,$ is a formal variable.
For a polynomial $\,f(\:\GG\:;\zz)\in\C[\:\GG\:]^{\:S_\bla}\?\ox\C[\zz]\,$,
\vvn.07>
denote by $\,[\:f\>]\,$ its class in $\,\Hd^*_T(\Fla\>;\C)\,$.
Notice that for each $\,i=1\lc N\>$,
\vvn.3>
\beq
\label{c1}
c_1(E_i)\,=\,[\:\gm_{i,1}\<\lsym+\gm_{i,\:\la_i}]\kern-.4em
\vv.2>
\eeq
is the equivariant first Chern class of the vector bundle $\,E_i\>$
over $\,\Fla\>$ with fiber $\,F_i\:/F_{i-1}\,$.

\vsk.2>
Given a point $\,\zz^0\?\in\C^n\:$, consider the algebra
\vvn.2>
\beq
\label{Hrelz0}
\Hd^*_T(\Fla\>;\C)_{\:\zz^0}\,=\,\C[\:\GG\:]^{\:S_\bla}\?\Big/\Bigl\bra\,
\prod_{i=1}^N\,\prod_{j=1}^{\la_i}\,(u-\ga_{\ij})\,=\,\prod_{a=1}^n\,(u-z^0_a)
\Bigr\ket\,.\kern-2em
\vv.2>
\eeq
The evaluation map
$\,\C[\:\GG\:]^{\:S_\bla}\?\ox\C[\zz]\to\C[\:\GG\:]^{\:S_\bla}\>$,
$\,f(\:\GG\:;\zz)\mapsto f(\:\GG\:;\zz^0)\,$, identifies
\vv.16>
$\,\Hd^*_T(\Fla\>;\C)_{\:\zz^0}\,$ with the quotient
$\,\Hd^*_T(\Fla\>;\C)/\bra\zz=\zz^0\:\ket\,$. Denote by $\,[\:f\>]_{\:\zz^0}$
\vvn.16>
the class of $\,f(\:\GG)\in\C[\:\GG\:]^{\:S_\bla}\:$ in
$\,\Hd^*_T(\Fla\>;\C)_{\:\zz^0}\,$.

\vsk.4>
Recall the polynomials $\,V_{\<I}(\:\GG)\,$, $\,I\<\in\Il\>$, given by
formula \eqref{VIx}\:.

\begin{lem}
\label{Hfree}
The algebra $\,\Hd^*_T(\Fla\>;\C)$ is a free module over
\vvn.07>
$\,\Hd^*_T({pt};\C)=\C[\:\zz\:]\:$ generated by the classes
$\,[\:V_{\<I}\:]\,$, $\,I\<\in\Il\>$. For every $\,\zz^0\?\in\C^n$,
\vv.07>
the classes $\,[\:V_{\<I}\:]_{\:\zz^0}\>$, $\,I\<\in\Il\>$,
give a basis of $\,\Hd^*_T(\Fla\>;\C)_{\:\zz^0}\>$. In particular,
$\,\:\dim\:\Hd^*_T(\Fla\>;\C)_{\:\zz^0}=\:n\:!/(\la_1!\ldots\la_N!\:)\>$
for every $\,\zz^0$.
\end{lem}
\begin{proof}
The statement follows from Propositions \ref{PA1} and \ref{PA2}.
\end{proof}

\vsk.3>
Recall the polynomial $\,R_\bla(\zz)\,$, see \eqref{RQ}\:. Let
\vvn.3>
\beq
\label{VGz}
V_\bla(\:\GG\:;\zz)\,=\>\prod_{i=1}^{N-1}\:\prod_{j=1}^{\la_i}
\:\prod_{a=\la^{(i)}+1}^n\!(\gm_{\ij}\<-z_a)\,,\kern-1em
\vv.3>
\eeq
cf.~\eqref{Vl}\:. Then $\,V_\bla(\zz\:;\zz)=R_\bla(\zz)\,$ and
$\,V_\bla(\zz_{\si_I};\zz)=0\,\:$ for $\,I\<\ne\Imil\:$.

\vsk.3>
For a polynomial $\,f(\:\GG\:;\zz)\,$, consider the restrictions
$f(\zz_{\si_I};\zz)\,$, $\,I\<\in\Il\,$.

\begin{lem}
\label{Hcx}
The map $\,\Hd^*_T(\Fla\>;\C)\to\bigoplus_{I\in\Il}\C[\:z\:]\,$,
\vvn.07>
$\,[\:f\>]\,\mapsto\,\bigl(f(\zz_{\si_I};\zz)\,,\,I\<\in\Il\bigr)\:$
is well-de\-fined and injective. A collection
\vv.06>
$\,\bigl(F_I(\zz)\,,\,I\<\in\Il\bigr)\,$ belongs to the image of this map
if and only if for any \,$I\<\in\Il$ and any transposition $\,s_{\ab}\,$,
\vvn.2>
\beq
\label{FIab}
F_I(\zz)\big|_{\:z_a=z_b}=\,F_{s_{a\<,b}(I)}(\zz)\big|_{\:z_a=z_b}\:.
\eeq
For a collection $\,\bigl(F_I(\zz)\,,\,I\<\in\Il\bigr)\,$
obeying \eqref{FIab}, the function
\vvn.3>
\beq
\label{fFV}
f(\:\GG\:;\zz)\,=\>\sum_{I\in\Il}
\frac{F_I(\zz)\,V_\bla(\:\GG\:;\zz_{\si_I}\<)}{R_\bla(\zz_{\si_I}\<)}\;.\kern-1em
\vv.1>
\eeq
is a polynomial and $\,f(\zz_{\si_I};\zz)=F_I(\zz)\,$
for every $\,I\<\in\Il\>$.
\end{lem}
\begin{proof}
Straightforward.
\end{proof}

%We say that a collection of polynomials $\,(F_I(\zz)\,,\,I\<\in\Il)\,$
%represents a class in $\,\Hd^*_T(\Fla\>;\C)\,$ if it has property \eqref{FIab}\:,
%and we denote by $\,[\:F_I\>,\,I\<\in\Il\:]\,$ the corresponding class.
%\vsk.3>

\begin{lem}
\label{lemEc}
For any $\,f(\:\GG\:;\zz)\in\C[\:\GG\:]^{\:S_\bla}\?\ox\C[\zz]\,$, the function
\vvn.3>
\beq
\label{Ecf}
\Ec\bra f\ket(\zz)\,=\>\sum_{I\in\Il}\,
\frac{f(\zz_{\si_I},\zz)}{R_\bla(\zz_{\si_I}\<)}\kern-1em
\vv.1>
\eeq
is a polynomial depending only on the class of $\,f\,$
in $\,\Hd^*_T(\Fla\>;\C)\,$.
\end{lem}
\begin{proof}
Straightforward.
\end{proof}

The induced map
\beq
\label{Ec}
\Ec\::\:\Hd^*_T(\Fla\>;\C)\,\to\,\C[\:\zz\:]\,,\qquad
[\:f\>]\>\mapsto\>\Ec\bra f\ket\,,\kern-1em
\vv.2>
\eeq
is the {\it equivariant integration map on\/} $\,\Hd^*_T(\Fla\>;\C)\,$.

\vsk.4>
Identify $\,\C[\:\GG\:]^{\:S_\bla}\:$ with the subspace of
$\,\C[\:\GG\:]^{\:S_\bla}\?\ox\C[\zz]\,$ of polynomials not depending
\vv.16>
on $\,\zz\,$. For $\,f(\:\GG)\in\C[\:\GG\:]^{\:S_\bla}\>$, denote
$\,\:\Ec_{\zz^0}\bra f\ket=\Ec\bra f\ket(\zz^0)\,$.

\begin{lem}
\label{lemEcz0}
For any $\,f(\:\GG\:;\zz)\in\C[\:\GG\:]^{\:S_\bla}\?\ox\C[\zz]\>$ and
$\,\zz^0\?\in\C^n$, we have
\vvn.4>
\be
\Ec\bra f(\:\GG\:;\zz)\ket(\zz^0)\,=\,\Ec_{\zz^0}\bra f(\:\GG\:;\zz^0)\ket\,.
\kern-1em
\ee
\end{lem}
\begin{proof}
Straightforward.
\end{proof}

By Lemmas \ref{lemEc}, \ref{lemEcz0}, for every $\,\zz^0\?\in\C^n$,
there is a well-defined map
\vvn.4>
\beq
\label{Ecz0}
\Ec_{\zz^0}:\:\Hd^*_T(\Fla\>;\C)_{\:\zz^0}\:\to\,\C\,,\qquad
\Ec_{\zz^0}:[\:f\>]_{\:\zz^0}\>\mapsto\,\Ec_{\zz^0}\bra f\ket\,,\kern-2em
\vv.2>
\eeq
called the {\it integration map on\/} $\,\Hd^*_T(\Fla\>;\C)_{\:\zz^0}\,$.

\subsection{Stable envelope map}
\label{sec:Stab}
Recall the functions $\,\Wo_{\?I}(\TT\:;\zz)\,$,
\vvn.1>
$\,\WWo_{\?I}(\TT\:;\zz)\,$, see \eqref{WoI}\:, \eqref{WocI}\:.
Let $\,\:\St_{\:I}(\:\GG\:;\zz)\,$, $\Stop_{\:I}(\:\GG\:;\zz)\,$ be
\vv.06>
the polynomials respectively obtained from $\,\WWo_{\?I}(\TT\:;\zz)\,$,
$\Wo_{\?I}(\TT\:;\zz)\,$ by the substitution
\vvn.3>
\beq
\label{tgm}
(\:t^{(i)}_1\?\lc t^{(i)}_{\la^{(i)}}\:)\,=\,
(\gm_{1,1}\lc\gm_{1,\>\la_1}\:\lc\gm_{i\:,1}\lc\gm_{i\:,\>\la_i})\,,
\qquad i=1\lc N-1\,.\kern-2em
\vv.1>
\eeq
Notice that
\vvn-.1>
\beq
\label{StWo}
\St_{\:I}(\zz_{\si_J};\zz)\,=\>\WWo_{\?I}(\Si_J\:;\zz)\,,\qquad
\Stop_{\:I}(\zz_{\si_J}\:;\zz)\,=\>\Wo_{\?I}(\Si_J\:;\zz)\,.\kern-1.2em
\vv.5>
\eeq
By Lemma \ref{Womax} and formula \eqref{VGz}\:, we have
\vvn.4>
\begin{alignat}2
\label{StIma}
\St_{\:\Imil}(\:\GG\:;\zz)\>&{}=\:V_\bla(\:\GG\:;\zz)\,,\qquad
& \St_{\:\si_0(\Imil)}(\:\GG\:;\zz)\>&{}=\>1\,,
\\[6pt]
\notag
\Stop_{\:\Imil}(\:\GG\:;\zz)\>&{}=\>1\,, &
\Stop_{\:\si_0(\Imil)}(\:\GG\:;\zz)\>&{}=\:V_\bla(\:\GG\:;\zz_{\si_0}\<)\,,
\kern-1em
\\[-16pt]
\notag
\end{alignat}
where $\,\si_0\>$ is the longest permutation, $\,\si_0(a)=n+1-a\,$,
\vvn.1>
$\;a=1\lc n\,$. By Proposition \ref{lemStSch}, the polynomials
\vvn.06>
$\,\:\St_{\:I}(\:\GG\:;\zz)\,$, $\Stop_{\:I}(\:\GG\:;\zz)\,$ coincide with
the $A\:$-type double Schubert polynomials $\,\:\Sg_\si(\:\GG\:;\zz)\,$,
\vvn.2>
\beq
\label{StS}
\St_{\:I}(\:\GG\:;\zz)\,=\,\Sg_{\si_{\?\si_{\<0}\<(I)}}(\:\GG\:;\zz_{\si_0})\,,
\kern1.6em
\Stop_{\:I}(\:\GG\:;\zz)\,=\,\Sg_{\si_I}(\:\GG\:;\zz)\,.
\kern-2em
\vv.1>
\eeq
%cf.~\eqref{WdS}\:.

\begin{lem}
\label{lemStWV}
For any $\,I\<\in\Il\>$,
\vvn-.3>
\be
\St_{\:I}(\:\GG\:;\zz)\,=\>\sum_{J\in\Il}\frac{\WWo_{\?I}(\Si_J\:;\zz)\,
V_\bla(\:\GG\:;\zz_{\si_J}\<)}{R_\bla(\zz_{\si_J}\<)}\;,\kern-1em
\vv-.5>
\ee
and
\vvn-.2>
\beq
\label{StWV}
\Stop_{\:I}(\:\GG\:;\zz)\,=\>\sum_{J\in\Il}\frac{\Wo_{\?I}(\Si_J\:;\zz)\,
V_\bla(\:\GG\:;\zz_{\si_J}\<)}{R_\bla(\zz_{\si_J}\<)}\;.\kern-1em
\vv.1>
\eeq
\end{lem}
\begin{proof}
Fix a generic $\,\zz\,$ and consider both sides of each formula as functions
\vvn.07>
of $\,\:\GG\:$. Formulae \eqref{StIma} and Lemma \ref{DlW} imply that
$\,\:\St_{\:I}(\:\GG\:;\zz)\,$, $\,\Stop_{\:I}(\:\GG\:;\zz)\,$
\vv.06>
are in the span of the polynomials $\,V_\bla(\:\GG\:;\zz_{\<\si_J}\<)\,$,
$\,J\<\in\Il\,$. Then the statement follows from formulae \eqref{fFV}\:,
\eqref{StS}\:.
\end{proof}

\begin{lem}
\label{Stort}
The polynomials $\;\St_{\:I}(\:\GG\:;\zz)\:$ and $\;\Stop_{\:I}(\:\GG\:;\zz)\:$
are biorthogonal,
\vvn.4>
\beq
\label{ESSop}
\Ec\bra\:\St_{\:I}(\:\GG\:;\zz)\,\Stop_{\<J}(\:\GG\:;\zz)\ket\,=\,\dl_{\IJ}\,.
\vv.3>
\eeq
\end{lem}
\begin{proof}
The statement follows from formulae \eqref{Ecf}\:, \eqref{StS}\:,
and Lemma \ref{lemorto}.
\end{proof}

Lemma \ref{ESSop} is equivalent to the orthogonality relation for the Schubert
polynomials.

\vsk.4>
Let $\,\:\St_{\:I}$, $\Stop_{\:I}\>$ be the classes of
$\,\:\St_{\:I}(\:\GG\:;\zz)\,$, $\Stop_{\:I}(\:\GG\:;\zz)\,$
\vvn.06>
in $\,\Hd^*_T(\Fla\>;\C)\,$. Define the {\it stable envelope map\/} by the rule
\vvn.5>
\beq
\label{stmap}
\St_{\:\bla}\::\:\Cnnl\?\ox\C[\zz]\,\to\,\Hd^*_T(\Fla\>;\C)\,,\qquad
v_I\>\mapsto\>\St_{\:I}\>,\quad I\<\in\Il\,.\kern-2em
\vv.3>
\eeq

\begin{lem}
\label{StH}
The map $\;\St_{\:\bla}$ is an isomorphism of free $\,\C[\:\zz\:]$-modules.
\end{lem}
\begin{proof}
The statement follows from Lemma \ref{Stort}\:.
For any $\,f\<\in\C[\:\GG\:]^{\:S_\bla}\?\ox\C[\zz]\,$, one has
\vvn.4>
\be
\bigl(\:\St_{\:\bla}\bigr)^{\?-1}\<:\>[\:f\>]\,\mapsto\:
\sum_{J\in\:\Il}\Ec\bra f(\:\GG\:;\zz)\,\Stop_{\<J}(\:\GG\:;\zz)\ket\,v_J\,.
%\sum_{\IJ\in\:\Il}\?\frac{f(\zz_{\si_I};\zz)\,\Wo_{\!\<J}(\Si_I;\zz)}
%{R_\bla(\zz_{\si_I}\<)}\;v_J\,.
\kern-1em
\vv-1.4>
\ee
\end{proof}

\vsk-.2>
Let $\,\:\St_{\:I\<,\:\zz^0}$, $\,\Stop_{\:I\<,\:\zz^0}$ be the classes of
$\,\:\St_{\:I}(\:\GG\:;\zz^0)\,$, $\Stop_{\:I}(\:\GG\:;\zz^0)\,$
in $\,\Hd^*_T(\Fla\>;\C)_{\:\zz^0}\>$. Consider the map
\vvn.1>
\beq
\label{stmapz0}
\St_{\:\bla\:,\:\zz^0}\::\:\Cnnl\to\,\Hd^*_T(\Fla\>;\C)_{\:\zz^0}\>,\qquad
v_I\>\mapsto\>\St_{\:I\<,\:\zz^0}\>,\quad I\<\in\Il\,.\kern-2em
\vv.3>
\eeq

\begin{lem}
\label{SHz0}
For every $\,\zz^0\?\in\C^n$, the classes $\,\:\St_{\:I\<,\:\zz^0}\>$,
\vvn.05>
$\,I\<\in\Il\>$, give a basis of $\,\Hd^*_T(\Fla\>;\C)_{\:\zz^0}\>$.
That is, the map $\,\:\St_{\:\bla\:,\:\zz^0}$ is an isomorphism.
\end{lem}
\begin{proof}
The statement follows from Lemmas \ref{Stort}, \ref{lemEcz0}.
For any $\,f\<\in\C[\:\GG\:]^{\:S_\bla}\>$, one has
\vvn.4>
\be
\bigl(\:\St_{\:\bla\:,\:\zz^0}\?\bigr)^{\?-1}\<:\>[\:f\>]_{\:\zz^0}\>\mapsto\:
\sum_{J\in\:\Il}\Ec_{\:\zz^0}\bra\:f(\:\GG)\>\Stop_{\<J}(\:\GG\:;\zz^0)\ket\,v_J\,.
\kern-1em
\vv-1.4>
\ee
\end{proof}

\vsk-.2>
Consider a polynomial
\vvn-.4>
\beq
\label{Who}
\Who_{\!\<\bla}(\TT\:;\GG)\,=\>\prod_{i=1}^{N-1}\,\prod_{j=1}^{\la^{(i)}}\,
\prod_{k=1}^{\la_{i+1}}\,(t^{(i)}_j\?-\gm_{i+1,\:k})\,.
\vv.2>
\eeq
and its image $\,\bigl[\:\Who_{\!\<\bla}(\TT\:;\GG)\:\bigr]\,$ in
$\,\C[\:\TT\:]\ox \Hd^*_T(\Fla\>;\C)\,$. Clearly, for any $\,I\<,\<J\<\in\Il\,$,
\vvn.4>
\beq
\label{WhSi}
\Who_{\!\<\bla}(\Si_I\:;\zz_J)\,=\,\dl_{\IJ}\,R_\bla(\zz_{\si_I}\<)\,.\kern-1em
\vv.3>
\eeq

\begin{thm}
\label{lemWhtG}
We have
\be
\bigl[\:\Who_{\!\<\bla}(\TT\:;\GG)\:\bigr]\,=\,
\sum_{I\in\Il}\,\Wo_{\?I}(\TT\:;\zz)\,\St_{\:I}\>.
\ee
\end{thm}
\begin{proof}
The statement follows from \cite[Proposition 9.2]{TV6} and Lemmas
\ref{lemWWo}, \ref{Womax}.
\end{proof}

\begin{rem}
According to formulae \eqref{StS}, the classes $\,\:\St_{\:I}\:$,
$\,\:\Stop_{\:I}\:$, $\,\:I\<\in\Il\,$, are the equivariant fundamental classes
of the Schubert subvarieties in $\,\Fla\>$ defined for the opposite orderings
of the chosen basis of $\,\:\C^n$ in the standard way,
see \cite[Section 4.1]{Oh}, cf.~\cite[Section 6]{RV}. On the equivariant
fundamental classes of Schubert varieties see, in particular, \cite{R, FRW}.
\end{rem}

\begin{rem}
Stable envelope maps for Nakajima quiver varieties were introduced in
\cite{MO}\:. They were defined in \cite{MO} geometrically in terms of the
associated torus action. The map $\;\St_{\:\bla}\:$ \,given by formula
\eqref{stmap} is the limit as $\,h\to\infty\,$ of the stable envelope map of
\cite{MO} associated with the cotangent bundle $\,\tfl\>$ of the partial flag
variety. In \cite{RTV1}\:, the stable envelope map for $\,\tfl\>$ is described
in terms of the Chern roots of the bundles $\,\FF_1\lc\FF_{N-1}\,$ over
$\,\Fla\>$ with fibers $\,F_1\lc F_{N-1}$\,, respectively.
\end{rem}

\begin{rem}
Unlike the stable envelope map $\,\:\St_{\:\bla}\>$ defined by \eqref{stmap}\:,
the stable envelope map for the cotangent bundle $\,\tfl\>$ is not
an isomorphism of free $\C[\:\zz\:]$-modules $\,\Cnnl\?\ox\C[\zz]\,$ and
$\,\Hd^*_{\smash{T\times\:\C^{\<\times}}}\<(\tfl\>;\C)\,$, but only an embedding.
\end{rem}

\subsection{Quantum multiplication}
The quantum multiplication in $\,\Hd^*_T(\Fla\>;\C)\,$, see for example
\cite{Mi,BM}\:, is a deformation of the multiplication in
$\,\Hd^*_T(\Fla\>;\C)\,$ depending on quantum parameters. In the notation of
\cite{Mi}\:, the quantum parameters are $\,q^{\>\al_i}\?$, $\,i=\lc N-1\,$.
In the notation of this paper, we have $\,q^{\>\al_i}\!=p_{i+1}/p_i\,$.

\vsk.2>
For $\,Y\!\in \Hd^*_T(\Fla\>;\C)\,$, denote by $\>Y\<*_\pp\>$ the operator
\vvn.05>
of quantum multiplication by $\>Y\>$. Recall that the operators $\>Y\<*_\pp\>$
are $\,\C[\:\zz\:]\:$-module endomorphisms of $\,\Hd^*_T(\Fla\>;\C)\,$.

\vsk.2>
The dynamical operators $\,\Xo_1(\zz\:;\pp)\lc\Xo_{\?N}(\zz\:;\pp)\,$,
\vvn.08>
see \eqref{Xo}\:, are linear functions of $\,\zz\,$. Thus their action
on $\,\Cnnl\:$ extends to the $\,\C[\:\zz\:]\:$-\:linear action
on $\Cnnl\?\ox\C[\zz]\,$.

\vsk.2>
For $\,i=1\lc N\>$, denote
$\,D_i\<=\gm_{i,1}\<\lsym+\gm_{i,\>\la_i}\>$, cf.~\eqref{c1}\:.

\begin{thm}
\label{D*X}
The isomorphism $\;\St_{\:\bla}\:$ intertwines the dynamical operators
$\,\Xo_1(\zz\:;\pp)\lc\Xo_{\?N}(\zz\:;\pp)\:$ acting on $\Cnnl\?\ox\C[\zz]\>$ and
the operators of quantum multiplication $\,[D_1]\>*_\pp\:\lc[D_N]\>*_\pp\:$
acting on $\,\Hd^*_T(\Fla\>;\C)\,$,
\vvn.1>
\beq
\label{StXD}
\St_{\:\bla}\:\circ\Xo_i(\zz\:;\pp)\,=\,
[D_i]\>{*_\pp}\:\circ\>\St_{\:\bla}\,.\kern-1em
\vv.2>
\eeq
\end{thm}
\begin{proof}
Theorem~6.4 in \cite{Mi} describes the quantum multiplication
\vvn.04>
in $\,\Hd^*_T(\Fla\>;\C)\,$ by the equivariant divisor classes.
In the notation of this paper those classes equal
$\,\:\St_{\>\si_0\>s_{\la^{(i)}\!,\>\la^{(i)}\<+1}(\Imil)}\,$,
$\,\:i=1,\dots, N-1\,$, where $\,\si_0\,$ is the longest permutation.
By formulae \eqref{stmap}\:, \eqref{WocI}\:, \eqref{Wos}\:,
\vvn.3>
\beq
\label{SDz}
\St_{\>\si_0\>s_{\la^{(i)}\!,\>\la^{(i)}\<+1}(\Imil)}(\:\GG\:;\zz)\,=\,
D_1\<\lsym+D_i\<-z_{n-\la^{(i)}\<+1}\<\lsym-z_n\,.\kern-2em
\vv.2>
\eeq
Comparing formula \eqref{Xo} term by term with formula (6.1) in \cite {Mi}
\vv.07>
yields formula \eqref{StXD} for $\,i=1\lc N-1\,$. Formula \eqref{StXD} for
$\,i=N\>$ holds since $\,[\:D_1\<\lsym+D_N\:]=z_1\<\lsym+z_n\,$ in
\vv.06>
$\,\Hd^*_T(\Fla\>;\C)\,$ and
$\,X_1(\zz\:;\pp)\lsym+X_N(\zz\:;\pp)=z_1\?\lsym+z_n\,$.
Recall that for each $\,i=1\lc N\>$, the operator $\,z_i\>*_\pp\>$
is the ordinary multiplication by $\,z_i\,$,
\end{proof}

Since the operators of quantum multiplication act $\,\C[\zz]\:$-linearly on
\vv.05>
$\,\Hd^*_T(\Fla\>;\C)\,$, the quantum multiplication in $\,\Hd^*_T(\Fla\>;\C)\,$
can be restricted to $\,\Hd^*_T(\Fla\>;\C)_{\:\zz^0}$. Denote by
$\>Y\<{*}_{\pp\:,\:\zz^0}\:$ the operator
of quantum multiplication by $\>Y\!\in \Hd^*_T(\Fla\>;\C)_{\:\zz^0}\>$.

\begin{cor}
\label{D*Xz0}
For every $\,\zz^0\?\in\C^n$, the isomorphism $\;\St_{\:\bla\:,\:\zz^0}\:$
\vv.1>
intertwines the dynamical operators
$\,\Xo_1(\zz^0;\pp)\lc\Xo_{\?N}(\zz^0;\pp)\:$ acting on $\Cnnl$ and
\vv.1>
the operators of quantum multiplication
$\,[D_1]_{\:\zz^0}\>{*}_{\pp\:,\:\zz^0}\lc[D_N]_{\:\zz^0}\>{*}_{\pp\:,\:\zz^0}$
acting on $\,\Hd^*_T(\Fla\>;\C)_{\:\zz^0}\,$,
\vvn.4>
\beq
\label{StXDz0}
\St_{\:\bla\:,\zz^0}\:\circ\Xo_i(\zz^0;\pp)\,=\,
{[D_i]_{\:\zz^0}\>{*}_{\pp\:,\:\zz^0}}\:\circ\>\St_{\:\bla\:,\:\zz^0}\,.\kern-1em
\vv.4>
\eeq
\end{cor}

\subsection{Differential and difference equations}
Consider
\vvn.05>
the space $\,{\C^n\!\times\C^N}$ with coordinates $\,\zz\:,\pp\,$.
\vvn.07>
By Lemma \ref{Hfree}, we have a trivial vector bundle
$\,H_\bla\?\to\:\C^n\!\times\C^N$ with fiber over a point
$\,(\zz^0\?,\pp^{\:0})\,$ given by $\,\Hd^*_T(\Fla\>;\C)_{\:\zz^0}\,$.
\vv.07>
Let $\,U_\bla\?\to\,\C^n\!\times\C^N$ be the trivial vector bundle with fiber
$\,\Cnnl\:$.

\begin{lem}
\label{StUH}
The map
$\;\St_{\:\bla}^{\:\lzs}:\:\<U_\bla\to\:H_\bla\,$,
$\,(\zz^0\?,\pp^{\:0}\?,v)\>\mapsto\>
(\zz^0\?,\pp^{\:0}\?,\>\St_{\:\bla\:,\:\zz^0}\>v)\>$,
is an isomorphism of vector bundles.
\end{lem}
\begin{proof}
The statement follows from Lemma \ref{SHz0}.
\end{proof}

\vsk.2>
The {\it equivariant quantum differential equations\/} for sections
of $\,H_\bla\:$ is a system of differential equations
\vvn-.3>
\beq
\label{qDEQ}
\ka\>p_i\:\frac{\der f}{\der\:p_i}\,=\,
{[D_i]_{\:\zz}\>{*}_{\pp\:,\:\zz}}\>f\,,\qquad i=1\lc N\>,\kern-2em
\vv.2>
\eeq
where $\,\ka\,$ is the parameter of the equations. By Corollary \ref{D*Xz0},
the isomorphism $\,\:\St_{\:\bla}^{\:\lzs}\,$ identifies equations \eqref{qDEQ}
with the limiting dynamical differential equation \eqref{DEQo} for sections
of $\,U_\bla\:$,
\vvn-.4>
\beq
\label{DEQch}
\ka\>p_i\:\frac{\der f}{\der\:p_i}\,=\,\Xo_i(\zz\:;\pp)\>f\,,\qquad
i=1\lc N\>.\kern-1.9em
\vv.3>
\eeq
Furthermore, the isomorphism $\,\:\St_{\:\bla}^{\:\lzs}\,$ and the limiting \qKZ/
equations \eqref{Kio} for sections of $\,U_\bla$ define the \qKZ/
{\it difference equations in cohomology\/}
\vvn.3>
\beq
\label{Kich}
f(z_1\lc z_a\<+\ka\lc z_n;\pp\:;\ka)\,=\,
\KH_a(\zz\:;\pp\:;\ka)\,f(\zz\:;\pp\:;\ka)\,,\qquad a=1\lc n\,,\kern-2em
\vv.3>
\eeq
where for each $\,a=1\lc n\,$, and fixed $\,\zz,\pp\,$,
the operator $\,\KH_a(\zz\:;\pp\:;\ka)\,$ is a map of fibers,
%$\,\KH_1\lc\KH_n\,$
\vvn.5>
\begin{gather*}
\KH_a(\zz\:;\pp\:;\ka)\>:\>\Hd^*_T(\Fla\>;\C)_{\:\zz}\>\to\,
\Hd^*_T(\Fla\>;\C)_{\:(z_1,\:\ldots\:,\>z_a+\:\ka\:,\:\ldots\:,\>z_n)}\,,
\kern-1.1em
\\[6pt]
\KH_a(\zz\:;\pp\:;\ka)\,=\,
\St_{\:\bla,\:(z_1\?,\:\ldots\:,\>z_a+\:\ka\:,\:\ldots\:,\>z_n)}\:
\circ\:\Ko_a(\zz\:;\pp\:;\ka)\:\circ\:(\:\St_{\:\bla,\>\zz}\<)^{-1}\:,
\kern-1em
\\[-14pt]
\end{gather*}
and the operator $\,\Ko_a(\zz\:;\pp\:;\ka)\,$ is given by \eqref{Ko}\:.

\begin{thm}
\label{qcompat}
Quantum differential equations \eqref{qDEQ} and \qKZ/ difference equations
\eqref{Kich} define a compatible system of differential and difference equations
for sections of $\,H_\bla\>$.
\end{thm}
\begin{proof}
The statement follows from Theorem \ref{thm qkz n}.
\end{proof}

\subsection{Solutions with values in cohomology}
\label{s6.6}
Recall the isomorphism
\vvn.07>
$\,\:\St_{\:\bla}^{\:\lzs}:U_\bla\to\:H_\bla\:$ of vector bundles,
see Lemma \ref{StUH}. For any function $\,f(\zz\:;\pp\:;\ka)\,$
with values in $\,\Cnnl$, that is, a section of $\,U_\bla\,$, denote by
$\,\:\St_{\:\bla}\>f\,$ the corresponding section of $\,H_\bla\:$ with values
\vvn.3>
\be
\St_{\:\bla}\>f(\zz\:;\pp\:;\ka)\,=\,
\St_{\:\bla\:,\:\zz}\bigl(f(\zz\:;\pp\:;\ka)\bigr)\,.
\vv.3>
\ee
Recall the solutions $\,\Pso_{\?P}(\zz\:;\pp\:;\ka)\,$
%$\>P\?\in\C[\:\GGd^{\pm1}]^{\:S_\bla}\?\ox\C[\:\zzd^{\pm1}]\,$,
of the joint system of limiting dynamical differential equations \eqref{DEQo}
and \qKZ/ difference equations \eqref{Kio} labeled by Laurent polynomials
$\,P(\:\GGd;\zzd)\,$, see formula \eqref{PPIo} and Proposition \ref{PsiPsolo}.

\begin{thm}
\label{thm main}
The isomorphism $\;\St_{\:\bla}^{\:\lzs}\>$ transforms the solutions
$\,\Pso_{\?P}(\zz\:;\pp\:;\ka)\>$ of equations \eqref{DEQo}\:, \eqref{Kio}
to solutions of the joint system of quantum differential equations and \qKZ/
difference equations \eqref{qDEQ}\:, \eqref{Kich}\:. Namely, for each
$\>P\?\in\C[\:\GGd^{\pm1}]^{\:S_\bla}\?\ox\C[\:\zzd^{\pm1}]\,$, the function
$\,\:\St_{\:\bla}\Pso_{\?P}(\zz\:;\pp\:;\ka)\,$ is a solution of equations
\eqref{DEQo}\:, \eqref{Kio}\:.
\end{thm}
\begin{proof}
The theorem follows from Theorem \ref{D*X} and the definition of the
\qKZ/ difference equations in \eqref{Kich}.
\end{proof}

In this section we work out the geometric interpretation of the obtained
solutions. Notice that it would be interesting to relate our solutions with
constructions in \cite{C}, where Mellin-Barnes integral representations
of solutions of the quantum differential equations were constructed for
a class of smooth projective varieties.

\vsk.3>
For $\,i=1\lc N\>$, let $\,\gmd_{1,1}\lc\gmd_{1\<,\:\la_i}$ be virtual line
\vv.05>
bundles such that $\,\bigoplus_{j=1}^{\>\la_i}\gmd_{\ij}=\:F_i/F_{i-1}\,$,
\vv.05>
and $\,\zdd_1\lc\zdd_n\:$ correspond to the factors of the torus $\,T\,$.
Denote
$\,\:\GGd^{\pm1}\?=(\gmd_{1,1}^{\pm1}\:\lc\gmd_{N\?,\>\la_N}^{\pm1})\,$
and $\,\zzd^{\pm1}\?=(\zdd_1^{\pm1}\?\lc\zdd_n^{\pm1})\,$, cf.~\eqref{GG}\:.
\vv-.04>
Let $\,\C[\:\GGd^{\pm1}]^{\:S_\bla}$ be the space of Laurent polynomials in
$\,\GGd$ symmetric in $\,\gmd_{i,1}\lc\gmd_{i,\>\la_i}$ for each $\,i=1\lc N\>$.

\vsk.2>
Consider the equivariant $\>K\?$-theory algebra $\,K_T(\Fla\>;\C)\,$. Then
\vvn.2>
\beq
\label{Krelt}
K_T(\Fla\>;\C)\,=\,\C[\:\GGd^{\pm1}]^{\:S_\bla}\?\ox\C[\:\zzd^{\pm1}]\>\Big/
\Bigl\bra\,\prod_{i=1}^N\prod_{j=1}^{\la_i}\,(u-\gmd_{\ij})\,=\,
\prod_{a=1}^n\,(u-\zdd_a)\Bigr\ket\,,\kern-1.6em
\vv.3>
\eeq
where $\,u\,$ is a formal variable. That is, $\,K_T(\Fla\>;\C)\,$ coincides
with the algebra $\,\Kco_\bla\>$, see \eqref{Krelo}\:.
For a Laurent polynomial $\,P(\:\GGd^{\pm1};\zzd^{\pm1})\in
\C[\:\GGd^{\pm1}]^{\:S_\bla}\?\ox\C[\zz^{\pm1}]\,$,
\vvn.07>
denote by $\,[P\:]\,$ its class in $\,K_T(\Fla\>;\C)\,$.

\begin{lem}
\label{Kfree}
The algebra $\,K_T(\Fla\>;\C)$ is a free module over $\,\C[\:\zzd^{\pm1}]\,$.
\end{lem}
\begin{proof}
The statement follows from Propositions \ref{PA1} and \ref{PA2}.
\end{proof}

Denote by $\,\Srskl\<$ the space spanned over $\,\C\,$ by the solutions
$\,\:\St_{\:\bla}\Pso_{\?P}(\zz\:;\pp\:;\ka)\,$ of equations \eqref{qDEQ}\:,
\eqref{Kich}\:.
\vv.05>
By Proposition \ref{PsiPsolo}, each element of $\,\Srskl\<$ is a section of
$\,H_\bla\:$ holomorphic in $\,\pp\,$ provided a branch of $\,\:\log\:p_i\:$
\vv.04>
is fixed for each \,$i=1\lc N$, and entire in $\,\zz\,$.

\vsk.2>
Since the function $\,\Pso_{\?P}(\zz\:;\pp\:;\ka)\,$ depends only
\vv.06>
on the class of $\,P\:$ in $\,K_T(\Fla\>;\C)\,$, then so does the section
$\,\:\St_{\:\bla}\Pso_{\?P}(\zz\:;\pp\:;\ka)\,$ of $\,H_\bla\,$.
Thus the map
\vvn.3>
\beq
\label{mk}
\mkh\<:\:K_T(\Fla\>;\C)\:\to\,\Srskl\>,\qquad
[P\:]\>\mapsto\>\St_{\:\bla}\Pso_{\?P}\,,
\vv.4>
\eeq
is well-defined, cf.~\eqref{muko}\:. By Corollary \ref{muko=}, the map
$\,\mkh\<$ is an isomorphism of $\;\C[\:\zzd^{\pm1}]\:$-\:modules.

\vsk.2>
In Section \ref{secLevo} we introduced the Levelt fundamental solution
\vv.07>
$\,\:\Psho\<(\zz\:;\pp\:;\ka)\,$ of dynamical differential equations
\eqref{DEQo}\:. Consider the vector bundle $\,E\!H_\bla\?\to\:\C^n\!\times\C^N$
with fiber over a point $\,(\zz^0\?,\pp^{\:0})\,$ given by
$\,\:\End\:\bigl(\Hd^*_T(\Fla\>;\C)_{\:\zz^0}\<\bigr)\,$,
and its section $\,\:\St_{\:\bla}\Psho$ with values
\vvn.3>
\beq
\label{StPshp}
\St_\bla\Psho\<(\zz\:;\pp\:;\ka)\>=\,
\St_{\:\bla,\:\zz}\circ\Psho\<(\zz\:;\pp\:;\ka)\circ
\bigl(\:\St_{\:\bla,\:\zz}\bigr)^{\?-1}\:.\kern-2em
\vv.3>
\eeq
By Theorem \ref{Bthm} and Corollary \ref{D*Xz0}, for any section $\,f\>$ of
\vv.04>
$\,H_\bla\:$ not depending on $\,\pp\,$, the section $\,\:\St_\bla\Psho\?f\,$
of $\,H_\bla\:$ with values $\,\:\St_\bla\Psho\<(\zz\:;\pp\:;\ka)\>f(\zz)\,$
\vv.04>
is a solution of quantum differential equations \eqref{qDEQ}\:. We will call
\vv.04>
the section $\,\:\St_\bla\Psho$ of $\,E\!H_\bla\:$ the {\it Levelt fundamental
solution\/} of quantum differential equations \eqref{qDEQ}\:.

\vsk.2>
Recall the function $\,\:\Pspo\<(\zz\:;\pp\:;\ka)\,$ with
values in $\,\:\End\:\bigl(\Cnnl\bigr)\,$, see \eqref{Psopnd}\:.
Consider the section and $\,\:\St_{\:\bla}\Pspo$ of $\,E\!H_\bla\:$
with values
\vvn.3>
$$
\St_{\:\bla}\Pspo\<(\zz\:;\pp\:;\ka)\>=\,
\St_{\:\bla,\:\zz}\circ\Pspo\<(\zz\:;\pp\:;\ka)\circ
\bigl(\:\St_{\:\bla,\:\zz}\bigr)^{\?-1}\:.\kern-1em
\vv.3>
$$
Then by formulae \eqref{Psopnd} and \eqref{StXDz0}\:,
\vvn-.3>
\beq
\label{StPsho}
\St_{\:\bla}\Psho\<(\zz\:;\pp\:;\ka)\,=\,
\bigl(\:\St_{\:\bla}\Pspo\<(\zz\:;\pp\:;\ka)\<\bigr)
\,\prod_{i=1}^N\,p_i^{\:[D_i]_\zz/\ka}\:,
\kern-1em
\eeq
where $\,\:\prod_{i=1}^N p_i^{\:[D_i]_\zz/\ka}\<$ acts on the fiber
$\,\Hd^*_T(\Fla\>;\C)_{\:\zz}\>$ as multiplication by itself.
\vv.1>
The expression for $\,\:\St_{\:\bla}\Pspo\<(\zz\:;\pp\:;\ka)\,$
is given by formula \eqref{SPspo} below.

\vsk.3>
Recall the function $\,\Ao\<(\TT\:;\zz\:;\ka)\,$,
see \eqref{Ato}\:, \eqref{Ato1}\:,
\vvn.2>
\be
\Ao\<(\TT\:;\zz\:;\ka)\,=\,\prod_{i=1}^{N-1}\,
\prod_{a=1}^{\la^{(i)}}\;\biggl(\,\prod_{\satop{b=1}{b\ne a}}^{\la^{(i)}}
\>\Gm\bigl(1+(t^{(i)}_b\!-t^{(i)}_a)/\ka\bigr)
\prod_{c=1}^{\la^{(i+1)}}\:
\frac1{\Gm\bigl(1+(t^{(i+1)}_c\!-t^{(i)}_a)/\ka\bigr)}\,\biggr)\>,\kern-1em
\vv-.2>
\ee
where $\,\la^{(N)}\?=n\,$ and $\,t^{(N)}_a\?=z_a\,$, $\;a=1\lc n\,$.
\vv0>
For $\,I\<\in\Il\,$ and $\,\mb=(m_1\lc m_{N-1})\<\in\Z_{\ge0}^{N\<-1}$, set
\vvn-.3>
\beq
\label{JchIm}
\Jchh_{\!\?I\<,\:\mb}(\:\GG\:;\zz\:;\ka)\,=\>
\sum_{i=1}^{N-1}\!\sum_{\satop{\lb\in\:\smash{\Z^{\:\la\+1}}\!\vp{|^{1^1}}}
{\!|\:\lb^{(i)}\<|=\:m_i}}\<\sum_{J\in\Il}
\frac{\Ao\<(\Si_J\?-\lb\ka\:;\zz\:;\ka)\,\Wo_{\?I}(\Si_J\?-\lb\ka\:;\zz)
\,V_\bla(\:\GG\:;\zz_{\si_J}\<)}{\Aob\<(\Si_J;\zz)\,R_\bla(\zz_{\si_J}\<)}\;,
\kern-1em
\eeq
where $\;|\>\lb^{(i)}|=\:l^{\:(i)}_j\!\lsym+\:l^{\:(i)}_{\la^{(i)}}\,$,
$\;i=1\lc N-1\,$. cf.~\eqref{Jco}\:, \eqref{fFV}\:.
\vvn.1>
The functions $\,\Jch_{I\<,\>\lb}(\:\GG\:;\zz\:;\ka)\>$
are polynomials in $\,\:\GG\:$ and rational functions in $\,\zz,\ka\,$.

\begin{prop}
\label{Jchhreg}
The functions $\,\Jchh_{\!\?I\<,\:\mb}(\:\GG\:;\zz\:;\ka)\:$ are regular
\vvn.07>
if $\,{z_a\<-z_b\not\in\<\ka\>\Z_{\:\ne\:0}}\>$ for all $\,a\ne b\,$,
and have at most simple poles at the hyperplanes
$\,z_a\<-z_b\in\<\ka\>\Z_{\:\ne\:0}\>$.
\end{prop}
\begin{proof}
Recall the functions $\,\:\Pspo_{\?\IJ}(\zz\:;\pp\:;\ka)\,$ and
$\,\:\Omo_\bla(\pp\:;\ka)\,$, see \eqref{PspIJo}\:, \eqref{Omlo}\:,
\vv.07>
respectively. By formulae \eqref{Jco}\:, \eqref{PspIJo}\:, \eqref{StWV}\:,
\eqref{ortho}\:, \eqref{JchIm}\:, we have
\vvn.6>
\begin{align}
\label{PspIJH}
\kern-.3em
\sum_{J\in\Il}{}\:&
\Pspo_{\?\IJ}(\zz\:;\pp\:;\ka)\,\Stop_{\>J}(\:\GG\:;\zz)\,={}
\\[-9pt]
\notag
&\!\<{}=\,\Omo_\bla(\pp\:;\ka)\!\?\sum_{\mb\:\in\:\Z_{\ge0}^{N-1}\!}
\!\<\Jchh_{\!\?I\<,\:\mb}(\:\GG\:;\zz\:;\ka)\,\prod_{i=1}^{N-1}\,
\bigl(\<(-\:\ka)^{-\la_i-\la_{i+1}}\>p_{i+1}/p_i\:\bigr)^{\:m_i}.\kern-2em
\\[-15pt]
\notag
\end{align}
Formula \eqref{Pspo} and Theorem \ref{Bthmo} imply that the functions
\vvn.1>
$\,\:\Pspo_{\?\IJ}(\zz\:;\pp\:;\ka)\,$ are regular if
$\,{z_a\<-z_b\not\in\<\ka\>\Z_{\:\ne\:0}}\>$ for all $\,a\ne b\,$, and have
\vvn.06>
at most simple poles at the hyperplanes $\,z_a\<-z_b\in\<\ka\>\Z_{\:\ne\:0}\>$.
Thus the statement of Proposition \ref{Jchhreg} follows from formula
\eqref{PspIJH}\:.
\end{proof}

Define the section $\,\Jchh_{\!\<\mb}\>$ of the bundle $E\!H_\bla$ with values
\vvn.4>
\begin{gather}
\label{Jchm}
\Jchh_{\!\<\mb}(\zz\:;\ka)\::\:
\Hd^*_T(\Fla\>;\C)_{\:\zz}\>\to\>\Hd^*_T(\Fla\>;\C)_{\:\zz}\,,\kern-1.2em
\\[7pt]
\notag
\Jchh_{\!\<\mb}(\zz\:;\ka)\::\:Y\:\mapsto\sum_{I\in\Il}\:
\Ec_{\:\zz}\bigl\bra\:Y\>[\Jchh_{\!\?I\<,\:\mb}(\:\GG\:;\zz\:;\ka)]_{\:\zz}
\bigr\ket\,\St_{\:I\<,\>\zz}\,,
\kern-1em
\\[-15pt]
\notag
\end{gather}
where $\,\Ec_{\:\zz}$ is the integration map on
$\,\Hd^*_T(\Fla\>;\C)_{\:\zz}\,$, see \eqref{Ecz0}\:.

\begin{prop}
\label{lemSPspo}
We have
\vvn-.2>
\beq
\label{SPspo}
\St_{\:\bla}\Pspo\<(\zz\:;\pp\:;\ka)\,=\,
\Omo_\bla(\pp\:;\ka)\!\sum_{\mb\:\in\:\Z_{\ge0}^{N-1}\!}\!
\Jchh_{\!\<\mb}(\zz\:;\ka)\,\prod_{i=1}^{N-1}\>
\bigl(\<(-\:\ka)^{-\la_i-\la_{i+1}}\>p_{i+1}/p_i\bigr)^{\:m_i}.
\kern-2em
\eeq
\end{prop}
\begin{proof}
The statement follows from formulae \eqref{Pspo}\:, \eqref{ESSop}\:,
\eqref{PspIJH}\:,
\eqref{Jchm}\:.
\end{proof}

\begin{prop}
\label{lemEJc}
For any $\>Y\!\in \Hd^*_T(\Fla\>;\C)_{\:\zz}\>$,
\vvn.5>
\beq
\label{EJcSt}
\Ec_{\:\zz}\bigl\bra\<(\Jchh_{\!\<\:\mb}(\zz\:;\ka)\>Y\:)\,
\Stop_{\:I\<,\>\zz}\bigr\ket\>=\,
\Ec_{\:\zz}\bigl\bra\:Y\>[\Jchh_{\!\?I\<,\:\mb}(\:\GG\:;\zz\:;\ka)]_{\:\zz}
\bigr\ket\,.\kern-2em
\vv.2>
\eeq
In particular,
\vv-.1>
\beq
\label{EJc}
\Ec_{\:\zz}\bigl\bra\Jchh_{\!\<\:\mb}(\zz\:;\ka)\>Y\:\bigr\ket\>=\,
\Ec_{\:\zz}\bigl\bra\:Y\>
[\Jchh_{\!\?\Imil\?,\>\mb}(\:\GG\:;\zz\:;\ka)]_{\:\zz}\bigr\ket\,.
\kern-2em
\vv.4>
\eeq
\end{prop}
\begin{proof}
The statement follows from formulae \eqref{Jchm}\:, \eqref{ESSop}\:,
\eqref{StIma}\:.
\end{proof}

\subsection{The map $\,\Bcyh_\bla\,$}
\label{secBh}
Let $\,\Lo\!\<\subset\<\C^n\:$ be the complement of the union of
the hyperplanes
\vvn.4>
\beq
\label{zzZh}
z_a\<-z_b\<\in\<\ka\>\Z_{\ne0}\,,\qquad a,b=1\lc n\,,\quad a\ne b\,.\kern-.6em
\vv.3>
\eeq
Denote by $\,\Oc\>$ the ring of functions of $\,\zz\,$ holomorphic in
$\,\Lo$, and by $\,\Oc_{H_\bla}\?$ the space of sections of the bundle
$\,H_\bla\:$ holomorphic in $\,\zz\,$ for $\,\zz\in\Lo$ and not depending on
$\,\pp\,$. The stable envelope map $\,\:\St_{\:\bla}\:$ induces an isomorphism
$\,\Cnnl\?\ox_{\:\C}\Oc\to\Oc_{H_\bla}\>$,
$\;f\mapsto \St_{\:\bla}\>f\,$.

\vsk.3>
Let $\>\Srolh\:$ be the space of sections of \,$H_\bla\:$ that are solutions
\vvn.07>
of quantum differential equations \eqref{DEQch} holomorphic in $\,\zz\,$
for $\,\zz\in\Lo$. The space $\>\Srolh\:$ is a counterpart
of the space $\>\Srol\>$ of solutions of limiting dynamical differential
\vvn.07> equations \eqref{DEQo} introduced in Section \ref{secBo}.
The map $\,\:\St_{\:\bla}\:$ induces an isomorphism $\,\Srol\!\to\Srolh\>$.

\vsk.3>
Recall the Levelt solution $\,\:\St_{\:\bla}\Psho$ of quantum differential
equations \eqref{qDEQ}\:.

\begin{prop}
\label{StPshof}
For any $\,f\<\in\Oc_{H_\bla}\,$, the section $\;\St_{\:\bla}\Psho\?f\:$
with values
%defined by the rule
\vvn.5>
\be
(\:\St_{\:\bla}\Psho\?f\:)(\zz\:;\pp\:;\ka)\,=\,
\bigl(\:\St_{\:\bla}\Psho(\zz\:;\pp\:;\ka)\<\bigr)\,f(\zz)\kern-2em
\vv.5>
\ee
belongs to $\,\Srolh$. Moreover, the map
\vvn.2>
\beq
\label{muoh}
\muoh\?:\:\Oc_{H_\bla}\<\to\>\Srolh\:,\qquad
f\>\mapsto\,\St_{\:\bla}\Psho\?f\,,
\kern-2em
\vv.2>
\eeq
is an isomorphism.
\end{prop}
\begin{proof}
For any $\,f\<\in\Oc_{H_\bla}\,$, the section $\;\St_{\:\bla}\Psho\?f\,$ solves
\vv-.06>
quantum differential equations \eqref{qDEQ}. Since by formula \eqref{StPshp}\:,
\vv-.06>
Lemma \ref{StUH}, and Theorem \ref{Bthmo}, the section
$\;\St_{\:\bla}\Psho(\zz\:;\pp\:;\ka)\,$ is holomorphic in $\,\zz\,$ for
$\,\zz\in\Lo$, the section $\;\St_{\:\bla}\Psho\?f\,$ belongs to $\,\Srolh\>$.
Backwards, for any $\,f\<\in\Srolh$, the section $\,f^{\zeno}$ with values
\vvn.2>
\beq
\label{fzeno}
f^{\zeno}(\zz)\,=\,
\bigl(\:\St_{\:\bla}\Psho(\zz\:;\pp\:;\ka)\<\bigr)^{\?-1}\<f(\zz\:;\pp)
\kern-2em
\vv.4>
\eeq
does not depend on $\,\pp\,$ because $\,f(\zz\:;\pp)\,$ is a solution of
\vv.06>
quantum differential equations \eqref{DEQch}\:. Moreover,
$\,\bigl(\St_{\:\bla}\Psho\<(\zz\:;\pp\:;\ka)\bigr)^{\?-1}\:$ is holomorphic
in $\,\zz\,$ for $\,\zz\in\Lo\:$ since
\vv.12>
$\,\bigl(\det\:\St_{\:\bla}\Psho\<(\zz\:;\pp\:;\ka)\bigr)^{\?-1}\:$ is entire
in $\,\zz\,$. Thus $\,f^{\zeno}\!\<\in\<\Oc_{H_\bla}$ and
$\,f=\muoh f^{\zeno}\<$. Hence, the map \,$\muoh\,$ is an isomorphism.
\end{proof}

\vsk.1>
For a solution $\,f\<\in\<\Srolh\>$ of quantum differential equations
\vv.07>
\eqref{qDEQ}\:, we call the section $\,f^{\zeno}$ defined by \eqref{fzeno}
the {\it principal term\/} of $\,f\,$. For the solution
$\,\:\St_{\:\bla}\Pso_{\?P}\,$ corresponding to a Laurent polynomial
$\,P(\:\GGd\:;\zzd)\,$, see \eqref{mk}\:, the principal term is described
in Proposition \ref{SPsiPpro} below.

\enlargethispage{12pt}
\vsk.3>
Recall the functions $\,\Co_\bla(\zz\:;\ka)\,$ and $\Go_\bla(\zz\:;\ka)\,$,
see \eqref{CGGo}\:, \eqref{GGGo}\:,
\vvn.3>
\begin{gather*}
\Co_\bla(\zz\:;\ka)\,=\,\prod_{i=1}^N\>
\ka^{\:\left(\:\sum_{j=i+1}^N\la_j\,-\,\sum_{j=1}^{i-1}\:\la_j\<\right)
\sum_{a=\smash{\la^{\?(i-1)}}\<+1}^{\la^{(i)}}z_a\</\<\ka}\:,\kern-1em
\\[2pt]
\Go_\bla(\zz\:;\ka)\,=\,\prod_{i=1}^{N-1}\>\prod_{a=1}^{\la^{(i)}}\,
\prod_{b=\la^{(i)}+1}^{\la^{(i+1)}}\!\Gm\bigl(\<(z_a\<-z_b)/\ka\bigr)\,,
\kern-1em
\\[-17pt]
\end{gather*}
and the function $\,\Pdd(\:\GG\:;\zz\:;\ka)\,$ obtained from a Laurent
\vvn.1>
polynomial $\,P(\:\GGd\:;\zzd)\,$ by substituting the variables $\gmd_{\ij}\:$
and $\,\zdd_a\:$ with the exponentials $\,e^{\:2\:\pii\,\gm_{i\<,j}\</\<\ka}$
and $\,e^{\:2\:\pii\,z_a\</\<\ka}$, respectively. Set
\vvn.2>
%\begin{gather*}
%\label{CH}
\be
\CH_\bla(\:\GG\:;\zz\:;\ka)\,=\>
\sum_{I\in\Il}\frac{\Co_\bla(\zz_{\si_I};\ka)\,V_\bla(\:\GG\:;\zz_{\si_I}\<)}
{R_\bla(\zz_{\si_J}\<)}\;,\kern1.2em
%\kern-1em
%\\[6pt]
%\label{GH}
\GH_\bla(\:\GG\:;\zz\:;\ka)\,=\>
\sum_{I\in\Il}\Go_\bla(\zz_{\si_I};\ka)\,V_\bla(\:\GG\:;\zz_{\si_I}\<)\,,
%\kern-1em
\vv.2>
\ee
%\\[-22pt]
%\end{gather*}
and
\vvn-.5>
\be
%\label{PHd}
\PHd\?(\:\GG\:;\zz\:;\ka)\,=\>
\sum_{I\in\Il}\frac{\Pdd(\zz_{\si_I};\zz\:;\ka)\,V_\bla(\:\GG\:;\zz_{\si_I}\<)}
{R_\bla(\zz_{\si_J}\<)}\;,\kern-1em
\vv-.7>
\ee
cf.~\eqref{fFV}\:.

\begin{prop}
\label{SPsiPpro}
For a Laurent polynomial $\,P(\:\GGd\:;\zzd)\>$, the values of the principal
\vv.1>
term $\,\:\Pszo_{\?P}$ of the solution $\;\St_{\:\bla}\Pso_{\?P}$ are
\vvn.3>
\beq
\label{SPsoP8}
\Pszo_{\?P}(\zz\:;\ka)\,=\,\bigl[\>\PHd\?(\:\GG\:;\zz\:;\ka)\,
\CH_\bla(\:\GG\:;\zz\:;\ka)\,\GH_\bla(\:\GG\:;\zz\:;\ka)\bigr]_\zz\,.
\kern-2em
\vv.2>
\eeq
\end{prop}
\begin{proof}
By Proposition \ref{PsiPpro},
\vvn.4>
\beq
\label{SPsoP0}
\Pszo_{\?P}(\zz\:;\ka)\,=\?\sum_{I\in\:\Il}
\Ec_{\:\zz}\bigl\bra\PHd\?(\:\GG\:;\zz\:;\ka)\,
\CH_\bla(\:\GG\:;\zz\:;\ka)\,\GH_\bla(\:\GG\:;\zz\:;\ka)\,
\Stop_{\:I}(\:\GG\:;\zz)\bigr\ket\,\St_{\:I\<,\:\zz}\,.
\kern-1.2em
\vv.1>
\eeq
By Lemmas \ref{Stort}, \ref{SHz0}, for any section $\,f\,$ of $\,H_\bla\>$,
\vvn.4>
\be
[\:f(\:\GG\:;\zz)\:]_{\:\zz}\:=\?\sum_{I\in\:\Il}
\Ec_{\:\zz}\bigl\bra f(\:\GG\:;\zz)\,\Stop_{\:I}(\:\GG\:;\zz)\bigr\ket\,
\St_{\:I\<,\:\zz}\,,
\kern-1.2em
\ee
Hence, the right\:-hand sides of formulae \eqref{SPsoP8} and \eqref{SPsoP0}
coincide.
\end{proof}

The functions $\,\CH_\bla(\:\GG\:;\zz\:;\ka)\,$,
$\,\GH_\bla(\:\GG\:;\zz\:;\ka)\,$, and $\,\:\PHd\?(\:\GG\:;\zz\:;\ka)\,$
\vv.07>
are polynomials in $\,\:\GG$ analytically depending on $\,\zz\,$. The sections
of $\,H_\bla\:$ defined by their cohomology classes can be thought of as
the cohomology classes of the respective analytic functions of $\,\:\GG\:$,
\vvn.2>
\begin{gather}
\label{cogog}
\Co_\bla(\:\GG\:;\ka)\,=\,\prod_{i=1}^N\>
\ka^{\:\left(\:\sum_{j=i+1}^N\la_j\,-\,\sum_{j=1}^{i-1}\:\la_j\<\right)
\sum_{j=1}^{\smash{\la_i}}\?\gm_{i\<,j}\</\<\ka}\:,\kern-1em
\\[3pt]
\notag
\Gp_\bla(\:\GG\:;\ka)\,=\,\ka^{\:\la_{\{2\}}}\:
\prod_{i=1}^{N-1}\prod_{j=i+1}^N\>\prod_{k=1}^{\la_i}\,\prod_{l=1}^{\la_j}
\,\Gm\bigl(1+(\gm_{\ik}\<-\gm_{\jl})/\ka\bigr)\,,\kern-1em
\\[-15pt]
\notag
\end{gather}
and $\,\Pdd(\:\GG\:;\zz\:;\ka)\,$. The class of $\;\Co_\bla(\:\GG\:;\ka)\,$
\vv.07>
is a product of exponentials of the equivariant first Chern classes
\,$c_1(E_i)\,=\,[\:\gm_{i,1}\<\lsym+\gm_{i,\:\la_i}]\,$ of
\vv.1>
the vector bundles $\,E_i\>$ over $\,\Fla\>$ with fibers $\,F_i\:/F_{i-1}\,$,
\vv.04>
\:the class of $\;\Gp_\bla(\:\GG\:;\ka)\,$ is the equivariant Gamma\:-\:class
of $\,\Fla\>$, \:and the class of $\,\Pdd\?(\:\GG\:;\zz\:;\ka)\,$ is
the equivariant Chern character of the class of the Laurent polynomial
$\,P(\:\GGd\:;\zzd)\,$ in $\,K_T(\Fla\>;\C)\,$.

\vsk.2>
Define a map
\vvn-.2>
\beq
\label{Bcyh}
\Bcyh_\bla\::\:K_T(\Fla\>;\C)\>\to\>\Oc_{H_\bla}\kern-2em
\vv.2>
\eeq
that sends the class $\,\:[P\:]\,$ of the Laurent polynomial
\vv.07>
$\,P(\:\GGd\:;\zzd)\,$ in $\,K_T(\Fla\>;\C)\,$ to the section of $\,H_\bla\:$
with values $\,[\>\PHd\?(\:\GG\:;\zz\:;\ka)\,\CH_\bla(\:\GG\:;\zz\:;\ka)\,
\GH_\bla(\:\GG\:;\zz\:;\ka)\:]_{\:\zz}\,$.
\vv.06>
By Proposition \ref{SPsiPpro}, the map $\,\:\Bcyh_\bla\>$ sends
the class $\,[P\:]\,$ to the principal term of the solution
$\,\:\St_{\:\bla}\Pso_{\?P}\:$ of the joint system of quantum differential
equations \eqref{qDEQ} and \qKZ/ difference equations in cohomology
\eqref{Kich}.

\begin{thm}
\label{Strianglo}
Recall the space $\,\Oc_{H_\bla}\!\<$ of sections of $\,H_\bla\:$ holomorphic
\vv.06>
in $\,\zz\,$ for $\,\zz\in\Lo$ and not depending on $\,\pp\,$, and the space
\vv.06>
$\,\Srolh\<$ of solutions of quantum differential equations \eqref{qDEQ}
holomorphic in $\,\zz\,$ for $\,\zz\in\Lo$.
Then the map $\;\Bcyh_\bla\::\:K_T(\Fla\>;\C)\:\to\:\Oc_{H_\bla}\:$
is well-defined and the following diagram is commutative,
\vvn-.5>
\beq
\label{Scd}
\xymatrix{{}\phan{\Oc_{H_\bla}}\llap{$K_T(\Fla\>;\C)$}
\ar^-{\;\Bcyh_\bla\!\!}[rr]
\ar_{\smash{\lower.8ex\llap{$\ssize\mkh\,\:$}}}[dr]&&\,\:\Oc_{H_\bla}
\ar^{\smash{\lower.8ex\rlap{$\;\ssize\muoh$}}}[dl]\\
&\!\Srolh&}
\vv-.1>
\eeq
\end{thm}
\begin{proof}
The statement follows from Proposition \ref{trianglo} by applying the stable
\vv.07>
envelope map $\,\:\St_{\:\bla}\>$. In more detail, by the standard reasoning
the functions $\;\CH_\bla(\:\GG\:;\zz\:;\ka)\,$,
\vv.06>
$\;\GH_\bla(\:\GG\:;\zz\:;\ka)\,$, $\,\Pdd\?(\:\GG\:;\zz\:;\ka)\,$, and
$\,\Ec_{\:\zz}\bigl\bra\PHd\?(\:\GG\:;\zz\:;\ka)\,
\CH_\bla(\:\GG\:;\zz\:;\ka)\,\GH_\bla(\:\GG\:;\zz\:;\ka)\,
\vv.07>
\Stop_{\:I}(\:\GG\:;\zz)\bigr\ket\,$ are holomorphic in $\,\zz\,$
for $\,\zz\in\Lo$. Hence, the map $\,\:\Bcyh_\bla\>$ is well-defined. The
commutativity of diagram \eqref{cdo} follows from Proposition \ref{SPsiPpro}.
\end{proof}

\subsection{The nonequivariant case $\,\zz=0\,$}
\label{s6.7}
In this section, we will discuss solutions of the quantum differential
equations for the cohomology algebra $\,\Hd^*\<(\Fla\>;\C)\,$ of the partial flag
variety $\,\Fla\>$ by specializing the results obtained for the equivariant
case at $\,\zz=0\,$.

\vsk.2>
The cohomology algebra $\,\Hd^*\<(\Fla\>;\C)\,$ of the partial flag variety
$\,\Fla\>$ has the form
\vvn.3>
\beq
\label{Hrel0}
\Hd^*\<(\Fla\>;\C)\,=\,\C[\:\GG\:]^{\:S_\bla}\?\Big/\Bigl\bra
\,\prod_{i=1}^N\,\prod_{j=1}^{\la_i}\,(u-\ga_{\ij})\,=\,u^n\:\Bigr\ket\,,
\kern-2em
\vv.2>
\eeq
where $\,u\,$ is a formal variable. It is isomorphic to the algebra
\vv.1>
$\,\Hd^*\<(\Fla\>;\C)_{\:\zz^0}\>$ at $\,\zz^0\<=0\,$, see \eqref{Hrelz0}\:.
We denote the class of a polynomial $\,f(\:\GG)\in\C[\:\GG\:]^{\:S_\bla}$
\vv.1>
in $\,\Hd^*\<(\Fla\>;\C)\,$ by $\,[\:f\>]^{}_{\:0}\,$.

\vsk.2>
The integration map on $\,\Hd^*\<(\Fla\>;\C)\,$ coincides with the map
$\,\:\Ec_{\zz^0}\:$ at $\,\zz^0\?=0\,$, see \eqref{Ecz0}\:,
\vvn.5>
\beq
\label{Ec0}
\Ec_{\:0}:\:\Hd^*\<(\Fla\>;\C)\>\to\>\C\,,\qquad
[\:f\>]^{}_{\:0}\:\mapsto\>\Ec_{\:0}\bra f\ket\,,\kern-1em
\vv.1>
\eeq
where
\vvn-.3>
\be
\Ec_{\:0}\bra f\ket\,=\:\biggl(\,\sum_{I\in\Il}\:
\frac{f(\zz_{\si_I})}{R_\bla(\zz_{\si_I}\<)}\>\biggr)\bigg|_{\:\zz=0}\>,
\kern-1em
\vv-.4>
\ee
see Lemmas \ref{lemEc}, \ref{lemEcz0}.

\vsk.3>
Consider the polynomials
$\,\:\St_{\:I\<,\:0}\:(\:\GG)=\St_{\:I}(\:\GG\:;0)\,$ and
$\,\:\Stop_{\:I\<,\:0}\:(\:\GG)=\Stop_{\:I}(\:\GG\:;0)\,$.
By formulae \eqref{StIma}\:,
\vvn-.3>
\begin{gather*}
\St_{\:\si_0(\Imil),\:0}\:(\:\GG)\,=\,\Stop_{\:\Imil\<,\:0}\:(\:\GG)\,=\,1\,,
\kern-.6em
\\[4pt]
\St_{\:\Imil\<,\:0}\:(\:\GG)\,=\,\Stop_{\:\si_0(\Imil),\:0}\:(\:\GG)\,=\>
\prod_{i=1}^{N-1}\:\prod_{j=1}^{\la_i}\,\gm_{\ij}^{\:n-\la^{(i)}}.\kern-.6em
\\[-16pt]
\end{gather*}
%where $\,\si_0\>$ is the longest permutation, $\,\si_0(a)=n+1-a\,$,
%$\;a=1\lc n\,$.
In general, the polynomials $\,\:\St_{\:I\<,\:0}\:(\:\GG)\,$,
$\Stop_{\:I\<,\:0}\:(\:\GG)\,$ coincide with the $A\:$-type Schubert
polynomials, see \eqref{StS}\:,
\beq
\label{StS0}
\St_{\:I\<,\:0}\:(\:\GG)\,=\,\Sg_{\si_{\?\si_{\<0}\<(I)}}(\:\GG\:;0\:)\,,
\qquad \Stop_{\:I\<,\:0}\:(\:\GG)\,=\,\Sg_{\si_I}(\:\GG\:;0\:)\,.
\kern-2em
\vv.2>
\eeq
By Lemma \ref{Stort}, the polynomials $\;\St_{\:I\<,\:0}\:(\:\GG)\,$ and
$\;\Stop_{\:I\<,\:0}\:(\:\GG)\,$ are biorthogonal,
\vvn.5>
\beq
\label{ESSop0}
\Ec_{\:0}\bra\:\St_{\:I\<,\:0}\:(\:\GG)\,\Stop_{\<J,\:0}\:(\:\GG)\ket\,=\,\dl_{\IJ}\,.
\vv.3>
\eeq

\vsk.3>
Denote by $\,\:\St_{\:I\<,\:0}$, $\Stop_{\:I\<,\:0}\>$ the classes of
$\,\:\St_{\:I\<,\:0}\:(\:\GG)\,$, $\Stop_{\:I\<,\:0}\:(\:\GG)\,$
\vvn.06>
in $\,\Hd^*\<(\Fla\>;\C)\,$. They are the fundamental classes of the Schubert
subvarieties in $\,\Fla\>$ defined for the opposite orderings of the chosen
basis of $\,\:\C^n$ in the standard way,

\vsk.3>
The quantum multiplication in $\,\Hd^*\<(\Fla\>;\C)\,$, is a deformation of
\vv.06>
the multiplication in $\,\Hd^*\<(\Fla\>;\C)\,$ depending on quantum parameters
$\,\pp\,$. For $\>Y\!\in \Hd^*\<(\Fla\>;\C)\,$, let $\>Y\<*_\pp\:$ be
the operator of quantum multiplication by $\>Y\:$ depending on $\,\pp\,$.

\vsk.3>
For $\,i=1\lc N\>$, denote $\,D_i\<=\gm_{i,1}\<\lsym+\gm_{i,\>\la_i}\>$,
\vv.1>
cf.~\eqref{c1}\:. The quantum differential equations for
$\,\Hd^*\<(\Fla\>;\C)\:$-valued functions of $\,\pp\,$ is the following
system of compatible differential equations,
\vvn-.3>
\beq
\label{qDEQ0}
\ka\>p_i\:\frac{\der f}{\der\:p_i}\,=\,{[D_i]^{}_{\:0}\:*_\pp}\,f\,,\qquad
i=1\lc N\>,\kern-2em
\vv.4>
\eeq
where $\,\ka\,$ is a parameter of the equations. Equations \eqref{qDEQ0}
\vv.1>
coincide with the restriction of equivariant quantum differential equations
\vv.06>
\eqref{qDEQ} to the subbundle $\,\Hd^*\<(\Fla\>;\C)\times\C^N\!\to\:\C^N$
of the bundle $\,H_\bla\?\to\:\C^n\!\times\C^N$ located over the points
$\,(\:0\:,\pp\:)\in\C^n\!\times\C^N$ of the base.

\vsk.3>
Consider the $\,K\?$-theory algebra $\,K(\Fla\>;\C)\,$. Then
\vvn.2>
\beq
\label{Krelt1}
K(\Fla\>;\C)\,=\,\C[\:\GGd^{\pm1}]^{\:S_\bla}\?\ox\C[\:\zzd^{\pm1}]\>\Big/
\Bigl\bra\,\prod_{i=1}^N\prod_{j=1}^{\la_i}\,(u-\gmd_{\ij})\,=\,
\prod_{a=1}^n\,(u-1)\Bigr\ket\,,\kern-1.6em
\vv.2>
\eeq
where $\,u\,$ is a formal variable. The evaluation map
\vvn.5>
\be
\C[\:\GGd^{\pm1}]^{\:S_\bla}\?\ox\C[\zzd^{\pm1}]\>\to\>
\C[\:\GGd^{\pm1}]^{\:S_\bla}\>,\qquad
P(\:\GGd;\zzd)\mapsto P\bigl(\:\GGd\:;(1\lc1)\bigr)\,,\kern-1em
\vv.4>
\ee
induces the isomorphism of
$\,K(\Fla\>;\C)\to K_T(\Fla\>;\C)/\bra\:\zz=(1\lc1)\:\ket\,$.
\vv.1>
Denote by $\,[\:P\:]^{}_{\>1}\:$ the class of
$\,P(\:\GGd)\in\C[\:\GGd^{\pm1}]^{\:S_\bla}\:$ in $\,K(\Fla\>;\C)\,$.

\vsk.3>
Identify $\,\C[\:\GGd^{\pm1}]^{\:S_\bla}\:$ with the subspace of
\vv.07>
$\,\C[\:\GGd^{\pm1}]^{\:S_\bla}\?\ox\C[\zzd^{\pm1}]\,$ of Laurent polynomials
not depending on $\,\zz\,$. For each $\,P\in\C[\:\GGd^{\pm1}]^{\:S_\bla}\:$,
\vv.06>
the solution $\,\:\St_{\:\bla}\Pso_{\?P}(\zz\:;\pp\:;\ka)\,$
of equivariant quantum differential equations \eqref{qDEQ}\:,
see Section \ref{s6.6}, is regular at $\,\zz=0\,$, and the solution
$\,\:\St_{\:\bla}\Pso_{\?P}(\:0\:;\pp\:;\ka)\,$ of equations \eqref{qDEQ0}
depends only on the class of $\,P\:$ in $\,K(\Fla\>;\C)\,$. That is,
there is a well-defined map
\vvn.5>
\beq
\label{mk0}
\mko\?:\>K(\Fla\>;\C)\:\to\,\Hrsl\,,\qquad
[P\:]_{\>1}^{}\:\mapsto\>\St_{\:\bla}\Pso_{\?P}(\:0\:;\pp\:;\ka)\,,\kern-2em
\vv.4>
\eeq
where $\,\Hrsl\>$ is the space of solutions of equations \eqref{qDEQ0}\:.

\begin{prop}
\label{mkoiso}
The map $\,\mko\?$ is an isomorphism.
%see \eqref{mk0}\:, 
\end{prop}
\begin{proof}
The statement follows from Theorem \ref{detYo} \>at $\>\zz=0\,$.
\end{proof}

The Levelt fundamental solution $\,\:\St_\bla\Psho\<(\zz\:;\pp\:;\ka)\,$
\vv-.04>
of equivariant quantum differential equations \eqref{qDEQ}\:,
see Section \ref{s6.6}, is regular at $\,\zz=0\,$ and the function
\vv.06>
$\,\:\St_\bla\Psho\<(\:0\:;\pp\:;\ka)\,$ is
an $\,\:\End\:\bigl(\Hd^*(\Fla\>;\C)\bigr)$-valued
solution of quantum differential equations \eqref{qDEQ0}\:.
\vv-.05>
We call $\,\:\St_\bla\Psho\<(\:0\:;\pp\:;\ka)\,$ the Levelt fundamental
\vv.07>
solution of quantum differential equations \eqref{qDEQ0}\:.
It induces the map
\beq
\label{mh0}
\mho\?:\>\Hd^*(\Fla\>;\C)\:\to\,\Hrsl\,,\qquad
f\>\mapsto\>\St_\bla\Psho\<(\:0\:;\pp\:;\ka)\,f\,,\kern-2em
\vv.2>
\eeq
to the space $\,\Hrsl\>$ of solutions of equations \eqref{qDEQ0}\:.

\begin{prop}
\label{mhoiso}
The map $\,\mho\?$ is an isomorphism.
%see \eqref{mh0}\:,
\end{prop}
\begin{proof}
The statement follows from Proposition \ref{StPshof} \>at $\>\zz=0\,$.
\end{proof}

For a solution $\,f(\pp)\<\in\Hrsl\>$ of quantum differential equations
\eqref{qDEQ0}\:, we call the element
\vvn.4>
\beq
\label{fzen0}
f^{\zenoo}=\,
\bigl(\:\St_{\:\bla}\Psho(\:0\:;\pp\:;\ka)\<\bigr)^{\?-1}\<f(\pp)
\kern-2em
\vv.4>
\eeq
the {\it principal term\/} of $\,f(\pp)\,$. The principal term of the solution
\vv.07>
$\,\:\St_{\:\bla}\Pso_{\?P}(\:0\:;\pp\:;\ka)\,$ corresponding to a Laurent
polynomial $\,P(\:\GGd)\,$ is given by Proposition \ref{SPsiPpr0} below.

\vsk.3>
For a function $\,F(\:\GG)\,$ holomorphic in a neighborhood of $\;\GG\<=0\,$
\vv.06>
and symmetric in $\,\gm_{i,1}\lc\gm_{i,\>\la_i}$ for each $\,i=1\lc N\>$,
define its class $\,[\:F\>]_{\:0}^{}\>$ in $\,\Hd^*(\Fla\>;\C)\,$ by expanding
$\,F(\:\GG)\,$ in the power series about $\;\GG\<=0\,$ and replacing each term
\vv.04>
by the corresponding class in $\,\Hd^*(\Fla\>;\C)\,$. The resulting sum contains
\vv.05>
only finitely many nonzero terms and, hence, is well-defined. Alternatively,
the class $\,[\:F\>]_{\:0}\>$ can be evaluated as follows. Set
\vvn.5>
\beq
\label{FH}
\FH\?(\:\GG\:;\zz)\,=\>
\sum_{I\in\Il}\frac{F(\zz_{\si_I})\,V_\bla(\:\GG\:;\zz_{\si_I}\<)}
{R_\bla(\zz_{\si_J}\<)}\;.\kern1.2em
\eeq

\begin{lem}
\label{FFH}
The function $\,\FH(\:\GG\:;\zz)\:$ is regular at $\,\zz=0\>$ and
$\,\:[\:F\>]_{\:0}^{}\<=[\:\FH\?(\:\GG\:;0\:)\:]_{\:0}^{}\,$.
\end{lem}
\begin{proof}
By the standard facts on power series, it suffices to prove the statement
\vv.06>
for a polynomial $\,F(\:\GG)\,$. In the polynomial case,
\vv.1>
formulae \eqref{FH}\:, \eqref{VGz} yield $\,F(\:\GG)=\FH\?(\:\GG\:;\GG)\,$.
Then $\,\:[\:F(\:\GG\:)\:]_{\:0}^{}\<=[\:\FH\?(\:\GG\:;\GG\:)\:]_{\:0}^{}\<=
[\:\FH\?(\:\GG\:;0\:)\:]_{\:0}^{}\,$,
because $\,\FH\?(\:\GG\:;\zz\:)\,$ is symmetric in $\,\zz\,$.
\end{proof}

Recall the functions $\,\Co_\bla(\:\GG\:;\ka)\,$, $\,\Gp_\bla(\:\GG\:;\ka)\,$,
see \eqref{cogog}\:, and the function $\,\Pdd(\:\GG\:;\ka)\,$ obtained from a
Laurent polynomial $\,P(\:\GGd)\,$ by substituting the variables $\gmd_{\ij}\:$
with the exponentials $\,e^{\:2\:\pii\,\gm_{i\<,j}\</\<\ka}$.
\vv.1>
The class of $\;\Co_\bla(\:\GG\:;\ka)\,$ is a product of exponentials of the
first Chern classes \,$c_1(E_i)\,=\,[\:\gm_{i,1}\<\lsym+\gm_{i,\:\la_i}]\,$ of
the vector bundles $\,E_i\>$ over $\,\Fla\>$ with fibers $\,F_i\:/F_{i-1}\,$,
\:the class of $\;\Gp_\bla(\:\GG\:;\ka)\,$ is the Gamma\:-\:class of
$\,\Fla\>$, \:and the class of $\,\Pdd\?(\:\GG\:;\ka)\,$ is the Chern character
of the class of the Laurent polynomial $\,P(\:\GGd)\,$ in $\,K(\Fla\>;\C)\,$.

\begin{prop}
\label{SPsiPpr0}
For a Laurent polynomial $\,P(\:\GGd)\>$,
\vv.1>
the principal term $\,\:\Pszoo_{\?P}\!$ of the solution
$\;\St_{\:\bla}\Pso_{\?P}(\:0\:;\pp\:;\ka)$ equals
\beq
\label{SPsoP80}
\Pszoo_{\?P}\<(\ka)\,=\,\bigl[\>\Pdd\?(\:\GG\:;\ka)\,
\Co_\bla(\:\GG\:;\ka)\,\Gp_\bla(\:\GG\:;\ka)\bigr]_{\:0}^{}\,.
\kern-2em
\vv.3>
\eeq
\end{prop}
\begin{proof}
The statement follows from Proposition \ref{SPsiPpr0} \>at $\>\zz=0\,$ and
Lemma \ref{FFH}.
\end{proof}

\vsk.2>
Define a map
\vvn.4>
\beq
\label{Bcyho}
\Bcyho_\bla:\:K(\Fla\>;\C)\>\to\>\Hd^*(\Fla\>;\C)\,,\qquad
[P\:]^{}_{\>1}\to\>\bigl[\>\Pdd\?(\:\GG\:;\ka)\,
\Co_\bla(\:\GG\:;\ka)\,\Gp_\bla(\:\GG\:;\ka)\bigr]_{\:0}^{}\,,\kern-1.5em
\vv.4>
\eeq
By Proposition \ref{SPsiPpr0}, the map $\,\:\Bcyho_\bla\:$ sends
\vv.1>
the class $\,[P\:]^{}_{\>1}$ to the principal term of the solution
$\,\:\St_{\:\bla}\Pso_{\?P}(\:0\:;\pp\:;\ka)\,$ of quantum
differential equations \eqref{qDEQ0}\:.

\begin{thm}
\label{Striangl0}
Recall the space $\,\Hrsl$ of solutions of quantum differential equations
\eqref{qDEQ0}\:. The following diagram is commutative,
\vvn-.1>
\beq
\label{Scd0}
\xymatrix{{}\phan{\Hd^*(\Fla\>;\C)}\llap{$K(\Fla\>;\C)$}\,
\ar^-{\;\Bcyho_\bla\!\!}[rr]
\ar_{\smash{\lower.8ex\llap{$\ssize\mko\,\:$}}}[dr]&&\,\:\Hd^*(\Fla\>;\C)
\ar^{\smash{\lower.8ex\rlap{$\;\ssize\mho$}}}[dl]\\
&\!\Hrsl&}
\eeq
\end{thm}
\begin{proof}
The statement is equivalent to Proposition \ref{SPsiPpr0}.
\end{proof}

\subsection{Topological\:-\:enumerative morphism and $J$-\:function}
\label{sec:top}
Recall the Levelt fundamental solution $\,\:\St_\bla\Psho\<(\:0\:;\pp\:;\ka)\,$
of quantum differential equations \eqref{qDEQ0}\:. By Lemma \ref{SHz0}
and biorthogonality relation \eqref{Stort}\:, for any $\,I\<\in\Il\,$,
there exists a unique $\,\Hd^*(\Fla\>;\C)\:$-valued function
$\,\Jchh_{\!\?I}\<(\pp\:;\ka)\,$ such that for any $\,f\<\in \Hd^*(\Fla\>;\C)\,$,
\vvn.4>
\beq
\label{EJcStp0}
\Ec_{\:0}^{}\bigl\bra\<\bigl(\:\St_{\:\bla}\Psho(\:0\:;\pp\:;\ka)\>f\:\bigr)\,
\Stop_{\:I\<,\:0}\bigr\ket\>=\,
\Ec_{\:0}^{}\bigl\bra\Jchh_{\!\?I}\<(\pp\:;\ka)\>f\:\bigr\ket\,.\kern-2em
\vv.4>
\eeq
Since $\,\:\Stop_{\:\Imil\?,\:0}\<=1\,$, see \eqref{StIma}\:,
formula \eqref{EJcStp0} for $\,I=\Imil$ takes the form
\vvn.2>
\beq
\label{EJcStp0m}
\Ec_{\:0}^{}\bigl\bra\St_{\:\bla}\Psho(\:0\:;\pp\:;\ka)\>f\:\bigr\ket\>=\,
\Ec_{\:0}^{}\bigl\bra\Jchh_{\!\?\Imil}\<(\pp\:;\ka)\>f\:\bigr\ket\,.\kern-2em
\vv.2>
\eeq

\begin{prop}
\label{lemJchI}
For any $\,I\<\in\Il\>$,
\beq
\label{JchI}
\kern-.5em
\Jchh_{\!\?I}\?(\pp\:;\ka)\,=\,
\Omo_\bla(\pp\:;\ka)\!\?\sum_{\mb\:\in\:\Z_{\ge0}^{N-1}\!}\!
[\:\Jchh_{\!\?I\<,\:\mb}(\:\GG\:;0\:;\ka)\:]_{\:0}^{}\>
\prod_{i=1}^{N-1}\,
\bigl(\<(-\:\ka)^{-\la_i-\la_{i+1}}\>p_{i+1}/p_i\:\bigr)^{\:m_i}\>
\prod_{i=1}^N\,p_i^{\:[D_i]_0^{}/\ka}\:,\kern-1em
\eeq
%cf.~\eqref{Jchm}\:.
where the functions $\,\Jchh_{\!\?I\<,\:\mb}\<$ are given by \eqref{JchIm}\:,
$\;\Omo_\bla\?=e^{\:\sum_{\:i<j}^{\cirs}p_j\</(\ka\:p_i)}$
\vv.07>
in which the sum is taken over all pairs $\,i<j\,$ such that $\,\la_i\<=1\,$
and $\,\la_{i+1}\<\lsym=\la_j\<=0\,$,
%$\,\la_s\<=0\,$ for all $\,i<\?s\le j\,$,
see \eqref{Omlo}\:,
and $\,D_i\<=\gm_{i,1}\<\lsym+\gm_{i,\>\la_i}\>$.
\end{prop}
\begin{proof}
By Proposition \ref{Jchhreg}, the functions
\vv-.1>
$\,\Jchh_{\!\?I\<,\:\mb}(\:\GG\:;\zz\:;\ka)\,$ are regular at $\,\zz=0\,$.
Thus the classes $\;[\:\Jchh_{\!\?I\<,\:\mb}(\:\GG\:;0\:;\ka)\:]_{\:0}^{}\>$
\vv.06>
in formula \eqref{JchI} are well defined. Then the statement follows from
Proposition \ref{lemEJc} and formulae \eqref{EJcStp0}\:, \eqref{SPspo},
\eqref{StPsho}\:\:.
\end{proof}

By inspection, the function $\,\Jchh_{\!\?\Imil}\<(\pp\:;\ka)\,$ in this paper
\vv.05>
coincides with the nonequivariant $\:J$-\:function of $\,\Fla\>$ obtained
in \cite{BCK} up to a change of notation.

\vsk.2>
In more detail, formula \eqref{JchIm} for the function
$\,\Jchh_{\!\?\Imil\<,\:\mb}(\:\Gm\:;\zz\:;\ka)\,$ takes the form
\vvn.1>
\beq
\label{JchImil}
\Jchh_{\!\?\Imil\<,\:\mb}(\:\GG\:;\zz\:;\ka)\,=\>
\sum_{i=1}^{N-1}\!\sum_{\satop{\lb\in\:\smash{\Z^{\:\la\+1}}\!\vp{|^{1^1}}}
{\!|\:\lb^{(i)}\<|=\:m_i}}\<\sum_{J\in\Il}
\frac{\Ao\<(\Si_J\?-\lb\:;\zz\:;\ka)\,V_\bla(\:\GG\:;\zz_{\si_J}\<)}
{\Aob\<(\Si_J;\zz)\,R_\bla(\zz_{\si_J}\<)}\;,
\kern-1em
\eeq
because $\,\Wo_{\?\Imil}(\Si_J\?-\lb\ka\:;\zz)=1\,$, see Lemma \ref{Womax}.
\vv-.05>
In \eqref{JchImil}\:, $\,\mb=(m_1\lc m_{N-1})\<\in\Z_{\ge0}^{N\<-1}$,
$\,\:\lb=\bigl(\>l^{\:(1)}_{\>1}\!\lc l^{\:(1)}_{\:\la^{(1)}},\,\ldots\,,
l^{\:(N-1)}_{\>1}\!\lc l^{\:(N-1)}_{\:\la^{(N-1)}}\bigr)\in\Z^{\la\+1}\?$,
$\;|\>\lb^{(i)}|=\:l^{\:(i)}_j\!\lsym+\:l^{\:(i)}_{\la^{(i)}}\,$,
\vvn.3>
\be
\Ao\<(\TT\:;\zz\:;\ka)\,=\,\prod_{i=1}^{N-1}\,
\prod_{a=1}^{\la^{(i)}}\;\biggl(\,\prod_{\satop{b=1}{b\ne a}}^{\la^{(i)}}
\>\Gm\bigl(1+(t^{(i)}_b\!-t^{(i)}_a)/\ka\bigr)
\prod_{c=1}^{\la^{(i+1)}}\:
\frac1{\Gm\bigl(1+(t^{(i+1)}_c\!-t^{(i)}_a)/\ka\bigr)}\,\biggr)\>,\kern-1em
\vv-.1>
\ee
where $\,\la^{(N)}\?=n\,$, $\;t^{(N)}_a\?=z_a\,$, and
$\,\:l^{\:(i)}_a\!=0\,$ for $\,a>\<\la^{(i)}$ or $\,i=N\>$, and
\vv.07>
$\,\:\Si_J\>$ is given by \eqref{SiI}\:. Then the expression in
\cite[\:Theorem 1.3\:]{BCK} for $\,J_d^{\:F}\<(\hbar)\,$ with $\;l=N-1\,$,
\vv.05>
$\,\:d_i\<=m_i\,$, $\,\:d_{\ij}\<=\:l^{\:(i)}_j\<$, $\,\:\hbar=\ka\,$, and
$\,H_{\ij}\<=\gm_{a,\>j-\la^{\<(a-1)}}\,$ for $\,\la^{(a-1)}\!<j\le\la^{(a)}$,
\vv.07>
agrees term by term with the expression for the class
$\,\:[\:\Jchh_{\!\?I\<,\:\mb}(\:\GG\:;0\:;\ka)\:]_{\:0}^{}\>$
obtained from \eqref{JchImil}\:. Also, formula \eqref{JchI} for
$\,\Jchh_{\!\?\Imil}\<(\pp\:;\ka)\,$ is the same as the series expansion
of the nonequivariant $\:J$-\:function of $\,\Fla\>$ in \cite{BCK}\:.

\vsk.2>
Observe that relation \eqref{EJcStp0m} between the $\:J$-\:function
\vv-.1>
$\,\Jchh_{\!\?\Imil}\<(\pp\:;\ka)\,$ and the Levelt fundamental solution
\vv.06>
$\,\:\St_{\:\bla}\Psho(\:0\:;\pp\:;\ka)\,$ matches the general relation
between the $\:J$-\:function and the topological\:-enumerative morphism,
see Definition 5.3 and formula (5.8) in \cite{CV}, where the $\:J$-\:function
is defined in terms of the topological\:-enumerative morphism.
This naturally suggests the following conjecture.

\begin{conj}
\label{topconj0}
The Levelt fundamental solution $\,\:\St_{\:\bla}\Psho(\:0\:;\pp\:;\ka)\>$ of
\vv.05>
quantum differential equations \eqref{qDEQ0} is the topological\:-enumerative
morphism of $\,\Fla\>$.
\end{conj}

Observe that formulae \eqref{EJcStp0}\,--\,\eqref{JchI} are specializations
of respective formulae for the Levelt fundamental solution
$\,\:\St_\bla\Psho\<(\zz\:;\pp\:;\ka)\,$ of equivariant quantum differential
equations \eqref{qDEQ}\:. By Lemma \ref{SHz0} and biorthogonality
relation \eqref{Stort}\:, for any $\,I\<\in\Il\,$, there exists a unique
section $\,\Jchh_{\!\?I}\<(\zz\:;\pp\:;\ka)\,$ of the bundle $\,H_\bla\:$
such that for any $\,f\<\in \Hd^*(\Fla\>;\C)\,$ considered as a section of
$\,H_\bla\:$ not depending on $\,\pp\,$,
\vvn.4>
\beq
\label{EJcStpz}
\Ec_{\:\zz}^{}\bigl\bra\<\bigl(\:\St_{\:\bla}\Psho(\zz\:;\pp\:;\ka)\>
[\:f\>]_{\:\zz}^{}\bigr)\,\Stop_{\:I\<,\:\zz}\bigr\ket\>=\,
\Ec_{\:\zz}^{}\bigl\bra\Jchh_{\!\?I}\?(\zz\:;\pp\:;\ka)\>
[\:f\>]_{\:\zz}^{}\bigr\ket\,.\kern-2em
\vv.4>
\eeq
Since $\,\:\Stop_{\:\Imil\?,\:\zz}\<=1\,$, see \eqref{StIma}\:,
formula \eqref{EJcStp0} for $\,I=\Imil$ takes the form
\vvn.2>
\beq
\label{EJcStpzm}
\Ec_{\:\zz}^{}\bigl\bra\St_{\:\bla}\Psho(\zz\:;\pp\:;\ka)\>
[\:f\>]_{\:\zz}^{}\bigr\ket\>=\,
\Ec_{\:\zz}^{}\bigl\bra\Jchh_{\!\?\Imil}\<(\zz\:;\pp\:;\ka)\>
[\:f\>]_{\:\zz}^{}\bigr\ket\,.\kern-2em
\vv.2>
\eeq
Furthermore, by Proposition \ref{lemEJc} and formulae \eqref{EJcStpz}\:,
\eqref{SPspo}, \eqref{StPsho}\:,
\vvn.5>
\begin{align}
\label{JchIz}
\kern-.4em
&\;\Jchh_{\!\?I}\?(\zz\:;\pp\:;\ka)\,={}
\\[-2pt]
\notag
&{}=\,\Omo_\bla(\pp\:;\ka)\!\?\sum_{\mb\:\in\:\Z_{\ge0}^{N-1}\!}\!
[\:\Jchh_{\!\?I\<,\:\mb}(\:\GG\:;\zz\:;\ka)\:]_{\:z}^{}\>
\prod_{i=1}^{N-1}\,
\bigl(\<(-\:\ka)^{-\la_i-\la_{i+1}}\>p_{i+1}/p_i\:\bigr)^{\:m_i}\>
\prod_{i=1}^N\,p_i^{\:[D_i]_\zz^{}/\ka}\:,\kern-.1em
\end{align}
where the notation is the same as in formula \eqref{JchI}\:.
This suggests the equivariant version of Conjecture \ref{topconj0}.

\begin{conj}
\label{topconj}
The Levelt fundamental solution $\,\:\St_{\:\bla}\Psho(\zz\:;\pp\:;\ka)\>$ of
\vv.05>
equivariant quantum differential equations \eqref{qDEQ} is the equivariant
topological\:-enumerative morphism of $\,\Fla\>$.
\end{conj}

\subsection{Example $\,\bla=(1,n-1)\,$, $\,\Fla=\CP^{\>n-1}$}
\label{examplecp}
Throughout this section,
\vvn.06>
let $\,N=2\,$, $\,n\ge 2\,$, $\,\bla=(1,n-1)\,$,
$\,\Fla=\CP^{\>n-1}$.

\vsk.2>
Denote $\,\gm=\gm_{1,1}\>$. Identify $\,\:\C[\:\gm\:]\,$ with the subspace
\vv.07>
of $\,\:\C[\:\GG\:]^{\:S_\bla}$ consisting of polynomials not depending
on $\,\gm_{\:2,1}\lc\<\gm_{\:2,\:n-1}\>$, and $\,\:\C[\:\GG\:]^{\:S_\bla}$
\vv.07>
with the subspace of ${\,\:\C[\:\GG\:]^{\:S_\bla}\?\ox\C[\zz]\,}$ consisting
of polynomials not depending on $\,\zz\,$. Then the equivariant cohomology
algebra $\,\Hd^*_T(\CP^{\>n-1}\<;\C)\,$, see \eqref{Hrel}\:, can be presented
as follows
\vvn-.2>
\beq
\label{Hcp}
\Hd^*_T(\CP^{\>n-1}\<;\C)\,=\,\C[\:\gm\:]\ox\C[\zz]\>\Big/\Bigl\bra
\,\prod_{a=1}^n\,(\gm-z_a)\>=\>0\:\Bigr\ket\,.\kern-2em
\vv.2>
\eeq
Similarly, for any $\,\zz^0\?\in\C^n\:$, the algebra
$\,\Hd^*_T(\CP^{\>n-1}\<;\C)_{\:\zz^0}\,$ can be presented as follows
\vvn.3>
\beq
\label{Hrelzcp}
\Hd^*_T(\CP^{\>n-1}\<;\C)_{\:\zz^0}\,=\,\C[\:\gm\:]\>\Big/\Bigl\bra
\,\prod_{a=1}^n\,(\gm-z^0_a)\>=\>0\:\Bigr\ket\,.\kern-2em
\eeq

\vsk.3>
For any polynomial $\,f(\gm\:;\zz)\,$, the polynomial $\,\Ec\bra f\ket(\zz)\,$,
see \eqref{Ecf}\:, can be evaluated as an integral
\vvn-.3>
\beq
\label{Ecfi}
\Ec\bra f\ket(\zz)\,=\,\frac1{2\:\piit}\,\int\limits_{\!\!C\,}
\frac{f(\gm\:;\zz)}{(\gm-z_1)\dots(\gm-z_n)}\;d\:\gm\kern-1em
\vv.1>
\eeq
over a contour $\,\:C\>$ encircling $\,\zzz\>$ counterclockwise.
Similarly, for any polynomial $\,f(\gm)\,$ and $\,\zz^0\?\in\C^n\:$,
\vvn-.3>
\beq
\label{Ecz0i}
\Ec_{\zz^0}\bra f\ket\,=\,\frac1{2\:\piit}\,\int\limits_{\!\!C\,}
\frac{f(\gm)}{(\gm-z^0_1)\dots(\gm-z^0_n)}\;d\:\gm\kern-1em
\eeq
for a contour $\,\:C\>$ encircling $\,z^0_1\lc z^0_n\>$ counterclockwise.

\vsk.4>
Following Section \ref{exampleo}, denote by $\,[\:a\:]$ the element
\vv.07>$\bigl(\{a\},\{\:1\lc a-1,a+1\lc n\:\}\bigr)\<\in\Il$\,. The polynomials
$\,\:\St_{\:[a]}(\gm\:;\zz)\,$, $\Stop_{\:[a]}(\gm\:;\zz)\,$ are respectively
obtained from the functions $\>\WWo_{[a]}(t\:;\zz)\,$,
$\,\Wo_{[a]}(t\:;\zz)\,$, see \eqref{Wcao}\:, \eqref{Wao}\:,
by the substitution $\,\:t=\gm\,$,
\beq
\label{Stacp}
\St_{\:[a]}(\gm\:;\zz)\,=\prod_{c=a+1}^n\?(\gm-z_b)\,,\kern1.3em
\Stop_{\:[a]}(\gm\:;\zz)\,=\,\prod_{c=1}^{a-1}\,(\gm-z_b)\,,\kern1.7em
a=1\lc n\,.\kern-1em
\vv.3>
\eeq
The biorthogonality relation, see Lemma \ref{Stort},
\vvn.3>
\beq
\label{ESSopi}
\Ec\bra\:\St_{\:[a]}(\gm\:;\zz)\,\Stop_{\:[b]}(\gm\:;\zz)\ket\,=\,\dl_{\ab}
\kern-1em
\vv-.1>
\eeq
is clear from formula \eqref{Ecfi}\:.

\vsk.2>
Consider a polynomial
\vvn-.4>
\beq
\label{Wb}
\Wb_{\!\<\bla}(t\:;\gm\:;\zz)\,=\,
\frac1{t-\gm}\,
\biggl(\,\prod_{a=1}^n\,(t-z_a)\:-\:\prod_{a=1}^n\,(\gm-z_a)\<\biggr)\>,
\eeq
and its image $\,\bigl[\:\Wb_{\!\<\bla}(t\:;\gm\:;\zz)\:\bigr]\,$
in $\,\C[\:t\:]\ox \Hd^*_T(\CP^{\>n-1};\C)\,$.
Recall the polynomial $\,\Who_{\!\<\bla}(t\:;\GG)\,$, see \eqref{Who}\:,
and its image $\,\bigl[\:\Who_{\!\<\bla}(t\:;\GG)\:\bigr]\,$
in $\,\C[\:t\:]\ox \Hd^*_T(\CP^{\>n-1};\C)\,$.
By relations \eqref{Hrel}\:, \eqref{Hcp}\:,
\vvn.2>
\be
\bigl[\:\Who_{\!\<\bla}(t\:;\GG)\:\bigr]\,=\,
\bigl[\:\Wb_{\!\<\bla}(t\:;\gm\:;\zz)\:\bigr]\,.\kern-2em
\vv.3>
\ee
Thus Theorem \ref{lemWhtG} for $\,\CP^{\>n-1}$ reads as follows.
%We have
\vvn.2>
\beq
\label{WSt}
\bigl[\:\Wb_{\!\<\bla}(t\:;\gm\:;\zz)\:\bigr]\,=\,
\sum_{a=1}^n\,\Wo_{[a]}(t\:;\zz)\,\bigl[\>\St_{\:[a]}(\gm\:;\zz)\bigr]\,.
\kern-2em
\vv.1>
\eeq
In fact, formula \eqref{WSt} is the image
$\,\C[\:t\:]\ox \Hd^*_T(\CP^{\>n-1};\C)\,$ of the equality of polynomials
\vvn.2>
\beq
\label{eqWSt}
\Wb_{\!\<\bla}(t\:;\gm\:;\zz)\,=\,
\sum_{a=1}^n\,\Wo_{[a]}(t\:;\zz)\,\St_{\:[a]}(\gm\:;\zz)\,.\kern-2em
\eeq

\vsk.2>
Recall that for $\,\zz\<\in\C^n\:$ and $\,a=1\lc n$, the classes of the
\vv.06>
polynomials $\,\:\St_{\:[a]}(\gm\:;\zz)\,$, $\,\:\Stop_{\:[a]}(\gm\:;\zz)\,$
in $\,\Hd^*_T(\CP^{\>n-1}\<;\C)_{\:\zz}\,$ are denoted by
$\,\:\St_{\:[a]\:,\,\zz}\>$, $\,\:\Stop_{\:[a]\:,\,\zz}$.
\vvn-.12>
Consider a basis $\,v_{[\:1\:]}\lc v_{[\:n\:]}\:$ of $\,\Ctnl\:$,
where $\,v_{[a]}\<=v_2^{\ox(a-1)}\<\ox v_1\<\ox v_2^{\ox(n-a)}$.
For every $\,\zz\<\in\C^n\:$, the map
\vvn.4>
\beq
\label{stmapz1}
\St_{\:\bla\:,\:\zz}\::\:\Cnnl\to\,\Hd^*_T(\Fla\>;\C)_{\:\zz}\>,\qquad
v_{\:[a]}\>\mapsto\>\St_{\:[a]\:,\:\zz}\>,\quad a=1\lc n\,,\kern-2em
\vv.4>
\eeq
gives an isomorphism of $\,\Cnnl\:$ and $\,\Hd^*_T(\CP^{\>n-1}\<;\C)_{\:\zz}\,$.

\vsk.3>
Set $\,D_1\<=\gm\,$, $\,\:D_2=z_1\<\lsym+z_n\<-\gm\,$.
The operators of quantum multiplication $\,\:[D_1]_{\:\zz}\>{*}_{\pp\:,\:\zz}\:$,
$\,[D_2]_{\:\zz}\>{*}_{\pp\:,\:\zz}\:$ act on
$\,\Hd^*_T(\CP^{\>n-1}\<;\C)_{\:\zz}\,$ as follows
\vvn.3>
\begin{alignat}2
\label{Docp}
&[D_1]_{\:\zz}\>{*}_{\pp\:,\:\zz}\,\St_{\:[1]\:,\,\zz}\>=\,
z_1\>\St_{\:[1]\:,\,\zz}\:+\>\frac{p_2}{p_1}\,\St_{\:[n]\:,\,\zz}\,,
\\[4pt]
\notag
&[D_1]_{\:\zz}\>{*}_{\pp\:,\:\zz}\,\St_{\:[\:a\:]\:,\,\zz}\>=\,
z_a\>\St_{\:[\:a\:]\:,\,\zz}\:+\,\St_{\:[\:a-1\:]\:,\,\zz}\,,
&& a=2\lc n\,,\kern-1.4em
\\
\notag
&[D_2]_{\:\zz}\>{*}_{\pp\:,\:\zz}\,\St_{\:[\:b\:]\:,\,\zz}\>=\>
\Bigl(\<-\>[D_1]_{\:\zz}\>{*}_{\pp\:,\:\zz}\:+\sum_{c=1}^n\,z_c\Bigr)
\St_{\:[\:b\:]\:,\,\zz}\>,\qquad && b=1\lc n\,,\kern-1.4em
\\[-16pt]
\notag
\end{alignat}
see for instance \cite{TV7}\:. As Corollary \ref{D*Xz0} states, the map
\vv.07>
$\,\:\St_{\:\bla\:,\:\zz}\:$ intertwines the dynamical operators
$\,\Xo_1(\zz\:;\pp)\,$, $\>\Xo_2(\zz\:;\pp)\,$, see \eqref{Xo1}\:,
\vv.07>
acting on $\,\Cnnl\:$ and the operators of quantum multiplication,
$\,\:[D_1]_{\:\zz}\>{*}_{\pp\:,\:\zz}\:$,
$\,[D_2]_{\:\zz}\>{*}_{\pp\:,\:\zz}\:$,
\vvn.4>
\beq
\label{StXDz1}
\St_{\:\bla\:,\:\zz}\:\circ\Xo_i(\zz;\pp)\,=\,
{[D_i]_{\:\zz}\>{*}_{\pp\:,\:\zz}}\:\circ\>\St_{\:\bla\:,\:\zz}\,,
\qquad i=1,2\,.\kern-2em
\eeq

\vsk.3>
Consider the trivial vector bundle $\,H_\bla\?\to\:\C^n\!\times\C^{\:2}$
\vv.07>
with fiber over a point $\,(\zz^0\?,\pp^{\:0})\,$ given by
$\,\Hd^*_T(\CP^{\>n-1}\<;\C)_{\:\zz^0}\,$. The equivariant quantum differential
equations for sections of $\,H_\bla\:$ is the system of differential equations
\vvn.2>
\beq
\label{qDEQ1}
\ka\>p_i\:\frac{\der f}{\der\:p_i}\,=\,
{[D_i]_{\:\zz}\>{*}_{\pp\:,\:\zz}}\>f\,,\qquad i=1,2\>,\kern-2em
\vv.2>
\eeq
where $\,\ka\,$ is the parameter of the equations. The isomorphisms
\vv.06>
$\,\:\St_{\:\bla\:,\:\zz}\:$ identify equations \eqref{qDEQ1} with the limiting
\vv.06>
dynamical differential equation \eqref{DEQo1}\:. Furthermore, the isomorphisms
$\,\:\St_{\:\bla\:,\:\zz}\:$ and the limiting \qKZ/ equations \eqref{Kio1}
define the \qKZ/ difference equations in cohomology
\vvn.2>
\beq
\label{Kich1}
f(z_1\lc z_a\<+\ka\lc z_n;\pp\:;\ka)\,=\,
\KH_a(\zz\:;\pp\:;\ka)\,f(\zz\:;\pp\:;\ka)\,,\qquad a=1\lc n\,,\kern-2em
\vv.3>
\eeq
where for each $\,a=1\lc n\,$, and fixed $\,\zz,\pp\,$,
the operator $\,\KH_a(\zz\:;\pp\:;\ka)\,$ is a map of fibers,
\vvn.5>
\begin{gather*}
\KH_a(\zz\:;\pp\:;\ka)\>:\>\Hd^*_T(\Fla\>;\C)_{\:\zz}\>\to\,
\Hd^*_T(\Fla\>;\C)_{\:(z_1,\:\ldots\:,\>z_a+\:\ka\:,\:\ldots\:,\>z_n)}\,,
\kern-1.1em
\\[6pt]
\KH_a(\zz\:;\pp\:;\ka)\,=\,
\St_{\:\bla,\:(z_1\?,\:\ldots\:,\>z_a+\:\ka\:,\:\ldots\:,\>z_n)}\:
\circ\:\Ko_a(\zz\:;\pp\:;\ka)\:\circ\:(\:\St_{\:\bla,\>\zz}\<)^{-1}\:,
\kern-1em
\\[-14pt]
\end{gather*}
and the operator $\,\Ko_a(\zz\:;\pp\:;\ka)\,$ is given by \eqref{Ko1}\:.
\vv.03>
By Theorem \ref{qcompat}, quantum differential equations \eqref{qDEQ1} and
\qKZ/ difference equations \eqref{Kich1} is a compatible system of differential
and difference equations for sections of $\,H_\bla\>$.

\vsk.2>
Denote $\,\gmd=\gmd_{1,1}\>$. Identify $\,\:\C[\:\gmd^{\pm1}\:]\,$ with
the subspace of $\,\:\C[\:\GGd^{\pm1}]^{\:S_\bla}$ consisting of Laurent
polynomials not depending on $\,\gmd_{\:2,1}\lc\<\gmd_{\:2,\:n-1}\>$, and
$\,\:\C[\:\GGd^{\pm1}]^{\:S_\bla}$ with the subspace of
${\,\:\C[\:\GGd^{\pm1}]^{\:S_\bla}\?\ox\C[\zzd^{\pm1}]\,}$ consisting
\vv.1>
of polynomials not depending on $\,\zzd\,$. Then the equivariant $\>K\?$-theory
algebra $\,K_T(\CP^{\>n-1}\<;\C)\>$, see \eqref{Krelt}\:, can be presented
as in \eqref{Krelo11}\:,
\vvn.2>
\beq
\label{Kcp}
K_T(\CP^{\>n-1}\<;\C)\,=\,\C[\:\gmd^{\pm1}\?,\zzd^{\pm1}]\>\>\Big/\Bigl\bra
\,\prod_{a=1}^n\,(\gmd-\zdd_a)\>=\>0\,\Bigr\ket\,.\kern-1.6em
\eeq

\vsk.2>
Recall that for a Laurent polynomial $\,P(\gmd\:;\zzd)\,$, we denote
\vvn.3>
\be
\Pdd(\gm\:;\zz\:;\ka)\,=\,P(e^{\:2\:\pii\,\gm\</\<\ka}\:;\:
e^{\:2\:\pii\,z_1\</\<\ka}\lc e^{\:2\:\pii\,z_n\</\<\ka}\:)\,.\kern-1em
\ee

\vsk.3>
In Section \ref{exampleo}, we described solutions
$\,\Pso_{\?P}(\zz\:;\pp\:;\ka)\,$ of equations \eqref{DEQo1}\:, \eqref{Kio1}
\vv.06>
labeled by Laurent polynomials $\,P(\:\gmd\:;\zzd)\,$, see \eqref{PsoP1}\:.
\vv.06>
Denote by $\,\Srskl\<$ the space spanned over $\,\C\,$ by the solutions
$\,\:\St_{\:\bla}\Pso_{\?P}(\zz\:;\pp\:;\ka)\,$ of equations \eqref{DEQo1}\:,
\vv.06>
\eqref{Kio1}\:. Each element of $\,\Srskl\<$ is a section of $\,H_\bla\:$
holomorphic in $\,p_1\:,\>p_2\:$ provided branches of $\,\:\log\:p_1\:$ and
\vv.04>
$\,\:\log\:p_2\:$ are fixed, and entire in $\,\zz\,$,
\vv.04>
see Proposition \ref{PsiPsolo}. The space $\,\Srskl\<$ is
\vv-.1>
a $\;\C[\:\zzd^{\pm1}]\:$-\:module, with a Laurent polynomial
$\,P(\zzd)\,$ acting as multiplication by $\Pdd(\zz\:;\ka)\,$.

\vsk.3>
Since the function $\,\Pso_{\?P}(\zz\:;\pp\:;\ka)\,$ depends only
\vv.07>
on the class of $\,P\:$ in $\,K_T(\Fla\>;\C)\,$, then so does the section
$\,\:\St_{\:\bla}\Pso_{\?P}(\zz\:;\pp\:;\ka)\,$ of $\,H_\bla\,$.
Thus the map
\vvn.4>
\beq
\label{mk1}
\mkh\<:\:K_T(\Fla\>;\C)\:\to\,\Srskl\>,\qquad
[P\:]\>\mapsto\>\St_{\:\bla}\Pso_{\?P}\,,
\vv.3>
\eeq
is well-defined, cf.~\eqref{muko}\:. By Corollary \ref{muko=}, the map
$\,\mkh\<$ is an isomorphism of $\;\C[\:\zzd^{\pm1}]\:$-\:modules.

\vsk.2>
There is an integral formula for the function
\vv.06>
$\,\Pso_{\?P}(\zz\:;\pp\:;\ka)\,$, see \eqref{PsoPint}\:.
Taking into account formula \eqref{WSt}\:, we get an integral presentation
for the section $\,\:\St_{\:\bla}\Pso_{\?P}(\zz\:;\pp\:;\ka)\,$,
\vvn.5>
\beq
\label{StPsoPint}
\St_{\:\bla}\Pso_{\?P}(\zz\:;\pp\:;\ka)\,=\,
\frac1{2\:\piit\,\ka}\;\int\limits_{\!\!C\,}\<\Pdd(t\:;\zz\:;\ka)\,
\Pho_{\<\bla}(t\:;\zz\:;\pp\:;\ka)\,
\bigl[\:\Wb_{\!\<\bla}(t\:;\gm\:;\zz)\:\bigr]_\zz\,\:d\:t\,,\kern-1.4em
\vv-.6>
\eeq
where
\vvn-.1>
\beq
\label{Phbb}
\Pho_{\<\bla}(t\:;\zz\:;\pp\:;\ka)\,=\,
(p_2/\ka)^{\>\sum_{a=1}^n\<z_a\</\?\ka}\,
(\ka^{\:n}p_1/p_2)^{\:t/\?\ka}\,\prod_{a=1}^n\,\Gm\bigl(\<(t-z_a)/\ka\bigr)
\kern-2em
\vv.3>
\eeq
is the master function, see \eqref{Pho}\:, the polynomial
$\,\Wb_{\!\<\bla}(t\:;\gm\:;\zz)\,$ is given by \eqref{Wb}\:,
\vv.07>
a contour $\,C\>$ encircles the poles of the product
$\,\prod_{a=1}^n\<\Gm\bigl({(t-z_a)/\<\ka}\bigr)\,$ counterclockwise.
\vv.03>
For instance, $\,C\>$ can be the parabola
\be
C\,=\,\{\,\ka\>\bigl(\:A-s^2\?+s\>\sqrt{\<-1}\,\bigr)\ \,|\ \,s\in\R\,\}\,,
\kern-.6em
\vv.4>
\ee
where $\,A\,$ is a sufficiently large positive real number.

\vsk.3>
In Section \ref{s6.6}, we introduced the Levelt fundamental solution
\vv.04>
$\,\:\St_\bla\Psho$ of quantum differential equations. For the example of
\vv.04>
equations \eqref{qDEQ1}\:, $\,\:\St_\bla\Psho$ is the section of
the vector bundle $\,E\!H_\bla\?\to\:\C^n\!\times\C^{\:2}$ with
fiber over a point $\,(\zz^0\?,\pp^{\:0})\,$ given by
$\,\:\End\:\bigl(\Hd^*_T(\CP^{\>n-1};\C)_{\:\zz^0}\<\bigr)\,$.

By \eqref{StPsho}\:, the values of $\,\:\St_\bla\Psho$ have the form
\vvn.3>
\beq
\label{StPsho1}
\St_{\:\bla}\Psho\<(\zz\:;\pp\:;\ka)\,=\,
\bigl(\:\St_{\:\bla}\Pspo\<(\zz\:;\pp\:;\ka)\<\bigr)\>
(\:p_1/p_2)^{\:[\gm]_\zz^{}/\<\ka}\>p_2^{\>\sum_{c=1}^n\<z_c\</\<\ka}\:,
\kern-1.6em
\vv.3>
\eeq
where
$\,(\:p_1/p_2)^{\:[\gm]_\zz^{}/\<\ka}\>p_2^{\>\sum_{c=1}^n\<z_c\</\<\ka}$
\vv.07>
acts on the fiber $\,\Hd^*_T(\CP^{\>n-1};\C)_{\:\zz}\>$ as multiplication
\vv.07>
by itself and the section $\,\:\St_{\:\bla}\Pspo\<(\zz\:;\pp\:;\ka)\,$ of
\vv.05>
$\,E\!H_\bla\:$ is an entire function of $\,p_2/p_1\>$ equal to the identity
map at $\,p_2\<=0\,$. Formulae \eqref{JchIm}\>--\>\eqref{SPspo}\:,
\eqref{Ecz0i}\:, \eqref{WSt} yield an integral expression for
$\,\:\St_{\:\bla}\Pspo\<(\zz\:;\pp\:;\ka)\,$, see formulae \eqref{Jchm1}\:,
\eqref{SPspo1} below.

\vsk.3>
Recall the complement $\,\Lo\!\<\subset\<\C^n\:$ of the union of hyperplanes
\eqref{zzZh}\:,
\vvn.3>
\be
z_a\<-z_b\<\in\<\ka\>\Z_{\ne0}\,,\qquad a,b=1\lc n\,,\quad a\ne b\,.\kern-.6em
\vv.2>
\ee
For $\,l\<\in\Z_{\ge 0}\,$, let $\,\Jcht_{\!l}$ be the section of the bundle
$E\!H_\bla$ with values
\vvn.4>
\begin{gather}
\label{Jchm1}
\Jcht_{\!l}\?(\zz\:;\ka)\::\:
\Hd^*_T(\CP^{\>n-1};\C)_{\:\zz}\>\to\>\Hd^*_T(\CP^{\>n-1};\C)_{\:\zz}\,,
\kern-1.2em
\\[7pt]
\notag
\Jcht_{\!l}\?(\zz\:;\ka)\::\:[\:f\>]_{\:\zz}^{}\>\mapsto\,
%{\ka^{\>ln}(-1)^n\<}
\frac1{2\:\piit\<}\;\int\limits_{\!C_\zz}\<
f(t)\>\bigl[\:\Wb_{\!\<\bla}(t-l\ka\:;\gm\:;\zz)\:\bigr]_\zz
\,\prod_{a=1}^n\,\prod_{m=0}^l\,\frac1{t-z_a\<-m\ka}\;d\:t\,,\kern-1em
\\[-15pt]
\notag
\end{gather}
where the integral is over a contour $\,C_\zz\>$ encircling the points
$\,\zzz\>$ counterclockwise and separating them from the sets
$z_a\<+\ka\>\Z_{>0}\>$, $\,a=1\lc n\,$. It is assumed in \eqref{Jchm1} that
$\zz\in\Lo$. By Proposition \ref{lemSPspo},
\vvn-.64>
\beq
\label{SPspo1}
\St_{\:\bla}\Pspo\<(\zz\:;\pp\:;\ka)\,=\,\sum_{l=0}^\infty\,
\Jcht_{\!l}(\zz\:;\ka)\,(\:p_2/p_1\<)^{\:l}\>.\kern-2em
\vv-.1>
\eeq
Notice that $\,\Jcht_{\!0}\<(\zz\:;\ka)\,$ is the identity map
\vv.07>
because the classes in $\,\Hd^*_T(\CP^{\>n-1};\C)_{\:\zz}\,$ of a polynomial
$\,f\,$ and the polynomial
\vvn.1>
\beq
\label{Wbf}
f_*(\gm)\,=\,\frac1{2\:\piit}\,\int\limits_{\!C_\zz}\<
\frac{f(\:t)\>\Wb_{\!\<\bla}(t\:;\gm\:;\zz)}{(\:t-z_1)\dots(\:t-z_n)}\;d\:t
\kern-1em
\eeq
coincide, $\,\:[\:f_*\:]_{\:\zz}^{}\?=
[\:f\>]_{\:\zz}^{}\?\in\Hd^*_T(\CP^{\>n-1}\<;\C)_{\:\zz}\,$.

\vsk.3>
Formulae \eqref{Jchm1}\:, \eqref{SPspo1} are counterparts of formulae
\vv.06>
\eqref{Pspo1} for dynamical differential equations \eqref{DEQo1}\:.

\vsk.2>
Recall the space $\,\Oc_{H_\bla}\?$ of sections of $\,H_\bla\:$ holomorphic
\vv.07>
in $\,\zz\,$ for $\,\zz\in\Lo$ and not depending on $\,\pp\,$, and the space
$\>\Srolh\:$ of sections of \,$H_\bla\:$ that are solutions of quantum
\vv-.12>
differential equations \eqref{qDEQ1} holomorphic in $\,\zz\,$ for
$\,\zz\in\Lo$. By Proposition \ref{StPshof}, the Levelt solution
$\,\:\St_{\:\bla}\Psho$ defines an isomorphism
\vvn-.2>
\beq
\label{muoh1}
\muoh\?:\:\Oc_{H_\bla}\<\to\>\Srolh\:,\qquad
f\>\mapsto\,\St_{\:\bla}\Psho\?f\,.\kern-2em
\vv.2>
\eeq
of $\,\Oc_{H_\bla}\?$ and $\>\Srolh\:$. The section $\,f\,$ is called
\vv.06>
the principal term of $\,\:\St_{\:\bla}\Psho\?f\,$. The principal term of
the solution $\,\:\St_{\:\bla}\Pso_{\?P}(\:\zz\:;\pp\:;\ka)\,$ corresponding
to a Laurent polynomial $\,P(\:\gmd\:;\zzd)\,$ is described below,
see \eqref{SPsoP81}\:.

\vsk.2>
For a function $\>f(t)\,$ holomorphic in a neighbourhood of the points
\vv.07>
$\,\zzz\,$, define the polynomial $\>f_*(\gm)\,$ by the rule
\vvn-.2>
\beq
\label{Wbfa}
f_*(\gm)\,=\,\frac1{2\:\piit}\,\int\limits_{\!\!C\!}
\frac{f(\:t)\>\Wb_{\!\<\bla}(t\:;\gm\:;\zz)}
{(\:t-z_1)\dots(\:t-z_n)}\;d\:t\,,\kern-1em
\vv-.1>
\eeq
where a contour $\,\:C\>$ encircles the points $\,\zzz\>$ counterclockwise and
\vv.06>
$\>f(t)\,$ is holomorphic inside $\,C\>$, cf.~\eqref{Wbf}\:. Define the class
$\,\:[\:f\>]_{\:\zz}^{}\?\in\Hd^*_T(\CP^{\>n-1}\<;\C)_{\:\zz}\,$
by the rule $\,\:[\:f\>]_{\:\zz}^{}\?=[\:f_*]_{\:\zz}^{}\>$.

\vsk.3>
Set
\vvn-.7>
\beq
\label{cogog1}
\Cto_\bla(\:\gm\:;\zz\:;\ka)\,=\,
\ka^{\left.\left(\<n\:\gm\:-\sum_{c=1}^nz_c\<\right)\<\right/\<\ka},
\kern1,6em
\Gtp_\bla(\:\gm\:;\zz\:;\ka)\,=\,
\ka^{\:n-1}\,\prod_{a=1}^n\,\Gm\bigl(1+(\gm-z_a)/\ka\bigr)\,,\kern-1em
\vv.1>
\eeq
cf.~\eqref{cogog}\:. By Proposition \ref{SPsiPpro}, the principal term
$\,\:\Pszo_{\?P}\,$ of the solution $\;\St_{\:\bla}\Pso_{\?P}\,$ corresponding
to a Laurent polynomial $\,P(\:\gmd\:;\zzd)\,$ is
\vvn.5>
\beq
\label{SPsoP81}
\Pszo_{\?P}(\zz\:;\ka)\,=\,\bigl[\>\Pdd\?(\:\gm\:;\zz\:;\ka)\,
\Cto_\bla(\:\gm\:;\zz\:;\ka)\,\Gtp_\bla(\:\gm\:;\zz\:;\ka)\bigr]_\zz\,.
\kern-2em
\vv.5>
\eeq
The right\:-hand side of the last formula equals the product of the equivariant
\vv.06>
Chern character $\,\bigl[\>\Pdd\?(\:\gm\:;\zz\:;\ka)\bigr]_\zz\:$ of the class
of the Laurent polynomial $\,P(\:\gmd\:;\zzd)\,$ in $\,K_T(\CP^{\>n-1};\C)\,$,
the exponential $\,\:\bigl[\>\Cto_\bla(\:\gm\:;\zz\:;\ka)\bigr]_\zz\:$ of
\vv-.02>
the equivariant first Chern class $\,\:[\:\gm\:]_{\:\zz}^{}\:$ of the tangent
bundle of $\,\CP^{\>n-1}$, and is the equivariant Gamma\:-\:class
$\,\:\bigl[\>\Gtp_\bla(\:\gm\:;\zz\:;\ka)\bigr]_\zz\:$ of $\,\CP^{\>n-1}$.

\vsk.2>
Recall the map \eqref{Bcyh}\:,
\vvn-.1>
\be
\Bcyh_\bla\::\:K_T(\CP^{\>n-1};\C)\>\to\>\Oc_{H_\bla}\,,
\vv.5>
\ee
that sends the class $\,\:[P\:]\,$ of the Laurent polynomial
$\,P(\:\gmd\:;\zzd)\,$ in $\,K_T(\CP^{\>n-1};\C)\,$ to the section
of $\,H_\bla\:$ with values $\,\:\bigl[\>\Pdd\?(\:\gm\:;\zz\:;\ka)\,
\Cto_\bla(\:\gm\:;\zz\:;\ka)\,\Gtp_\bla(\:\gm\:;\zz\:;\ka)\bigr]_\zz\>$.
\vv.1>
By \eqref{SPsoP81}\:, the map $\,\:\Bcyh_\bla\>$ sends the class $\,[P\:]\,$
to the principal term of the solution $\,\:\St_{\:\bla}\Pso_{\?P}\:$ of
\vv.06>
the joint system of differential equations \eqref{qDEQ1} and difference
equations \eqref{Kich1}. Furthermore, by Theorem \ref{Strianglo},
the following diagram is commutative,
\vvn-.2>
\beq
\label{Scd1}
\xymatrix{{}\phan{\Oc_{H_\bla}}\llap{$K_T(\CP^{\>n-1}\>;\C)$}
\ar^-{\;\Bcyh_\bla\!\!}[rr]
\ar_{\smash{\lower.8ex\llap{$\ssize\mkh\,\:$}}}[dr]&&\,\:\Oc_{H_\bla}
\ar^{\smash{\lower.8ex\rlap{$\;\ssize\muoh$}}}[dl]\\
&\!\Srolh&}
\vv.1>
\eeq
where $\,\mkh\,$ and $\,\muoh\,$ are the maps \eqref{mk1} and \eqref{muoh1}\:,
respectively.

\vsk.3>
The topological\:-\:enumerative morphism for $\,\CP^{\>n-1}$ was studied in
\cite{CV}\:. To refer to formulae and statements in \cite{CV}\:, we will use
the superscript \cv/. To compare formulae in this paper with their
counterparts in \cite{CV}\:, one should make the following substitution,
\vvn.3>
\be
[\:\gm\:]\:=\:x\,,\qquad p_1=\:q^{-1}\:,\qquad p_2=\:1\,,\qquad \ka\:=\:-1\,.
\kern-1em
\vv.4>
\ee
Then the classes of polynomials
$\,\:\St_{\:[1]}(\gm\:;\zz)\lc\St_{\:[n]}(\gm\:;\zz)\,$
\vv.1>
in $\,\Hd^*_T(\CP^{\>n-1}\<;\C)\,$ coincide with the classes $\,g_1\lc g_n\:$
\vv.1>
given by (4.4)\cv/\!, formulae \eqref{StXDz1} for quantum multiplication agree
with formulae (5.13)\cv/\!, (5.14)\cv/\!, and quantum differential equations
\eqref{qDEQ1} and (5.19)\cv/\! match.

\vsk.3>
By the results of \cite{CV}\:, Conjecture \ref{topconj} holds true for
$\,\CP^{\>n-1}$.
\begin{thm}
\label{topcp}
The Levelt fundamental solution $\,\:\St_{\:\bla}\Psho(\zz\:;\pp\:;\ka)\>$ of
equivariant quantum differential equations \eqref{qDEQ1} is the equivariant
topological\:-enumerative morphism of $\,\:\CP^{\>n-1}\<$ restricted to the small
equivariant quantum locus.
\end{thm}
\begin{proof}
The statement follows from formulae \eqref{StPsho1}\>--\>\eqref{SPspo1} and
Theorem 6.4\>\cv/\!.
\end{proof}

According to Definition 5.3\>\cv/\!, the equivariant $\:J$-\:function for
\vv.06>
$\,\CP^{\>n-1}$ is a unique section $\,J(\zz\:;\pp\:;\ka)\,$ of $\,H_\bla\:$
such that for any $\,f\<\in \Hd^*(\CP^{\>n-1}\<;\C)\,$ considered as a section
of $\,H_\bla\:$ not depending on $\,\pp\,$,
\beq
\label{EJcStpzm1}
\Ec_{\:\zz}^{}\bigl\bra\<\bigl(\:\St_{\:\bla}\Psho(\zz\:;\pp\:;\ka)\<\bigr)\>
[\:f\>]_{\:\zz}^{}\bigr\ket\>=\,\Ec_{\:\zz}^{}
\bigl\bra J(\zz\:;\pp\:;\ka)\>[\:f\>]_{\:\zz}^{}\bigr\ket\,.\kern-2em
\vv.5>
\eeq
Formulae \eqref{Ecz0i}\:, \eqref{StPsho1}\>--\>\eqref{SPspo1}\:, and
\vvn.3>
\be
\frac1{2\:\piit}\,\int\limits_{\!\!C\,}\frac{\Wb_{\!\<\bla}(t\:;\gm\:;\zz)}
{(\gm-z_1)\dots(\gm-z_n)}\;d\:\gm\,=\,1\kern-1em
\vv.2>
\ee
for a contour $\,\:C\>$ encircling $\,\zzz\>$ counterclockwise, yield the
\vv.06>
following expression for the equivariant $\:J$-\:function for $\,\CP^{\>n-1}$,
\beq
\label{Jfun}
J(\zz\:;\pp\:;\ka)\,=\,p_2^{\>\sum_{c=1}^n\<z_c\</\<\ka}
\:(\:p_1/p_2)^{\:[\gm]_\zz^{}/\<\ka}\:
\biggl[\;\sum_{l=0}^\infty\,(\:p_2/p_1\<)^{\:l}
\,\:\prod_{a=1}^n\,\prod_{m=1}^l\,\frac1{\gm-z_a\<-m\ka}\,\biggr]_{\<\zz}\:.
\kern-1em
\vv.2>
\eeq
This expression matches the formula for the equivariant $\:J$-\:function for
$\,\CP^{\>n-1}$ obtained in \cite{Gi1}\:, \cite{LLY}\:.

\vsk.3>
All the results of this section admit straightforward specialization to
the nonequivariant case by setting $\,\zz=0\,$.

\appendix
\section{Polynomial identities}
\label{AppA}
Recall $\,\la=(\la_1\lc\la_N)\,$, $\,|\:\bla\:|=n\,$, the set $\,\Il\,$,
$\,\Imil\?\in\Il\,$, and the permutation
\vvn.1>
$\,\si_I\?\in\<S_n\,$, $\,I\<\in\Il\,$, such that $\,\si_I(\Imil)=I\,$.
\vvn.1>
Let $\,S_\bla\<=\Sla\!\subset S_n\>$ be the isotropy subgroup of $\,\Imil\,$.
Each coset in $\,S_n/\<S_\bla\,$ contains exactly one permutation of the form
$\,\si_I\,$, $\,I\<\in\Il\,$.

\vsk.3>
Consider the variables $\,\xx=(\xxx)\,$, $\,\yy=(\yyy)\,$. Define
\vvn.3>
\beq
\label{Vl}
V_\bla(\xx\:;\yy)\,=\>\prod_{a=1}^{\la^{(i)}}\,
\prod_{b=\smash{\la^{(i)}+1}}^n\!\?(x_a\<-y_b)\,.\kern-1.8em
\vv.4>
\eeq
Clearly,
$\,V_\bla(\xx\:;\yy)\in\C[\:\xx\:]^{\:S_\bla}\?\otimes\C[\:\yy\:]^{\:S_\bla}$.
For $\,\si\in\<S_n\,$, denote $\,\xx_\si=(x_{\si(1)}\lc x_{\si(n)})\,$.

\begin{lem}
\label{LA}
We have $\,V_\bla(\xx\:;\xx_\si\<)=0\,$ unless $\,\si\<\in\<S_\bla\,$, and
$\,V_\bla(\xx_{\si_I};\xx_{\si_I}\<)=(-1)^{\:|\:\si_I|}V_\bla(\xx\:;\xx)\,$.
\end{lem}
\begin{proof}
Straightforward.
\end{proof}

For any function $f(\xx)$, define
\vvn-.4>
\be
\Sym^{\:\bla}_{\>\xx}f(\xx)\,=\:\sum_{I\in\Il\!}\,f(\xx_{\si_I})\,.
\vv-.6>
\ee

\begin{prop}
\label{PA1}
We have
\vvn-.2>
\beq
\label{VV}
V_\bla(\xx\:;\zz)\,=\,\Sym^{\:\bla}_{\>\yy}\biggl(\:
\frac{V_\bla(\xx\:;\yy)\,V_\bla(\yy\:;\zz)}{V_\bla(\yy\:;\yy)}\:\biggr)\,.\kern-1em
\vv.3>
\eeq
\end{prop}
\begin{proof}
By formula \eqref{Vl}\:, the expression
$\,V_\bla(\xx\:;\yy)\,V_\bla(\yy\:;\zz)/\:V_\bla(\yy\:;\yy)\,$ is invariant
under permutations of $\,\yyy\,$ by elements of the subgroup $\,S_\bla\,$.
Thus the right-hand side of formula \eqref{VV} is a symmetric polynomial in
$\,\yyy\:$ of degree zero in each of the variables $\,\yyy\,$, hence a constant.
The constant can be found by evaluating this polynomial at $\,\yy=\xx\,$ or
$\,\yy=\zz\,$.
\end{proof}

\begin{prop}
\label{PA2}
For any polynomial $\;f\<\in\C[\xx]^{\:S_\bla}\:$, we have
\vv.5>
\beq
\label{fVVl}
f(\xx)\,=\,\Sym^{\:\bla}_{\>\yy}\biggl(\:
\frac{V_\bla(\xx\:;\yy)}{V_\bla(\yy\:;\yy)}\;\Sym^{\:\bla}_{\>\xx}\biggl(
\frac{f(\xx)\,V_\bla(\yy\:;\xx)}{V_\bla(\xx\:;\xx)}\:\biggr)\?\biggr)
\,.\kern-1.6em
\vv.4>
\eeq
\end{prop}
\begin{proof}
Changing the order of symmetrizations in the right-hand side of
formula \eqref{fVVl} and applying formula \eqref{VV} yields
\vvn.5>
\begin{align*}
& \Sym^{\:\bla}_{\>\yy}\biggl(\:\frac{V_\bla(\xx\:;\yy)}{V_\bla(\yy\:;\yy)}\,
\biggl(\Sym^{\:\bla}_{\>\zz}\biggl(
\frac{f(\zz)\,V_\bla(\yy\:;\zz)}{V_\bla(\zz\:;\zz)}\:\biggr)\?\biggr)
\bigg|_{\:\zz=\xx}\>\biggr)\:={}
\\[12pt]
&\,{}=\:\biggl(\Sym^{\:\bla}_{\>\zz}\biggl(\frac{f(\zz)}{V_\bla(\zz\:;\zz)}\,
\Sym^{\:\bla}_{\>\yy}\biggl(\:
\frac{V_\bla(\xx\:;\yy)\,V_\bla(\yy\:;\zz)}{V_\bla(\yy\:;\yy)}\:\biggr)\?\biggr)
\?\biggr)\bigg|_{\:\zz=\xx}\!\!\!={}
\\[10pt]
&\,{}=\:\biggl(\Sym^{\:\bla}_{\>\zz}\biggl(
\frac{f(\zz)\,V_\bla(\xx\:;\zz)}{V_\bla(\zz\:;\zz)}\:\biggr)\?\biggr)
\bigg|_{\:\zz=\xx}\!\!\!=\>f(\xx)\,.
\\[-28pt]
\end{align*}
\end{proof}

The expression $\;\Sym^{\:\bla}_{\>\xx}\bigl(f(\xx)\,V_\bla(\yy\:;\xx)/\:
V_\bla(\xx\:;\xx)\bigr)\,$
\vvn.12>
is a symmetric polynomial in $\,\xxx\,$. Thus Lemma \ref{LA} and
Proposition \ref{PA2} show that the polynomials $V_\bla(\xx\:;\yy_{\si_I})\,$,
$\,I\<\in\Il\,$, give a basis of the space $\,\C[\xx]^{\:S_\bla}$ as a free
module over the ring of symmetric polynomials $\,\C[\xx]^{\:S_n}$.

\begin{prop}
\label{PA3}
We have
\vvn-.2>
\beq
\label{detV}
\det\:\bigl(V_\bla(\xx_{\si_I};\yy_{\si_J}\<)\bigr)_{\:\IJ\in\:\Il}=
\,\prod_{a=1}^{n-1}\,\prod_{b=a+1}^n\bigl(\<(x_a\<-x_b)\>(y_b\<-y_a)
\bigr)^{d^{(2)}_\bla}\<,\kern-2em
\vv-.1>
\eeq
where
\vvn-.16>
\be
d^{(2)}_\bla\>=\,\frac{(n-2)\:!}{\la_1\:!\ldots\la_N\:!}\,
\sum_{i=1}^{N-1}\sum_{j=i+1}^N\la_i\>\la_j\,.\kern-2em
\vv.2>
\ee
\end{prop}
\begin{proof}
Consider the matrix $\,M(\xx\:;\yy)=
\bigl(V_\bla(\xx_{\si_I};\yy_{\si_J\<})\bigr)_{\:\IJ\in\:\Il}\:$.
\vv.1>
Let $\,F(\xx\:;\yy)=\det\:M(\xx\:;\yy)\,$. Formula \eqref{VV} reads
$\,M(\xx\:;\zz)=M(\xx\:;\yy)\>\bigl(M(\yy\:;\yy)^{-1}\>M(\yy\:;\zz)\,$, thus
\vvn.3>
\beq
F(\xx\:;\zz)\,=\,\frac{F(\xx\:;\yy)\>F(\yy\:;\zz)}{F(\yy\:;\yy)}\;,
\vv.2>
\eeq
so that $\,F(\xx\:;\yy)=G(\xx)\,\Gt(\yy)\,$ for some functions $\,G\>,\Gt\,$.
\vv.07>
Also, $\,F(\xx\:;\yy)=F(\yy\:;\xx)\,$, since the matrix $\,M(\yy\:;\xx)$ is
the transpose of $\,M(\xx\:;\yy)\,$. Hence the functions $\,G\>$ and $\,\Gt\>$
are proportional and one can take $\,\Gt=G\,$.
Finally, the matrix $\,M(\xx\:;\xx)\,$ is diagonal and
\vvn.3>
\be
\bigl(G(\xx)\bigr)^2\:=\,F(\xx\:;\xx)\,=\>
\prod_{I\in\Il\!}\>V_\bla(\xx_{\si_I};\xx_{\si_I})\,=\,\prod_{a=1}^n\,
\prod_{\satop{b=1}{b\ne a}}^n\,(x_a\<-x_b)^{\>d^{(2)}_\bla}\<,\kern-2em
\vv-.8>
\ee
which implies formula \eqref{detV}\:.
\end{proof}

For \,$I\<\in\Il\,$, set
\vvn-.2>
\beq
\label{VIx}
V_I(\xx)\,=\,
V_\bla\bigl(\xxx\:;\si_I(n)-1\lc \si_I(1)-1\bigr)\,.\kern-2em
\vv.3>
\eeq
We use the polynomials $\,V_I(\xx)\,$ in Section \ref{seclaur} to give
a basis of the algebra $\,\Kc_\bla\>$. By formula \eqref{detV}\:,
\vvn-.5>
\beq
\label{detVI}
\det\:\bigl(V_I(\xx_{\si_J}\<)\bigr)_{\:\IJ\in\:\Il}=\,
\prod_{k=2}^{n-1}\,j^{(n-j)\>d^{(2)}_\bla}\;
\prod_{a=1}^{n-1}\,\prod_{b=a+1}^n(x_b\<-x_a)^{\>d^{(2)}_\bla}\<.
\kern-2em
\vvgood
\eeq

\section{Schubert polynomials}
\label{appSch}
Consider the operators $\,\:\Dlx_1\lc\Dlx_{\:n-1}\,$ acting on functions
of $\,\xxx\,$:
\vvn.3>
\be
\Dlx_{\:a}\:f(\xx)\,=\,\frac{f(\xx)-f(x_1\lc x_{a+1},x_a\lc x_n)}{x_a\<-x_{a+1}}\;.
\kern-1.6em
\ee
They satisfy the nil\:-\:Coxeter relations,
\vvn.4>
\beq
\label{nCx}
(\Dlx_{\:a})^2=\:0\,,\kern1.8em
\Dlx_{\:a}\>\Dlx_{\:b}=\:\Dlx_{\:b}\>\Dlx_{\:a}\,,\quad |\:a-b\:|>1\,,
\kern 1.8em
\Dlx_{\:a}\>\Dlx_{\:a+1}\>\Dlx_{\:a}=\:\Dlx_{\:a+1}\>\Dlx_{\:a}\>\Dlx_{\:a+1}\,,
\kern-1.8em
\vv.4>
\eeq
for any $\,a,b\,$.

\vsk.3>
Given a permutation $\,\si\<\in\<S_n\,$, define the operator $\,\:\Dlx_\si\,$
\vv.1>
as follows. For $\,a=1\lc n-1\,$, let $\,s_a\:$ be the transposition of $\,a\,$
\vv.1>
and $\,a+1\,$. If $\;\si=s_{a_1}\ldots\:s_{a_k}$ is the reduced presentation,
then $\,\:\Dlx_\si\<=\Dlx_{\:a_1}\<\ldots\>\Dlx_{\:a_k}$. In particular,
$\,\:\Dlx_{\:s_a}\!\<=\Dl_{\:a}\,$. The operators $\,\:\Dlx_\si\>$ are well-defined
due to relations \eqref{nCx}\:.

\begin{lem}
\label{sitau}
For $\,\si,\tau\<\in\<S_n\>$, we have
$\,\Dlx_\si\>\Dlx_\tau=\:\Dlx_{\:\si\tau}\:$
if $\,|\:\si\>\tau\:|=|\:\si\:|+|\:\tau\:|\,$, and
$\,\Dlx_\si\>\Dlx_\tau=\:0\,$, otherwise.
\end{lem}
\begin{proof}
The statement follows from relations \eqref{nCx}\:.
\end{proof}

Define the operators $\,\:\Dly_\si\,$, $\,\si\<\in\<S_n\>$, acting on functions
of $\,\yyy\,$ similarly.

\vsk.4>
The $A\:$-type double Schubert polynomials $\,\:\Sg_\si\>$, $\,\si\<\in\<S_n\>$,
see \cite[\:Chapter~2\:]{L}\:,
% where $\,\:\Sg_\si\,$ are denoted as $\>Y_\si\>$
\vv.07>
are defined as follows. For the longest permutation $\;\si_0\,$,
$\,\si_0(i)=n+1-i\,$, $\;i-1\lc n\,$, set
\vvn.4>
\be
\Sg_{\si_0}\<(\xx\:;\yy)\,=\,
\prod_{i=1}^{n-1}\,\prod_{j=1}^{n-i}\,(x_i\<-y_j)\,,\kern-.4em
\vv.2>
\ee
For any $\,\:\si\<\in\<S_n\>$, set
$\,\:\Sg_\si(\xx\:;\yy)=\Dlx_{\:\si^{-1}\si_0}\Sg_{\si_0}(\xx\:;\yy)\,$.

\begin{lem}
\label{Dlxy}
For any $\,\:\si\<\in\<S_n\>$, we have $\;\Dlx_\si\>\Sg_{\si_0}(\xx\:;\yy)=
(-1)^{|\si|}\,\Dly_{\si_0\:\si^{-1}\si_0}\<\Sg_{\si_0}(\xx\:;\yy)\,$.
\end{lem}
\begin{proof}
The statement follow by induction on the length of $\,\si\,$
from the equalities
\vvn.4>
\be
\Dlx_{\:s_a}\<\Sg_{\si_0}(\xx\:;\yy)\,=\,
\frac{\Sg_{\si_0}(\xx\:;\yy)}{x_a\<-y_{n-a}}\,=\>
-\>\Dly_{s_{n-a-1}}\<\Sg_{\si_0}(\xx\:;\yy)\kern-.4em
\vv.3>
\ee
and $\,s_{n-a-1}=\:\si_0\,s_a\:\si_0\,$ for every $\,a=1\lc n-1\,$.
\end{proof}

\begin{prop}
\label{Sgxyx}
For any $\,\:\si\<\in\<S_n\>$, we have
\vvn.3>
\be
\Sg_\si(\xx\:;\yy)\,=\,
(-1)^{|\si_0|-|\si|}\>\Dly_{\si\si_0}\Sg_{\si_0}(\xx\:;\yy)\,=\,
(-1)^{|\si|}\>\Sg_{\si^{-1}}(\yy\:;\xx)\,.
\ee
\end{prop}
\begin{proof}
The statement follows from the equality
$\,\:\Sg_{\si_0}(\yy\:;\xx)= (-1)^{|\si_0|}\>\Sg_{\si_0}(\xx\:;\yy)\,$
and Lemma \ref{Dlxy}.
\end{proof}

Denote by $\,\si_\bla\:$ the longest element of the subgroup
$\,S_\bla\<=\Sla\!\subset S_n\>$,
\vvn.3>
\be
\si_\bla(a)=\la^{(i-1)}\?+\la^{(i)}\?+1-a\,,\qquad
\la^{(i-1)}\?+1<a\le\la^{(i)}\:,\quad i=1\lc N\>.\kern-.8em
\vv.3>
\ee

\begin{lem}
\label{sibla}
We have
\vvn-.8>
\be
\Dlx_{\:\si_\bla}\<\Sg_{\si_0}(\xx\:;\yy)\,=\,
\prod_{a=1}^{\la^{(i)}}\,
\prod_{b=\smash{\la^{(i)}+1}}^n\!\?(x_a\<-y_{n-b+1})\,.\kern-1.8em
\vv-.4>
\ee
\end{lem}
\begin{proof}
Straightforward.
\end{proof}

Recall the polynomials $\,\:\St_{\:I}(\:\GG\:;\zz)\,$,
\vv.1>
$\Stop_{\:I}(\:\GG\:;\zz)\,$, see Section \ref{sec:Stab}.

\begin{prop}
\label{lemStSch}
For any $\,I\<\in\Il\>$,
\vvn.3>
\beq
\label{StSch}
\St_{\:I}(\:\GG\:;\zz)\,=\,\Sg_{\si_{\?\si_{\<0}\<(I)}}(\:\GG\:;\zz_{\si_0})\,,
\kern1.6em
\Stop_{\:I}(\:\GG\:;\zz)\,=\,\Sg_{\si_I}(\:\GG\:;\zz)\,.
\kern-2em
\vv.1>
\eeq
\end{prop}
\begin{proof}
By formula \eqref{StIma} and Lemmas \ref{DlW},
\vvn.3>
\be
\Stop_{\:I}(\:\GG\:;\zz)\,=\,(-1)^{|\:\si_{\<\si_0(I)}\<|-|\si_I|}\>
\Dlz_{\:\si_I\:\si_{\?\si_0(I)}^{\smash{-1}}}\!V_\bla(\:\GG\:;\zz_{\si_0}\<)\,,
\kern-1em
\ee
where the operator $\,\:\Dlz_\si\,$ acts on functions of $\,\zz\,$.
By formula \eqref{VGz} and Lemmas \ref{sibla}, \ref{Dlxy},
\vvn.3>
\be
V_\bla(\:\GG\:;\zz_{\si_0}\<)\,=\,(-1)^{|\si_\bla|}\>
\Dlz_{\:\si_0\>\si_\bla\si_0}\Sg_{\si_0}(\xx\:;\yy)
\vv.2>
\ee
because $\,\:\si_\bla^{\smash{-1}}\<=\si_\bla\,$.
Since $\,\:\si_{\?\si_0(I)}=\:\si_0\>\si_\bla\,$ and
$\,\:|\:\si_I\>\si_{\?\si_0(I)}^{\smash{-1}}|+
|\:\si_{\si_0(I)}\>\si_0\:|=|\:\si_I\>\si_0\:|\,$,
\vvn.2>
\be
\Stop_{\:I}(\:\GG\:;\zz)\,=\,(-1)^{|\si_0|-|\si|}\>
\Dlz_{\si\si_0}\Sg_{\si_0}(\:\GG\:;\zz)\,=\,\Sg_{\si_I}(\:\GG\:;\zz)\,.
\vv.3>
\ee
by Lemma \ref{sitau} and Proposition \ref{Sgxyx}, that proves the second
\vv.1>
equality in \eqref{StSch}\:. The first equality in \eqref{StSch} follows
from the relation
$\,\:\St_{\:I}(\:\GG\:;\zz)=\Stop_{\:\si_0(I)}(\:\GG\:;\zz_{\si_0}\<)\,$.
\end{proof}

\section{Proofs of Lemmas \ref{Womax} and \ref{WImis}}
\label{AppB}

\begin{proof}[Proof of Lemma \ref{Womax}]
\strut\;The proof is by induction on $\,N\:$.
To show that $\;\Wo_{\Imil}(\TT\:;\zz)=1\,$,
\,set \>$f_1\<=1\,$, and for \,$N\ge 2\,$, write
\vvn.3>
\be
f_N(\TT^{(1)}\<\lc\TT^{(N-1)};\zz)\,=\,\Wo_{\Imil}(\TT\:;\zz)\,,
\vv.4>
\ee
indicating the dependence on $\,N\:$ explicitly.
\vvn.1>
By formula \eqref{WoI}, the function $\>f_N$ is a polynomial in
\vvn.1>
$\,t_1^{(N-1)}\<\lc t_{\la^{(N-1)}}^{(N-1)}\:$ of degree zero in each of
\vvgood
these variables, hence a constant. Evaluating
$\>f_N(\TT^{(1)}\<\lc\TT^{(N-1)};\zz)\,$ at
$\,\TT^{(N-1)}\?=(z_1\lc z_{\la^{(N-1)}}\<)\,$ yields
\vvn.3>
\be
f_N(\TT^{(1)}\<\lc\TT^{(N-1)};\zz)\,=\,
f_{N-1}(\TT^{(1)}\<\lc\TT^{(N-2)};z_1\lc z_{\la^{(N-1)}}\<)\,.
\vv.3>
\ee
Thus $\>f_N(\TT^{(1)}\<\lc\TT^{(N-1)};\zz)\:=\:1\;$ by the induction
assumption.

\vsk.3>
The proof of formula \eqref{WoIma} is similar.
Set \>$f_1(\zz)=1\,$, and for \,$N\ge 2\,$, write
\vvn.3>
\be
f_N(\TT^{(1)}\<\lc\TT^{(N-1)};\zz)\,=\,\Wo_{\si_0(\Imil)}(\TT\:;\zz)\,,
\vv.3>
\ee
Define polynomials \,$g_1\lc g_{N-1}\,$,
\vvn.3>
\be
g_i(\TT^{(i)};\zz)\,=\,\prod_{a=1}^{\la^{(i)}}\,
\prod_{b=\la^{(i)}+1}^{\la^{(i+1)}}\!(t_a^{(i)}\?-z_{n-b+1})\,.
\vv.2>
\ee
By formula \eqref{WoI}, the function $\>f_N$ is a polynomial in
$\,t_1^{(N-1)}\<\lc t_{\la^{(N-1)}}^{(N-1)}\,$ of degree $\,\la_N\,$
in each of these variables, and $\>f_N$ is divisible by
$\,g_{N-1}(\TT^{(N-1)};\zz)\,$. Hence,
\vvn.4>
\be
f_N(\TT^{(1)}\<\lc\TT^{(N-1)};\zz)\,=\,
r_N(\TT^{(1)}\<\lc\TT^{(N-2)};\zz)\,g_{N-1}(\TT^{(N-1)};\zz)
\vv.3>
\ee
for some polynomial $\,r_N\,$. Evaluating both sides at
$\,\TT^{(N-1)}\?=(z_{\la_N+1}\>\lc z_n)\,$
\vvn.4>
\be
r_N(\TT^{(1)}\<\lc\TT^{(N-2)};\zz)\,=\,
f_{N-1}(\TT^{(1)}\<\lc\TT^{(N-2)};z_{\la_N+1}\>\lc z_n)\,.
\vv.2>
\ee
Thus $\>f_N(\TT^{(1)}\<\lc\TT^{(N-1)};\zz)\:=\:
g_1(\TT^{(1)};\zz)\ldots g_{N-1}(\TT^{(N-1)};\zz)\,$,
which proves formula \eqref{WoIma}\:.
\end{proof}

\begin{proof}[Proof of Lemma \ref{WImis}]
\strut\;The proof is by induction on \,$N$, similarly to that of
Lemma \ref{Womax}. Set \>$f_1\<=1\,$, and for \,$N\ge 2\,$, write
\vvn.3>
\be
f_N(\TT^{(1)}\<\lc\TT^{(N-1)};\zz)\,=\,\Wo_{\Imil}(\TT\:;\zz)\,,\qquad
f^{(i)}_N(\TT^{(1)}\<\lc\TT^{(N-1)};\zz)\,=\,
\Wo_{\?s_{\la^{(i)}\!,\>\la^{(i)}\<+1}(\Imil)}(\TT\:;\zz)\,,
\vv.2>
\ee
indicating the dependence on $\,N\:$ explicitly. By Lemma \ref{Womax},
$\>f_N(\TT^{(1)}\<\lc\TT^{(N-1)};\zz)=1\,$.

\vsk.3>
By formula \eqref{WoI}, the function $\>f^{(N-1)}_N$ is a symmetric linear
function in $\,t_1^{(N-1)}\<\lc t_{\la^{(N-1)}}^{(N-1)}\,$ vanishing at
$\,\TT^{(N-1)}\?=(z_1\lc z_{\la^{(N-1)}}\<)\,$ Hence,
\be
f^{(N-1)}_N(\TT^{(1)}\<\lc\TT^{(N-1)};\zz)\,=\,
r_N(\TT^{(1)}\<\lc\TT^{(N-2)};\zz)\>
\biggl(\,\sum_{j=1}^{\la^{(N-1)}}\;(\:t^{(N-1)}_j\?-z_j)\:\biggr)
\ee
for some polynomial $\,r_N\,$. Evaluating both sides at
$\,\TT^{(N-1)}\?=(z_1\lc z_{\la^{(N-1)}-1}\:,\:0\:)\,$ yields
\vvn.4>
\be
r_N(\TT^{(1)}\<\lc\TT^{(N-2)};\zz)\,=\,
f_{N-1}(\TT^{(1)}\<\lc\TT^{(N-2)};z_1\lc z_{\la^{(N-1)}-1}\:,\:0\:)\,=\,1
\vv.3>
\ee
by Lemma \ref{Womax}, which proves Lemma \ref{Wos} for $\;i=N-1\,$.

\vsk.3>
For $\;i\ne N-1\,$, formula \eqref{WoI} implies that
\vvn.1>
the function $\>f^{(i)}_N$ is a polynomial in
$\,t_1^{(N-1)}\<\lc t_{\la^{(N-1)}}^{(N-1)}\:$ of degree zero in each of
\vvn.1>
these variables, hence a constant. Evaluating
$\>f^{(i)}_N(\TT^{(1)}\<\lc\TT^{(N-1)};\zz)\,$ at
$\,\TT^{(N-1)}\?=(z_1\lc z_{\la^{(N-1)}}\<)\,$ yields
\be
f^{(i)}_N(\TT^{(1)}\<\lc\TT^{(N-1)};\zz)\,=\,
f^{(i)}_{N-1}(\TT^{(1)}\<\lc\TT^{(N-2)};z_1\lc z_{\la^{(N-1)}}\<)\,,=\,
\sum_{j=1}^{\la^{(i)}}\;(\:t^{(i)}_j\?-z_j)\,,
\vv.4>
\ee
where the second equality holds by the induction assumption.
Lemma \ref{Wos} is proved.
\end{proof}

\section{Stirling's formula}
\label{AppC}

The next two lemmas are used in the proof of Lemma \ref{lemresh}.
The first lemma is well known.

\begin{lem}
\label{Stira}
For $\,|\<\arg\:(h/\ka)\:|<\pi$ and any $\,\al\<\in\C\,$, we have
\vvn.3>
\be
\lim_{h\to\infty}\>\frac{\Gm(\al+h/\ka)\>}{(h/\ka)^{\>\al}\,\Gm(h/\ka)}\,=\,1
\kern-2em
\vv-.2>
\ee
locally uniformly in $\,\al$.
\end{lem}

\begin{lem}
\label{Stirb}
For any $\,\al\:,r\?\in\C\,$, $l\in\Z\,$,
we have
\vvn.3>
\be
\lim_{h\to\infty}\>\frac{r^{\:l}\,\Gm(\al+l+h/\ka)\>}
{(h/\ka)^{\>l}\,\Gm(\al+h/\ka)\>\Gm(l+1)}\,=\,\frac{r^{\:l}}{\Gm(l+1)}
\kern-2em
\vv.4>
\ee
uniformly in $\,\:l$ and locally uniformly in $\,\al\:,r\:$.
\end{lem}
\begin{proof}
Let $\,y=h/\ka\,$ and $\>f(x,y,\al)=(1-xy^{-1})^{\?-\:\al-y}\:$.
For $\,r\<<|\:y\:|\,$ by Cauchy's formula,
\vvn.4>
\be
\frac{r^{\:l}\,\Gm(\al+l+y)\>}{y^{\>l}\,\Gm(\al+y)\>\Gm(l+1)}\,=\,
\frac{r^{\:l}}{2\>\piit}\,\int\limits_{\!\!|x|=r}\!
\frac{f(x,y,\al)}{x^{\:l+1}}\,d\:x\,.
\vv.2>
\ee
Since $\,\lim\limits_{y\to\infty}\:f(x,y,\al)=e^x\>$ locally uniformly
in $x,\al\,$, and
\vvn-.1>
\be
\frac{r^{\:l}}{\Gm(l+1)}\,=\,\frac{r^{\:l}}{2\>\piit}\,\int\limits_{\!\!|x|=r}
\!\frac{e^x}{x^{\:l+1}}\,d\:x\,,
\vv-.5>
\ee
Lemma \ref{Stirb} follows.
\end{proof}

\section{Polynomiality of solutions}
\label{AppE}

Consider dynamical differential equations \eqref{DEQ} for nonpositive integer
values of $\,h/\ka\,$. Recall the $\,\:\End\:\bigl(\Cnnl\bigr)$-\:valued Levelt
fundamental solution $\,\Psh(\zz\:;h;\qq\:;\ka)\,$, see Theorem \ref{Bthm},
Proposition \ref{lemPsp}, and formula \eqref{Xi0} for the dynamical
Hamiltonians $\,X_1\lc X_n\:$ at $\,\qq=\0\,$.

\vsk.2>
Say that a rational function $\,f(\zz)\,$ is over integers if $\,f(\zz)\,$ is
a ratio of polynomials in $\,\zz\,$ with integer coefficients.

\begin{thm}
\label{thm E1}
Let $\,m\<\in\<\Z_{\ge 0}\:$ and $\,\:h=\<-\:m\ka\,$.
Then the Levelt fundamental solution has the form
\vvn-.5>
\beq
\label{Psipol}
\Psh(\zz\:;\<-\:m\ka\:;\qq\:;\ka)\,=\,P(\zz\</\<\ka\:;m;\qq)\,\:
\prod_{i=1}^{N-1}\prod_{j=i+1}^N (1-q_j/q_i)^{-\:m\:\la_i}\,
\prod_{i=1}^N\,q_i^{\>\smash{X_i(\zz\</\<\ka\:;\:-m\:;\>\0)}}\>,
\kern-2em
\vv.1>
\eeq
where $\,\:P(\zz\:;m;\qq)\:$ is a polynomial in the variables
\vv.04>
$\,\:q_2/q_1\lc q_n/q_{n-1}$ with coefficients being rational functions
of $\,\:\zz$ over integers. The degree of $\,\:P(\zz\:;m;\qq)\:$
in the variable $\,q_{i+1}/q_i\:$ is at most $\,\:m\>(n-i)\>\la^{(i)}$.
The funciton $\,\:P(\zz\:;m;\qq)\:$ is holomorphic in $\,\zz\>$ if
$\,{z_a\<-z_b\not\in\Z_{\:\ne\:0}}\>$ for all $\,\:{1\le a<b\le n}\,$.
\vv.04>
The singularities of $\,\:P(\zz\:;m;\qq)$ at the hyperplanes
$\,z_a\<-z_b\in\Z_{\:\ne\:0}\>$ are simple poles.
For any $\,\:\zz$ such that $\,z_a\<-z_b\not\in\<\ka\>\Z_{\:\ne\:0}\>$ for all
\vv-.06>
$\,\:1\le a<b\le n\,$, the columns of $\,\,\Psh(\zz\:;\<-\:m\ka\:;\qq\:;\ka)\:$
form a basis of the space of $\;\Cnnl\<$-valued solutions of dynamical
differential equations \eqref{DEQ}\:.
\end{thm}
\begin{proof}
By formulae \eqref{Psipnd}\:, \eqref{PspIJ}\:,
%\eqref{Omh}\:,
\vv.04>
the function $\,\:\Psh(\zz\:;\<-\:m\ka\:;\qq\:;\ka)\,$ has the form
\eqref{Psipol} with $\,P(\zz\:;m;\qq)\,$ being a power series in
$\,q_2/q_1\lc q_n/q_{n-1}\,$. The first step of the proof is to show
that the power series terminate so that $\,P(\zz\:;m;\qq)\,$ is a polynomial
in $\,q_2/q_1\lc q_n/q_{n-1}\,$.

\vsk.2>
Denote by $\,P_{\<\IJ}(\zz\:;m;\qq)\,$ the entries of $\,P(\zz\:;m;\qq)\,$
\vv.06>
in the standard basis $\,\{\:v_I\,,\alb\,I\<\in\Il\:\}\,$ of $\,\Cnnl$.
By formula \eqref{PspIJ}\:,
\vvn-.5>
\beq
\label{PIJ}
P_{\<\IJ}(\zz\:;m;\qq)\,=\,
\sum_{\lb\in\Z_{\ge0}^{\:\la\+1}\!\!}\Jc_{\IJ,\,\lb}(\zz\:;-\:m)\,
\prod_{i=1}^{N-1}\,(q_{i+1}/q_i)^{\>\sum_{a=1}^{\la^{(i)}}l_a^{(i)}},\kern-2em
\eeq
where the coefficients $\,\Jc_{\IJ,\,\lb}(\zz\:;-\:m)\,$ are given by formula
\eqref{JIJm} below.

\vsk.3>
Denote $\,\:[\:x,\<j\>]\:=\:x\>(x-1)\ldots(x-j+1)\,$. Recall
\vvn.3>
\be
\Zc_\bla\,=\,\{\>\lb\<\in\Z^{\la\+1}_{\ge 0}\,\:|\ \,
l^{\:(i)}_{\:a}\?\ge l^{\:(i+1)}_{\:a}\<,\ \;i=1\lc N-1\,,
\ \;a=1\lc\la^{(i)}\:\}\,,\kern-2em
\vv.3>
\ee
see \eqref{Zc}\:. For $\,\:\lb\in\Z_{\ge0}^{\:\la\+1}\<$,
\:set $\,F_\lb(\zz\:;m)=0\,$ if $\,\:\lb\<\not\in\<\Zc_\bla\,$, and
\vvn.3>
\begin{align}
\label{Fzm}
F_\lb(\zz\:;m)\,&{}=\>\prod_{i=1}^{N-1}\,\prod_{a=1}^{\la^{(i)}}\,
\biggl(\:\frac{(-1)^{\:l^{\:(i)}_{\:a}\!-\>l^{\:(i+1)}_{\:a}}\>
[\:m\>,l^{\:(i)}_{\:a}\!\<-l^{\:(i+1)}_{\:a}\:]}
{(\:l^{\:(i)}_{\:a}\!\<-l^{\:(i+1)}_{\:a})\:!}\,\times{}
\\[4pt]
\notag
&\>{}\times\>\prod_{\satop{b=1}{b\ne a}}^{\la^{(i)}}\,
\frac{[\:z_b-z_a\?+l^{\:(i)}_{\:a}\!\<-l^{\:(i)}_{\:b}\?,m+1\:]}
{[\:z_b-z_a\:,m+1\:]}\;
\prod_{\satop{c=1}{c\ne a}}^{\la^{(i+1)}}\frac{[\:z_c\<-z_a\:,m+1\:]}
{[\:z_c\<-z_a\?+l^{\:(i)}_{\:a}\!-l^{\:(i+1)}_{\:c}\?,m+1\:]}\,\:\biggr)
\kern-.3em
\\[-17pt]
\notag
\end{align}
if $\,\:\lb\<\in\?\Zc_\bla\,$. Here $\,\la^{(N)}\?=n\,$ and
$\,\:l^{(N)}_a\!=0\,\:$ for all $\,a=1\lc n\,$.

\vsk.2>
Notice that the factor $\,\:[\:m\>,l^{\:(i)}_{\:a}\!\<-l^{\:(i+1)}_{\:a}\:]\,$
\vv.06>
in formula \eqref{Fzm} equals zero if
$\,\:l^{\:(i)}_{\:a}\!\<-l^{\:(i+1)}_{\:a}\?>m\,$. Therefore,
$\,F_\lb(\zz\:;m)=0\,$ unless $l^{\:(i)}_{\:a}\!\<-l^{\:(i+1)}_{\:a}\?\le m\,$
\vv.06>
for all $\,a,i\,$, and in particular, $\,F_\lb(\zz\:;m)=0\,$ unless
$\,l^{\:(i)}_{\:a}\!\le m\>(n-i)\,$ for all $\,a,i\,$.

\vsk.3>
Denote
\vvn-.1>
\be
\Zc_{\bla,\:m}\,=\,\{\>\lb\<\in\Z^{\la\+1}_{\ge 0}\,\:|\ \,
0\le l^{\:(i)}_{\:a}\!\<-l^{\:(i+1)}_{\:a}\?\le m\,,\ \;i=1\lc N-1\,,
\ \;a=1\lc\la^{(i)}\:\}\,,\kern-2em
\vv.3>
\ee
By formulae \eqref{Jc}\:, \eqref{A}\:, \eqref{Fzm}\:,
\beq
\label{JIJm}
\Jc_{\IJ,\,\lb}(\zz\:;-\:m)\,=\,\sum_{K\in\Il}\>
\frac{F_\lb(\zz_{\si_K}\<;m)\,W_I(\Si_K\?-\lb;\zz\:;-\:m)\,
\WW_J(\Si_K;\zz\:;-\:m)}
{R_\bla(\zz_{\si_K}\<)\,Q_\bla(\zz_{\si_K};-\:m)\,
c_\bla(\Si_K\?-\lb\:;-\:m)\,c_\bla(\Si_K;-\:m)}
\vv.2>
\eeq
if $\,\:\lb\<\in\?\Zc_{\bla,\:m}\,$, and $\,\Jc_{\IJ,\,\lb}(\zz\:;-\:m)=0\,$
if $\,\:\lb\<\not\in\<\Zc_{\bla,\:m}\,$. Therefore,
\vvn.3>
\beq
\label{PIJm}
P_{\<\IJ}(\zz\:;m;\qq)\,=\,
\sum_{\lb\in\Zc_{\bla,m}\!\!}\Jc_{\IJ,\,\lb}(\zz\:;-\:m)\,
\prod_{i=1}^{N-1}\,(q_{i+1}/q_i)^{\>\sum_{a=1}^{\la^{(i)}}l_a^{(i)}},\kern-2em
\vv.1>
\eeq
so that $\,P(\zz\:;m;\qq)\,$ is a polynomial in $\,\:q_2/q_1\lc q_n/q_{n-1}$
\vv.06>
of degree at most $\,m\>(n-i)\>\la^{(i)}\:$ in the variable $\,q_{i+1}/q_i\,$,
$\,i=1\lc N\>$. Moreover, by formulae \eqref{Fzm}\:, \eqref{hWI}\:,
\eqref{cla}\:, \eqref{RQ}\:, all the coefficients
$\,\Jc_{\IJ,\,\lb}(\zz\:;-\:m)\,$ are rational functions of $\,\zz\,$
over integers. The regularity properties of $\,P(\zz\:;m;\qq)\,$ follow from
Theorem \ref{Bthm}, item (\:ii\:)\:.

\vsk.2>
The columns of $\,\:\Psh(\zz\:;\<-\:m\ka\:;\qq\:;\ka)\,$ form a basis of
\vv.07>
the space of $\;\Cnnl\<$-valued solutions of dynamical differential equations
\eqref{DEQ} by Theorem \ref{Bthm}, see formula \eqref{detPsh}\:.
\end{proof}

Theorem \ref{thm E1} allows us to describe the monodromy of the system of
dynamical differential equations \eqref{DEQ} if $\,h/\ka\,$ is a nonpositive
integer.

\begin{cor}
\label{cor monodr}
Let $\,\:h\<\in\<\ka\>\Z_{\le 0}\,$. Then the monodromy of the system
of dynamical differ\-en\-tial equations \eqref{DEQ} is abelian.
\vv.07>
The monodromy representation is generated by the matrices
$\;e^{\>2\:\pii\,\smash{X_1(\zz\</\<\ka\:;\:h\</\<\ka\:;\>\0)}}\:\lc
\>e^{\>2\:\pii\,\smash{X_N(\zz\</\<\ka\:;\:h\</\<\ka\:;\>\0)}}\:$.
For any $\,\:\zz$ such that $\,z_a\<-z_b\not\in\<\ka\>\Z\>$
for all $\,\:1\le a<b\le n\,$, the monodromy representation is
the direct sum of one-dimensional representations generated by the solutions
$\,\,\Psi_{\<I}(\zz\:;\<-\:m\ka\:;\qq\:;\ka)\,$, $\,\:I\<\in\Il\,$,
defined by \eqref{mcF}\:.
\end{cor}

Notice that for integer values of $\,h/\ka\,$, the factors
$\,(1-q_j/q_i)^{\la_ih/\<\ka}$ in formula \eqref{Oml} are rational functions
\vv.04>
and the only nontrivial monodromy of the solutions
$\,\:\Psi_{\<J}(\zz\:;\<-\:m\ka\:;\qq\:;\ka)\,$ comes from the factors
$\,q_{i\vp{|^1}}^{\:\smash{\sum_{a\in J_i}\!\<z_a\</\<\ka}}\:$
in formula \eqref{PsiPsd}\:.

\vsk.2>
The results of this appendix were motivated by the corresponding statements in
\cite{OP} on the quantum differential equation of the Hilbert scheme of points
in the plane.

\end{document}